%% file: main.tex
\documentclass[11pt]{article}
\usepackage{a4wide}
\usepackage[utf8]{inputenc}
\usepackage[T1]{fontenc}

\input{packages}
\input{macros}

\begin{document}
\title{\huge{Almost-linear time parameterized algorithm for rankwidth via dynamic rankwidth}}
\author{Tuukka Korhonen\thanks{Department of Informatics, University of Bergen, Norway (\texttt{tuukka.korhonen@uib.no}). Tuukka Korhonen was supported by the Research Council of Norway via the project BWCA (grant no. 314528).} \and
Marek Sokołowski\thanks{Institute of Informatics, University of Warsaw, Poland (\texttt{marek.sokolowski@mimuw.edu.pl}). The work of Marek Soko{\l}owski on this manuscript is a part of a project that has received funding from the European Research Council (ERC), grant agreement No 948057 --- BOBR.}}
\date{}
\maketitle
\thispagestyle{empty}
\pagenumbering{roman}

 \begin{textblock}{20}(-1.9, 8.2)
  \includegraphics[width=40px]{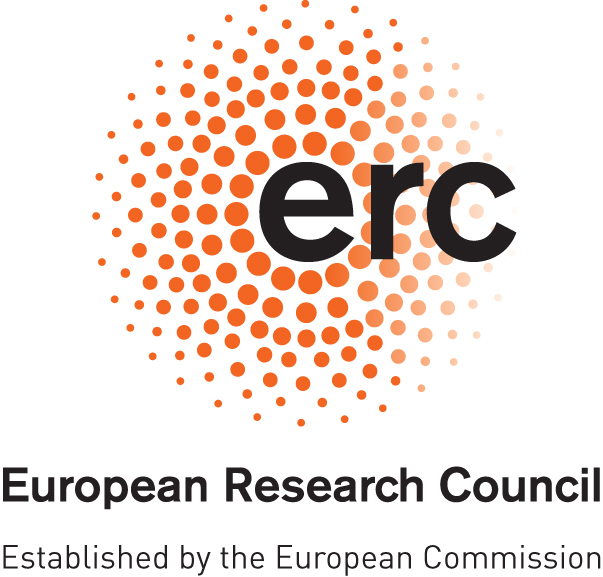}%
 \end{textblock}
 \begin{textblock}{20}(-2.15, 8.5)
  \includegraphics[width=60px]{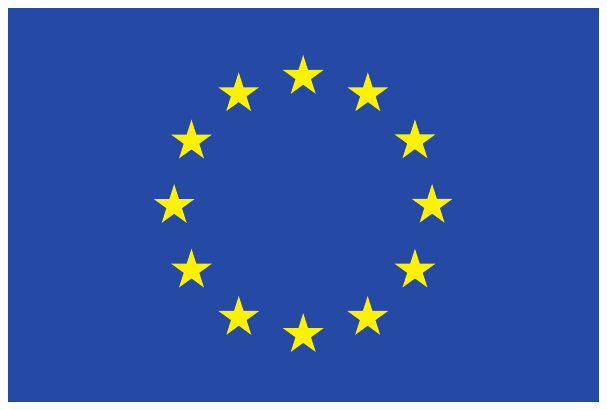}%
 \end{textblock}

\begin{abstract}
\input{abstract.tex}
\end{abstract}

\newpage
\tableofcontents
\newpage
\pagenumbering{arabic}

\input{intro.tex}

\input{overview.tex}

\input{prelim.tex}

\input{datastructure.tex}

\input{refinement.tex}

\input{automata.tex}

\input{together.tex}

\input{comp.tex}

\input{dealternation.tex}

\input{computing-closures.tex}

\input{conclusion.tex}

\paragraph*{Acknowledgements.}
We thank Micha\l{} Pilipczuk for helpful discussions on this project.

\bibliographystyle{alpha}
\bibliography{references}

\appendix

\input{omitted-proofs1.tex}

\input{cw_appendix.tex}

\input{totally-pure.tex}

\end{document}

%% file: packages.tex
\usepackage{inconsolata}
\usepackage{libertine}

\usepackage{amsthm}
\usepackage{amssymb}
\usepackage{amsmath}
\usepackage{mathtools}
\usepackage{mathrsfs}
\usepackage{thmtools,thm-restate}
\usepackage{bm}

\usepackage{tikz}
\usepackage{graphicx}
\graphicspath{ {./images/} }

\usepackage{algorithm}
\usepackage[noend]{algpseudocode}
\usepackage{caption}
\usepackage{subcaption}
\usepackage[noadjust]{cite}
\usepackage{enumitem}
\usepackage{xspace}
\usepackage{xcolor}
\usepackage{xparse}
\usepackage{xfrac}
\usepackage[disable]{todonotes}
\usepackage{textpos}

\usepackage[colorlinks = true,linkcolor = blue,urlcolor = blue, citecolor = blue]{hyperref}

\usepackage[capitalize, nameinlink]{cleveref}

%% file: macros.tex
\newcommand\tk[1]{\todo[inline,size=\scriptsize,backgroundcolor=yellow]{#1 - \textbf{Tuukka}}}

\newcommand\ms[1]{\todo[inline,size=\scriptsize,backgroundcolor=green]{#1 - \textbf{Marek}}}

\newtheorem{theorem}{Theorem}[section]

\newtheorem{lemma}[theorem]{Lemma}
\newtheorem{observation}[theorem]{Observation}
\newtheorem{corollary}[theorem]{Corollary}

\newtheorem{problem}[theorem]{Problem}

\crefname{claim}{Claim}{Claims}
\crefname{problem}{Problem}{Problems}
\newtheorem{claim}[theorem]{Claim}

\newenvironment{claimproof}
 {\proof}
 {\endproof}

\newcommand{\Oh}[2][]{{\cal O}_{#1}(#2)}

\newcommand{\Tpref}{\ensuremath{T_{\mathrm{pref}}}\xspace}
\newcommand{\App}{\ensuremath{\mathsf{App}}}
\newcommand{\height}{\ensuremath{\mathsf{height}}}

\newcommand{\funrestriction}[2]{#1\vert_{#2}}

\newcommand{\C}{\mathbb{C}}
\newcommand{\D}{\mathbb{D}}
\newcommand{\F}{\mathbb{F}}
\renewcommand{\H}{\mathbb{H}}

\newcommand{\N}{\mathbb{N}}

\newcommand{\Z}{\mathbb{Z}}

\newcommand{\Cc}{\mathcal{C}}
\newcommand{\Dc}{\mathcal{D}}

\newcommand{\Kc}{\mathcal{K}}

\newcommand{\Tc}{\mathcal{T}}

\newcommand{\Wc}{\mathcal{W}}
\newcommand{\Vc}{\mathcal{V}}

\newcommand{\BB}{\mathfrak{B}}

\renewcommand{\leq}{\leqslant}
\renewcommand{\geq}{\geqslant}

\renewcommand{\le}{\leqslant}
\renewcommand{\ge}{\geqslant}

\newcommand{\CMSO}{\mathsf{CMSO}}
\newcommand{\LinCMSO}{\mathsf{LinCMSO}}
\newcommand{\MSO}{\mathsf{MSO}}

\newcommand{\oE}{\vec{E}}
\newcommand{\oed}{\vec{e}}
\newcommand{\ouv}{\vec{uv}}
\newcommand{\oxy}{\vec{xy}}
\newcommand{\oyx}{\vec{yx}}
\newcommand{\oyz}{\vec{yz}}
\newcommand{\ozy}{\vec{zy}}

\newcommand{\ovu}{\vec{vu}}
\newcommand{\ozx}{\vec{zx}}
\newcommand{\oab}{\vec{ab}}
\newcommand{\oba}{\vec{ba}}
\newcommand{\obc}{\vec{bc}}
\newcommand{\ocb}{\vec{cb}}
\newcommand{\olp}{\vec{lp}}
\newcommand{\opl}{\vec{pl}}
\newcommand{\oap}{\vec{ap}}

\newcommand{\oxp}{\vec{xp}}
\newcommand{\opx}{\vec{px}}
\newcommand{\ocax}{\vec{c_1 x}}
\newcommand{\ocbx}{\vec{c_2 x}}
\newcommand{\oxca}{\vec{x c_1}}

\newcommand{\oApp}{\ensuremath{\vec{\mathsf{App}}}}

\newcommand{\boldcup}{\bm{\bigcup}}

\newcommand{\reps}{\mathcal{R}}
\newcommand{\repse}{\mathcal{E}}
\newcommand{\dmap}{\mathcal{F}}
\newcommand{\PT}{\mathcal{P}_3}

\newcommand{\leafs}{L}
\newcommand{\leafe}{\vec{L}}
\newcommand{\lmap}{\lambda}
\newcommand{\lparts}{\mathcal{L}}

\newcommand{\prt}{\mathcal{C}}
\newcommand{\prtb}{\mathcal{B}}

\newcommand{\compl}[1]{\overline{#1}}

\newcommand{\GF}{\text{GF}}

\newcommand{\NA}{\ensuremath{\mathsf{Annot}}}

\newcommand{\cut}{\mathsf{cut}}
\newcommand{\aep}{\mathsf{aep}}
\newcommand{\aes}{\mathsf{aes}}

\newcommand{\Tconn}{\ensuremath{T_{\mathrm{conn}}}\xspace}

\newcommand{\width}{\ensuremath{\mathsf{width}}}

\newcommand{\cutrk}{\ensuremath{\mathsf{cutrk}}}
\newcommand{\trivialpartition}[1]{\ensuremath{\mathsf{TrivPart}}(#1)}
\newcommand{\reppart}{\ensuremath{\mathsf{RepPart}}}

\newcommand{\prdesc}{\overline{u}}

\newcommand{\pred}{\ensuremath{\mathsf{pred}}}

\newcommand{\autom}{\mathcal{A}}

\newcommand{\pg}{\mathcal{G}}

\newcommand{\opkg}{\ensuremath{\mathsf{op}}}

\newcommand{\Expr}{\ensuremath{\mathsf{Expr}}}

\newcommand{\eudesc}{\overline{e}}

\newcommand{\sumof}[1]{\langle {#1} \rangle}
\newcommand{\swap}{\mathrm{swap}}

\newcommand{\symd}{\triangle}

\newcommand{\fullset}{{\rm FS}}
\newcommand{\lin}[1]{\mathbf{#1}}

\newcommand{\mix}{\mathsf{M}}
\newcommand{\AutomataReps}{\mathsf{reps}}

\crefname{arde}{ENC}{ENC}
\crefname{ardse}{SENC}{SENC}

%% file: abstract.tex
We give an algorithm that given a graph $G$ with $n$ vertices and $m$ edges and an integer $k$, in time $\Oh[k]{n^{1+o(1)}} + \Oh{m}$ either outputs a rank decomposition of $G$ of width at most $k$ or determines that the rankwidth of $G$ is larger than $k$; the $\Oh[k]{\cdot}$-notation hides factors depending on $k$.
Our algorithm returns also a $(2^{k+1}-1)$-expression for cliquewidth, yielding a $(2^{k+1}-1)$-approximation algorithm for cliquewidth with the same running time.
This improves upon the $\Oh[k]{n^2}$ time algorithm of Fomin and Korhonen~[STOC~2022].

The main ingredient of our algorithm is a fully dynamic algorithm for maintaining rank decompositions of bounded width: We give a data structure that for a dynamic $n$-vertex graph $G$ that is updated by edge insertions and deletions maintains a rank decomposition of $G$ of width at most $4k$ under the promise that the rankwidth of $G$ never grows above $k$.
The amortized running time of each update is $\Oh[k]{2^{\sqrt{\log n} \log \log n}}$.
The data structure furthermore can maintain whether $G$ satisfies some fixed $\CMSO_1$ property within the same running time.
We also give a framework for performing ``dense'' edge updates inside a given set of vertices $X$, where the new edges inside $X$ are described by a given $\CMSO_1$ sentence and vertex labels, in amortized \mbox{$\Oh[k]{|X| \cdot 2^{\sqrt{\log n} \log \log n}}$} time.
Our dynamic algorithm generalizes the dynamic treewidth algorithm of Korhonen, Majewski, Nadara, Pilipczuk, and Soko{\l}owski~[FOCS~2023].

%% file: intro.tex
\section{Introduction}
\label{sec:intro}
Decomposing a graph into a tree-like structure along separators or cuts with simple structure is a popular paradigm in graph algorithms.
While treewidth~\cite{DBLP:journals/jct/RobertsonS84} is the most prominent graph parameter associated with such decompositions, the second most prominent is arguably the \emph{rankwidth}.

A \emph{rank decomposition} of a graph $G$ is a pair $(T,\lmap)$, where $T$ is a tree whose every non-leaf node has degree $3$ and $\lmap$ is a bijection from $V(G)$ to the leaves of $T$.
For an edge $uv$ of $T$, the \emph{width} of $uv$ is defined as follows.
Let $L_u$ be the leaves of $T$ that are closer to $u$ than to $v$ and let $L_v$ be the leaves that are closer to $v$ than to $u$.
Then, the width of $uv$ is the $\GF(2)$-rank of the $|L_u| \times |L_v|$ matrix that describes the adjacencies between $\lmap^{-1}(L_u)$ and $\lmap^{-1}(L_v)$ in $G$ with zeros and ones.
The width of the rank decomposition $(T,\lmap)$ is the maximum width of an edge of it, and the rankwidth of a graph $G$ is the minimum width of a rank decomposition of it.

Rankwidth was introduced by Oum and Seymour~\cite{OumS06} to approximate a graph parameter called \emph{cliquewidth}\footnote{We provide the definition of cliquewidth in \Cref{sec:cliquewidth}. For this introduction it is enough to know that cliquewidth and rankwidth are functionally tied to each other, and the cliquewidth of a graph $G$ is equal to the smallest $k$ for which there exists a decomposition called a ``$k$-expression'' for $G$.}.
They showed that if a graph has rankwidth $k$, then its cliquewidth is between $k$ and $2^{k+1}-1$, and gave a polynomial-time algorithm that constructs a \emph{$(2^{k+1}-1)$-expression} witnessing that the cliquewidth is at most $(2^{k+1}-1)$ when given a rank decomposition of width $k$.

Our main contribution is the following theorem about computing rankwidth exactly and approximating cliquewidth.
We use the $\Oh[k]{\cdot}$-notation to hide factors depending only on $k$.

\begin{theorem}
\label{the:altrw}
There is an algorithm that, given an $n$-vertex $m$-edge graph $G$ and an integer $k$, in time $\Oh[k]{n \cdot 2^{\sqrt{\log n} \log \log n}} + \Oh{m}$, either outputs a rank decomposition of $G$ of width at most $k$ or determines that the rankwidth of $G$ is larger than $k$.
The algorithm also outputs a $(2^{k+1}-1)$-expression for cliquewidth of $G$ within the same running time.
\end{theorem}

\Cref{the:altrw} improves upon the $\Oh[k]{n^2}$ time algorithm of Fomin and Korhonen~\cite{DBLP:conf/stoc/FominK22}, and is a subpolynomial $2^{\sqrt{\log n} \log \log n} = n^{o(1)}$ factor away from concluding the long line of work on computing rankwidth in the setting where $k$ is bounded~\mbox{\cite{OumS06,OumS07,CourcelleO07,Oum08,HlinenyO08,DBLP:journals/siamdm/JeongKO21,DBLP:conf/stoc/FominK22}}.
Moreover, if the average degree of the input graph is higher than $f(k) \cdot 2^{\sqrt{\log n} \log \log n}$ for the function $f(k)$ hidden by the $\Oh[k]{\cdot}$-notation (which is a very natural case when we are interested in rankwidth), then our algorithm works in truly linear $\Oh{m}$ time.
Before further comparing our algorithm to the previous algorithms, let us discuss the motivation for computing rankwidth.

\paragraph{Applications of rankwidth.}
Cliquewidth was introduced by Courcelle, Engelfriet, and Rozenberg~\cite{CourcelleER93} in their study of logic and automata on graphs and was defined in its present form by Courcelle~\cite{DBLP:journals/apal/Courcelle95}.
A closely related parameter \emph{NLC-width} was also investigated by Wanke~\cite{Wanke94}.
Cliquewidth and rankwidth can be regarded as generalizations of treewidth that are suitable for dense graphs: A graph of treewidth $k$ has rankwidth at most $k+1$~\cite{DBLP:journals/jgt/Oum08} and cliquewidth at most $3 \cdot 2^{k-1}$~\cite{DBLP:journals/siamcomp/CorneilR05}, while for example the complete graphs have unbounded treewidth but rankwidth at most $1$.
More generally, $n$-vertex graphs with more than $kn$ edges have treewidth more than $k$, while dense classes of graphs with bounded rankwidth include for example the cographs and the distance-hereditary graphs~\cite{DBLP:journals/ijfcs/GolumbicR00}.

Many graph problems that are NP-hard in general can be solved efficiently on graphs of bounded treewidth using dynamic programming.
The celebrated theorem of Courcelle~\cite{Courcelle90} states that every graph problem expressible in \emph{Counting Monadic Second Order Logic} ($\CMSO_2$) can be solved in $\Oh[k]{n}$ time when a graph is given together with a tree decomposition of width $k$.
Combined with the $\Oh[k]{n}$ time algorithm for computing optimum-width tree decompositions by Bodlaender~\cite{bodlaender-tw-opt}, this yields $\Oh[k]{n}$ time algorithms for many classical NP-hard graph problems on graphs of treewidth~$k$.

The version of Courcelle's theorem for cliquewidth by Courcelle, Makowsky, and Rotics~\cite{CourcelleMR00} states that every graph problem that is expressible in a variant of $\CMSO_2$ called $\CMSO_1$ can be solved in $\Oh[k]{n}$ time when a graph is given together with a $k$-expression for it.
The difference between $\CMSO_1$ and $\CMSO_2$ is that while in $\CMSO_2$ one can quantify over sets of vertices and edges, $\CMSO_1$ allows quantification only over sets of vertices.
$\CMSO_1$ together with its optimization variant called $\LinCMSO_1$~\cite{CourcelleMR00}
captures graph problems such as $k$-colorability for fixed $k$, maximum independent set, maximum clique, minimum dominating set, minimum feedback vertex set, and longest induced path.
It is known that the boundary of tractability for $\CMSO_2$ is characterized precisely by bounded treewidth and the boundary of tractability for $\CMSO_1$ by bounded cliquewidth/rankwidth~\cite{DBLP:journals/apal/Seese91,CourcelleO07,DBLP:conf/lics/KreutzerT10,DBLP:conf/soda/KreutzerT10}.

Our \Cref{the:altrw} combined with the ``Courcelle's theorem for cliquewidth''~\cite{CourcelleMR00} implies the following corollary about algorithms for problems on graphs of bounded rankwidth.
\begin{corollary}
\label{cor:mso}
Every graph problem that can be expressed in $\LinCMSO_1$ logic can be solved in time $\Oh[k]{n \cdot 2^{\sqrt{\log n} \log \log n}}+\Oh{m}$ on graphs of rankwidth $k$.
\end{corollary}

In particular, thanks to our \Cref{the:altrw}, all aforementioned NP-hard graph problems can be solved in time $\Oh[k]{n \cdot 2^{\sqrt{\log n} \log \log n}}+\Oh{m}$ on graphs of rankwidth $k$, improving upon previous $\Oh[k]{n^2}$ time.
Beyond optimization, Courcelle's theorem for cliquewidth has been extended also to counting problems~\cite{DBLP:journals/dam/CourcelleMR01}.
In addition to these meta-theorems, there has been a significant amount of work in designing hand-crafted algorithms working on rank decompositions and cliquewidth expressions for various problems, see e.g.~\cite{DBLP:journals/dam/FischerMR08,DBLP:journals/dam/Bui-XuanTV10,DBLP:journals/dam/GanianH10,DBLP:journals/talg/CoudertDP19,DBLP:journals/siamdm/Lampis20,DBLP:journals/tocl/GroheN23}.

\paragraph{Computing rankwidth and cliquewidth.}
When Oum and Seymour~\cite{OumS06} introduced rank\-width, they also gave an algorithm for $3$-approximating rankwidth in $\Oh{8^k n^9 \log n}$ time, which by the relation between rankwidth and cliquewidth implied a $2^{\Oh{k}}$-approximation for cliquewidth with the same running time.
Their algorithm works in the general setting of branchwidth of connectivity functions.
This setting not only captures rankwidth, but also branchwidth of (hyper)graphs and matroids~\cite{DBLP:journals/jct/RobertsonS91}, and carving width~\cite{DBLP:journals/combinatorica/SeymourT94}.
This was the first $f(k)$-approximation for cliquewidth with running time $f(k) n^{\Oh{1}}$ for any function $f(k)$, and even presently, to the best of our knowledge, all known algorithms for approximating cliquewidth work via rankwidth.
Computing cliquewidth was shown to be NP-complete by Fellows, Rosamond, Rotics, and Szeider~\cite{DBLP:journals/siamdm/FellowsRRS09} and the NP-completeness of rankwidth was observed by Oum~\cite{Oum08}.

Oum gave in~\cite{Oum08} two algorithms improving the running time of the Oum-Seymour $3$-ap\-prox\-imation, first running in $\Oh{8^k n^4}$ time and second running in $\Oh[k]{n^3}$ time.
The latter algorithm combines the approach of~\cite{OumS06} with a~related work of Hlin{\v{e}}n{\'y} that provides a~$3$-approximation of the branchwidth of matroids in $\Oh[k]{n^3}$ time~\cite{Hlineny05}.
An exact algorithm computing rankwidth in time $n^{\Oh{k}}$ was given by Oum and Seymour~\cite{OumS07}, and in time $\Oh[k]{n^3}$ by Courcelle and Oum~\cite{CourcelleO07} by using vertex-minors.
The algorithm of Courcelle and Oum does not provide the corresponding rank decomposition, but this caveat was removed by Hlin\v{e}n{\'y} and Oum~\cite{HlinenyO08} by giving a $\Oh[k]{n^3}$ time exact algorithm that also constructs the rank decomposition.
An alternative constructive $\Oh[k]{n^3}$ time exact algorithm was developed by Jeong, Kim, and Oum~\cite{DBLP:journals/siamdm/JeongKO21}, who gave an algorithm that directly constructs an optimum-width rank decomposition by dynamic programming on rank decompositions, analogous to the algorithm of Bodlaender and Kloks for treewidth~\cite{DBLP:journals/jal/BodlaenderK96}.
Their algorithm extends to the general setting of branchwidth of ``subspace arrangements'' over finite fields.
In 2017 Oum asked whether there exists an $\Oh[k]{n^c}$ time (approximation) algorithm for rankwidth for $c<3$~\cite{Oum17}, which was answered affirmatively by Fomin and Korhonen with an $\Oh[k]{n^2}$ time algorithm~\cite{DBLP:conf/stoc/FominK22}.
Our \Cref{the:altrw} further improves upon their algorithm.
We give an overview of the algorithms for rankwidth in \Cref{table:rw_history}.

\begin{table}[!t]
\begin{center}
\begin{tabular}{|c|c|c|c|}
\hline
Reference & APX & TIME & Remarks \\ \hline
\cite{OumS06} & $3k+1$ & $\Oh{8^k n^9\log{n}} $ & Works for connectivity functions\\
\cite{OumS07} & exact & $\Oh{n^{8k+12} \log n}$ & Works for connectivity functions\\
\cite{Oum08} & $3k+1$ & $\Oh{8^k n^4}$ & \\
\cite{Oum08} & $3k-1$ & $\Oh[k]{n^3}$ & \\
\cite{CourcelleO07} & exact & $ \Oh[k]{n^3} $ & Does not provide a decomposition\\ 
\cite{HlinenyO08} & exact & $ \Oh[k]{n^3} $ & \\
\cite{DBLP:journals/siamdm/JeongKO21} & exact & $\Oh[k]{n^3}$ & Works for spaces over finite fields\\ 
\cite{DBLP:conf/stoc/FominK22} & exact &  $\Oh[k]{n^2}$ & \\
This paper & exact & $\Oh[k]{n \cdot 2^{\sqrt{\log n} \log \log n}} + \Oh{m}$ & \\
\hline
\end{tabular}
\end{center}
\caption{Overview of algorithms for computing rankwidth.
Here $n$ is the number of vertices, $m$ is the number of edges, and $k$ is the rankwidth of the input graph.
Unless otherwise specified, each of the algorithms outputs in $\Oh{\text{TIME}}$ a decomposition of width given in the APX column.
All the functions on $k$ hidden by the $\Oh[k]{\cdot}$-notation here are at least double-exponential. Most of this table is from a similar table in~\cite{DBLP:conf/stoc/FominK22}.}
\label{table:rw_history}
\end{table}

\paragraph{Dynamic rankwidth.}
As both a main ingredient in proving \Cref{the:altrw} and a contribution of independent interest, we give a data structure for efficiently maintaining rank decompositions of dynamic graphs under edge insertions and deletions, under the promise that the rankwidth of the graph never grows above a given parameter $k$.
The data structure can also maintain any finite-state dynamic programming scheme on the rank decomposition.
We formalize this by stating that it can maintain the value of any $\LinCMSO_1$ sentence on the graph.
In particular, we prove the following.

\begin{theorem}
\label{the:dynsimple}
There is a data structure that is initialized with an integer $k$ and an empty $n$-vertex dynamic graph $G$, and maintains a rank decomposition of $G$ of width at most $4k$ under edge insertions and deletions, under the promise that the rankwidth of $G$ never exceeds $k$.
The amortized initialization time is $\Oh[k]{n \log^2 n}$ and the amortized update time is $2^{\Oh[k]{\sqrt{\log n \log \log n}}}$.
Furthermore, when initialized with a $\LinCMSO_1$ sentence $\varphi$ of constant length, the data structure maintains the value of $\varphi$ on $G$.
\end{theorem}

We observe that $2^{\Oh[k]{\sqrt{\log n \log \log n}}} \le \Oh[k]{2^{\sqrt{\log n} \log \log n}}$ by a simple tradeoff trick.
We stated \Cref{the:altrw} in the latter form for simplicity.
We also note that \Cref{the:dynsimple} immediately implies a $4$-approximation algorithm for rankwidth with time complexity $\Oh[k]{(n+m) \cdot 2^{\sqrt{\log n} \log \log n}}$, simply by inserting the edges of the input graph one by one into the data structure an returning the final rank decomposition\footnote{Note that the edges can be inserted in an order so that the graph held by the data structure never has rankwidth more than $k+1$, assuming the input graph has rankwidth at most $k$. This happens, for instance, if the edges are inserted in the lexicographic order: the edge $uv$ is inserted before the edge $u'v'$ if either $\min(u, v) < \min(u', v')$, or $\min(u, v) = \min(u', v')$ and $\max(u, v) < \max(u', v')$.}.
To prove \Cref{the:altrw}, we use a bit more technical version of \Cref{the:dynsimple} where we assert that the rank decomposition held by the data structure is represented as an ``annotated rank decomposition'' (defined in \Cref{sec:annot}).
Even after this the proof of \Cref{the:altrw} requires non-trivial additional work.

A theorem similar to \Cref{the:dynsimple}, but for treewidth and $\CMSO_2$ instead of rankwidth and $\CMSO_1$, was recently given by Korhonen, Majewski, Nadara, Pilipczuk, and Soko\l{}owski~\cite{dyntw}.
As rankwidth generalizes treewidth, the high-level approach of our \Cref{the:dynsimple} is similar to the high-level approach of the data structure of~\cite{dyntw} (yielding the similar running times).
However, making the approach work for rankwidth requires developing an extensive amount of new machinery for rankwidth, with several new algorithmic and structural insights.
We provide a comparison between techniques in \Cref{the:dynsimple} and in the result of~\cite{dyntw} at the end of \Cref{ssec:ov:dynrw}.
We note that also formally speaking, \Cref{the:dynsimple} is a generalization the result of~\cite{dyntw}: The setting of treewidth and $\CMSO_2$ logic can be reduced to the setting of rankwidth and $\CMSO_1$ logic by considering instead of a graph $G$ the graph $G'$ obtained by subdividing every edge of $G$ once and adding two degree-1 vertices adjacent to each non-subdivision vertex.
Then every edge update of $G$ can be simulated by two edge updates of $G'$, the rankwidth of $G'$ is at most the treewidth of $G$ plus one, and every $\CMSO_2$ sentence about $G$ can be translated into a $\CMSO_1$ sentence about $G'$.

An apparent caveat in the statement of \Cref{the:dynsimple} is that even though the $n$-vertex complete graph has rankwidth $1$, it takes $\Omega(n^2)$ edge insertion operations to build it in the data structure.
To address this caveat, we introduce a framework for dense updates that can manipulate many edges efficiently.
For example, given two sets of vertices $A$ and $B$, we can add all possible edges between $A$ and $B$ in $|A \cup B| \cdot 2^{\Oh[k]{\sqrt{\log n \log \log n}}}$ amortized time.
More generally, we define that an \emph{edge update sentence} is a tuple $\eudesc = (\varphi, X, X_1, \ldots, X_p)$, where $\varphi$ is a $\CMSO_1$ sentence with $p+1$ free set variables, and $X_i \subseteq X \subseteq V(G)$.
Such edge update sentence re-defines all adjacencies inside the induced subgraph $G[X]$ by setting an edge between $u,v \in X$ if and only if $G$, together with the interpretations of the $p+1$ free variables as $(\{u,v\}, X_1, \ldots, X_p)$, satisfies $\varphi$.
Then we define that $|\eudesc| = |X|$ and that the length of $\eudesc$ is the length of $\varphi$.
Now, \Cref{the:dynsimple} can be generalized as follows.

\begin{theorem}
\label{the:dynfull}
The data structure of \Cref{the:dynsimple}, when furthermore initialized with a given integer $d$, can also support the following operations:
\begin{itemize}
\item $\mathsf{Update}(\eudesc)$: Given an edge update sentence $\eudesc$ of length at most $d$, either returns that the graph resulting from applying $\eudesc$ to $G$ would have rankwidth more than $k$, or applies $\eudesc$ to update $G$. Runs in $|\eudesc| \cdot 2^{\Oh[k,d]{\sqrt{\log n \log \log n}}}$ amortized time.
\item $\LinCMSO_1(\varphi,X_1,\ldots,X_p)$: Given a $\LinCMSO_1$ sentence $\varphi$ of length at most $d$ with $p$ free set variables and $p$ vertex subsets $X_1,\ldots,X_p \subseteq V(G)$, returns the value of $\varphi$ on $(G,X_1,\ldots,X_p)$. Runs in time $\Oh[d]{1}$ if $X_1,\ldots,X_p = \emptyset$, and in time $\sum_{i=1}^p |X_i| \cdot 2^{\Oh[k,d]{\sqrt{\log n \log \log n}}}$ otherwise.
\end{itemize}
\end{theorem}

We included also the $\LinCMSO_1(\varphi,X_1,\ldots,X_p)$ operation to the statement of \Cref{the:dynfull} to allow determining whether two given vertices $u$ and $v$ are adjacent in $2^{\Oh[k,d]{\sqrt{\log n \log \log n}}}$ time, as this becomes a non-trivial problem in this setting.
We discuss further extensions of \Cref{the:altrw} and \Cref{the:dynsimple} in \Cref{sec:concl}.

\paragraph{Organization.}
We start by giving an overview of our proofs in \Cref{sec:overview}.
We discuss notation and preliminary results in \Cref{sec:preli}.
Then, we give our framework of annotated rank decompositions and prefix-rebuilding operations in \Cref{sec:annot}.
In \Cref{sec:refi} we give the main tools for maintaining rank decompositions of dynamic graphs, although delaying significant ingredients to \Cref{sec:dealternation,sec:computing-closures}.
In \Cref{sec:rwautom} we introduce rank decomposition automata and give results about them.
Then, in \Cref{sec:dynrw} we finish the proofs of \Cref{the:dynsimple,the:dynfull}.
In \Cref{sec:comp} we prove \Cref{the:altrw}.
Then in \Cref{sec:dealternation} we prove a result called ``Dealternation Lemma'', which is used in \Cref{sec:refi}.
In \Cref{sec:computing-closures} we give results related to computing exact rankwidth by dynamic programming, which are used in \Cref{sec:refi,sec:dynrw,sec:comp}.
Finally, we conclude in \Cref{sec:concl}.

%% file: overview.tex
\section{Overview}
\label{sec:overview}
In this section we give an overview of our algorithms.
We start by giving an overview of the proof of \Cref{the:dynsimple} in \Cref{ssec:ov:dynrw}.
We omit many important ingredients, of which two major ones we overview in \Cref{ssec:ov:dealt,ssec:overview-optimum-width}.
We overview the proof of \Cref{the:altrw} in \Cref{ssec:ov:altrw}.

\input{overview_dynrw.tex}

\input{overview-dealt.tex}

\input{overview-jko.tex}

\input{overview-linear-algo.tex}

%% file: overview_dynrw.tex
\subsection{Dynamic rankwidth}
\label{ssec:ov:dynrw}
Suppose we maintain a dynamic $n$-vertex graph $G$ under edge insertions and deletions, and the rankwidth of $G$ is guaranteed to stay at most $k$.
Our goal will be to maintain a \emph{rooted rank decomposition} $\Tc = (T,\lmap)$ of $G$ of width at most $4k$ and height at most $h = 2^{\Oh[k]{\sqrt{\log n \log \log n}}}$.
Rooted rank decompositions are defined like rank decompositions, except the tree $T$ is a rooted binary tree.
Even the existence of rank decompositions with such parameters is not immediately obvious, but indeed Courcelle and Kant{\'{e}}~\cite{DBLP:conf/wg/CourcelleK07} show that a rank decomposition of width $k$ can be turned into a rooted rank decomposition of width at most $2k$ and height $\Oh{\log n}$.

It is essential for our algorithm to also maintain dynamic programming schemes on the rank decomposition.
We need this to support both the $\LinCMSO_1$ queries and various internal operations of our data structure.
We will formalize dynamic programming as automata processing the tree $T$, so that the state of a node can be computed in $\Oh[k]{1}$ time from the states of its children.
For this, we need to store additional information about the graph $G$ in the rank decomposition, for which we next define \emph{annotated rank decompositions}.

For an edge $xy$ of the tree $T$, we denote by $\lparts(\Tc)[\oxy] \subseteq V(G)$ the vertices of $G$ that are mapped to leaves of $T$ closer to $x$ than $y$.
If $\Tc$ has width $\le 4k$, there exists a set $R \subseteq \lparts(\Tc)[\oxy]$ with $|R| \le 2^{4k}$ so that for every $v \in \lparts(\Tc)[\oxy]$ exists $r \in R$ so that $N(r) \setminus \lparts(\Tc)[\oxy] = N(v) \setminus \lparts(\Tc)[\oxy]$.
We define that such $R$ is a \emph{representative} of $\lparts(\Tc)[\oxy]$, and a minimal such $R$ a \emph{minimal representative} of $\lparts(\Tc)[\oxy]$.
An annotated rank decomposition stores for every oriented edge $\oxy$ of $T$ a minimal representative $\reps(\oxy)$ of the set $\lparts(\Tc)[\oxy]$.
It also stores for every edge $xy$ of $T$ the bipartite graph $\repse(xy) = G[\reps(\oxy), \reps(\oyx)]$, encoding the adjacencies between $\lparts(\Tc)[\oxy]$ and $\lparts(\Tc)[\oyx]$.
Furthermore, for every (oriented) path $xyz$ of length $3$ in $T$, it stores the function $\dmap(xyz) \colon \reps(\oxy) \to \reps(\oyz)$, mapping each $v \in \reps(\oxy)$ to $r \in \reps(\oyz)$ so that $N(r) \setminus \lparts(\Tc)[\oyz] = N(v) \setminus \lparts(\Tc)[\oyz]$.
It can be shown that an annotated rank decomposition of $G$ uniquely defines the graph $G$.
Indeed in our algorithm we do not store the graph $G$ explicitly; we only maintain an annotated rank decomposition of it.

Then, the slightly more formal definition of a \emph{rank decomposition automaton} is that it is a tree automaton working on $T$, where the state of a node $x$ can be computed from the states of its children and the annotations $\reps,\repse,\dmap$ around $x$ in $\Oh[k]{1}$ time.
Note that the total size of the annotations around $x$ is $\Oh[k]{1}$.
A significant part of this article is to show that various dynamic programming routines on rank decompositions and cliquewidth expressions can be formulated as rank decomposition automata.
The formal definitions about annotated rank decompositions are in \Cref{sec:annot} and about rank decomposition automata in \Cref{sec:rwautom}.

After this detour to dynamic programming, let us return to the problem of dynamic maintenance of the annotated rank decomposition $\Tc$ of $G$.
We say that a set $\Tpref \subseteq V(T)$ is a \emph{prefix} of $T$ if it induces a connected subtree of $T$ that contains the root.
Suppose $\Tc$ has width at most $4k$ and height at most $h$, and there comes an update to insert or delete an edge between two vertices $u$ and $v$, which turns $G$ into $G'$.
Let $\Tpref$ be the minimal prefix of $T$ that contains $\lmap(u)$ and $\lmap(v)$.
We have that $|\Tpref| \le 2 \cdot h \le 2^{\Oh[k]{\sqrt{\log n \log \log n}}}$.
Now, we can turn $\Tc$ into an annotated rank decomposition of $G'$ by only re-computing annotations inside $\Tpref$, which can be done in $\Oh[k]{|\Tpref|}$ time.
The states of the maintained automata also need to be re-computed only for nodes in $\Tpref$, which also works in $\Oh[k]{|\Tpref|}$ time.
Therefore, we manage to update $\Tc$ in the desired time bound.
The only issue is that the width of $\Tc$ could increase in this process.

We observe that the width of $\Tc$ can increase only by one, and moreover, only the widths of edges of $\Tc$ inside $\Tpref$ can increase, so suppose now that $\Tc$ has width $4k+1$ and all edges of width $4k+1$ are inside $\Tpref$.
For convenience, we remove $\lambda(u)$ and $\lambda(v)$ from $\Tpref$ so that all leaves of $\Tc$ are outside of $\Tpref$; this maintains that all edges of width $4k+1$ are inside $\Tpref$ since the edges of $\Tc$ incident to $\lambda(u)$ and $\lambda(v)$ have width at most $1$ at all times.
To reduce the width, we design a \emph{refinement operation}, that takes as input a~prefix $\Tpref$ of $\Tc$ not containing any leaves of $\Tc$ such that all edges incident to a~vertex outside of $\Tpref$ have width at most $4k$ and, in some sense, locally re-computes the decomposition for the prefix $\Tpref$.
More accurately, the goal is that applying refinement to $\Tpref$ reduces the width back to at most $4k$, assuming $G$ has rankwidth at most $k$, and runs in amortized time proportional to $|\Tpref|$.
The refinement can increase the height of $\Tc$ by $\Oh{\log n}$.
To not allow the height of $\Tc$ to spiral out of control, we then use a \emph{height reduction scheme}, that by repeatedly applying refinement decreases the height back to at most~$h$.


Then we delve into the details of the refinement operation.
Let $\Tpref$ be the given prefix, and let us define $\oApp(\Tpref)$ as the set of oriented edges $\oxy$ of $T$ with $x \notin \Tpref$ and  $y \in \Tpref$.
Note that for every $v \in V(G)$ there exists a~unique edge $\oxy \in \oApp(\Tpref)$ such that $v \in \lparts(\Tc)[\oxy]$.
Then we say that a \emph{closure} of $\Tpref$ is a partition $\prt$ of $V(G)$, so that for each $C \in \prt$ there exists $\oxy \in \oApp(\Tpref)$ with $C \subseteq \lparts(\Tc)[\oxy]$.
A rank decomposition of a closure $\prt$ is a pair $\Tc^{\star} = (T^{\star},\lmap^{\star})$, where $T^{\star}$ is a cubic tree with $|\prt|$ leaves and $\lmap^{\star}$ is a bijection from $\prt$ to the leaves of $T^{\star}$.
The width of $\Tc^{\star}$ can be naturally defined analogously to the definition for rank decompositions of graphs.
Before giving any arguments about how to find closures $\prt$ with desirable properties, let us describe how the refinement operation uses a closure $\prt$ to transform $\Tc$ into a new annotated rank decomposition $\Tc'$.

Assume $\Tc^{\star}$ is a rank decomposition of $\prt$ of width at most $2k$.
We use the method of~\cite{DBLP:conf/wg/CourcelleK07} to turn $\Tc^{\star}$ into a rooted rank decomposition $\Tc^{\star\star}$ of width at most $4k$ and height at most $\Oh{\log n}$.
Then, for each $C \in \prt$, we construct from $\Tc$ a rooted rank decomposition $\Tc^{C}$ with $|C|$ leaves corresponding to $C$ by repeatedly deleting all leaves not corresponding to vertices in $C$ and contracting degree-2 nodes.
In particular, if $\Tc$ contains an edge $xy$ with $\lparts(\Tc)[\oxy] \cap C \neq \emptyset$ and $\lparts(\Tc)[\oyx] \cap C \neq \emptyset$, then $\Tc^{C}$ contains an edge $x'y'$ with $\lparts(\Tc^C)[\vec{x'y'}] = \lparts(\Tc)[\oxy] \cap C$ and $\lparts(\Tc^C)[\vec{y'x'}] = \lparts(\Tc)[\oyx] \cap C$ (and all edges of $\Tc^C$ are like that).
Then we construct $\Tc'$ by taking $\Tc^{\star\star}$ and identifying the root of each $\Tc^C$ with the leaf of $\Tc^{\star\star}$ corresponding to $C$.
Without assuming anything about $\prt$, we can deduce that the height of $\Tc'$ is at most $\Oh{\log n}$ more than the height of $\Tc$ and all edges of $\Tc'$ corresponding to edges of $\Tc^{\star\star}$ have width at most $4k$.
However, we do not know anything about the widths of edges not in $\Tc^{\star\star}$, and we do not know how this transformation could be implemented efficiently.

The first requirement for efficient implementation of this transformation is that $|\prt|$ is not too large.
We prove the existence of a closure with an even stronger property.
We say that $\prt$ is a \emph{$k$-closure} if the rankwidth of $\prt$ is at most $2k$.
We also say that $\prt$ is \emph{$c$-small} if for all $\oxy \in \oApp(\Tpref)$ there are at most $c$ parts $C \in \prt$ with $C \subseteq \lparts(\Tc)[\oxy]$.
We prove the following lemma in \Cref{sec:refi}.
\begin{lemma}[\Cref{lem:smallclosures}]
\label{lem:ov:smallclosure}
For every $k,\ell \in \N$ there exists $c \in \N$ so that if $\Tc$ has width at most $\ell$ and $G$ has rankwidth at most $k$, then for any prefix $\Tpref$ of $\Tc$ there exists a $c$-small $k$-closure $\prt$ of $\Tpref$.
\end{lemma}

As $|\oApp(\Tpref)| \le \Oh{|\Tpref|}$, this implies $|\prt| \le \Oh[k]{|\Tpref|}$.
In \Cref{ssec:ov:dealt}, we will highlight the main tool we develop for proving \Cref{lem:ov:smallclosure} --- the \emph{Dealternation Lemma} for rankwidth, analogous to the Dealternation Lemma for treewidth of Boja{ń}czyk and Pilipczuk~\cite{DBLP:journals/lmcs/BojanczykP22}, and overview its proof, which is fully presented in \Cref{sec:dealternation}.

We then require one more property of $\prt$, which will be useful in both addressing the issue of the widths of the edges of $\Tc'$ coming from the decompositions $\Tc^{C}$ and in the efficient implementation of the refinement.
For a node $x$ of $\Tc$ with parent $p$ we denote $\lparts(\Tc)[x] = \lparts(\Tc)[\oxp]$, and when $x$ is the root $\lparts(\Tc)[x] = V(G)$.
We say that $\prt$ \emph{cuts} a node $x$ of $\Tc$ if more than one part in $\prt$ intersects $\lparts(\Tc)[x]$, and define $\cut(\prt)$ to be the set of nodes cut by $\prt$.
Note that $\cut(\prt)$ is a prefix of $T$ and $\Tpref \subseteq \cut(\prt)$.
We define that $\prt$ is a \emph{minimal} $c$-small $k$-closure if among all $c$-small $k$-closures of $\Tpref$, $\prt$ primarily minimizes $\sum_{C \in \prt} \cutrk(C)$ and secondarily minimizes $|\cut(\prt)|$.
Here, $\cutrk(C)$ denotes the rank of the $|C| \times |V(G)\setminus C|$ matrix describing adjacencies between $C$ and $V(G)\setminus C$.

We observe that if $\Tc$ has a node $x \in V(T) \setminus \Tpref$ and $C \in \prt$ intersects $\lparts(\Tc)[x]$, then in $\Tc'$ there is a node $x^C$ with $\lparts(\Tc')[x^C] = \lparts(\Tc)[x] \cap C$, and moreover all nodes of $\Tc'$ coming from the decompositions $\Tc^{C}$ can be characterized like this.
Therefore, to prove that $\Tc'$ has width at most $4k$, it suffices to prove that for all such $x$ it holds that $\cutrk(\lparts(\Tc)[x] \cap C) \le 4k$.
We prove the following stronger statement in \Cref{sec:refi}.
\begin{lemma}[\Cref{lem:closlink,lem:mincloswidthbound}]
\label{lem:ov:closlink}
Let $\prt$ be a minimal $c$-small $k$-closure of $\Tpref$, $C \in \prt$, and $x \in V(T) \setminus \Tpref$ with $\lparts(\Tc)[x] \cap C \neq \emptyset$.
Then $\cutrk(\lparts(\Tc)[x] \cap C) \le \cutrk(\lparts(\Tc)[x])$, with equality only if $\lparts(\Tc)[x] \subseteq C$.
\end{lemma}

As $\cutrk(\lparts(\Tc)[x]) \le 4k$ for $x \in V(T) \setminus \Tpref$, this implies that $\Tc'$ has width at most $4k$ when $\prt$ is minimal.
The proof of \Cref{lem:ov:closlink} makes use of the submodularity of the $\cutrk$ function.
It can be considered to be a rankwidth analog of the techniques developed for improving tree decompositions by Korhonen and Lokshtanov~\cite{DBLP:conf/stoc/KorhonenL23}.
Let us then assume that $\prt$ is a minimal $c$-small $k$-closure.

We then use the fact that in \Cref{lem:ov:closlink} the equality holds only if $\lparts(\Tc)[x] \subseteq C$.
The nodes of $\Tc$ can be partitioned into three groups based on $\Tpref$ and $\prt$ --- those in $\Tpref$, those in $\cut(\prt) \setminus \Tpref$, and those in $V(T) \setminus \cut(\prt)$.
If $x \in V(T) \setminus \cut(\prt)$, then there exists $C \in \prt$ so that $\lparts(\Tc)[x] \subseteq C$.
In this case the resulting decomposition $\Tc'$ will contain the exactly same subtree rooted at $x$ as $\Tc$, so maximal such subtrees can be copied from $\Tc$ to $\Tc'$ by changing just one pointer, copying also the annotations and the automata states, and the number of them is $\Oh{|\cut(\prt)|}$.
If $x \in \cut(\prt) \setminus \Tpref$, then a node $x^C$ corresponding to $x$ is constructed for every $C \in \prt$ that intersects $\lparts(\Tc)[x]$.
Because $\prt$ is $c$-small, there are at most $c$ such nodes $x^C$, and by \Cref{lem:ov:closlink} for all of them it holds that $\cutrk(\lparts(\Tc')[x^C]) < \cutrk(\lparts(\Tc)[x])$.
Thus, we can think that we replace each node in $\cut(\prt) \setminus \Tpref$ by at most $c$ nodes that each has smaller width (the width of a node is the width of the edge between it and its parent), which motivates to use the following potential function for amortized analysis:

\[\Phi(\Tc) = \sum_{x \in V(T)} (2c)^{\width_{\Tc}(x)} \cdot \height_{\Tc}(x),\]
where $\height_{\Tc}(x)$ denotes the height of $x$ in $\Tc$, i.e., the distance from $x$ to the deepest leaf in its subtree.
Let us not focus on the $\height(x)$ factor at this point, but note that by the above discussion, the factor $(2c)^{\width_{\Tc}(x)}$ achieves that for every $x \in \cut(\prt) \setminus \Tpref$,
\[\sum_{C \in \prt\, \mid\, \lparts(\Tc)[x] \cap C \neq \emptyset} (2c)^{\width_{\Tc'}(x^C)} \cdot \height_{\Tc'}(x^C) < (2c)^{\width_{\Tc}(x)} \cdot \height_{\Tc}(x),\]
implying that the potential decreases proportionally to the number of nodes in $\cut(\prt) \setminus \Tpref$, which justifies implementing the refinement operation in time proportional to $|\cut(\prt)|$.
Before going into more analysis of the potential and the height reduction, let us discuss this implementation.

Given $\Tpref$, we wish to find in time $\Oh[k]{|\cut(\prt)|}$ some representation of a minimal $c$-small $k$-closure $\prt$ of $\Tpref$.
We observe that for each oriented edge $\oxy \in \oApp(\cut(\prt))$, there is unique $C \in \prt$ so that $\lparts(\Tc)[\oxy] \subseteq C$.
Therefore, we define the \emph{appendix edge partition} $\aep(\prt)$ of $\prt$ to be the partition of $\oApp(\cut(\prt))$ into $|\prt|$ parts naturally corresponding to $\prt$.
As $|\oApp(\cut(\prt))| \le \Oh{|\cut(\prt)|}$, we can represent $\aep(\prt)$ in $\Oh{|\cut(\prt)|}$ space.
Also, a rank decomposition $\Tc^{\star}$ of $\prt$ can be represented in $\Oh{|\cut(\prt)|}$ space by associating the leaves with the parts of $\aep(\prt)$.
We compute these objects by the following lemma, which we prove in \Cref{sec:computing-closures} and overview in \Cref{ssec:overview-optimum-width}.
\begin{lemma}[Informal statement of \Cref{lem:closureprds}]
\label{lem:ov:closureprds}
By maintaining an automaton on $\Tc$, we can support an operation that given a prefix $\Tpref$, in time $\Oh[k]{|\cut(\prt)|}$ returns $\cut(\prt)$, $\aep(\prt)$, and a rank decomposition $\Tc^{\star}$ of $\prt$ of width at most $2k$, for some minimal $c$-small $k$-closure $\prt$ of $\Tpref$, or concludes that the rankwidth of $G$ is more than $k$.
\end{lemma}

After turning $\Tc^{\star}$ into a log-height decomposition $\Tc^{\star\star}$, we can compute based on $\Tc^{\star\star}$ and $\aep(\prt)$ a ``recipe'' of size $\Oh{|\cut(\prt)|}$ on how the subtrees of $\Tc$ hanging below edges $\oxy \in \oApp(\cut(\prt))$ should be re-arranged to transform $\Tc$ into $\Tc'$.
Even after this, the problem of turning $\Tc$ into an annotated rank decomposition $\Tc'$ efficiently turns out to not be straightforward, as we need to compute the annotations for $\Tc'$.
In \Cref{sec:annot} we give a divide-and-conquer type algorithm for computing these annotations based on the recipe in $\Oh[k]{|\cut(\prt)| \log n}$ time (\Cref{lem:rearangmain,lem:prdsrearangmain}).
This concludes the overview on how $\Tc$ is transformed into $\Tc'$ in $\Oh[k]{|\cut(\prt)| \log n}$ time.

We then return to the potential function $\Phi(\Tc)$.
The main idea of the amortized analysis of our algorithm is that each edge update can increase the potential by at most $\Oh[k]{h^2}$ (recall that $h$ is the bound on the height of $\Tc$) and the refinement operation can increase the potential by at most $\Oh[k]{h |\Tpref| \log n}$ and decreases it proportionally to $|\cut(\prt) \setminus \Tpref|$.
The fact that edge updates increase the potential by at most $\Oh[k]{h^2}$ is straightforward from the facts that the update affects the widths of at most $\Oh{h}$ nodes, and the contribution of each node to potential is at most $\Oh[k]{h}$.

The analysis of the potential change caused by refinement is based on case-analysis of nodes of $\Tc'$: If a node is in a subtree directly copied from $\Tc$ to $\Tc'$, then nothing changes.
If a node is of type $x^C$ for $x \in \cut(\prt) \setminus \Tpref$ and $C \in \prt$, then its potential can be charged from the potential of the corresponding node $x$ as argued earlier, and this even decreases the potential proportionally to $|\cut(\prt) \setminus \Tpref|$.
If a node comes from $\Tc^{\star\star}$, then its height is initially $\Oh{\log n}$, but can increase when we attach trees $\Tc^{C}$ as its descendants.
We observe that each $\Tc^C$ can increase the height of at most $\Oh{\log n}$ such nodes, so the total potential of such nodes is bounded by $\Oh[k]{|\Tc^{\star\star}| \log n} + \sum_{C \in \prt} \Oh[k]{\height(\Tc^C) \log n} \le \Oh[k]{h |\Tpref| \log n}$ (recall that $|\prt| \le \Oh[k]{|\Tpref|}$).
These arguments imply that if the height of $\Tc$ stays at most $h$, then the amortized running time of each update is $\Oh[k]{h^2 \log^2 n}$.
It remains to give the height reduction scheme to maintain this height bound.

Suppose the height of $\Tc$ increased above $h$ by an application of the refinement operation.
We wish to argue that whenever the height is more than $h$, there is a prefix $\Tpref$, so that if we apply the refinement on $\Tpref$, the potential $\Phi(\Tc')$ of the resulting decomposition is smaller than $\Phi(\Tc)$, and moreover, the running time of the refinement operation is $\Oh[k]{(\Phi(\Tc) - \Phi(\Tc')) \cdot \log n}$.
For this, we prove the following more fine-grained bound on $\Phi(\Tc')$:
\[\Phi(\Tc') \le \Phi(\Tc) - \sum_{x \in \Tpref} \height_{\Tc}(x) - |\cut(\prt)| + \log n \cdot \Oh[k]{|\Tpref| + \sum_{\oxy \in \oApp(\Tpref)} \height_{\Tc}(x)}\]
This is not very hard to deduce from the construction of $\Tc'$ and the arguments for bounding the potential change given earlier, but we omit giving a more detailed argument here.
Then, it suffices to prove that if $\Tc$ has height more than $h$, we can find a non-empty prefix $\Tpref$ so that according to the above formula, $\Phi(\Tc') \le \Phi(\Tc) - |\cut(\prt)|$.
For this, we use the following result about binary trees proved implicitly in~\cite{dyntw} for their height reduction of dynamic treewidth.
\begin{lemma}[\cite{dyntw}]
\label{lem:ov:heightredu}
Let $c \ge 2$ and $T$ be a binary tree with $n$ nodes.
If the height of $T$ is at least $2^{\Omega(\sqrt{\log n \log c})}$ then there exists a non-empty prefix $\Tpref$ of $T$ so that \[c \cdot \left(|\Tpref| + \sum_{\oxy \in \oApp(\Tpref)} \height_{T}(x)\right) \le \sum_{x \in \Tpref} \height_T(x).\]
\end{lemma}

By plugging in $c = f(k) \log n$ for a suitable function $f(k)$, the existence of a desired prefix $\Tpref$ follows whenever $\height(\Tc) \ge 2^{\Omega(\sqrt{\log n \log (f(k) \log n)})} \ge 2^{\Omega_k(\sqrt{\log n \log \log n})}$, which is the claimed bound for $h$.
Then, the height-reduction scheme consists of applying refinement operations on such prefixes $\Tpref$ until the height is decreased below $h$.
As the running time is proportional to the potential decrease, these operations are ``free'' from the viewpoint of amortized analysis.
This concludes the overview of our dynamic algorithm, up to the Dealternation Lemma and the proof of \Cref{lem:ov:closureprds}, which we will overview in~\Cref{ssec:ov:dealt,ssec:overview-optimum-width}, respectively.

\paragraph{Comparison to dynamic treewidth of~\cite{dyntw}.}
Our approach for dynamic rankwidth is inspired by the approach for dynamic treewidth of~\cite{dyntw}.
In particular, we design a refinement operation with similar properties to their refinement operation, so that we can then use the height-reduction scheme encapsulated in~\Cref{lem:ov:heightredu} to control the height of the decomposition.
As the combinatorics of treewidth and rankwidth are different, the definitions and structural results used for our refinement operation are different from those of the refinement operation of~\cite{dyntw}.
In particular, the concept of closures and \Cref{lem:ov:smallclosure,lem:ov:closlink}, along with the Dealternation Lemma, are novel structural results about rankwidth.
Somewhat surprisingly, in the end our rankwidth version of the refinement operation turned out to be more elegant than the treewidth version, which has a more complicated construction of the resulting decomposition $\Tc'$, resulting also in a more complicated analysis of the potential.
From the more low-level side, manipulating rank decompositions and maintaining automata on them is much more complicated and less researched task than that on tree decompositions.
We consider the concept of annotated rank decompositions, along with the efficient algorithms for manipulating them (particularly \Cref{lem:rearangmain}), an important contribution of this work, which we will highlight further in~\Cref{ssec:overview-optimum-width}.

%% file: overview-dealt.tex
\subsection{Dealternation Lemma}
\label{ssec:ov:dealt}
We now overview a~crucial combinatorial result regarding optimum-width rank decompositions that lies at the heart of the dynamic rankwidth data structure: the Dealternation Lemma for rankwidth, proved in \cref{sec:dealternation}.
Then we sketch how the Dealternation Lemma is used in the proof of \cref{lem:ov:smallclosure}.

Our Dealternation Lemma essentially states the following: Whenever $\Tc^b$ is some rooted rank decomposition of a~graph $G$ of unoptimal (but bounded) width, there exists a~rank decomposition $\Tc$ of optimum width in which every subtree $\lparts(\Tc^b)[x] \subseteq V(G)$ of $\Tc^b$ can be decomposed into a~bounded number of ``simple'' pieces of $\Tc$.
An~analog of this statement for tree decompositions is proved in the work of Bojańczyk and Pilipczuk~\cite[Lemma 3.7]{DBLP:journals/lmcs/BojanczykP22}.

Formally, if $\Tc = (T, \lambda)$ is a~rooted rank decomposition of $G$ and $F \subseteq V(G)$, we say that $F$ is a~\emph{tree factor} of $\Tc$ if $F = \lparts(\Tc)[v]$ for some $v \in V(T)$, and a~\emph{context factor} if $F$ is nonempty and of the form $F_1 \setminus F_2$, where both $F_1$ and $F_2$ are tree factors.
Then the Dealternation Lemma reads as follows:

\begin{lemma}[\cref{lem:dealt}] 
  \label{lem:overview-subdealt-graphs}
  There exists a~function $f(\ell)$ so that if $G$ is a~graph and $\Tc^b = (T^b, \lambda^b)$ is a~rooted rank decomposition of $G$ of width $\ell$, then there exists a~rooted rank decomposition $\Tc$ of $G$ of optimum width so that for every node $x \in V(T^b)$, the set $\lparts(\Tc^b)[x]$ can be partitioned into a~disjoint union of at most $f(\ell)$ factors of $\Tc$.
\end{lemma}

\paragraph*{Subspace arrangements.}
In the following sections, we make heavy use of a~generalization of the notion of rankwidth to linear spaces over finite fields.
Let $\F$ be a~finite field; throughout the work we assume $\F = \GF(2)$.
For two linear subspaces $V_1, V_2 \subseteq \F^d$, let $V_1+V_2$ denote their sum and $V_1 \cap V_2$ denote their intersection. For $V \subseteq \F^d$, let $\dim(V)$ denote the dimension of $V$.
Any family $\Vc = \{V_1, \dots, V_n\}$ of linear subspaces of $\F^d$ is called a~\emph{subspace arrangement}.
For convenience, let $\sumof{\Vc} \coloneqq V_1 + \ldots + V_n$.
Let also $\lin{e}_1, \dots, \lin{e}_d$ denote the canonical basis of $\F^d$.

A~rank decomposition (or more properly, a \emph{branch decomposition}) $\Tc = (T, \lambda)$ of $\Vc$ is defined as for graphs or partitions, only that we assign subspaces $V_i$ to the leaves of $T$.
For an edge $xy$ of $T$, let $\lparts(\Tc)[\vec{xy}] \subseteq \Vc$ denote the subfamily of linear subspaces assigned to the leaves of $T$ closer to $x$ than~$y$.
If $\Tc$ is rooted, we define $\lparts(\Tc)[x]$ analogously to \Cref{ssec:ov:dynrw}.
Then the \emph{width} of an edge $xy$ is $\dim(\sumof{\lparts(\Tc)[\vec{xy}]} \cap \sumof{\lparts(\Tc)[\vec{yx}]})$, and then the width of $\Tc$ and the rankwidth of $\Vc$ are defined naturally.
The definitions of tree and context factors also lift naturally to the setting of rooted rank decompositions of subspace arrangements.

As observed in~\cite{DBLP:journals/tit/JeongKO17}, any undirected graph $G$ can be converted to an~equivalent subspace arrangement $\Vc$ as follows: Let $V(G) = \{v_1, \dots, v_n\}$.
Then for $i \in [n]$ define $V_i$ to be the subspace of $\F^n$ spanned by the two vectors $\lin{e}_i$ and $\sum_{v_j \in N(v_i)} \lin{e}_j$, and set $\Vc = \{V_1, \dots, V_n\}$.
Then, for any rank decomposition $\Tc$ of $G$ of width $\ell$, the isomorphic rank decomposition $\Tc'$ of $\Vc$ has width~$2\ell$.
So the rankwidth of $\Vc$ is equal to twice the rankwidth of $G$.
Hence, the Dealternation Lemma can be rephrased in the language of subspace arrangements:

\begin{lemma}[\cref{lem:subspace-dealternation}]
  \label{lem:overview-subdealt}
  There exists a~function $f(\ell)$ so that if $G$ is a~graph and $\Tc^b = (T^b, \lambda^b)$ is a~rooted rank decomposition of $G$ of width $\ell$, then there exists a~rooted rank decomposition $\Tc$ of $G$ of optimum width so that for every node $x \in V(T^b)$, the set $\lparts(\Tc^b)[x]$ can be partitioned into a~disjoint union of at most $f(\ell)$ factors of $\Tc$.
\end{lemma}

For convenience, we will henceforth write $\Vc_x$ as a~shorthand for $\lparts(\Tc^b)[x]$.
From now on we will only focus on the proof of \cref{lem:overview-subdealt}.

\paragraph*{Outline of the proof of the Dealternation Lemma.}
The proof of \Cref{lem:overview-subdealt} is inspired by its treewidth counterpart in \cite{DBLP:journals/lmcs/BojanczykP22}: Similarly to how their Dealternation Lemma can be viewed as a~purely combinatorial understanding of the treewidth algorithm by Bodlaender and Kloks~\cite{DBLP:journals/jal/BodlaenderK96}, our proof relies heavily on the combinatorial understanding of the rankwidth algorithm by Jeong, Kim and Oum~\cite{DBLP:journals/siamdm/JeongKO21}.
In fact, as a~starting point of the proof, we invoke their result:

\begin{theorem}[{\cite[Proposition 4.6]{DBLP:journals/siamdm/JeongKO21}}]
  \label{thm:overview-totally-pure}
  Let $\Tc^b = (T^b, \lambda^b)$ be a~rooted rank decomposition of a~subspace arrangement $\Vc$.
  Then there exists a~rooted rank decomposition $\Tc = (T, \lambda)$ of the same subspace arrangement $\Vc$ of optimum width that is ``totally pure'' with respect to $\Tc^b$.
\end{theorem}

We delay the precise definition of ``totally pure'' to \cref{sec:totally-pure-decomp}.
Intuitively though, $\Tc$ is totally pure with respect to $\Tc^b$ if $\Tc$ excludes, for all $x \in V(T^b)$, specific local ``complicated'' patterns defined in terms of $\Vc_x$ and $\Vc \setminus \Vc_x$.
We now lift \cref{thm:overview-totally-pure} to show that, in fact, for all $x \in V(T^b)$ the entire decomposition $\Tc$ admits a~simple and bounded-size description in terms of $\Vc_x$ and $\Vc \setminus \Vc_x$.

Let $v \in V(T)$ and $x \in V(T^b)$.
We say that $v$ is: (i) \emph{$x$-full} if $\lparts(\Tc)[v] \subseteq \Vc_x$; (ii) \emph{$x$-empty} if $\lparts(\Tc)[v]$ is disjoint from $\Vc_x$; and (iii) \emph{$x$-mixed} otherwise.
Now, a non-leaf node $v$ of $T$ is an~\emph{$x$-leaf point} if one child of $v$ is $x$-empty and the other is $x$-full; and $v$ is an~\emph{$x$-branch point} if both children of $v$ are $x$-mixed.
We define the \emph{$x$-mixed skeleton} of $\Tc$ as a~(possibly empty) rooted tree $T^\mix$ with $V(T^\mix)$ comprising the $x$-leaf points and the $x$-branch points of $T$, with two vertices $u, v \in V(T^\mix)$ connected by an edge if the simple path between $u$ and $v$ in $T$ is internally disjoint from $V(T^\mix)$ (see \cref{fig:skeleton-def}).
We then prove that:

\begin{lemma}[\cref{lem:small-mixed-skeleton}]
  \label{lem:overview-small-mixed-skeleton}
  There exists a~function $f(\ell)$ so that if $\Tc^b = (T^b, \lambda^b)$ is a~rooted rank decomposition of $\Vc$ of width $\ell$ and $\Tc$ is an~optimum-width decomposition of $\Vc$ that is totally pure with respect to $\Tc^b$, then, for every $x \in V(T^b)$, the $x$-mixed skeleton of $\Tc$ has at most $f(\ell)$ nodes.
\end{lemma}
We omit the proof in this overview; however, the proof proceeds by selecting any rank decomposition $\Tc$ of $\Vc$ of optimum width that is totally pure with respect to $\Tc^b$ (its existence is asserted by \cref{thm:overview-totally-pure}) and verifying that it actually satisfies all the conditions of \cref{lem:overview-small-mixed-skeleton}.
However, the work is far from done as $\Tc$ might not meet the requirements of the Dealternation Lemma.

Fix $x \in V(T^b)$.
Observe that every $x$-full node $v$ yields a~tree factor $\lparts(\Tc)[v] \subseteq \Vc_x$ (and every tree factor that is a~subset of $\Vc_x$ is like this).
On the other hand consider a~vertical path $v_1 v_2 \dots v_{p+1}$ in $\Tc$ without any $x$-leaf points or $x$-branch points such that $v_{p+1}$ is $x$-mixed.
For $i \in [p]$, let $v'_i$ be the child of $v_i$ different than $v_{i+1}$.
It can be shown that each node $v'_i$ is either $x$-full or $x$-empty.
Also whenever, for $1 \leq a \leq b \leq p$, the nodes $v'_a, \dots, v'_b$ are $x$-full, then $\lparts(\Tc)[v_a] \setminus \lparts(\Tc)[v_{b+1}] \subseteq \Vc_x$ is a~context factor (and all context factors that are subsets of $\Vc_x$ are like this).
This creates an issue: If, for example, $\lparts(\Tc)[v'_i]$ is $x$-full for odd $i \in [p]$ and $x$-empty for even $i \in [p]$ (in other words, the sequence $v'_1, \dots, v'_p$ alternates between $x$-full and $x$-empty nodes), we will not be able to partition $\lparts(\Tc^b)[x]$ into fewer than $\frac{p}{2} - O(1)$ factors.
Hence our strategy is to improve the decomposition by \emph{dealternating} all such heavily alternating paths --- that is, reorder the nodes along the path so as to bunch the $x$-full nodes $v'_i$ into a~small number of contiguous blocks, bounded by some constant $c_\ell \geq 1$ dependent only on $\ell$.
This reordering is highly non-trivial --- utmost care needs to be taken to avoid increasing the width of the decomposition --- and its implementation adapts to the setting of rankwidth the toolchain of \cite{DBLP:journals/lmcs/BojanczykP22}, which in turn encapsulates the technique of \emph{typical sequences} of \cite{DBLP:journals/jal/BodlaenderK96}.
After this is done, we show that we can partition $\lparts(\Tc^b)[x]$ into at most $\Oh{|V(T^\mix)| \cdot c_\ell}$ factors of $\Tc$, where $T^\mix$ is the $x$-mixed skeleton of $\Tc$.
Then we repeat the process for each $x \in V(T^b)$. So for $V(T^b) = \{x_1, \dots, x_n\}$, we perform $n$ phases, where in the $j$th phase, we perform the dealternation as above to produce a~partitioning of $\lparts(\Tc^b)[x_j]$ into a~small number of factors of $\Tc$.

This strategy comes with non-trivial requirements: (1) dealternation should not increase the width of the decomposition, (2) for $1 \leq i < j \leq n$, the reordering performed during the $j$th phase should preserve all factors in the already constructed partitioning of $\lparts(\Tc^b)[x_i]$, and (3) for $1 \leq j < i \leq n$, the $j$th phase should not blow up the size of the $x_i$-mixed skeleton of $\Tc$.
While it appears hard to ensure all these conditions at once, this feat can fortunately be achieved.
Hence, using our approach, we ultimately arrive at the following improvement step:

\begin{lemma}[Local Dealternation Lemma, \cref{lem:local-dealternation}]
  \label{lem:over-local-dealternation}
  There exists a~function $f(t,\ell)$ such that the following holds.
  Let $x \in V(T^b)$ and assume that the $x$-mixed skeleton of $\Tc$ has $t$ nodes. Then there exists a~rooted rank decomposition $\Tc'$ of $\Vc$ of optimum width such that:
  \begin{itemize}
    \item the set $\lparts(\Tc^b)[x]$ is a~disjoint union of at most $f(t, \ell)$ factors of $\Tc'$;
    \item for every $y \in V(T^b)$, the $y$-mixed skeletons of $\Tc$ and $\Tc'$ are equal; and
    \item if $y$ is not an~ancestor of $x$ and $F \subseteq \lparts(\Tc^b)[y]$ is a~factor of $\Tc$, then $F$ is also a~factor of $\Tc'$.
  \end{itemize}
\end{lemma}

\cref{lem:over-local-dealternation} is proved in \cref{ssec:final-local-dealternation-lemma}.
Then the Dealternation Lemma follows by sorting the nodes of $V(T^b)$ in the order of non-inc\-reasing distance from the root of $T^b$ and applying \cref{lem:over-local-dealternation}  for each $x \in V(T^b)$ in this order.

\paragraph*{Dealternation Lemma to \cref{lem:ov:smallclosure}.}
We now briefly describe how the Dealternation Lemma implies \cref{lem:ov:smallclosure}.
Fix $k, \ell \in \N$ and let $c \coloneqq f(\ell)$, where the function $f$ is as in the statement of the Dealternation Lemma.
Let $\Tc$ be a~rank decomposition of $G$ of width $\ell$ and assume that $G$ has width $k$.
By applying \Cref{lem:overview-subdealt-graphs}, let $\Tc' = (T',\lmap')$ be a rooted rank decomposition of $G$ of width $k$ so that for every node $t \in V(T)$ the set $\lparts(\Tc)[t]$ can be partitioned into a disjoint union of at most $c$ factors of $\Tc'$.
Then for each $a \in \App_T(\Tpref)$ let $\prt_a$ be the partition of $\lparts(\Tc)[a]$ into at most $c$ parts that are factors of $\Tc'$, and let $\prt = \bigcup_{a \in \App_T(\Tpref)} \prt_a$.
It remains to show that $(G[\prt], \prt)$ has rankwidth at most $2k$. The bound on the rankwidth of $(G[\prt], \prt)$ can be shown in several ways: For instance, one can construct a~rank decomposition $\Tc'' = (T'', \lmap'')$ of $\prt$ from $\Tc'$ by: (1) setting $T'' \coloneqq T'$, (2) choosing for every $C \in \prt$ an~arbitrary vertex $v \in C$ and assigning $\lmap''(C) \coloneqq \lmap'(v)$, and (3) removing from $T''$ all leaves without any assigned parts of $\prt$ and contracting degree-$2$ vertices.
It then can be proved that such a~constructed $\Tc''$ has width at most $2k$.
Therefore, $\prt$ is a~$k$-closure of $\Tpref$, and its $c$-smallness follows directly from the construction.

%% file: overview-jko.tex
\subsection{Automata for optimum-width decompositions}
\label{ssec:overview-optimum-width}
We then overview an~algorithmic result displaying the strength of the model of rank decomposition automata; namely that there exists a~rank decomposition automaton computing the exact value of the rankwidth of the underlying graph.

\begin{lemma}[Informal]
  \label{lem:overview-decomposition-aut}
  Fix integers $k \leq \ell$.
  Suppose $\Tc$ is an~annotated rank decomposition of width at most $\ell$ of a~dynamic graph $G$.
  By maintaining an~automaton on $\Tc$, we can support an operation that returns whether the rankwidth of $G$ is at most $k$.
\end{lemma}

We will then use \Cref{lem:overview-decomposition-aut} to show that given an~annotated rank decomposition of small (but possibly non-optimal) width of a~graph, we can efficiently construct a~rank decomposition of this graph of optimum width:

\begin{lemma}[\cref{lem:linear-decomposition-recovery}]
  \label{lem:overview-get-decomposition}
There is an~algorithm that, given as input an~annotated rank decomposition $\Tc$ of width $\ell$ of a graph $G$ and an integer $k$, in time $\Oh[\ell]{|\Tc|}$ either determines that $G$ has rankwidth larger than $k$, or outputs a (non-annotated) rank decomposition $(T,\lambda)$ of $G$ of width at most~$k$.
Moreover, the resulting decomposition can be annotated in time $\Oh[\ell]{|\Tc| \log |\Tc|}$.
\end{lemma}

Later, we will show how both \cref{lem:overview-decomposition-aut,lem:overview-get-decomposition} are used in the proof of \cref{lem:ov:closureprds} announced in \cref{ssec:ov:dynrw}, i.e., that we can maintain an~automaton on $\Tc$ supporting the following operation: given a~prefix $\Tpref$ of $\Tc$, find a~minimal $c$-small $k$-closure $\prt$ of $\Tpref$.
Moreover, \cref{lem:overview-get-decomposition} is crucially used in the proof of \cref{the:altrw}; we overview that result in \cref{ssec:ov:altrw}.

The results in this section build on (and improve upon) an algorithm of Jeong, Kim, and Oum~\cite{DBLP:journals/siamdm/JeongKO21} for computing rankwidth exactly in $\Oh[k]{n^3}$ time by using a dynamic programming procedure that can be regarded as a rankwidth analog of the Bodlaender-Kloks dynamic programming for treewidth \cite{DBLP:journals/jal/BodlaenderK96}.
\Cref{lem:overview-get-decomposition} showcases the strength of annotated rank decompositions, as the corresponding algorithm for rank decompositions given in~\cite{DBLP:journals/siamdm/JeongKO21} works in $\Oh[\ell]{|\Tc|^2}$ time.

\paragraph*{Exact rankwidth automaton.}

%

The automaton announced in the statement of \cref{lem:overview-decomposition-aut} effectively reimplements the subroutine of \textsc{Branch-Width Compression} from the cubic-time rankwidth algorithm of Jeong, Kim and Oum~\cite{DBLP:journals/siamdm/JeongKO21}: Given a~subspace arrangement $\Vc$, $|\Vc| = n$, comprising subspaces of $\F^n$ and a~rank decomposition $\Tc^b$ of $\Vc$ of width at most $\ell$, find a~rank decomposition of $\Vc$ of width at most $k$ if one exists.
However, due to the fact that their algorithm operates on subspaces of $n$-dimensional linear spaces explicitly, their subroutine works in time $\Oh[k]{n^2}$ --- and even this complexity is only achieved after a~cubic-time preprocessing of $\Vc$.
In the restricted case of rankwidth of graphs, we are able to break the quadratic time barrier by manipulating the implicit representations of these spaces and optimize the time complexity of our implementation to linear, and even represent the algorithm as an~automaton on~$\Tc^b$.

We now give a~short overview of \textsc{Branch-Width Compression} in \cite{DBLP:journals/siamdm/JeongKO21}.
Recall that $\Vc_x = \lparts(\Tc^b)[x]$ for $x \in V(T^b)$.
The \emph{boundary space} at $x$ is $B_x \coloneqq \sumof{\Vc_x} \cap (\Vc \setminus \Vc_x)$; so we have $\ell = \max_{x \in V(T^b)} \dim(B_x)$.
Let $\BB_x$ be an~ordered basis of $B_x$ --- any sequence of $\dim(B_x)$ vectors of $B_x$ spanning $B_x$.
Then any vector of $B_x$ can be uniquely represented in the basis $\BB_x$ using $\dim(B_x)$ bits as a~linear combination of vectors of $\BB_x$, and any subspace of $B_x$ can be represented using at most $\dim(B_x)^2$ bits as a~span of at most $\dim(B_x)$ vectors of $B_x$. 
If $x$ is a~non-leaf node with two children $c_1, c_2$, then we define $B'_x = B_x + B_{c_1} + B_{c_2}$, and we let $\BB'_x$ to be an~ordered basis of $B'_x$ whose prefix is $\BB_x$.
In the algorithm, all subspaces of $B_x$ are represented in the basis $\BB_x$, and all subspaces of $B'_x$ are represented in the basis $\BB'_x$.
For any node $x$ of $T^b$ with parent $p$, let $M_{\vec{xp}}$ be the transition matrix from the basis $\BB_x$ to the basis $\BB'_p$, i.e., the unique $|\BB'_p| \times |\BB_x|$ matrix such that, for every vector $\lin{v} \in \F^{\dim(B_x)}$, we have $\sum_{i=1}^{\dim(B_x)} \lin{v}_i (\BB_x)_i = \sum_{i=1}^{\dim(B'_p)} (M_{\vec{xp}} \lin{v})_i (\BB'_p)_i$.
Note that $M_{\vec{xp}}$ can be represented using $\dim(B_x) \cdot \dim(B'_p) = \Oh[\ell]{1}$ bits, even though $B_x$ and $B'_p$ are subspaces of a~highly-dimensional space $\F^n$.

In the algorithm the authors compute, for every $x \in V(T^b)$, the \emph{full set at $x$ of width $k$ with respect to $\Tc^b$}, denoted $\fullset_k(x)$, which is a~family of objects representing heavily compressed versions of rank decompositions of $\Vc_x$ that are totally pure with respect to $\Tc^b$.\footnote{This mirrors an~analogous definition of a~full set in the work of Bodlaender and Kloks \cite{DBLP:journals/jal/BodlaenderK96}.}
They show that:
%
\begin{itemize}[nosep]
  \item for a~leaf $l$ of $T^b$, the set $\fullset_k(l)$ can be constructed in time $\Oh[\ell]{1}$ given only $\dim(B_l)$;
  \item for a~non-leaf $x$ of $T^b$ with two children $c_1, c_2$, the set $\fullset_k(x)$ can be constructed in time $\Oh[\ell]{1}$ given $\fullset_k(c_1)$, $\fullset_k(c_2)$, the transition matrices $M_{\vec{c_1 x}}$, $M_{\vec{c_2 x}}$ and the value $\dim(B_x)$;
  \item $\fullset_k(r) \neq \emptyset$ for the root $r$ of $T^b$ if and only if the rankwidth of $\Vc$ is at most $k$; and
  \item if $\fullset_k(r) \neq \emptyset$, then a~rank decomposition of $\Vc$ of width at most $k$ can be reconstructed in time $\Oh[\ell]{n}$ from the values $\fullset_k(x)$ for $x \in V(T^b)$ and the transition matrices $M_{\vec{xp}}$.
\end{itemize}

So, assuming access to the transition matrices $M_{\vec{xp}}$ for all non-root $x \in V(T^b)$ with parent $p$, the entire \textsc{Branch-Width Compression} can be implemented in time $\Oh[\ell]{n}$.
However, it seems quite hard to determine these matrices efficiently from a~general subspace arrangement $\Vc$: \cite{DBLP:journals/siamdm/JeongKO21} determines the ordered bases $\BB_x$, $\BB'_x$ explicitly and computes the transition matrices $M_{\vec{xp}}$ from these bases afterwards.
This approach unfortunately requires $\Omega(n^2)$ time and space since we need $\Omega(n^2)$ bits of memory to simply store all the ordered bases.
However, in the setting of rank decompositions of graphs, we can work around this issue using annotated rank decompositions.
The following lemma (not proved here) encapsulates the key technical idea of our approach.

\begin{lemma}[informal statement of \cref{lem:transition-matrices}]
  \label{lem:overview-transition-matrices}
  Suppose $\Tc$ is a~rooted annotated rank decomposition of a~graph $G$ and $\Tc^b$ is the isomorphic rank decomposition of the subspace arrangement $\Vc$ equivalent to $G$.
  Then there exist two families of ordered bases $\{\BB_x\}_{x \in V(T^b)}$, $\{\BB'_x\}_{x \in V(T^b)}$, such that for every $x \in V(T^b)$ with parent $p$ and children $c_1, c_2$, we can uniquely determine the transition matrices $M_{\vec{c_1 x}}$, $M_{\vec{c_2 x}}$ and the value $\dim(B_x)$ in time $\Oh[\ell]{1}$ from the annotations of $\Tc$ around $x$.
\end{lemma}


Recalling the model of rank decomposition automata defined before, observe that we can encode the algorithm of Jeong, Kim and Oum as a~rank decomposition automaton running on $\Tc$:


\begin{lemma}[informal statement of \cref{lem:exact-rank-automaton}]
  \label{lem:overview-automaton}
%
  There exists a~rank decomposition automaton such that, for any graph $G$ with annotated rank decomposition $\Tc^b$ of width $\ell$, the state of the automaton at node $x \in V(\Tc^b)$ is exactly $\fullset_k(x)$.
  Each state of the automaton can be evaluated in time $\Oh[\ell]{1}$.
\end{lemma}

This essentially resolves \cref{lem:overview-decomposition-aut}.
With the help of the rank decomposition reconstruction subroutine from \cite{DBLP:journals/siamdm/JeongKO21}, our algorithm can also output a~non-annotated rank decomposition of $G$ of width at most $k$ in linear time, yielding the first part of \Cref{lem:overview-get-decomposition}.
By \cref{lem:rearangmain}, the output decomposition of \Cref{lem:overview-get-decomposition} can be annotated in $\Oh[\ell]{|\Tc| \log |\Tc|}$ time using a~divide-and-conquer type algorithm.


\paragraph*{Closure automaton.}
We then briefly sketch the proof of \cref{lem:ov:closureprds} as an~application of \cref{lem:overview-automaton}:
Assuming we maintain appropriate automaton on a~decomposition $\Tc$, we can support an~operation that given a~prefix $\Tpref$, returns an~encoding of a~minimal $c$-small $k$-closure  $\prt$ of $\Tpref$.

The automaton we will construct and maintain is a~\emph{closure automaton}.
For fixed $c$ and $k$ it computes, for all edges $\vec{xy}$ of $\Tc$, the family $\AutomataReps^{c,k}(\vec{xy})$ of all partitions $\prt_{\vec{xy}}$ of $\lparts(\Tc)[\vec{xy}]$ into at most $c$ parts that can be extended to a~minimal $c$-small $k$-closure of some prefix $\Tpref$ with $\vec{xy} \in \oApp(\Tpref)$.
The main challenge is how to represent $\prt_{\vec{xy}}$ --- storing the partitioning of $\lparts(\Tc)[\vec{xy}]$ explicitly is obviously impractical, and even storing $\aep(\prt_{\vec{xy}})$ turns out to be too expensive in our algorithm.
Instead, for every set $C \in \prt_{\vec{xy}}$ we only keep a~carefully selected minimal representative of $C$.
Then, with some extensive bookkeeping, we can compute the family $\AutomataReps^{c,k}(\vec{xy})$ for all $\vec{xy}$ in time $\Oh[c,k]{1}$.

Then, given a~prefix $\Tpref$, we want to find a~closure $\prt$ of $\Tpref$ such that: (i) for every $\vec{xy} \in \oApp(\Tpref)$, (the representation of) the subfamily of $\prt$ restricted to $\lparts(\Tc)[\vec{xy}]$ belongs to $\AutomataReps^{c,k}(\vec{xy})$, (ii) the rankwidth of $\prt$ is at most $2k$.
This can be achieved in time $\Oh[c,k]{|\Tpref|}$ by using the exact rankwidth automaton from \cref{lem:overview-automaton} and applying on it standard dynamic programming techniques on automata.
With enough care, this dynamic programming allows us to find a~\emph{minimal} closure $\prt$.
Then, restoring the objects $\cut(\prt)$, $\aep(\prt)$ and the rank decomposition of $\prt$ of width at most $2k$ are straightforward (even if technical) tasks that can be done in total time $\Oh[c,k]{|\cut(\prt)|}$.


%% file: overview-linear-algo.tex
\subsection{Almost-linear time algorithm for rankwidth}
\label{ssec:overview-linear-algo}
\label{ssec:ov:altrw}
Then we show how to compute a~rank decomposition of an~$n$-vertex, $m$-edge graph $G$ of width at most $k$ in time $\Oh[k]{n \cdot 2^{\sqrt{\log n} \log \log n}} + \Oh{m}$, if such a~decomposition exists (\cref{the:altrw}).
The full exposition of this algorithm can be found in \cref{sec:comp}.


In this section we assume that the input graph $G$ is bipartite, with the bipartition $V(G) = A \cup B$; in \cref{ssec:bipartite-reduction} we show that the general case can be reduced to the bipartite case by using a construction of Courcelle~\cite{DBLP:journals/japll/Courcelle06}.
Also assume that $G$ has rankwidth at most $k$.
Let $X \triangle Y = (X \cup Y) \setminus (X \cap Y)$ denote the symmetric difference of sets.
We say that two vertices $u, v \in V(G)$ are \emph{twins} if $N(u) = N(v)$, and \emph{$c$-near-twins} for $c \in \N$ if $|N(u) \triangle N(v)| \leq c$.
The main idea of our algorithm is to exploit the presence of many twins and near-twins in bipartite graphs of small rankwidth.


Consider the following auxiliary problem, which we call {\sc Twin Flipping}.
As input we are given an~annotated rank decomposition $\Tc$ of width at most $k$ of a~bipartite graph $G = (A, B, E)$, $E \subseteq A \times B$; a~set $X \subseteq A$ with the property that every vertex of $X$ has a~twin in $A \setminus X$; and a~set of pairs $F \subseteq X \times B$.
Let $n = |A| + |B|$ and assume $|F| \leq \Oh[k]{n}$.
The task is to construct an~annotated rank decomposition of $G' \coloneqq (A, B, E \triangle F)$ of width at most $k$, assuming it exists.
Define a~function $T(n)$ with the property that {\sc Twin Flipping} can be solved in time $\Oh[k]{T(n)}$.
Then we have:
\begin{lemma}
\label{lem:overview-twin-flipping}
$T(n) \leq n \cdot 2^{o(\sqrt{\log n} \log \log n)}$.
\end{lemma}
\begin{proof}[Sketch of the proof]
Consider $G$ to be a~dynamic graph described by an~annotated rank decomposition, initially $\Tc$, maintained by the dynamic rankwidth data structure of \Cref{the:dynsimple}. Then for each $(u, v) \in F$ in the lexicographic order, flip the adjacency between $u$ and $v$ (add the edge $uv$ to $G$ if not present, remove it otherwise).
It can be shown that the rankwidth of the dynamic graph never grows above $k+1$ during this process, so the data structure can perform the initialization and all the updates in time $n \cdot 2^{\Oh[k]{\sqrt{\log n \log \log n}}} = \Oh[k]{n \cdot 2^{o(\sqrt{\log n} \log \log n)}}$\footnote{The fact that the~data structure can be efficiently initialized with an~annotated rank decomposition $\Tc$ is not stated explicitly in \cref{the:dynsimple}, but this follows readily from the discussion in \Cref{ssec:ov:dynrw} and we actually prove this in \cref{lem:finaldynrwds}.}, maintaining a $4$-approximate decomposition, which can be finally turned into optimal decomposition by \Cref{lem:overview-get-decomposition}.
\end{proof}


We also define another auxiliary problem, {\sc Twin Detection}: construct an~efficient data structure that, when initialized with a~bipartite graph $G = (A, B, E)$ with $B = \{v_1, \dots, v_{|B|}\}$, supports the following query: given a~set $X \subseteq A$ and a~subinterval $[\ell, r]$ of $[1, |B|]$, return the partition of $X$ into the equivalence classes of twins in the induced subgraph $G[X, \{v_\ell, \dots, v_r\}]$.
In \cref{lem:twinds} we propose such a~data structure with initialization time $\Oh{n + m}$ and query time $\Oh{|X| \log n}$.
The implementation uses as a~black box a~linear-time suffix array construction algorithm of \cite{DBLP:journals/jacm/KarkkainenSB06}.

In the third and final auxiliary problem, {\sc Near-Twin Pairing}, we get as input an~annotated rank decomposition $\Tc$ of width at most $k$ of a~bipartite graph $G=(A,B,E)$ with $|B| \geq 2$.
On output we should produce: (i) $t = \max(1, \frac{|B|}{\Oh[k]{1}})$ pairwise disjoint pairs of vertices $(u_1, v_1), \dots, (u_t, v_t)$ of $B$ such that $u_i$ and $v_i$ are $\Oh[k]{\frac{|A|}{|B|}}$-near-twins for all $i \in [t]$, and (ii) the sets $N(u_i) \triangle N(v_i)$ for each $i \in [t]$.
We show in \cref{lem:neartwins} the solution of this problem in time $\Oh[k]{n}$.

We also use the following straightforward fact: If $\Tc$ is an~annotated rank decomposition of $G$ of width $k$, and a~graph $G^\star$ is created from $G$ by cloning a~vertex $v$ (creating a~new vertex $v^\star$ such that $N_{G^\star}(v^\star) = N_G(v)$), then $\Tc$ can be transformed into an annotated rank decomposition $\Tc^\star$ of $G^\star$ of the same width in time $\Oh{1}$.


The main ingredient of our algorithm is the following result:
\begin{lemma}
  \label{lem:main-thm-from-flipping}
  A~decomposition of $G$ of width at most $k$ can be found in time $\Oh[k]{T(n) \log^2 n} + \Oh{m}$.
\end{lemma}
\begin{proof}[Sketch of the proof]
  In time $\Oh{n + m}$, initialize the data structure for {\sc Twin Detection} on the input graph $G = (A,B,E)$.
  Also suppose $B = \{v_1, \dots, v_{|B|}\}$.
  We now design a~recursive algorithm that takes as input a~subset $A' \subseteq A$ and a~subinterval $[\ell, r] \subseteq [1, |B|]$, and returns an~annotated rank decomposition of width $k$ of $G[A', B']$, where $B' = \{v_\ell, \dots, v_r\}$.
  
  The base case is $\ell = r$; then the graph is a~forest and we can construct its rank decomposition of width at most $k$ in time $\Oh{|A'|}$.
  So suppose that $\ell < r$.
  We resolve this case in several steps.
  
  \medskip
  \noindent\textbf{Step 1: Filter out the twins.}
  We query the data structure for {\sc Twin Detection} on $X = A'$ and the interval $[\ell, r]$ in time $\Oh{|A'| \log n}$; the result of the query can be represented as a~subset $A'' \subseteq A'$ with no twins in $G[A'', B']$, and a~mapping $\eta \,\colon\, A' \to A''$ such that for every $v \in A' \setminus A''$, $\eta(v)$ is a~twin of $v$ in $G[A', B']$.
  Since $A''$ has no twins in $G[A'', B']$ and $G[A'', B']$ has rankwidth at most $k$, we can show that $|A''| \leq \Oh[k]{|B'|}$; this statement is proved as \cref{lem:bipartitetwins}, but has appeared before in various forms and generalizations \cite{DBLP:conf/mfcs/PaszkeP20,DBLP:journals/ejc/BonnetFLP24}.  
  Hence we will now only compute an annotated rank decomposition $\Tc$ of $G[A'', B']$ since it is straightforward to add the vertices of $A' \setminus A''$ to $\Tc$ as soon as $\Tc$ is constructed.
  
  \medskip
  \noindent \textbf{Step 2: Recurse on $B$.}
  Let $\delta \approx \frac12(\ell + r)$ and let $B_1 = \{v_\ell, \dots, v_\delta\}$ and $B_2 = \{v_{\delta+1}, \dots, v_r\}$.
  For each $i \in [2]$, we construct an~annotated rank decomposition $\Tc_i$ of $G[A'', B_i]$ recursively.
  
  \medskip
  \noindent \textbf{Step 3: Merge the decompositions.}
  The final step -- merging $\Tc_1$ and $\Tc_2$ into an~annotated rank decomposition $\Tc$ of $G[A'', B_1 \cup B_2]$ -- is quite non-trivial.
  In fact, we will perform this step recursively by implementing a~subroutine taking as input a~subset $B'_2 \subseteq B_2$ and a~rank decomposition $\Tc'_2$ of $G[A'', B'_2]$ of width at most $k$ and returning an~analogous decomposition $\Tc'$ of $G[A'', B_1 \cup B'_2]$.
  
  First, if $|B'_2| = 1$, then we model the problem at hand as an~instance of {\sc Twin Flipping} as follows: assume $B'_2 = \{v\}$. Choose an arbitrary vertex $u \in B_1$ in $G[A'', B_1]$ and clone it, naming the clone~$v$. Denote the updated graph $G^\star$ and let $\Tc^\star$ be an~annotated rank decomposition of $G^\star$. 
  Then $\Tc'$ is exactly the result of the {\sc Twin Flipping} problem for the graph $G^\star$ with sides $B_1 \cup \{v\}$ and $A''$, the decomposition $\Tc^\star$, the set $X = \{v\}$ and the set of edges $F = \{wv \mid w \in N(u) \triangle N(v)\}$.
  We can easily see that the time required to resolve case is $\Oh[k]{T(|A''| + |B_1| + 1)} = \Oh[k]{T(|B'|)}$.
  
  Now suppose $|B'_2| \geq 2$.
  Then by {\sc Near-Twin Pairing} applied to the decomposition $\Tc'_2$ of $G[A'', B'_2]$ we get $t = \max(1, \frac{|B'_2|}{\Oh[k]{1}})$ pairwise disjoint pairs of vertices $(u_i, v_i), \dots, (u_v, v_t)$ of $B$ such that $|N(u_i) \triangle N(v_i)| \leq \Oh[k]{\frac{|A''|}{|B'_2|}}$ for each $i \in [t]$.
  Therefore, $\sum_{i = 1}^t |N(u_i) \triangle N(v_i)| \leq \Oh[k]{|A''|}$.
  Let $B_2^{\rm del} = B'_2 \setminus \{v_1, \dots, v_t\}$ and $\Tc_2^{\rm del}$ be the rank decomposition of $G[A'', B_2^{\rm del}]$, easily constructed from $\Tc_2'$.
  We run the subroutine recursively for $B_2^{\rm del} \subseteq B_2$ and $\Tc_2^{\rm del}$ and get the decomposition $\Tc^{\rm del}$ of $G[A'', B_1 \cup B_2^{\rm del}]$.
  Create a~new graph $G^\star$ from $G[A'', B_1 \cup B_2^{\rm del}]$ by cloning, for each $i \in [t]$, the vertex $u_i$ and naming the clone $v_i$; let also $\Tc^\star$ be the decomposition of $G^\star$.
  Finally, let $F = \{wv_i \mid i \in [t], w \in N(u_i) \triangle N(v_i)\}$ and apply {\sc Twin Flipping} to the graph $G^\star$, its decomposition $\Tc^\star$, the set $X = \{v_1, \dots, v_t\}$ and the set of flipped edges $F$, resulting in the sought decomposition $\Tc'$.
  Tracing all the steps described above, excluding the recursive call on the subset $B_2^{\rm del}$, we find that that these steps can be performed in total time $\Oh[k]{T(|A''| + |B_1| + |B'_2|)} = \Oh[k]{T(|B'|)}$.
  
  Since each recursive call takes time $\Oh[k]{T(|B'|)}$ and $B'$ decreases in size by a multiplicative factor of $1 - \frac{1}{\Oh[k]{1}}$ on each level of recursion, we get that the recursion terminates after $\Oh[k]{\log |B'|}$ levels and so the entire decomposition-merging subroutine takes total time $\Oh[k]{T(|B'|) \log |B'|}$.
  
  \medskip
  \noindent \textbf{Summary.} The recursive reconstruction of an~annotated rank decomposition of $G[A', B']$, where $|B'| = r - \ell + 1$, takes time $\Oh[k]{T(|B'|) \log |B'|}$, excluding the time spent in the two recursive calls for subsets of $B'$.
  Thus, the total running time of the entire recursive scheme across all levels of recursion is $\Oh[k]{T(n) \log^2 n}$.
  Including the time required to instantiate the instance of {\sc Twin Detection}, we get the final time complexity of $\Oh[k]{T(n) \log ^2 n} + \Oh{m}$.
\end{proof}

So \cref{the:altrw} holds by \cref{lem:overview-twin-flipping,lem:main-thm-from-flipping}.
Moreover, an~$\Oh[k]{n \log^{\Oh{1}} n}$ time algorithm for {\sc Twin Flipping} would immediately imply an~improved $\Oh[k]{n \log^{\Oh{1}} n} + \Oh{m}$ time algorithm for finding rank decompositions of graphs of width at most $k$.

%% file: prelim.tex
\section{Preliminaries}
\label{sec:preli}
We present definitions and preliminary results in this section.

We use $\log$ to denote the base-2 logarithm.
We use $\N$ to denote the set of non-negative integers, and $\Z$ the set of all integers.
For two integers $a$ and $b$ with $a \le b$ we denote by $[a,b]$ the set of integers $\{a,\ldots,b\}$ and for $n \in \N$ we denote by $[n]$ the set $\{1,\ldots,n\}$.
For two sets $A$ and $B$, we denote their symmetric difference by $A \symd B = (A \cup B) \setminus (A \cap B)$.

\paragraph{Graphs and trees.}
For a graph $G$, we denote by $V(G)$ the set of its vertices and $E(G)$ the set of its edges.
We assume that there is a total order on the set $V(G)$, for example by representing vertices as integers.
All graphs in this paper are undirected and we normally treat edges as undirected, i.e., for $uv \in E(G)$ it holds that $uv = vu$, but we associate with $G$ the set of \emph{oriented edges} $\oE(G)$, which for every $uv \in E(G)$ contains $\ouv$ and $\ovu$ which denote, respectively, the orienting of $uv$ towards $v$ and the orienting of $uv$ towards $u$.
The set of neighbors of a vertex $v$ in $G$ is denoted by $N_G(v)$ and neighbors of a set of vertices $X$ by $N_G(X) = \bigcup_{v \in X} N_G(v) \setminus X$.
Closed neighborhoods are denoted by $N_G[v] = N_G(v) \cup \{v\}$ and $N_G[X] = N_G(X) \cup X$.
We drop the subscript if the graph is clear from the context.
We call two vertices $u$ and $v$ \emph{twins} if $N(u) = N(v)$.
A path of length $\ell \ge 1$ is an ordered sequence $v_1 v_2\ldots v_{\ell-1} v_{\ell}$ of $\ell$ distinct vertices so that any two consecutive vertices are adjacent.
We denote by $\PT(G)$ the set of paths of length $3$ in $G$.

We denote the subgraph of $G$ induced by $X \subseteq V(G)$ by $G[X]$, and the subgraph induced by $V(G) \setminus X$ by $G - X$.
When $X,Y \subseteq V(G)$ are disjoint, we denote by $G[X,Y]$ the bipartite graph with vertex set $X \cup Y$ that contains the edges of $G$ with one endpoint in $X$ and one endpoint in $Y$.
A partition of a set $X$ is a set of non-empty disjoint subsets of $X$ so that $X$ equals their union.
For a partition $\prt$ of $X$ we use the notation $\boldcup \prt = \bigcup_{C \in \prt} C = X$.  
For a graph $G$ and a partition $\prt$ of a subset of $V(G)$ we denote by $G[\prt]$ the graph with vertex set $V(G[\prt]) = \boldcup \prt$ and edge set $E(G[\prt]) = \{uv \in E(G) \mid u \in C_1 \in \prt, v \in C_2 \in \prt, C_1 \neq C_2\}$.

A tree is a connected acyclic graph.
We often call vertices of trees \emph{nodes} to distinguish them from vertices of graphs.
A subtree of a tree $T$ is a subgraph of $T$ that is connected.
\emph{Contracting} a degree-2 node in a tree means contracting one of the edges incident to it.
A leaf is a node of a tree with degree $1$, except the root of a rooted tree is never a leaf.
A cubic tree is a tree where every non-leaf node has degree $3$, and which has at least two leaves, and a subcubic tree is a tree where each node has degree at most $3$.
A binary tree is a rooted tree where each node has either $0$ or $2$ children, and which has at least two leaves.
Note that cubic trees and binary trees correspond to each other: We can make a cubic tree into a binary tree by subdividing an edge and placing the root on the subdivision vertex, and we can make a binary tree into a cubic tree by contracting the root.
In a~rooted tree, a~vertical path is a~path $x_1 x_2 \dots x_k$ where $x_{i+1}$ is a~child of $x_i$ for each $i \in [k - 1]$.

A node $x$ of a rooted tree is a descendant of a node $y$ if the unique path from $x$ to the root contains $y$.
If $x$ is a descendant of $y$, then $y$ is an ancestor of $x$.
Note that every node is both a descendant and an ancestor of itself.
An oriented edge $\vec{xy}$ of a rooted tree $T$ is directed towards the root if $y$ is the parent of $x$, and away from the root otherwise.
We say that an oriented edge $\vec{x y}$ of a tree $T$ is a \emph{predecessor} of an oriented edge $\vec{z w}$ if either $\vec{x y} = \vec{z w}$ or there is a path in $T$ between $y$ and $z$ that avoids $x$ and $w$.
The set of predecessors of $\vec{z w}$ is denoted by $\pred_T(\vec{zw})$.
If $\vec{x y}$ is a predecessor of $\vec{z w}$ then we say $\vec{z w}$ is a \emph{successor} of $\vec{x y}$.
If $\vec{x y},\vec{y z} \in \oE(T)$ with $x \neq z$, then $\vec{x y}$ is called a \emph{child} of $\vec{yz}$.

We denote by $\leafs(T)$ the set of leaves of a tree $T$, and by $\leafe(T)$ the oriented edges $\vec{lp} \in \oE(T)$ where $l$ is a leaf of $T$, which will be called \emph{leaf edges}.
For an oriented edge $\oxy \in \oE(T)$, we denote by $\leafs(T)[\oxy] \subseteq \leafs(T)$ the subset of the leaves of $T$ that are closer to $x$ than $y$.
The set $\leafe(T)[\oxy] \subseteq \leafe(T)$ is defined analogously, i.e., $\leafe(T)[\oxy] = \{\vec{lp} \in \oE(T) \mid l \in \leafs(T)[\oxy]\}$.
When $T$ is rooted and $t \in V(T)$, we use $\leafs(T)[t]$ to denote the set of leaves that are descendants of $t$.
The set $\leafe(T)[t]$ is defined analogously.

Let $\Tconn \subseteq V(T)$ be a set of nodes that induces a connected subtree of a tree $T$.
We say that a node $a \in V(T)$ is an \emph{appendix} of $\Tconn$ if $a$ is not in $\Tconn$ but a neighbor of $a$ is.
We denote the set of appendices of $\Tconn$ by $\App_T(\Tconn)$.
The oriented edges $\oxy \in \oE(T)$ with $x \in \App_T(\Tconn)$ and $y \in \Tconn$ are called the \emph{appendix edges} of $\Tconn$ and the set of them is denoted by $\oApp_T(\Tconn)$.
If $T$ is rooted and $\Tconn$ contains the root, then $\Tconn$ is called a \emph{prefix} of $T$.
The set $\Tconn$ is called \emph{leafless} if it is disjoint from $\leafs(T)$.

The \emph{height} of a node $x$ in a rooted tree $T$ is the number of nodes on a longest path from $x$ to a leaf and is denoted by $\height_T(x)$.
The height of $T$ is the height of its root.

\paragraph{Rank decompositions.}
A partitioned graph is a pair $(G,\prt)$, where $G$ is a graph and $\prt$ is a partition of $V(G)$.
A \emph{rank decomposition} of a partitioned graph $(G, \prt)$ is a pair $\Tc = (T,\lmap)$, where $T$ is a cubic tree and $\lmap$ is a bijection $\lmap \colon \prt \rightarrow \leafe(T)$.
A rank decomposition of a graph $G$ is a rank decomposition of $(G, \trivialpartition{V(G)})$, where $\trivialpartition{V(G)}$ denotes the partition of $V(G)$ into sets of size $1$.
The bijection $\lmap$ is called the \emph{leaf mapping}.
In the case of graphs, we may treat $\lmap$ as a function $\lmap \colon V(G) \rightarrow \leafe(T)$.
We define that there is no rank decomposition of a partitioned graph with less than $2$ parts or a graph with less than $2$ vertices.
For an oriented edge $\oxy \in \oE(T)$, we denote by $\lparts(\Tc)[\oxy] = \bigcup_{\olp \in \leafe(T)[\oxy]} \lmap^{-1}(\olp)$ the union of the parts of $\prt$ that are mapped to leaf edges that are closer to $x$ than $y$.
A \emph{rooted rank decomposition} of a partitioned graph is defined like a rank decomposition, but the tree $T$ is a binary tree.
When $\Tc$ is a rooted rank decomposition and $t \in V(T)$, we denote by $\lparts(\Tc)[t] = \bigcup_{\olp \in \leafe(T)[t]} \lmap^{-1}(\olp)$ the union of the parts of $\prt$ that are mapped to descendants of $t$.

Let $G$ be a graph and $A \subseteq V(G)$.
We denote $\compl{A} = V(G) \setminus A$.
We denote by $\cutrk_G(A)$ the rank of the $|A|\times |\compl{A}|$ 0-1-matrix over the binary field $\GF(2)$ describing adjacencies between vertices in $A$ and vertices in $\compl{A}$ in $G$.
The \emph{width} of an edge $xy \in E(T)$ of a rank decomposition is $\cutrk_G(\lparts(\Tc)[\oxy]) = \cutrk_G(\lparts(\Tc)[\oyx])$, and the width of a rank decomposition is the maximum width of its edge.
The rankwidth of a graph is the minimum width of a rank decomposition of it.

We will use the following properties of the $\cutrk_G$ function.

\begin{lemma}[\cite{OumS06}]
\label{lem:symsubmod}
For any graph $G$, the function $\cutrk_G \colon 2^{V(G)} \rightarrow \N$ is symmetric and submodular, that is,
\begin{enumerate}
\item\label{lem:symsubmod:sym} $\cutrk_G(A) = \cutrk_G(\compl{A})$ for all $A \subseteq V(G)$ and
\item\label{lem:symsubmod:submod} for all $A,B \subseteq V(G)$ it holds that $\cutrk_G(A \cup B) + \cutrk_G(A \cap B) \le \cutrk_G(A) + \cutrk_G(B)$.
\end{enumerate}
\end{lemma}

We will refer to \Cref{lem:symsubmod:sym} as the symmetry of $\cutrk$ and to \Cref{lem:symsubmod:submod} as the submodularity of $\cutrk$.

Let us also recall a known lemma that rank decompositions can be transformed into logarithmic height without increasing the width much.
This lemma was shown by Courcelle and Kant{\'{e}}~\cite{DBLP:conf/wg/CourcelleK07}, but we will also give a proof of it in~\Cref{sec:appendix1} in order to demonstrate the $\Oh{|V(T)| \log |V(T)|}$ running time.
\begin{restatable}{lemma}{lemlogdepthdecomp}
\label{lem:logdepthdecomp}
There is an algorithm that given a (rooted) rank decomposition $(T,\lmap)$ of a partitioned graph $(G,\prt)$ of width $k$, in time $\Oh{|V(T)| \log |V(T)|}$ returns a rooted rank decomposition of $(G,\prt)$ of height $\Oh{\log |V(T)|}$ and width at most $2k$.
\end{restatable}

In \Cref{lem:logdepthdecomp} we assume that the leaf mapping $\lmap \colon \prt \rightarrow \leafe(T)$ is represented in $\Oh{|V(T)|}$ space, for example as a mapping from pointers representing parts in $\prt$ to $\leafe(T)$.

\paragraph{Representatives.}
A \emph{representative} of a set $A \subseteq V(G)$ in a graph $G$ is a set $R_A \subseteq A$ so that for every $a \in A$ there exists $r \in R_A$ with $N_G(a) \setminus A = N_G(r) \setminus A$.
Such set $R_A$ is a \emph{minimal representative} of $A$ if no subset of it is a representative of $A$.
A cut of a graph $G$ is a pair $(A,B)$ so that $V(G)$ is the disjoint union of $A$ and $B$.
We say that vertices $u,v \in A$ are twins over a cut $(A,B)$ if $N(u) \cap B = N(v) \cap B$.
A representative graph of a cut $(A,B)$ is a bipartite graph $G[R_A,R_B]$, where $R_A$ is a representative of $A$ and $R_B$ a representative of $B$.
A minimal representative graph of a cut is defined by requiring $R_A$ and $R_B$ to be minimal representatives.
We observe that $\cutrk_{G[R_A,R_B]}(R_A) = \cutrk_{G}(A)$.

We will need the following lemma about how minimal representative graphs are isomorphic to each other.

\begin{lemma}
\label{lem:uniqisom}
Let $(A,B)$ be a cut of a graph $G$, $R_A^1,R_A^2$ minimal representatives of $A$, and $R_B^1,R_B^2$ minimal representatives of $B$.
The graphs $G[R_A^1,R_B^1]$ and $G[R_A^2,R_B^2]$ are isomorphic to each other, and moreover if $R_B^1 = R_B^2$, then there is a unique isomorphism that is identity on $R_B^1 = R_B^2$.
\end{lemma}
\begin{proof}
For every $v \in R_A^1$ there exists by definition exactly one $u \in R_A^2$ so that $N(v) \cap B = N(u) \cap B$, so we can map such $u$ and $v$ to each other, and similarly for $R_B^1$ and $R_B^2$.
This is not necessarily the only isomorphism because both sides can be permuted, e.g., when $G$ is a perfect matching between $A$ and $B$.
However, it becomes unique if we fix the mapping for one side.
\end{proof}

We also recall the following well-known lemma, which allows to make use of rank decompositions in dynamic programming.

\begin{lemma}
\label{lem:repsbound}
Let $A \subseteq V(G)$ and $R_A$ a minimal representative of $A$.
Then $|R_A| \le 2^{\cutrk_G(A)}$.
\end{lemma}
\begin{proof}
Follows from the fact that a matrix of rank $k$ over $\GF(2)$ can have at most $2^k$ distinct rows.
\end{proof}

%% file: datastructure.tex
\section{Annotated rank decompositions and prefix rebuilding}
\label{sec:annot}
In this section we introduce our notion of annotated rank decompositions and the notion of prefix-rebuilding updates to manipulate them.
Definitions comprise a large part of this section, but we also give (slightly non-trivial) proofs on the implementations of these manipulations.

\subsection{Annotated rank decompositions}
An annotated rank decomposition is a tuple $\Tc = (T, U, \reps, \repse, \dmap)$, where 
\begin{itemize}
\item $T$ is a cubic tree and $U$ is a set,
\item $\reps$ is a function that maps each oriented edge $\oxy \in \oE(T)$ to a non-empty set $\reps(\oxy) \subseteq U$,
\item $U$ is the disjoint union of the sets $\reps(\olp)$ over the leaf edges $\olp \in \leafe(T)$,
\item $\repse$ is a function that maps each edge $xy \in E(T)$ to a bipartite graph $\repse(xy)$ with bipartition $(\reps(\oxy), \reps(\oyx))$ and with no twins over this bipartition, and
\item $\dmap$ is a function that maps each path of length three $xyz \in \PT(T)$ to a \emph{representative map} $\dmap(xyz) \colon \reps(\oxy) \rightarrow \reps(\oyz)$.
\end{itemize}

For an oriented edge $\oxy \in \oE(T)$, we denote by $\lparts(\Tc)[\oxy] = \bigcup_{\olp \in \leafe(T)[\oxy]} \reps(\olp)$ the union of the elements of $U$ on the leaf edges that are closer to $x$ than $y$.
Let $(G,\prt)$ be a partitioned graph.
We define that an annotated rank decomposition $\Tc = (T, U, \reps, \repse, \dmap)$ \emph{encodes} $(G,\prt)$ if
\begin{enumerate}
\item\label[arde]{prop:ard:enclp} $\prt = \{\reps(\olp) \mid \olp \in \leafe(T)\}$, and in particular $V(G) = U$,
\item\label[arde]{prop:ard:encedge} for all $C \in \prt$ the graph $G[C]$ is edgeless,
\item\label[arde]{prop:ard:minrep} for all $\oxy \in \oE(T)$ the set $\reps(\oxy)$ is a minimal representative of $\lparts(\Tc)[\oxy]$ in $G$ and $\repse(xy) = G[\reps(\oxy), \reps(\oyx)]$, and
\item\label[arde]{prop:ard:edges} for all $xyz \in \PT(T)$ and $u \in \reps(\oxy)$ it holds that $N_G(u) \cap \reps(\ozy) = N_G(\dmap(xyz)(u)) \cap \reps(\ozy)$.
\end{enumerate}

We will call these the properties \Cref{prop:ard:enclp,prop:ard:encedge,prop:ard:minrep,prop:ard:edges}.
Let us then prove that the partitioned graph encoded by $\Tc$ is uniquely defined by $\Tc$.
The proof contains useful properties of annotated rank decompositions that will be implicitly used later.

\begin{lemma}
If an annotated rank decomposition encodes a partitioned graph, then it uniquely determines the partitioned graph it encodes.
\end{lemma}
\begin{proof}
Suppose $\Tc = (T, U, \reps, \repse, \dmap)$ encodes $(G,\prt)$.
The partition $\prt = \{\reps(\olp) \mid \olp \in \leafe(T)\}$ is uniquely defined by $\Tc$ by \Cref{prop:ard:enclp}.
Let $u,v$ be distinct vertices in $V(G) = U$.
If $u$ and $v$ are in the same part of $\prt$ then by \Cref{prop:ard:encedge} there is no edge between $u$ and $v$.

Then suppose $u \in \reps(\vec{l_1 p_1})$ and $v \in \reps(\vec{l_2 p_2})$ with $l_1 \neq l_2$.
Let $x_1 = l_1, x_2, \ldots, x_{t-1}, x_t = l_2$ be the unique path in $T$ between $l_1$ and $l_2$.
For $i \in [t-2]$ let $u_i = \dmap(x_i x_{i+1} x_{i+2}) \circ \ldots \circ \dmap(x_1 x_2 x_3)(u)$, where $\circ$ denotes the function composition.
Let us prove by induction that for every $i \in [t-2]$, it holds that
\begin{equation}
\label{lem:ardeuniq:eq1}
N_G(u) \cap \reps(\vec{x_{i+2} x_{i+1}}) = N_G(u_i) \cap \reps(\vec{x_{i+2} x_{i+1}}).
\end{equation}
For $i=1$ it holds by \Cref{prop:ard:edges}.
Then, for $i\ge 2$ we have by \Cref{prop:ard:minrep} and induction assumption that 
\[N_G(u) \cap \lparts(\Tc)[\vec{x_{i+1} x_{i}}] = N_G(u_{i-1}) \cap \lparts(\Tc)[\vec{x_{i+1} x_{i}}],\]
which implies
\[N_G(u) \cap \reps(\vec{x_{i+2} x_{i+1}}) = N_G(u_{i-1}) \cap \reps(\vec{x_{i+2} x_{i+1}})\]
because $\reps(\vec{x_{i+2} x_{i+1}}) \subseteq \lparts(\Tc)[\vec{x_{i+1} x_{i}}]$.
This yields \Cref{lem:ardeuniq:eq1} by \Cref{prop:ard:edges} by applying the function $\dmap(x_i x_{i+1} x_{i+2})$ to $u_{i-1} \in \reps(\vec{x_i x_{i+1}})$.

Now, $u_{t-2} \in \reps(\vec{p_2 l_2})$ and $u$ is adjacent to $v$ if and only if $u_{t-2}$ is adjacent to $v$, and therefore by \Cref{prop:ard:minrep} we have that $uv \in E(G)$ if and only if $u_{t-2} v \in E(\repse(p_2 l_2))$.
\end{proof}

We say that an annotated rank decomposition encodes a graph $G$ if it encodes the partitioned graph $(G,\trivialpartition{V(G)})$.

At this point, let us make a few remarks about the choices of these definitions.
We note that it would have been natural to require an additional property that $\reps(\oxy) \subseteq \reps(\ozx) \cup \reps(\vec{w x})$, where $zxy, wxy \in \PT(T)$.
However, this property turns out to be too strong in that in some cases we do not know if it could be maintained efficiently.
We also note that an equivalent alternative to storing the functions $\dmap(xyz)$ would be to store the graphs $G[\reps(\oxy), \reps(\ozy)]$: The function $\dmap(xyz)$ can be computed given $\repse(yz)$ and $G[\reps(\oxy), \reps(\ozy)]$, and conversely the graph $G[\reps(\oxy), \reps(\ozy)]$ can be computed given $\dmap(xyz)$ and $\repse(yz)$.
We choose to store $\dmap(xyz)$ because it is more explicit for the purpose of tracking representatives along the decomposition.


We define $|\Tc| = |T|$.
The width of an annotated rank decomposition is the maximum of $\cutrk_{\repse(xy)}(\reps(\oxy))$.
If the width of an annotated rank decomposition is $\ell$, then by \Cref{lem:repsbound} we have $|\reps(\oxy)| \le 2^{\ell}$ for all oriented edges $\oxy$.
It follows that an annotated rank decomposition of width $\ell$ can be represented in space $2^{\Oh{\ell}} |\Tc|$.

Let $\Tc' = (T',\lmap)$ be a rank decomposition of $(G,\prt)$.
We say that an annotated rank decomposition $\Tc = (T, V(G), \reps, \repse, \dmap)$ \emph{corresponds} to $\Tc'$ if $T = T'$ and for all $\olp \in \leafe(T)$ it holds that $\reps(\olp) = \lmap^{-1}(\olp)$.
Note that there is a unique rank decomposition of $(G,\prt)$ that the annotated rank decomposition $\Tc$ corresponds to.
We also observe that if $\Tc$ corresponds to $\Tc'$, then the widths of $\Tc$ and $\Tc'$ are equal.
When talking about annotated rank decompositions we sometimes use definitions that are defined for rank decompositions but not explicitly for annotated rank decompositions, in which case these definitions refer to the rank decomposition that the annotated rank decomposition corresponds to.

We define a \emph{rooted annotated rank decomposition} in the same way as an annotated rank decomposition, except the tree $T$ is a binary tree instead of a cubic tree.
If $r \in V(T)$ is the root and $x,y \in V(T)$ are its two children, then we require that $\reps(\vec{xr}) = \reps(\vec{ry})$, $\reps(\vec{yr}) = \reps(\vec{rx})$, and the functions $\dmap(xry)$ and $\dmap(yrx)$ are identity functions.
We observe that an annotated rank decomposition of width $\ell$ can be turned in $\Oh[\ell]{1}$ time into a corresponding rooted annotated rank decomposition, and vice versa.

We assume that the tree $T$ of an annotated rank decomposition is represented by an adjacency list and the functions $\reps$, $\repse$, $\dmap$ as tables.
For rooted annotated rank decompositions the adjacency list furthermore contains information on which adjacent node is the parent, and we also always store a pointer to the root node.
We also assume that the representation contains a table so that given $u \in U$ we can find $\olp \in \leafe(T)$ so that $u \in \reps(\olp)$ in constant time.

\subsection{Prefix-rebuilding updates}
\label{subsec:prefixrebuilding}
We will maintain a rooted annotated rank decomposition that encodes the dynamic graph $G$ we are maintaining.
All updates to the decomposition will be done via \emph{prefix-rebuilding updates}, which informally speaking change a prefix of a rooted annotated rank decomposition, but keep everything else intact.
The updates to the graph $G$ will also be made via prefix-rebuilding updates to the decomposition.
In particular, we will not maintain $G$ explicitly, but instead $G$ will be represented by the decomposition we are maintaining.

We then define prefix-rebuilding updates formally.
An update that changes a rooted annotated rank decomposition $\Tc = (T, U, \reps, \repse, \dmap)$ into another rooted annotated rank decomposition $\Tc' = (T', U, \reps', \repse', \dmap')$ is a prefix-rebuilding update if there exists a leafless prefix $\Tpref$ of $T$ and a leafless prefix $\Tpref'$ of $T'$ so that
\begin{itemize}
\item $T - \Tpref = T' - \Tpref'$,
\item for all $\oxy \in \oE(T-\Tpref)$ it holds that $\reps(\oxy) = \reps'(\oxy)$,
\item for all $\olp \in \leafe(T)$ there exists $p' \in V(T')$ so that $\vec{lp'} \in \leafe(T')$ and $\reps(\olp) = \reps'(\vec{lp'})$,
\item for all $xy \in E(T - \Tpref)$ it holds that $\repse(xy) = \repse'(xy)$, and
\item for all $xyz \in \PT(T - \Tpref)$ it holds that $\dmap(xyz) = \dmap'(xyz)$.
\end{itemize}

We say that such $\Tpref$ is the prefix of $T$ associated with the update and $\Tpref'$ the prefix of $T'$ associated with the update.
We observe that a prefix-rebuilding update never changes the partition of $U$ associated with the leaves of the decomposition.
We also note that because both $T$ and $T'$ are binary trees with the same number of leaves, $|\Tpref| = |\Tpref'|$ must hold.

The purpose of prefix-rebuilding updates will be to argue that such updates, along with re-computing various auxiliary information stored in the decomposition, can be implemented in time proportional to $|\Tpref|$ instead of time proportional to $|V(T)|$.
For example, bottom-up dynamic programming on the decomposition would need to be recomputed only for the nodes in $\Tpref'$.
Next we introduce some definitions to more formally facilitate this.

We define the tuple of \emph{annotations} of $\Tc'$ with respect to the prefix $\Tpref'$ to be the triple \[\NA(\Tc',\Tpref') = (\funrestriction{\reps'}{\oE(T') \setminus \oE(T' - \Tpref')}, \funrestriction{\repse'}{E(T') \setminus E(T'-\Tpref')}, \funrestriction{\dmap'}{\PT(T') \setminus \PT(T' - \Tpref)}).\]

Then, we say that the \emph{description} of the prefix-rebuilding update that changes $\Tc$ into $\Tc'$ is the triple
\[\prdesc = (\Tpref, T^\star, \NA(\Tc',\Tpref')),\]
where $\Tpref$ and $\NA(\Tc',\Tpref')$ are as defined above, and $T^\star = T'[\Tpref' \cup \App_{T'}(\Tpref')]$.
Note that $\Tpref' = V(T^\star) \setminus \leafs(T^\star)$ and $\leafs(T^\star) = \App_T(\Tpref) = \App_{T'}(\Tpref') = \App_{T^\star}(\Tpref')$.
We observe that the resulting rooted annotated rank decomposition $\Tc'$ is uniquely determined by $\Tc$ and $\prdesc$.
We denote $|\prdesc| = |\Tpref|$ and observe that if $\Tc'$ has width $\ell$, then $\prdesc$ can be represented in space $\Oh[\ell]{|\prdesc|}$.

Next we show that rooted annotated rank decompositions can be maintained efficiently under prefix-rebuilding updates.

\begin{lemma}
\label{lem:basicprds}
Suppose a representation of a rooted annotated rank decomposition $\Tc$ is already stored.
Then, given a description $\prdesc$ of a prefix-rebuilding update that changes $\Tc$ into $\Tc'$ of width $\ell$, the representation of $\Tc$ can be turned into a representation of $\Tc'$ in time $\Oh[\ell]{|\prdesc|}$.
\end{lemma}
\begin{proof}
Let $\Tc = (T,U,\reps,\repse,\dmap)$, $\Tc' = (T', U, \reps', \repse',\dmap')$, and $\prdesc = (\Tpref, T^\star, \NA(\Tc',\Tpref'))$.
We first use $T^\star$ to compute for all $a \in \App_T(\Tpref)$ the parent $\pi(a)$ of $a$ in $T^\star$, which is also the parent of $a$ in $T'$.
Then, we construct $T'$ by taking $T^\star$ and for each $a \in \App_T(\Tpref)$ attaching the subtree of $T$ rooted at $a$ as a child of $\pi(a)$.
This can be done by $\Oh{1}$ pointer changes for each such $a$, so we constructed $T'$ in time $\Oh{|\Tpref|}$.
In this process also the annotations in the subtrees below such appendices $a$ are preserved, so to construct the rest of the annotations of $\Tc'$ we just copy the annotations from $\NA(\Tc',\Tpref')$ in $\Oh[\ell]{|\Tpref|}$ time.
\end{proof}

\subsection{Prefix-rebuilding data structures}
To formalize the notion of a rooted annotated rank decomposition that maintains some auxiliary information under prefix-rebuilding updates, we define \emph{prefix-rebuilding data structures}.
For $\ell \in \N$, an $\ell$-prefix-rebuilding data structure with overhead $\tau$ is a data structure that maintains a rooted annotated rank decomposition $\Tc$ of width at most $\ell$ that encodes a dynamic graph $G$, and supports the following queries:
\begin{itemize}
\item $\mathsf{Initialize}(\Tc)$: Initialize the data structure with the given rooted annotated rank decomposition $\Tc$. Assumes that $\Tc$ encodes a graph $G$ and the width of $\Tc$ is at most $\ell$. Runs in time $\Oh{\tau \cdot |\Tc|}$.
\item $\mathsf{Update}(\prdesc)$: Given a description $\prdesc$ of a prefix-rebuilding update that changes $\Tc$ into $\Tc'$, apply this update to $\Tc$. Assumes that $\Tc'$ encodes a graph $G'$ and the width of $\Tc'$ is at most $\ell$. Runs in time $\Oh{\tau \cdot |\prdesc|}$.
\end{itemize}

Note that an $\ell$-prefix-rebuilding data structure with overhead $\tau = \Oh[\ell]{1}$ supporting the two queries mentioned above can be readily implemented by \Cref{lem:basicprds}.
The purpose of this definition is to give a template for data structures that implement also other queries in addition to the aforementioned two.
As an immediate example, let us give a prefix-rebuilding data structure for maintaining the $\height_T(t)$ function.

\begin{lemma}
\label{lem:utilityprds}
Let $\ell \in \N$.
There exists an $\ell$-prefix-rebuilding data structure with overhead $\Oh[\ell]{1}$ that maintains a rooted annotated rank decomposition $\Tc = (T,V(G),\reps,\repse,\dmap)$ that encodes a dynamic graph $G$, and additionally supports the following query:
\begin{itemize}
\item $\mathsf{Height}(t)$: Given a node $t \in V(T)$, returns $\height_T(t)$ in time $\Oh{1}$.
\end{itemize}
\end{lemma}
\begin{proof}
In the $\mathsf{Initialize}(\Tc)$ query we compute $\height_T(t)$ by bottom-up dynamic programming for every node $t \in V(T)$.
This runs in $\Oh[\ell]{|\Tc|}$ time.
The $\mathsf{Update}(\prdesc)$ query is implemented by first using \Cref{lem:basicprds} to construct $\Tc' = (T',U,\reps',\repse',\dmap')$, and then computing $\height_{T'}(t)$ for all $t \in \Tpref'$ by bottom-up dynamic programming, where $\Tpref'$ is the prefix of $T'$ associated with the update.
This runs in $\Oh[\ell]{|\prdesc|}$ time.
Then the $\mathsf{Height}(t)$ query can be implemented by simply returning the already stored height of the node $t$.
\end{proof}

Let us then clarify our assumptions about prefix-rebuilding data structures.
We assume that the stored decomposition $\Tc$ always encodes a graph.
We also assume that the data structure explicitly represents the current decomposition $\Tc$ at all times, so that we can access it and for example retrieve a copy of $\Tc$ in $\Oh[\ell]{|\Tc|}$ time.

As the final lemma of this subsection we give a prefix-rebuilding data structure for making certain straightforward manipulations of descriptions of prefix-rebuilding updates.

\begin{lemma}
\label{lem:prdscompose}
Let $\ell \in \N$.
There exists an $\ell$-prefix-rebuilding data structure with overhead $\Oh[\ell]{1}$ that maintains a rooted annotated rank decomposition $\Tc$ that encodes a dynamic graph $G$, and additionally supports the following queries:
\begin{itemize}
\item $\mathsf{Reverse}(\prdesc)$: Given a description $\prdesc$ of a prefix-rebuilding update that changes $\Tc$ into $\Tc'$, return a description of a prefix-rebuilding update that changes $\Tc'$ into $\Tc$. Runs in time $\Oh[\ell]{|\prdesc|}$.
\item $\mathsf{Compose}(\prdesc_1,\prdesc_2)$: Given two descriptions of prefix-rebuilding updates, $\prdesc_1$ and $\prdesc_2$, so that $\prdesc_1$ changes $\Tc$ into $\Tc'$ and $\prdesc_2$ changes $\Tc'$ into $\Tc''$ and both $\Tc'$ and $\Tc''$ have width at most $\ell$, return a description of a prefix-rebuilding update that changes $\Tc$ into $\Tc''$. Runs in time $\Oh[\ell]{|\prdesc_1| + |\prdesc_2|}$.
\end{itemize}
\end{lemma}
\begin{proof}
We maintain $\Tc = (T,V(G),\reps,\repse,\dmap)$ by using \Cref{lem:basicprds}.

The $\mathsf{Reverse}$ query is implemented as follows.
Let $\prdesc = (\Tpref, T^\star, \NA(\Tc',\Tpref'))$, where $\Tpref' = V(T^\star) \setminus \leafs(T^\star)$.
We observe that now
\[\prdesc^r = (\Tpref', T[\Tpref \cup \App_T(\Tpref)], \NA(\Tc, \Tpref))\]
is a description of a prefix-rebuilding update that changes $\Tc'$ into $\Tc$, and it can be computed from $\prdesc$ and $\Tc$ in $\Oh[\ell]{|\prdesc|}$ time.

The $\mathsf{Compose}$ query is implemented as follows.
Let $\prdesc_1 = (\Tpref^1, T^\star_1, \NA(\Tc',\Tpref^{1'}))$ and $\prdesc_2 = (\Tpref^2, T^\star_2, \NA(\Tc'', \Tpref^{2'}))$, where $\Tpref^{i'} = V(T^\star_i) \setminus \leafs(T^\star_i)$ for $i \in [2]$.
Let $\Tpref^{\circ} = \Tpref^1 \cup (\Tpref^2 \setminus \Tpref'^1)$ and $\Tpref^{\circ'} = \Tpref^{2'} \cup (\Tpref^{1'} \setminus \Tpref^2)$.

We use the $\mathsf{Reverse}$ query to compute a description $\prdesc^r_1$ that turns $\Tc'$ into $\Tc$, then we use \Cref{lem:basicprds} to turn $\Tc$ into $\Tc'$, then again $\mathsf{Reverse}$ to compute a description $\prdesc^r_2$ that turns $\Tc''$ into $\Tc'$, and then \Cref{lem:basicprds} to turn $\Tc'$ into $\Tc'' = (T'', V(G), \reps'', \repse'', \dmap'')$.
This runs in time $\Oh[\ell]{|\prdesc_1| + |\prdesc_2|}$.
Then, we observe that
\[\prdesc_{\circ} = (\Tpref^{\circ}, T''[\Tpref^{\circ'} \cup \App_T(\Tpref^{\circ'})], \NA(\Tc'', \Tpref^{\circ'}))\]
is a description of a prefix-rebuilding update that changes $\Tc$ into $\Tc''$.
We can compute $\prdesc_{\circ}$ from $\Tc''$, $\prdesc_1$, and $\prdesc_2$ in $\Oh[\ell]{|\Tpref^{\circ}| + |\Tpref^{\circ'}|} = \Oh[\ell]{|\prdesc_1| + |\prdesc_2|}$ time.
We return $\prdesc_{\circ}$ and finally turn $\Tc''$ back into $\Tc$ with $\prdesc^r_1$ and $\prdesc^r_2$.
\end{proof}

\subsection{Prefix-rearrangement descriptions}
In our algorithm we wish to re-arrange rooted annotated rank decompositions by prefix-rebuilding updates without worrying about the details on what happens to the annotations $\reps$, $\repse$, and $\dmap$ stored in them.
In this subsection we show that prefix-rebuilding updates that are described without the tuple of new annotations and which do not change the graph encoded by the decomposition can be efficiently turned into prefix-rebuilding updates with descriptions as defined in \Cref{subsec:prefixrebuilding}.
In particular, we introduce \emph{prefix-rearrangement descriptions} as a more high-level versions of descriptions of prefix-rebuilding updates, and show that they can be turned efficiently into descriptions of prefix-rebuilding updates.

Let $\Tc = (T,U,\reps,\repse,\dmap)$ be a rooted annotated rank decomposition that encodes a graph $G$.
We define that a prefix-rearrangement description is a pair $\prdesc = (\Tpref,T^\star)$, where $\Tpref$ is a leafless prefix of $T$ and $T^\star$ is a binary tree with $\leafs(T^\star) = \App_T(\Tpref)$.
A prefix-rebuilding update \emph{corresponds} to $(\Tpref,T^\star)$ if it changes $\Tc$ into a rooted annotated rank decomposition $\Tc' = (T',U,\reps',\repse',\dmap')$ so that
\begin{itemize}
\item $\Tc'$ encodes $G$,
\item $\Tpref$ is the prefix of $T$ associated with the update, and
\item $T^\star = T'[\Tpref' \cup \App_{T'}(\Tpref')]$, where $\Tpref'$ is the prefix of $T'$ associated with the update.
\end{itemize}

In other words, a prefix-rearrangement description is like a prefix-rebuilding description but it does not contain the triple of new annotations, and it is required to maintain the graph $G$ encoded by the decomposition.
It can be observed that $\Tc$ and the prefix-rearrangement description uniquely determine the resulting tree $T'$, and in particular the rank decomposition to which $\Tc'$ corresponds, but not necessarily the annotations in $\Tc'$.

We again denote $|\prdesc| = |\Tpref|$.
The rest of this subsection is devoted to showing that given a prefix-rearrangement description $\prdesc$, a description of a prefix-rebuilding update that corresponds to $\prdesc$ can be computed in $\Oh[\ell]{|\prdesc| \log |\prdesc|}$ time, where $\ell$ is the maximum of the widths of $\Tc$ and $\Tc'$.
We start with several auxiliary lemmas.
In these lemmas we mostly manipulate unrooted annotated rank decompositions.

We first observe that we can efficiently remove leaves from an annotated rank decomposition.

\begin{lemma}
\label{lem:dropleafs}
There is an algorithm that given an annotated rank decomposition $\Tc$ of width $\ell$ that encodes a partitioned graph $(G,\prt)$ and a subset $\prt' \subseteq \prt$ with $|\prt'| \ge 2$, in time $\Oh[\ell]{|\Tc|}$ returns an annotated rank decomposition of width at most $\ell$ that encodes $(G[\prt'],\prt')$.
\end{lemma}
\begin{proof}
Let $\Tc = (T,V(G),\reps,\repse,\dmap)$ and $G' = G[\prt']$.
We will construct an annotated rank decomposition $\Tc'$ that encodes $(G',\prt')$ and has width at most $\ell$.

First we construct for all $\oxy \in \oE(T)$ a set $\reps'''(\oxy) \subseteq \lparts(\Tc)[\oxy] \cap V(G')$ that is a minimal representative of $\lparts(\Tc)[\oxy] \cap V(G')$ in $G$, along with functions $\phi(\oxy) \colon \reps'''(\oxy) \rightarrow \reps(\oxy)$ that satisfy $N_G(u) \cap \reps(\oyx) = N_G(\phi(\oxy)(u)) \cap \reps(\oyx)$ for all $u \in \reps'''(\oxy)$.
These can be computed by dynamic programming that follows the mapping $\dmap$ by two depth-first searches on $\Tc$, first computing for edges pointing towards an arbitrarily chosen root, and second for edges pointing away from the root.
Then we construct the graph $\repse'''(xy) = G[\reps'''(\oxy), \reps'''(\oyx)]$ for each $xy \in E(T)$ from $\repse(xy)$ with the help of the functions $\phi(\oxy)$ and $\phi(\oyx)$.
Then, by using $\repse'''(xy)$ we can compute for each $\oxy$ a subset $\reps''(\oxy) \subseteq \reps'''(\oxy)$ that is a minimal representative of $\lparts(\Tc)[\oxy] \cap V(G')$ in $G'$ (instead of $G$).
We also compute the graphs $\repse''(xy) = G[\reps''(\oxy), \reps''(\oyx)]$ for each $xy \in E(T)$.

We also construct for all $xyz \in \PT(T)$ a function $\dmap''(xyz) \colon \reps''(\oxy) \rightarrow \reps''(\oyz)$ so that for all $u \in \reps''(\oxy)$ it holds that $N_{G'}(u) \cap \reps''(\ozy) = N_{G'}(\dmap(xyz)(u)) \cap \reps''(\ozy)$.
This can be constructed by first using $\phi(\oxy)$ and $\dmap(xyz)$ to compute $v \in \reps(\oyz)$ so that $N_G(v) \cap \reps(\ozy) = N_G(u) \cap \reps(\ozy)$ and then using $v$, $\repse(yz)$, and $\phi(\ozy)$ to compute $N_{G'}(u) \cap \reps''(\ozy)$.

We observe that $\Tc'' = (T,V(G'), \reps'',\repse'',\dmap'')$ almost satisfies all the properties required to be an annotated rank decomposition that encodes $G'$: the only issue is that some of the sets $\reps''(\oxy)$ can be empty.
We construct $\Tc'$ from $\Tc''$ by deleting all edges $xy$ where either $\reps''(\oxy)$ or $\reps''(\oyx)$ is empty, deleting all thus created isolated nodes, and finally contracting all degree-2 nodes.
Note that the annotations can be modified in a straightforward way when contracting.

Because $\repse''(xy)$ is isomorphic to an induced subgraph of $\repse(xy)$ for all $xy \in E(T)$, the width of the resulting decomposition is at most the width of $\Tc$.
Also, because $|\reps(\oxy)| \le 2^{\ell}$ for all $\oxy \in \oE(T)$, the algorithm can be implemented in $\Oh[\ell]{|\Tc|}$ time.
\end{proof}

Then, we will observe that a~certain type of induced subgraph finding problem can be solved by dynamic programming on annotated rank decompositions.
Let $G$ and $H$ be graphs, and $\gamma$ a function $\gamma \colon V(G) \rightarrow 2^{V(H)}$.
We say that $H$ is a \emph{labeled induced subgraph} of $(G,\gamma)$ if $G$ has an induced subgraph $G[X]$ so that $G[X]$ is isomorphic to $H$ with an isomorphism $\phi \colon X \rightarrow V(H)$ so that $\phi(x) \in \gamma(x)$ for all $x \in X$.
The pair $(X,\phi)$ will be called the \emph{witness} of the labeled induced subgraph.
The following lemma will be proven in \Cref{subsec:appendixcmso} by encoding the problem in $\CMSO_1$ logic.

\begin{restatable}{lemma}{lemlabelindsub}
\label{lem:labelindsub}
There is an algorithm that given an annotated rank decomposition $\Tc$ of width $\ell$ that encodes a partitioned graph $(G,\prt)$, a graph $H$, and a function $\gamma \colon V(G) \rightarrow 2^{V(H)}$, in time $\Oh[\ell,H]{|\Tc|}$ either returns a witness of $H$ as a labeled induced subgraph of $(G,\gamma)$ or returns that $(G,\gamma)$ does not contain $H$ as a labeled induced subgraph.
\end{restatable}

Then we need an algorithm that, given an annotated rank decomposition that encodes a partitioned graph $(G,\prt)$ and a vertex $v \in V(G)$, outputs $N_G(v)$.

\begin{lemma}
\label{lem:getnbs}
There is an algorithm that, given an annotated rank decomposition $\Tc$ of width $\ell$ that encodes a partitioned graph $(G,\prt)$ and a vertex $v \in V(G)$, in time $\Oh[\ell]{|\Tc|}$ returns $N_G(v)$.
\end{lemma}
\begin{proof}
We run a depth-first search that starts at the leaf edge $\olp \in \leafe(T)$ with $v \in \reps(\olp)$, and for each successor $\oxy$ of $\olp$ computes $u \in \reps(\oxy)$ that represents $v$ by following the mapping $\dmap$ along the depth-first search.
After this, the neighbors of $v$ can be determined from the graphs $\repse(l'p')$ of the leaf edges $\vec{l'p'} \in \leafe(T)$.
As $|\reps(\oxy)| \le 2^{\ell}$ for all $\oxy \in \oE(T)$, both steps take $\Oh[\ell]{|\Tc|}$ time.
\end{proof}

The following lemma will be the main lemma towards the main algorithm of this subsection.
It performs the update in the setting when the prefix-rearrangement description completely describes the new tree.
After that, we will reduce the general case to this.

\begin{lemma}
\label{lem:rearangmain}
There is an algorithm that given an annotated rank decomposition $\Tc'$ of width at most $\ell$ that encodes a partitioned graph $(G,\prt)$ and a rank decomposition $(T,\lmap)$ of $(G,\prt)$ of width at most $\ell$, in time $\Oh[\ell]{|V(T)| \log |V(T)|}$ returns an annotated rank decomposition $\Tc$ that encodes $(G,\prt)$ and corresponds to $(T,\lmap)$.
\end{lemma}
\begin{proof}
The idea of the algorithm will be to work recursively by picking an edge $xy \in E(T)$ that corresponds to a balanced cut between the leaves of $T$, then in time $\Oh[\ell]{|V(T)|}$ computing a minimal representative of the cut $(\lparts(T,\lmap)[\oxy], \lparts(T,\lmap)[\oyx])$ of $G$, then recursively constructing annotated rank decompositions on both sides of this cut, and finally combining them.
To make this idea work, we need to keep the ``boundary'' of the subtree of $T$ that we are currently working on small, and explicitly encode all adjacencies from a representative of the boundary to all other vertices.

More formally, we define a \emph{decomposition-boundary-pair}:
Let $\Tc'$ be an annotated rank decomposition that encodes a partitioned graph $(G',\prt')$, $B$ a graph with $V(G') \subseteq V(B)$ so that $B[V(G')]$ is edgeless, and $\prtb$ a partition of $V(B) \setminus V(G')$ so that for all $C \in \prtb$ the graph $B[C]$ is edgeless.
We call the pair $(B, \prtb)$ a \emph{boundary representation} and the pair $(\Tc', (B, \prtb))$ a decomposition-boundary-pair.
The pair $(\Tc', (B, \prtb))$ encodes a partitioned graph $(G,\prt)$ where $V(G) = V(B)$, $E(G) = E(B) \cup E(G')$, and $\prt = \prtb \cup \prt'$.
In particular, the edges in the subgraph induced by $V(G')$ come from $\Tc'$, and the other edges come from $B$.
Note that we allow $V(G') = V(B)$, in which case $\prtb = \emptyset$ and $B$ is edgeless.

Then we give our algorithm.
We will describe a recursive algorithm that takes as input
\begin{itemize}
\item an annotated rank decomposition $\Tc'$ of width at most $\ell$ and a boundary representation $(B, \prtb)$ with $|\prtb| \le 4$ and $|C| \le 2^{\ell}$ for all $C \in \prtb$, so that the decomposition-boundary-pair $(\Tc', (B,\prtb))$ encodes a partitioned graph $(G,\prt)$, where $|\prt| \ge 2$, and
\item a rank decomposition $(T,\lmap)$ of $(G,\prt)$ of width at most $\ell$,
\end{itemize}
and outputs
\begin{itemize}
\item an annotated rank decomposition $\Tc$ that encodes $(G,\prt)$ and corresponds to $(T,\lmap)$.
\end{itemize}
The base case is that $|\prt| \le 3$.
In this case $|V(G)| \le 3 \cdot 2^{\ell}$, so we can first explicitly construct $(G,\prt)$ from $(\Tc', (B,\prtb))$, and then from $(G,\prt)$ and $(T,\lmap)$ construct an annotated rank decomposition $\Tc$ that corresponds to $(T,\lmap)$ in a straightforward way in time $\Oh[\ell]{1}$.


Then we consider the case when $|\prt| \ge 4$.
Let us first pick the edge of $T$ along which we do recursion.
We say that a leaf of $T$ is a \emph{boundary leaf} if it corresponds to a part of $\prt$ that is in $\prtb$.
By our assumption there are at most $4$ boundary leaves.
If there are exactly $4$ boundary leaves, we pick $xy \in E(T)$ so that both $\leafs(T)[\oxy]$ and $\leafs(T)[\oyx]$ contain $2$ boundary leaves (note that this can always be done by a walking argument on the decomposition).
Otherwise, we pick $xy \in E(T)$ so that $|\leafs(T)[\oxy]| \le \frac{2}{3} |\leafs(T)|$ and $|\leafs(T)[\oyx]| \le \frac{2}{3} |\leafs(T)|$ (this can also be done by a similar argument).
In both of the cases, such $xy$ can be found in time $\Oh{|V(T)|}$.

Let $(X,Y) = (\lparts(T,\lmap)[\oxy],\lparts(T,\lmap)[\oyx])$ be the cut of $G$ corresponding to $xy$.
Next we compute a minimal representative $(R_X, R_Y)$ of $(X,Y)$.
Such a representative corresponds to a largest set of vertices $R \subseteq V(G)$ so that in the graph $G[R \cap X, R \cap Y]$ there are no twins over the bipartition $(R \cap X, R \cap Y)$.
Because the width of $(T,\lmap)$ is at most $\ell$, we have by \Cref{lem:repsbound} that $|R_X|,|R_Y| \le 2^{\ell}$.
Therefore, we compute such largest $R$ by a combination of brute-force and \Cref{lem:labelindsub}: We guess a graph isomorphic to $G[R \cap X, R \cap Y]$ and how the vertices in $\boldcup \prtb$ are mapped into this graph.
Then we use the graph $B$ and the cut $(X,Y)$ to compute for each vertex in $V(G')$ how it could be mapped to this graph so that it is consistent with the already guessed mapping, and based on that construct an instance of labeled induced subgraph and apply \Cref{lem:labelindsub} with $\Tc'$ to find such $R \subseteq V(G)$.
Note that multiple such $R$ could exist, but we pick arbitrarily a single one found by this procedure.
As $|R| \le 2 \cdot 2^{\ell}$ and $|\boldcup \prtb| \le 4 \cdot 2^{\ell}$, the running time of this step is $\Oh[\ell]{|V(T)|}$.

Then we describe the recursive call.
We describe the call only for the $X$-side of the cut, but it is analogous for the side of $Y$, with notation using $Y$ in the subscript instead of $X$.
Let $\prt_X$ be the partition obtained from $\prt$ by first removing all parts that are subsets of $Y$, and then inserting the part $R_Y$. Let also $G_X = G[X \cup R_Y]$.
Because $X$ and $Y$ are non-empty, we have that $|\prt_X| \ge 2$.
Then, a rank decomposition $(T_X, \lmap_X)$ of $(G_X, \prt_X)$ is obtained from $(T,\lmap)$ by cutting along $xy$, taking the side with $X$ in the leaves, and mapping $R_Y$ to the new leaf created by this cutting, and all other parts of $\prt_X$ to the same leaves they were previously mapped.
The new leaf to which $R_Y$ is mapped will be called $y$, so $T_X$ is an induced subgraph of $T$.
(Similarly, $T_Y$ is an induced subgraph of $T$, with $V(T_X) \cap V(T_Y) = \{x,y\}$.)
Because $R_Y$ is a representative of $Y$ it follows that the width of $(T_X,\lmap_X)$ is at most $\ell$.
Both $\prt_X$ and $(T_X,\lmap_X)$ can be constructed in $\Oh[\ell]{|V(T)|}$ time.

We will recursively call the algorithm with $(T_X,\lmap_X)$, and for this we must construct a decompo\-sition-boundary-pair that encodes $(G_X, \prt_X)$.
To deal with technicalities, if $|\prt_X| \le 5$, we actually do not apply a recursive call but instead construct $(G_X, \prt_X)$ explicitly in time $\Oh[\ell]{|V(T)|}$ by using \Cref{lem:getnbs}, and construct the annotations for $(T_X,\lmap_X)$ in a straightforward way in time $\Oh[\ell]{1}$.
Then, assume $|\prt_X| \ge 6$.
The new boundary representation $(B_X,\prtb_X)$ is constructed by first removing all vertices in $Y$ from $B$ and from all sets in $\prtb$, then inserting to $\prtb$ the set $R_Y$ as a new part, and then inserting to the graph $B$ the vertices $R_Y$ and all edges between $R_Y$ and $X$, which can be computed in time $\Oh[\ell]{|V(T)|}$ by \Cref{lem:getnbs} and the fact that $|R_Y| \le 2^{\ell}$.
The fact that $|R_Y| \le 2^{\ell}$ also implies that the assumption that all parts of $\prtb_X$ have size at most $2^{\ell}$ holds.
We also have to argue that $|\prtb_X| \le 4$.
If $|\prtb| \le 3$, this holds by the fact that we inserted only one new part.
If $|\prtb| = 4$, then by the selection of $xy$ there are two parts of $\prtb$ that are subsets of $Y$, and in fact in this case we have $|\prtb_X| \le 3$.
We will use this fact also later in the analysis of the overall time complexity.
The annotated rank decomposition $\Tc'_X$ of the decompo\-sition-boundary-pair is constructed from $\Tc'$ in $\Oh[\ell]{|V(T)|}$ time by applying \Cref{lem:dropleafs}, in particular, by deleting the parts that are subsets of $Y$.
Here we use $|\prt_X| \ge 6$ to guarantee that $\Tc'_X$ has at least two leaves.
We then observe that $(\Tc'_X, (B_X,\prtb_X))$ is a decomposition-boundary pair that encodes $(G_X,\prt_X)$ and satisfies all assumptions required by the recursion.

Then, let $\Tc_X = (T_X,V(G_X),\reps_X,\repse_X, \dmap_X)$ and $\Tc_Y = (T_Y, V(G_Y),\reps_Y, \repse_Y, \dmap_Y)$ be the annotated rank decompositions obtained by the recursive calls.
We describe the construction of the annotated rank decomposition $\Tc = (T,V(G),\reps,\repse,\dmap)$.
First, for every $\oed \in \oE(T_X) \setminus \{\oxy\}$ we set $\reps(\oed) \coloneqq \reps_X(\oed)$, and for every $\oed \in \oE(T_Y) \setminus \{\oyx\}$ we set $\reps(\oed) \coloneqq \reps_Y(\oed)$.
Observe that this sets representatives for all oriented edges of $T$, and that $\reps(\oxy) = R_X$ and $\reps(\oyx) = R_Y$.
Then, for every $e \in E(T_X) \setminus \{xy\}$ we set $\repse(e) \coloneqq \repse_X(e)$ and for every $e \in E(T_Y) \setminus \{xy\}$ we set $\repse(e) \coloneqq \repse_Y(e)$.
We set $\repse(xy) \coloneqq G[R_X, R_Y]$, which can be computed in time $\Oh[\ell]{|V(T)|}$ by \Cref{lem:getnbs}.

At this point, we note that from the fact that $(R_X,R_Y)$ is a minimal representative of $(X,Y)$, and by induction on the recursion, it follows that for all $\oab \in \oE(T)$, the set $\reps(\oab)$ is a minimal representative of $\lparts(\Tc)[\oab]$, and that for all $ab \in E(T)$, $\repse(ab) = G[\reps(\oab), \reps(\oba)]$.
In particular, $\Tc$ satisfies the property \Cref{prop:ard:minrep}.
Also the properties \Cref{prop:ard:enclp,prop:ard:encedge} are clearly satisfied.

Then we construct $\dmap$.
First, for every $abc \in \PT(T_X)$ so that $c \neq y$ we have $\reps_X(\oab) = \reps(\oab)$ and $\reps_X(\obc) = \reps(\obc)$, so we can set $\dmap(abc) \coloneqq \dmap_X(abc)$.
Analogously, for every $abc \in \PT(T_Y)$ so that $c \neq x$ we set $\dmap(abc) \coloneqq \dmap_Y(abc)$.
Then consider $txy \in \PT(T_X)$ for arbitrary such $t \in V(T_X)$.
We have that $\repse_X(xy) = G[\reps_X(\oxy), R_Y]$ and $\reps_X(\oxy)$ is a minimal representative of $\lparts(\Tc)[\oxy]$.
By \Cref{lem:uniqisom}, $\repse_X(xy)$ is isomorphic to $G[R_X, R_Y]$ with an isomorphism that is identity on on $R_Y$, and such an isomorphism is unique.
We find such isomorphism $\phi \colon \reps_X(\oxy) \cup R_Y \rightarrow R_X \cup R_Y$ in $\Oh[\ell]{1}$ time.
Then we construct $\dmap(txy)$ by letting $\dmap(txy)(r) = \phi(\dmap_X(txy)(r))$ for all $r \in \reps(\vec{tx})$.
For $tyx \in \PT(T_Y)$ we construct $\dmap(tyx)$ analogously.

It can be observed that this construction can be implemented in $\Oh[\ell]{|V(T)|}$ time.
It remains to show that the constructed annotated rank decomposition $\Tc$ indeed encodes $(G,\prt)$ and corresponds to $(T,\lmap)$.
We observe that by construction $\Tc$ corresponds to $(T,\lmap)$.
For showing that $\Tc$ encodes $(G,\prt)$, it remains to show \Cref{prop:ard:edges}.

\begin{claim}
For all $abc \in \PT(T)$ and $u \in \reps(\oab)$, we have $N_G(u) \cap \reps(\ocb) = N_G(\dmap(abc)(u)) \cap \reps(\ocb)$.
\end{claim}
\begin{claimproof}
For all $abc \in \PT(T)$ except of form $abc \in \{txy,tyx\}$ this holds because it holds for $\Tc_X$ and $\Tc_Y$ and the graphs $G_X$ and $G_Y$ are induced subgraphs of $G$.

Then consider the case $abc = txy$.
Recall that $R_Y = \reps(\oyx) = \reps_X(\oyx)$ and let $\phi$ be the unique isomorphism from $\repse_X(xy) = G[\reps_X(\oxy),R_Y]$ to $G[R_X, R_Y]$ that is identity on $R_Y$.
We have that
\begin{align*}
N_G(u) \cap R_Y &= N_G(\dmap_X(txy)(u)) \cap R_Y && \text{(by \Cref{prop:ard:edges} on $\Tc_X$)}\\
&=N_G(\phi(\dmap_X(txy)(u))) \cap R_Y && \text{(by isomorphism)}\\
&=N_G(\dmap(txy)(u)) \cap R_Y && \text{(by construction)}
\end{align*}
The case of $abc = tyx$ is similar.
\end{claimproof}
This concludes the proof that the output of the algorithm is as claimed.

Then we analyze the time complexity of the algorithm.
We already analyzed that a single recursive call takes $\Oh[\ell]{|V(T)|}$ time.
It remains to observe that if $|\prtb| = 4$, then in the child calls it holds that $|\prtb_X|,|\prtb_Y| \le 3$, and that if $|\prtb| \le 3$, then in the child calls it holds that $|\leafs(T_X)|,|\leafs(T_Y)| \le \frac{2}{3} |\leafs(T)|+1$.
Because $T$ is a cubic tree we have $|V(T)| = \Oh{|\leafs(T)|}$, and therefore a standard analysis of divide-and-conquer algorithms gives the total time complexity $\Oh[\ell]{|V(T)| \log |V(T)|}$.
\end{proof}

Then we will present one more auxiliary lemma that will be used in reducing the general case to the case of \Cref{lem:rearangmain}.

Let $\Tc = (T,V(G),\reps,\repse,\dmap)$ be a rooted annotated rank decomposition that encodes a partitioned graph $(G,\prt)$.
Given a leafless connected node set $\Tconn \subseteq V(T) \setminus \leafs(T)$, we denote by $\reppart(\Tc,\Tconn)$ the partition $\{\reps(\oap) \mid \oap \in \oApp_T(\Tconn)\}$ naturally associated with the appendix edges of $\Tconn$.
Obtaining an annotated rank decomposition of the partitioned graph $(G[\reppart(\Tc,\Tconn)], \reppart(\Tc,\Tconn))$ from $\Tc$ is almost straightforward, but we have to deal with a technical issue arising from the fact that some representatives in $\Tc$ inside the subtree $T[\Tconn]$ are not necessarily in $\boldcup \reppart(\Tc,\Tconn)$.

\begin{lemma}
\label{lem:annotdecomppart}
Let $\Tc = (T,V(G),\reps,\repse,\dmap)$ be a rooted annotated rank decomposition of width $\ell$ that encodes $(G,\prt)$ and whose representation is already stored.
There is an algorithm that given a leafless connected node set $\Tconn \subseteq V(T) \setminus \leafs(T)$, in time $\Oh[\ell]{|\Tconn|}$ returns an annotated rank decomposition of width at most $\ell$ that encodes the partitioned graph $(G[\reppart(\Tc,\Tconn)], \reppart(\Tc,\Tconn))$.
\end{lemma}
\begin{proof}
We denote $(G',\prt) = (G[\reppart(\Tc,\Tconn)], \reppart(\Tc,\Tconn))$.

Let $T' = T[\Tconn \cup \App_T(\Tconn)]$ be the subtree of $T$ induced by $\Tconn$ and its appendices.
We construct $\Tc' = (T', V(G'), \funrestriction{\reps}{\oE(T')}, \funrestriction{\repse}{E(T')}, \funrestriction{\dmap}{\PT(T')})$ in a straightforward way in $\Oh[\ell]{|\Tconn|}$ time.
It can be observed that $\Tc'$ is almost an annotated rank decomposition that encodes $(G',\prt)$: the only issue is that some representatives are not from the set $V(G')$.
This issue can be fixed by finding for every representative $u \in \reps(\oxy)$ with $u \notin V(G')$ a representative $u' \in V(G')$ with $N(u') \cap \lparts(\Tc)[\oyx] = N(u) \cap \lparts(\Tc)[\oyx]$, and replacing $u$ with $u'$ in $\reps(\oxy)$, $\repse(xy)$, and in the representative maps that concern the edge $\oxy$.
This can be done in $\Oh[\ell]{|\Tconn|}$ time by a 2-phase dynamic programming that first finds such representatives $u'$ on oriented edges pointing towards the root, and then on oriented edges pointing towards the leaves.
Finally, it is straightforward to turn the obtained rooted annotated rank decomposition into unrooted.
\end{proof}

Then we give the main algorithm of this subsection.

\begin{lemma}
\label{lem:prdsrearangmain}
There exists an $\ell$-prefix-rebuilding data structure that maintains a rooted annotated rank decomposition $\Tc$ and additionally supports the following query:
\begin{itemize}
\item $\mathsf{Translate}(\Tpref,T^\star)$: Given a prefix-rearrangement description $(\Tpref,T^\star)$ on the decomposition $\Tc$, in time $\Oh[\ell,\ell']{|\Tpref| \log |\Tpref|}$ returns a description of a corresponding prefix-rebuilding update, where $\ell'$ is the width of the resulting rooted annotated rank decomposition.
\end{itemize}
\end{lemma}
\begin{proof}
We maintain the rooted annotated rank decomposition $\Tc = (T,V(G),\reps,\repse,\dmap)$ that encodes a graph $G$ by making use of \Cref{lem:basicprds}.
It remains to describe how the $\mathsf{Translate}(\Tpref,T^\star)$ query is implemented.
Throughout the proof we will use $\pi \colon \oApp_T(\Tpref) \rightarrow \leafe(T^\star)$ to denote the bijection that maps an appendix edge $\oap \in \oApp_T(\Tpref)$ to the corresponding edge $\vec{ap'} \in \leafe(T^\star)$.

Consider the partitioned graph $(G^\star,\prt^\star) = (G[\reppart(\Tc,\Tpref)],\reppart(\Tc,\Tpref))$.
We apply \Cref{lem:annotdecomppart} to obtain an annotated rank decomposition $\Tc^{\star \star}$ of width at most $\ell$ that encodes $(G^\star,\prt^\star)$.
Then let $\lmap^\star \colon \prt^\star \rightarrow \leafe(T^\star)$ be the function that maps $\reps(\oap)$ to $\pi(\oap)$ for all appendix edges $\oap \in \oApp_T(\Tpref)$.
We observe that $(T^\star, \lmap^\star)$ is a rooted rank decomposition of $(G^\star,\prt^\star)$ of width at most $\ell'$.
Then we apply \Cref{lem:rearangmain} with $\Tc^{\star \star}$ and $(T^\star, \lmap^\star)$ to obtain a rooted annotated rank decomposition $\Tc^{\star} = (T^\star,V(G^\star),\reps^\star,\repse^\star,\dmap^\star)$ that encodes $(G^\star,\prt^\star)$ and corresponds to $(T^\star,\lmap^\star)$.
Note that even though \Cref{lem:rearangmain} works with unrooted decompositions, it is simple to make it work for rooted decompositions by unrooting $(T^\star, \lmap^\star)$ before applying it and then rooting the returned decomposition at the corresponding place.
So far all the steps have taken $\Oh[\ell,\ell']{|\Tpref| \log |\Tpref|}$ time.
It remains to attach the subtrees below the appendices of $\Tpref$ to $\Tc^\star$.

We consider an annotated rank decomposition $\Tc' = (T', V(G), \reps', \repse', \dmap')$ that is constructed as follows.
We start with $\Tc^\star$, and then for every appendix edge $\oap \in \oApp_T(\Tpref)$, we attach the subtree of $T$ below $\oap$ to $T'$ so that $\oap$ is identified with $\pi(\oap) = \vec{ap'}$, and also copy all annotations associated to that subtree in $\Tc$ to $\Tc'$.
We do not copy the annotations on the edge $ap$, in particular, it will hold that $\reps'(\vec{p'a}) = \reps^\star(\vec{p'a})$ and $\reps'(\vec{ap'}) = \reps^\star(\vec{ap'}) = \reps(\oap)$.

The functions $\dmap'(tap')$ where $\vec{ap'} = \pi(\oap)$ for $\oap \in \oApp_T(\Tpref)$ and $t$ is a child of $a$ can be copied from $\Tc$ in the natural way, so it remains to construct the functions $\dmap'(p'at)$.
It holds that $\reps(\vec{ap}) = \reps'(\vec{ap'})$ and both $\repse(ap)$ and $\repse'(ap')$ are representative graphs of $(\lparts(\Tc)[\vec{ap}], \lparts(\Tc)[\vec{pa}])$.
Therefore, by \Cref{lem:uniqisom} let $\phi \colon \reps'(\vec{ap'}) \cup \reps'(\vec{p'a}) \rightarrow \reps(\vec{ap}) \cup \reps(\vec{pa})$ be the unique isomorphism between $\repse'(ap')$ and $\repse(ap)$ that is identity on $\reps'(\vec{ap'}) = \reps(\vec{ap})$.
Such $\phi$ can be computed in $\Oh[\ell]{1}$ time.
We construct $\dmap'(p'at)$ by setting $\dmap'(p'at)(u) \coloneqq \dmap(pat)(\phi(u))$ for all $u \in \reps'(\vec{p'a})$.

Clearly, this construction of $\Tc'$ can be implemented by a prefix-rebuilding update that corresponds to the given prefix-rearrangement description, and the description of this prefix-rebuilding update can be computed according to the previous discussion in $\Oh[\ell,\ell']{|\Tpref| \log |\Tpref|}$ time.
It remains to show that $\Tc'$ encodes $G$.

The properties \Cref{prop:ard:enclp,prop:ard:encedge} obviously hold.
Then we show \Cref{prop:ard:minrep}.

\begin{claim}
\label{lem:prdsrearangmain:claim1}
For all $\oxy \in \oE(T')$ it holds that $\reps'(\oxy)$ is a minimal representative of $\lparts(\Tc')[\oxy]$ and $\repse'(xy) = G[\reps'(\oxy), \reps'(\oyx)]$.
\end{claim}
\begin{claimproof}
First, suppose that $\oxy \in \oE(T') \setminus \oE(T^\star)$.
We have $\lparts(\Tc')[\oxy] = \lparts(\Tc)[\oxy]$, $\reps'(\oxy) = \reps(\oxy)$, and $\repse'(xy) = \repse(xy)$, so the claim holds because $\Tc$ encodes $G$.

Then suppose $\oxy \in \oE(T^\star)$.
Because $\Tc^\star$ encodes $(G^\star,\prt^\star)$, $\reps'(\oxy)$ is a minimal representative of $\lparts(\Tc^\star)[\oxy]$ in $G^\star$ and $\repse'(xy) = G^\star[\reps'(\oxy), \reps'(\oyx)]$.
To obtain that $\reps'(\oxy)$ is a minimal representative of $\lparts(\Tc')[\oxy]$ and $\repse'(xy) = G[\reps'(\oxy), \reps'(\oyx)]$, it suffices to argue that $\lparts(\Tc^\star)[\oxy]$ is a representative of $\lparts(\Tc')[\oxy]$ in $G$ and $G[\lparts(\Tc^\star)[\oxy], \lparts(\Tc^\star)[\oyx]] = G^\star[\lparts(\Tc^\star)[\oxy], \lparts(\Tc^\star)[\oyx]]$.
The former follows from the fact that for each $\oap \in \oApp_{T}(\Tpref)$ the set $\reps(\oap) = \reps'(\pi(\oap))$ is a representative of $\lparts(\Tc)[\oap] = \lparts(\Tc')[\pi(\oap)]$.
The latter follows from the definition of $G^\star$ and the fact that for each $\oap \in \oApp_{T}(\Tpref)$ either $\reps(\oap) \subseteq \lparts(\Tc^\star)[\oxy]$ or $\reps(\oap) \subseteq \lparts(\Tc^\star)[\oyx]$.
\end{claimproof}

The next claim will imply \Cref{prop:ard:edges}.

\begin{claim}
\label{lem:prdsrearangmain:claim2}
For all $xyz \in \PT(T')$ and $u \in \reps'(\oxy)$, it holds that \[N_G(u) \cap \reps'(\ozy) = N_G(\dmap'(xyz)(u)) \cap \reps'(\ozy).\]
\end{claim}
\begin{claimproof} 
When $y \notin V(T^\star)$ or when $xyz = tap'$ for $\vec{ap'} = \pi(\oap)$ with $\oap \in \oApp_T(\Tpref)$ this holds by the property \Cref{prop:ard:edges} of $\Tc$.
Also when $xyz \in \PT(T^\star)$ this holds by the property \Cref{prop:ard:edges} of $\Tc^\star$.
It remains to consider the case of $xyz = p'at$ for $\vec{ap'} = \pi(\oap)$ with $\oap \in \oApp_T(\Tpref)$.

Let $\phi$ be the unique isomorphism between $\repse'(ap')$ and $\repse(ap)$ that is identity on $\reps'(\vec{ap'}) = \reps(\vec{ap})$, and recall that $\reps'(\vec{ta}) = \reps(\vec{ta})$.
We have that $N_G(u) \cap \reps(\vec{ap}) = N_G(\phi(u)) \cap \reps(\vec{ap})$.
Because $\reps(\vec{ap})$ is a representative of $\lparts(\Tc)[\vec{ap}]$, this implies that $N_G(u) \cap \lparts(\Tc)[\vec{ap}] = N_G(\phi(u)) \cap \lparts(\Tc)[\vec{ap}]$.
Now, $\reps'(\vec{ta}) \subseteq \lparts(\Tc)[\vec{ta}] \subseteq \lparts(\Tc)[\vec{ap}]$, so we get that 
\begin{align*}
N_G(u) \cap \reps'(\vec{ta}) &= N_G(\phi(u)) \cap \reps'(\vec{ta})\\
&= N_G(\dmap(pat)(\phi(u))) \cap \reps'(\vec{ta}) && \text{(by \Cref{prop:ard:edges} on $\Tc$)}\\
&= N_G(\dmap'(p'at)(u)) \cap \reps'(\vec{ta}). && \text{(by construction of $\dmap'$)}
\end{align*}
\end{claimproof}
Hence the construction of $\Tc'$ is correct.
\end{proof}

\subsection{Edge update descriptions}
The dynamic graph $G$ in our algorithm is represented by an annotated rank decomposition that encodes $G$, and therefore we use prefix-rebuilding updates to update $G$.
In this section we give a higher-level formalism for describing edge updates, and show that it can be translated to corresponding descriptions of prefix-rebuilding updates efficiently.

Let $\Tc = (T,U,\reps,\repse,\dmap)$ be a rooted annotated rank decomposition that encodes a graph $G$.
An \emph{edge update description} is a quadruple $\prdesc = (W, \Tpref, \reps^\star, \repse^\star)$, where 
\begin{itemize}
\item $W \subseteq V(G)$,
\item $\Tpref$ is a prefix of $T$ so that if $\olp \in \leafe(T)$ and $\reps(\olp) \subseteq W$ then $l \in \Tpref$,
\item $\reps^\star$ is a function that maps each $\oxy \in \oE(T[\Tpref])$ to a non-empty set $\reps^\star(\oxy) \subseteq \lparts(\Tc)[\oxy]$,
\item $\repse^\star$ is a function that maps each $xy \in E(T[\Tpref])$ to a bipartite graph $\repse^\star(xy)$ with bipartition $(\reps^\star(\oxy), \reps^\star(\oyx))$, each $xyz \in \PT(T[\Tpref])$ to a bipartite graph $\repse^\star(xyz)$ with bipartition $(\reps^\star(\oxy), \reps^\star(\ozy))$, and each $xyz \in \PT(T)$ with $x \in \App_T(\Tpref)$ and $y,z \in \Tpref$ to a bipartite graph $\repse^\star(xyz)$ with bipartition $(\reps(\oxy), \reps^\star(\ozy))$.
\end{itemize}

We say that $\prdesc$ \emph{describes} a graph $G'$ if 
\begin{itemize}
\item $V(G') = V(G)$, 
\item for all $u,v \in V(G)$ with $u \notin W$ or $v \notin W$ we have $uv \in E(G')$ if and only if $uv \in E(G)$,
\item for all $\oxy \in \oE(T[\Tpref])$ the set $\reps^\star(\oxy)$ is a representative of $\lparts(\Tc)[\oxy]$ in $G'$,
\item for all $xy \in E(T[\Tpref])$ it holds that $\repse^\star(xy) = G'[\reps^\star(\oxy), \reps^\star(\oyx)]$,
\item for all $xyz \in \PT(T[\Tpref])$ it holds that $\repse^\star(xyz) = G'[\reps^\star(\oxy), \reps^\star(\ozy)]$, and
\item for all $xyz \in \PT(T)$ with $x \in \App_T(\Tpref)$ and $y,z \in \Tpref$, $\repse^\star(xyz) = G'[\reps(\oxy), \reps^\star(\ozy)]$.
\end{itemize}

Note that $\reps^\star(\oxy)$ is not required to be a minimal representative and the graphs in the image of $\repse^\star$ are allowed to have twins over the bipartition.

We observe that if $\prdesc$ describes some graph $G'$, then $G'$ is uniquely determined by $\prdesc$ and $G$.
In particular, by making use of the $\repse^\star(xyz) = G'[\reps^\star(\oxy), \reps^\star(\ozy)]$ graphs, the description $\prdesc$ can be turned into an annotated rank decomposition that encodes $G'[W]$.
We denote $|\prdesc| = |\Tpref|$.
We define that the \emph{width} of $\prdesc$ is the maximum of $\cutrk_{\repse^\star(xy)}(\reps^\star(\oxy))$ over all $\oxy \in \oE(T[\Tpref])$.
Note that if $\prdesc$ has width $\ell$ then it can be represented in space $\Oh[\ell]{|\Tpref|}$.

We say that a prefix-rebuilding update \emph{corresponds} to an edge update description $\prdesc$ if $\prdesc$ describes a graph $G'$, the update turns $\Tc$ into a rooted annotated rank decomposition $\Tc' = (T',U',\reps',\repse',\dmap')$ so that $\Tc'$ encodes $G'$ and $T' = T$, and the prefix of $T$ associated with the update is $\Tpref \setminus \leafs(T)$.
Note that such update can change the width of an edge $xy \in E(T)$ only if $W$ intersects both $\lparts(\Tc)[\oxy]$ and $\lparts(\Tc)[\oyx]$, in particular, only if $xy \in E(T[\Tpref])$.
It follows that the width of $\Tc'$ is at most the maximum of the widths of $\Tc$ and $\prdesc$.

We then give the algorithm to translate edge update descriptions into descriptions of prefix-rebuilding updates.

\begin{lemma}
\label{lem:transleudescprdesc}
There exists an $\ell$-prefix-rebuilding data structure with overhead $\Oh[\ell]{1}$ that maintains a rooted annotated rank decomposition $\Tc$ that encodes a dynamic graph $G$ and additionally supports the following query:
\begin{itemize}
\item $\mathsf{Translate}(\prdesc)$: Given an edge update description $\prdesc$ of width $\ell'$ that describes a graph $G'$, in time $\Oh[\ell,\ell']{|\prdesc|}$ returns a description of a corresponding prefix-rebuilding update.
\end{itemize}
\end{lemma}
\begin{proof}
We maintain the rooted annotated rank decomposition $\Tc = (T,V(G),\reps,\repse,\dmap)$ that encodes $G$ by making use of \Cref{lem:basicprds}.
It remains to describe how the $\mathsf{Translate}(\prdesc)$ query is implemented.

Denote $\prdesc = (W, \Tpref, \reps^\star, \repse^\star)$ and $\Tc = (T,U,\reps,\repse,\dmap)$.
We construct $\Tc' = (T,U,\reps',\repse',\dmap')$ as follows.
First, for every $\oxy \in \oE(T[\Tpref])$ we compute a set $\reps^{\star\star}(\oxy) \subseteq \reps^{\star}(\oxy)$ so that $\reps^{\star\star}(\oxy)$ is a minimal representative of $\lparts(\Tc)[\oxy]$ in $G'$.
This can be computed in $\Oh[\ell']{1}$ time by using $\repse^{\star}(xy)$.
We also compute $\repse^{\star\star}(xy) = \repse^{\star}[\reps^{\star\star}(\oxy),\reps^{\star\star}(\oyx)]$ for all $xy \in E(T[\Tpref])$.
Then we construct $\reps'$ by setting $\reps'(\oxy) = \reps^{\star\star}(\oxy)$ if $\oxy \in \oE(T[\Tpref])$ and $\reps'(\oxy) = \reps(\oxy)$ otherwise.
We also construct $\repse'$ by setting $\repse'(xy) = \repse^{\star\star}(xy)$ if $xy \in E(T[\Tpref])$ and $\repse'(xy) = \repse(xy)$ otherwise.

Because all edges $xy$ so that both $\lparts(\Tc)[\oxy]$ and $\lparts(\Tc)[\oyx]$ intersect $W$ are in $E(T[\Tpref])$, $\Tc'$ satisfies \Cref{prop:ard:minrep}.
It remains to construct $\dmap'$.

When both $xy$ and $yz$ are not in $E(T[\Tpref])$ we let $\dmap'(xyz) = \dmap(xyz)$.
This satisfies \Cref{prop:ard:edges} because $\reps'(\oxy) = \reps(\oxy)$, $\reps'(\oyz) = \reps(\oyz)$, $\repse'(xy) = \repse(xy)$, and $\repse'(yz) = \repse(yz)$.
Let $xyz \in \PT(T[\Tpref])$ and let $u \in \reps'(\oxy)$.
By using $\repse^\star(xyz)$ we can compute $N_{G'}(u) \cap \reps^{\star \star}(\ozy)$, and then find $v \in \reps'(\oyz) = \reps^{\star \star}(\oyz)$ so that $N_{G'}(u) \cap \reps^{\star \star}(\ozy) = N_{G'}(v) \cap \reps^{\star \star}(\ozy)$ and set $\dmap(xyz)(u) = v$.
This clearly satisfies \Cref{prop:ard:edges}.
The same idea works for computing $\dmap(xyz)$ when $x$ or $z$ is not in $\Tpref$.



We observe that this construction can be implemented with a prefix-rebuilding update so that $\Tpref \setminus \leafs(T)$ is the prefix of $T$ associated with the update.
Moreover, the description of the prefix-rebuilding update can be computed in $\Oh[\ell,\ell']{|\Tpref|}$ time.
\end{proof}

%% file: refinement.tex
\section{Refinement}
\label{sec:refi}
In this section we introduce the refinement operation that will be used for improving the rank decomposition, and give the height reduction scheme by using the refinement operation.

\input{closures.tex}
\subsection{Refinement operation}
We start by introducing the potential function we use for the amortized analysis of the algorithm.

In a rooted rank decomposition $\Tc = (T,\lmap)$ of a graph $G$, let us say that the \emph{width} of a node $t \in V(T)$ is the width of the edge between the node and the parent, and denote it by $\width_{\Tc,G}(t) = \cutrk_G(\lparts(\Tc)[t])$.
The width of the root node is defined to be $0$.
Let $f$ be the function from \Cref{lem:smallclosures}.
Then we let the $\ell$-potential of $t$ with respect to $G$ be 
\[\Phi_{\ell,\Tc,G}(t) = (2 \cdot f(\ell))^{\width_{\Tc,G}(t)} \cdot \height_{T}(t),\]
and the $\ell$-potential of $\Tc$ with respect to $G$ be
\[\Phi_{\ell,G}(\Tc) = \sum_{t \in V(T)} \Phi_{\ell,\Tc,G}(t).\]
We will omit the graph $G$ from the subscript in these notations if it is clear from the context.

For a set of nodes $S \subseteq V(T)$ we will denote $\height_T(S) = \sum_{t \in S} \height_T(t)$ and $\Phi_{\ell,\Tc}(S) = \sum_{t \in S} \Phi_{\ell,\Tc}(t)$.

Then we give the refinement operation formulated as a prefix-rebuilding data structure.

\begin{lemma}
\label{lem:refimain}
Let $k \in \N$ and $\ell \ge 4k+1$.
There exists an $\ell$-prefix-rebuilding data structure with overhead $\Oh[\ell]{1}$ that maintains a rooted annotated rank decomposition $\Tc = (T,V(G),\reps,\repse,\dmap)$ that encodes a dynamic graph $G$ and supports the following operation:
\begin{itemize}
\item $\mathsf{Refine}(\Tpref)$: Given a leafless prefix $\Tpref$ of $T$ so that $\Tpref$ contains all nodes of width $>4k$, returns either that the rankwidth of $G$ is greater than $k$, or a description $\prdesc$ of a prefix-rebuilding update so that the rooted rank decomposition $\Tc'$ to which $\Tc$ corresponds to after applying $\prdesc$ has the following properties:
\begin{enumerate}
\item $\Tc'$ encodes $G$,
\item $\Tc'$ has width at most $4k$, and
\item\label{lem:refimain:enum:potbound} the following inequality holds:
\[\Phi_{\ell}(\Tc') \le \Phi_{\ell}(\Tc) - \height_T(\Tpref) + \log |\Tc| \cdot \Oh[\ell]{|\Tpref| + \height_T(\App_T(\Tpref))}.\]
\end{enumerate}
In the former case, the running time of $\mathsf{refine}(\Tpref)$ is $\Oh[\ell]{|\Tpref|}$, and in the latter case the running time and therefore also $|\prdesc|$ is bounded by
\[\log |\Tc| \cdot \Oh[\ell]{\Phi_{\ell}(\Tc)-\Phi_{\ell}(\Tc') + \log |\Tc| \cdot (|\Tpref| + \height_T(\App_T(\Tpref)))}.\]
\end{itemize}
\end{lemma}
\begin{proof}
We use \Cref{lem:basicprds} for maintaining a representation of $\Tc$.
Let $c = f(\ell)$, where $f$ is the function from \Cref{lem:smallclosures}, in particular, so that if $G$ has rankwidth at most $k$ then there exists a $c$-small $k$-closure of $\Tpref$.
We maintain the $\ell$-prefix-rebuilding data structure from \Cref{lem:closureprds} with these values of $c$ and $k$ and the $\ell$-prefix-rebuilding data structure from \Cref{lem:prdsrearangmain}, by simply relaying all prefix-rebuilding updates also to these data structures.
In particular, they will always store the exactly same rooted annotated rank decomposition $\Tc$.

Then we describe how the $\mathsf{Refine}(\Tpref)$ operation is implemented.
First we apply the $\mathsf{Closure}(\Tpref)$ operation of the data structure of \Cref{lem:closureprds}.
If it returns that no $c$-small $k$-closure of $\Tpref$ exists, then by \Cref{lem:smallclosures} the rankwidth of $G$ is more than $k$ and we can return immediately.
Otherwise, it returns a representation of a minimal $c$-small $k$-closure $\prt$ of $\Tpref$, containing in particular the sets $\cut_T(\prt)$ and $\aep_T(\prt)$, and a rooted rank decomposition $\Tc^{**} = (T^{**},\lmap^{**})$ of $(G[\prt],\prt)$ of width at most $2k$, where $\lmap^{**}$ is represented as a function $\lmap^{**} \colon \aep_T(\prt) \rightarrow \leafe(T^{**})$.
We immediately use \Cref{lem:logdepthdecomp} to turn $\Tc^{**}$ into a rooted rank decomposition $\Tc^{*} = (T^{*}, \lmap^{*})$ of width at most $4k$ and height at most $\Oh{\log n}$.

Let us describe the construction of the rooted rank decomposition $\Tc' = (T',\lmap')$.
For this, we denote by $(T,\lmap)$ the rooted rank decomposition that $\Tc$ corresponds to.

First, for each part $C \in \prt$ (represented by $\aes_T(C) \in \aep_T(\prt)$) we construct a rooted rank decomposition $\Tc_C = (T_C, \lmap_C)$ as follows.
The tree $T_C$ is obtained by first taking the subtree of $T$ induced by nodes $t \in V(T)$ with $\lparts(\Tc)[t] \cap C \neq \emptyset$ and iteratively contracting all resulting degree-2 nodes.
Then we set $\lmap_C \coloneqq \funrestriction{\lmap}{C}$.
Now $\Tc_C$ is a rooted rank decomposition of $G[C]$ so that for every node $t \in V(T_C)$ there exists a node $t' \in V(T)$ with $\lparts(\Tc_C)[t] = \lparts(\Tc)[t'] \cap C$.
Then, the rooted rank decomposition $\Tc' = (T',\lmap')$ is constructed by taking $\Tc^*$ and for each $C \in \prt$ attaching $\Tc_C$ to $T^*$ by identifying the root of $T_C$ with the leaf $\lmap^*(C)$ of $T^*$.
It can be observed that $\Tc'$ is a rooted rank decomposition of $G$, for every $t \in V(T') \cap V(T^*)$ it holds that $\lparts(\Tc')[t] = \lparts(\Tc^*)[t]$, and for every $t \in V(T') \cap V(T_C)$ for $C \in \prt$ it holds that $\lparts(\Tc')[t] = \lparts(\Tc_C)[t]$.

\begin{claim}
\label{lem:refimain:claim:runtimeclaim}
A description $\prdesc$ of a prefix-rebuilding update that turns $\Tc$ into a rooted annotated rank decomposition that corresponds to $\Tc'$ can be computed in $\Oh[\ell]{|\cut_T(\prt)| \log |\Tc|}$ time.
\end{claim}
\begin{claimproof}
We will show that such prefix-rearrangement description can be computed in $\Oh[\ell]{|\cut_{T}(\prt)|}$ time.
This then implies the claim by applying the $\mathsf{Translate}$ query of the prefix-rebuilding data structure of \Cref{lem:prdsrearangmain}.

Recall that for every $t \in \App_T(\cut_T(\prt))$ we have that $\lparts(\Tc)[t] \subseteq C$ for some $C \in \prt$, so the subtree rooted at $t$ can be copied verbatim from $\Tc$ to $\Tc'$.
It follows that we can set the prefix of the prefix-rearrangement description to be $\cut_T(\prt)$.
It remains to construct the tree $T^\star$ of the description.
We construct it by first taking $T^*$, and then for every leaf of it that corresponds to a part $C \in \prt$ constructing the prefix of $T_C$ that is not copied verbatim.
In particular, let $C \subseteq \lparts(\Tc)[a]$ for $a \in \App_T(\Tpref)$, and denote by $\cut_{T,a}(\prt) \subseteq \cut_T(\prt)$ the nodes that are cut by $\prt$ and are descendants of $a$.
By using the mapping $\lmap^* \colon \aep_T(\prt) \rightarrow \leafe(T^*)$ we can construct the prefix of $T_C$ that is not copied verbatim in $\Oh{|\cut_{T,a}(\prt)|}$ time, also finding out how the subtrees that are copied verbatim are attached to the prefix.
Because $\prt$ is $c$-small, the total time sums up to $\Oh{c \cdot |\cut_{T}(\prt)|} = \Oh[\ell]{|\cut_{T}(\prt)|}$.
\end{claimproof}

For bounding the width of $\Tc'$ and analyzing the potential, let us relate the nodes in each of the trees $T_C$ to nodes in $T$.
Let us denote by $\pi_C \colon V(T_C) \rightarrow V(T)$ the mapping that maps each node $t \in V(T_C)$ to a node $t' \in V(T)$ so that $\lparts(\Tc_C)[t] = \lparts(\Tc)[t'] \cap C$ and $t'$ minimizes $\height_T(t')$ under this condition (this defines $\pi_C(t)$ uniquely).
Note that $\pi_C$ is an injection, and if $a \in \App_T(\Tpref)$ so that $C \subseteq \lparts(\Tc)[a]$, then $\pi_C(t)$ is a descendant of $a$ for all $t \in V(T_C)$.

\begin{claim}
For all $t \in V(T_C)$ it holds that $\cutrk_G(\lparts(\Tc_C)[t]) \le \cutrk_G(\lparts(\Tc)[\pi_C(t)])$, and moreover if $\pi_C(t) \in \cut_T(\prt)$ then $\cutrk_G(\lparts(\Tc_C)[t]) < \cutrk_G(\lparts(\Tc)[\pi_C(t)])$.
\end{claim}
\begin{claimproof}
First suppose that $\pi_C(t) \notin \cut_T(\prt)$.
In that case, $\lparts(\Tc_C)[t] = \lparts(\Tc)[\pi_C(t)]$ because $C$ intersects $\lparts(\Tc)[\pi_C(t)]$ but does not cut $\pi_C(t)$.
Then, if $\pi_C(t) \in \cut_T(\prt)$, \Cref{lem:mincloswidthbound} implies that $\cutrk_G(\lparts(\Tc_C)[t]) < \cutrk_G(\lparts(\Tc)[\pi_C(t)])$ because $\prt$ is linked because it is minimal.
\end{claimproof}

It follows that $\Tc'$ has width at most $4k$: All nodes in $V(T^*)$ have width at most $4k$, and for all $C \in \prt$ and $t \in V(T_C)$ we have that $\pi_C(t) \notin \Tpref$ implying $\cutrk_G(\lparts(\Tc)[\pi_C(t)]) \le 4k$ and therefore $\cutrk_G(\lparts(\Tc')[t]) \le 4k$.

To bound $\Phi_{\ell}(\Tc')$, first note that
\[\Phi_\ell(\Tc') = \Phi_{\ell,\Tc'}(V(T^*) \setminus \leafs(T^*)) + \sum_{C \in \prt} \Phi_\ell(\Tc_C).\]
Let us first bound the latter term.

\begin{claim}
\label{lem:refimain:claim:potbound1}
\[\sum_{C \in \prt} \Phi_\ell(\Tc_C) \le \Phi_\ell(\Tc) - \Phi_{\ell,\Tc}(\Tpref) - |\cut_T(\prt)\setminus \Tpref|\]
\end{claim}
\begin{claimproof}
We observe that $\height_{T_C}(t) \le \height_T(\pi_C(t))$ for all $C \in \prt$ and $t \in V(T_C)$, which implies $\Phi_{\ell,\Tc_C}(t) \le \Phi_{\ell,\Tc}(\pi_C(t))$ for all such $t$, and moreover when $\pi_C(t) \in \cut_T(\prt)$ it holds that
\begin{align*}
\Phi_{\ell,\Tc_C}(t) &= (2 \cdot f(\ell))^{\width_{\Tc_C}(t)} \cdot \height_{T_C}(t)\\
&\le (2 \cdot f(\ell))^{\width_{\Tc}(\pi_C(t))-1} \cdot \height_{T}(\pi_C(t))\\
&\le \Phi_{\ell,\Tc}(\pi_C(t)) / (2 \cdot f(\ell)).
\end{align*}

Then, for $x \in V(T)$, let us denote by $\pi^{-1}(x)$ the set of nodes in $\bigcup_{C \in \prt} V(T_C)$ that are mapped to $x$ by $\pi_C$, i.e., $\pi^{-1}(x) = \{t \in V(T_C) \mid C \in \prt \text{ and } \pi_C(t) = x\}$.
We observe that if $x \in \Tpref$ then $\pi^{-1}(x) = \emptyset$, if $x \in \cut_T(\prt) \setminus \Tpref$ then $|\pi^{-1}(x)| \le f(\ell)$ because $\prt$ is $c$-small for $c = f(\ell)$, and if $x \in V(T) \setminus \cut_T(\prt)$ then $|\pi^{-1}(x)| = 1$.

By putting these two observations together we obtain
\begin{align*}
\sum_{C \in \prt} \Phi_\ell(\Tc_C) &\le \Phi_{\ell,\Tc}(V(T)\setminus \cut_T(\prt)) + \sum_{t \in \cut_T(\prt)\setminus \Tpref} f(\ell) \cdot \Phi_{\ell,\Tc}(t) / (2 \cdot f(\ell))\\
&\le \Phi_{\ell,\Tc}(V(T)\setminus \cut_T(\prt)) + \Phi_{\ell,\Tc}(\cut_T(\prt)\setminus \Tpref)/2\\
&\le \Phi_\ell(\Tc) - \Phi_{\ell,\Tc}(\Tpref) - |\cut_T(\prt)\setminus \Tpref|.
\end{align*}
\end{claimproof}

Then we bound the former term.

\begin{claim}
\label{lem:refimain:claim:potbound2}
\[\Phi_{\ell,\Tc'}(V(T^*)) \le \log |\Tc| \cdot \Oh[\ell]{|\Tpref|+\height_T(\App_T(\Tpref))}\]
\end{claim}
\begin{claimproof}
For each node $t \in V(T^*)$, let $\Gamma(t) \in \prt$ be a part of $\prt$ so that $\lmap^*(\Gamma(t)) \in \leafe(T^*)[t]$, and among such parts $\Gamma(t)$ maximizes $\height(T_{\Gamma(t)})$.
Such $\Gamma(t)$ is not necessary unique, in which case we assign some such $\Gamma(t)$ arbitrarily.
Because the height of $T^*$ is at most $\Oh{\log |\prt|} \le \Oh{\log |\Tc|}$, we have that $\height_{T'}(t) \le \Oh{\log |\Tc|} + \height(T_{\Gamma(t)})$, implying that
\begin{align*}
\height_{T'}(V(T^*)) \le \Oh{|V(T^*)| \log |\Tc|} + \sum_{t \in V(T^*)} \height(T_{\Gamma(t)}).
\end{align*}
We observe that if $\Gamma(t) \subseteq \lparts(\Tc)[a]$ for $a \in \App_T(\Tpref)$, then $\height(T_{\Gamma(t)}) \le \height_T(a)$.
Because $\prt$ is $c$-small, for each $a \in \App_T(\Tpref)$ there are at most $c$ such sets $\Gamma(t)$.
Also, because $T^*$ has height at most $\Oh{\log |\Tc|}$, each $C \in \prt$ can be the set $\Gamma(t)$ for at most $\Oh{\log |\Tc|}$ nodes in $V(T^*)$.
From these observations it follows that
\begin{align*}
\sum_{t \in V(T^*)} \height(T_{\Gamma(t)}) &\le \sum_{a \in \App_T(\Tpref)} \Oh{\height_T(a) \cdot c \cdot \log |\Tc|}\\
&\le \Oh[\ell]{\height_T(\App_T(\Tpref)) \cdot \log |\Tc|}
\end{align*}

Then the conclusion of the claim follows from $|V(T^*)| \le 2 \cdot |\prt| \le 4c \cdot |\Tpref| \le \Oh[\ell]{|\Tpref|}$.
\end{claimproof}

By putting \Cref{lem:refimain:claim:potbound1,lem:refimain:claim:potbound2} together, we obtain
\begin{align}
\Phi_\ell(\Tc') \le \Phi_\ell(\Tc) &- \Phi_{\ell,\Tc}(\Tpref) - |\cut_T(\prt)\setminus \Tpref| \nonumber\\
&+ \log |\Tc| \cdot \Oh[\ell]{|\Tpref|+\height_T(\App_T(\Tpref))}\label{lem:refimain:runtimeequ},
\end{align}
which by $\Phi_{\ell,\Tc}(\Tpref) \ge \height_T(\Tpref)$ implies the desired potential bound of \Cref{lem:refimain:enum:potbound}.

Let us then prove the running time bound of $\mathsf{Refine}(\Tpref)$ in the lemma statement.
The algorithm consists of calling the data structure of \Cref{lem:closureprds}, applying \Cref{lem:logdepthdecomp}, and constructing the description $\prdesc$ of the prefix-rebuilding update, which by \Cref{lem:refimain:claim:runtimeclaim} all take at most $\Oh[\ell]{|\cut_T(\prt)| \log |\Tc|}$ time.
We can rearrange \Cref{lem:refimain:runtimeequ} into
\[\Phi_{\ell,\Tc}(\Tpref) + |\cut_T(\prt)\setminus \Tpref| \le \Phi_\ell(\Tc) - \Phi_\ell(\Tc') + \log |\Tc| \cdot \Oh[\ell]{|\Tpref|+\height_T(\App_T(\Tpref))},\]
which by $|\cut_T(\prt)| \le \Phi_{\ell,\Tc}(\Tpref) + |\cut_T(\prt)\setminus \Tpref|$ implies
\[|\cut_T(\prt)| \le \Phi_\ell(\Tc) - \Phi_\ell(\Tc') + \log |\Tc| \cdot \Oh[\ell]{|\Tpref|+\height_T(\App_T(\Tpref))},\]
which yields the desired running time.
\end{proof}

\subsection{Height reduction}
The main combinatorial ingredient for our height reduction scheme is the following lemma, which is proved implicitly in \cite[Section 6]{dyntw}.

\begin{lemma}[\cite{dyntw}]
\label{lem:heighttree}
Let $c \ge 2$ and $T$ be a binary tree with $n$ nodes.
If the height of $T$ is at least $2^{\Omega(\sqrt{\log n \log c})}$ then there exists a non-empty prefix $\Tpref$ of $T$ so that \[c \cdot \left(|\Tpref| + \height_T(\App_T(\Tpref))\right) \le \height_T(\Tpref).\]
Moreover, if a representation of $T$ is already stored and supports the function $\height_T(t)$ for $t \in V(T)$ in $\Oh{1}$ time, then such $\Tpref$ can be found in $\Oh{|\Tpref|}$ time.
\end{lemma}

Then, our height reduction scheme is formulated as a prefix-rebuilding data structure as follows.

\begin{lemma}
\label{lem:heightreduprds}
Let $k \in \N$ and $\ell \ge 4k+1$.
There exists an $\ell$-prefix-rebuilding data structure with overhead $\Oh[\ell]{1}$ that maintains a rooted annotated rank decomposition $\Tc = (T,V(G),\reps,\repse,\dmap)$ that encodes a dynamic graph $G$ of rankwidth at most $k$, supports the operation $\mathsf{Refine}(\Tpref)$ from \Cref{lem:refimain}, and additionally supports the following operation under the promise that the width of $\Tc$ is at most $4k$:
\begin{itemize}
\item $\mathsf{ImproveHeight}()$: Updates $\Tc$ through a sequence of prefix-rebuilding updates so that the resulting annotated rank decomposition $\Tc'$ encodes $G$, has height $2^{\Oh[\ell]{\sqrt{\log n \log \log n}}}$ and width at most $4k$, and returns the corresponding sequence of descriptions of prefix-rebuilding updates.
All of the intermediate decompositions also have width at most $4k$.
It holds that $\Phi_{\ell}(\Tc') \le \Phi_{\ell}(\Tc)$ and the running time of $\mathsf{ImproveHeight}()$ is $\Oh[\ell]{(\Phi_\ell(\Tc) - \Phi_\ell(\Tc')) \log |\Tc|}$.
\end{itemize}
\end{lemma}
\begin{proof}
We maintain a representation of $\Tc$ by \Cref{lem:basicprds}, and additionally maintain the prefix-rebuilding data structures $\D^{\mathsf{height}}$ given by \Cref{lem:utilityprds} and $\D^{\mathsf{refine}}$ given by \Cref{lem:refimain}, so that all prefix-rebuilding updates that are applied to $\Tc$ are also relayed to $\D^{\mathsf{height}}$ and $\D^{\mathsf{refine}}$, in particular, so that they store the exactly same rooted annotated rank decomposition $\Tc$.
The $\mathsf{Refine}(\Tpref)$ operation is implemented by using $\D^{\mathsf{refine}}$.
It remains to implement the $\mathsf{ImproveHeight}()$ operation.

Let us choose $c_0 = \Oh[\ell]{1}$ based on $\ell$ so that the inequality of \Cref{lem:refimain:enum:potbound} in \Cref{lem:refimain} is true in the form
\begin{equation}
\label{lem:heightreduprds:eq1}
\Phi_{\ell}(\Tc') \le \Phi_{\ell}(\Tc) - \height_T(\Tpref) + \frac{c_0}{2} \cdot \log |\Tc| \cdot \left(|\Tpref| + \height_T(\App_T(\Tpref))\right).
\end{equation}

Let $c = c_0 \cdot \log |\Tc|$.
First, if $\height(\Tc) \le 2^{\Oh{\sqrt{\log n \log c}}} \le 2^{\Oh[\ell]{\sqrt{\log n \log \log n}}}$, where the constant in the ${\cal O}$-notation depends on the constant in the $\Omega$-notation in \Cref{lem:heighttree}, then the height of $\Tc$ is already small enough and we do not update $\Tc$ and return an empty sequence of descriptions of prefix-rebuilding updates.
Otherwise, we use the algorithm from \Cref{lem:heighttree} with the $\height_T(t)$ operation supplied from $\D^{\mathsf{height}}$ to find a non-empty prefix $\Tpref$ of $T$ so that 
\begin{equation}
\label{lem:heightreduprds:eq2}
c_0 \cdot \log |\Tc| \cdot \left(|\Tpref| + \height_T(\App_T(\Tpref))\right) \le \height_T(\Tpref).
\end{equation}

Then we apply the $\mathsf{Refine}$ operation with this $\Tpref$ and apply the resulting prefix-rebuilding update to $\Tc$, relaying it also to $\D^{\mathsf{height}}$ and $\D^{\mathsf{refine}}$.
By putting \Cref{lem:heightreduprds:eq1,lem:heightreduprds:eq2} together, we obtain that the resulting decomposition $\Tc'$ satisfies
\[\Phi_{\ell}(\Tc') \le \Phi_\ell(\Tc) - \frac{c_0}{2} \cdot \log |\Tc| \cdot \left(|\Tpref| + \height_T(\App_T(\Tpref))\right).\]
Because $\Tpref$ is non-empty, we have in particular $\Phi_\ell(\Tc') < \Phi_\ell(\Tc)$.
The time complexity of the application of \Cref{lem:heighttree} is $\Oh{|\Tpref|} = \Oh{\Phi_\ell(\Tc)-\Phi_\ell(\Tc')}$.
The time complexity of the application of the $\mathsf{Refine}$ operation and the size of the description of the update is bounded by
\begin{align*}
&\log |\Tc| \cdot \Oh[\ell]{\Phi_{\ell}(\Tc)-\Phi_{\ell}(\Tc') + \log |\Tc| \cdot (|\Tpref| + \height_T(\App_T(\Tpref)))}\\
&= \Oh[\ell]{(\Phi_{\ell}(\Tc)-\Phi_{\ell}(\Tc')) \log |\Tc|},
\end{align*}
which is also the time it takes to apply the prefix-rebuilding updates, implying that the total time complexity is $\Oh[\ell]{(\Phi_{\ell}(\Tc)-\Phi_{\ell}(\Tc')) \log |\Tc|}$.
The width of $\Tc'$ is guaranteed to be at most $4k$ by \Cref{lem:refimain}.

Applying this update did not necessarily decrease the height of $\Tc$, but we can run it again repeatedly until it decreases the height to $2^{\Oh[\ell]{\sqrt{\log n \log \log n}}}$.
Because $\Phi_\ell(\Tc') < \Phi_\ell(\Tc)$, the number of such iterations is bounded by $\Phi_\ell(\Tc)$, and moreover, as the running time of a single iteration is bounded by $\Oh[\ell]{(\Phi_{\ell}(\Tc)-\Phi_{\ell}(\Tc')) \log |\Tc|}$, the running time of any sequence of such iterations is bounded by $\Oh[\ell]{(\Phi_{\ell}(\Tc)-\Phi_{\ell}(\Tc'')) \log |\Tc|}$, where $\Tc''$ is the final decomposition.
Because all of the updates were obtained from the $\mathsf{Refine}$ operation, all of the rank decompositions in the sequence of updates have width at most $4k$.
\end{proof}

%% file: closures.tex
\subsection{Closures}
The main graph-theoretic ingredient of the refinement operation is the concept of \emph{closures}.

Let $\Tc = (T,\lmap)$ be a rooted rank decomposition of a graph $G$, $\Tpref$ a leafless prefix of $T$, and $k$ a positive integer.
A \emph{$k$-closure} of $\Tpref$ is a partition $\prt$ of $V(G)$ so that
\begin{enumerate}
\item for each $C \in \prt$ there exists $a \in \App_T(\Tpref)$ so that $C \subseteq \lparts(\Tc)[a]$, and
\item the partitioned graph $(G[\prt], \prt)$ has rankwidth at most $2k$.
\end{enumerate}

We will show that if $G$ has rankwidth at most $k$, then for any $\Tpref$ there exists a $k$-closure with specific properties.
This will be then used in the refinement operation.

\paragraph{Small closures.}
We say that a $k$-closure $\prt$ is $c$-small for some integer $c$ if for every $a \in \App_T(\Tpref)$ there exist at most $c$ parts $C \in \prt$ with $C \subseteq \lparts(\Tc)[a]$.
In this subsection we show that if $G$ has rankwidth $k$ and $\Tc$ has width $\ell$, then there exists a $f(\ell)$-small $k$-closure of any prefix $\Tpref$ of $T$.
For this we will first prove the Dealternation Lemma for rankwidth, which will be an analogue of a similar lemma for treewidth given in~\cite{DBLP:journals/lmcs/BojanczykP22}.
We postpone the proof of this lemma to \Cref{sec:dealternation}, but let us state it here.

We say that a set $F \subseteq V(G)$ is a \emph{tree factor} of $\Tc$ if $F = \lparts(\Tc)[t]$ for some node $t \in V(T)$.
Similarly, we say that $F \subseteq V(G)$ is a \emph{context factor} of $\Tc$ if it is not a tree factor but it can be written as $F = F_1 \setminus F_2$, where $F_1$ and $F_2$ are tree factors of $\Tc$.
A set $F \subseteq V(G)$ is a \emph{factor} of $\Tc$ if it is either a tree factor or a context factor of $\Tc$.

\begin{restatable}{lemma}{dealternationlemma}
\label{lem:dealt}
There exists a function $f(\ell)$ so that if $G$ is a graph of rankwidth $k$ and $\Tc$ a rooted rank decomposition of $G$ of width $\ell$, then there exists a rooted rank decomposition $\Tc'$ of $G$ of width $k$ so that for every node $t \in V(T)$, the set $\lparts(\Tc)[t]$ can be partitioned into a disjoint union of $f(\ell)$ factors of $\Tc'$.
\end{restatable}

Next we use the Dealternation Lemma to prove the existence of $f(\ell)$-small $k$-closures.

\begin{lemma}
\label{lem:smallclosures}
There exists a function $f(\ell)$, so that if $G$ is a graph of rankwidth $k$, $\Tc = (T,\lmap)$ is a rooted rank decomposition of $G$ of width $\ell$, and $\Tpref$ a leafless prefix of $T$, then there exists a $f(\ell)$-small $k$-closure $\prt$ of $\Tpref$.
\end{lemma}
\begin{proof}
By applying \Cref{lem:dealt}, let $\Tc' = (T',\lmap')$ be a rooted rank decomposition of $G$ of width $k$ so that for every node $t \in V(T)$ the set $\lparts(\Tc)[t]$ can be partitioned into a disjoint union of $f(\ell)$ factors of $\Tc'$.
Then for each $a \in \App_T(\Tpref)$ let $\prt_a$ be the partition of $\lparts(\Tc)[a]$ into $f(\ell)$ parts that are factors of $\Tc'$, and let $\prt = \bigcup_{a \in \App_T(\Tpref)} \prt_a$.
It remains to show that $(G[\prt], \prt)$ has rankwidth at most $2k$.

Observe that if all factors in $\prt$ would be tree factors, then we would directly get that $(G[\prt], \prt)$ has rankwidth at most $k$ by using the same rank decomposition truncated to the roots of the factors.
Therefore, our goal is to change $\Tc'$ so that all factors in $\prt$ become tree factors and the width increases to at most $2k$.

Let us say that an edge $ab \in E(T')$, where $b$ is the parent of $a$ in $T'$, is \emph{processed} if either of the following conditions holds:
\begin{itemize}
	\item there exists a tree factor $F \in \prt$ that intersects both $\lparts(\Tc')[\oab]$ and $\lparts(\Tc')[\oba]$; or
	\item $\lparts(\Tc)[a]$ is a~tree factor, and there is no context factor in $T'$ of the form $\lparts(\Tc)[g] \setminus \lparts(\Tc)[a]$ for a~strict ancestor $g$ of $b$.
\end{itemize}
Otherwise, $ab$ is \emph{unprocessed}.
We will make changes to $\Tc'$ while maintaining an invariant that every processed edge has width at most $2k$ and every unprocessed edge has width at most $k$.
Suppose there is a node $x \in V(T')$ and a descendant $y$ of $x$ so that $C = \lparts(\Tc')[x] \setminus \lparts(\Tc')[y]$ is a context factor $C \in \prt$.
Note that $x$ is not $y$ nor a child of $y$ because otherwise $C$ would be a tree factor.
Let $p_x$ be the parent of $x$ (or $p_x = x$ if $x$ is the root of $T'$) and $p_y$ be the parent of $y$ in $T'$.
Note that all edges on the simple path between $p_x$ and $y$ are unprocessed.

We will change $\Tc'$ into a new rooted rank decomposition $\Tc''$ so that the number of context factors decreases but the invariant is maintained.
In particular, $\Tc''$ is constructed by cutting off the subtree rooted at $y$ by  cutting the edge between $y$ and $p_y$, and putting it back so that $x$ and $y$ have the same parent in the resulting decomposition.
For this, the edge $xp_x$ will be subdivided, or if $x$ is the root a new root will be created so that $x$ and $y$ are its children.
Let $p'$ be the new common parent of $x$ and $y$.
Also, the degree-2 node $p_y$ created by cutting the edge $yp_y$ is contracted (\Cref{fig:small-closures-replacement}).

\begin{figure}
	\centering
	\begin{subfigure}{0.3\textwidth}
		\centering
		\includegraphics[height=4.5cm]{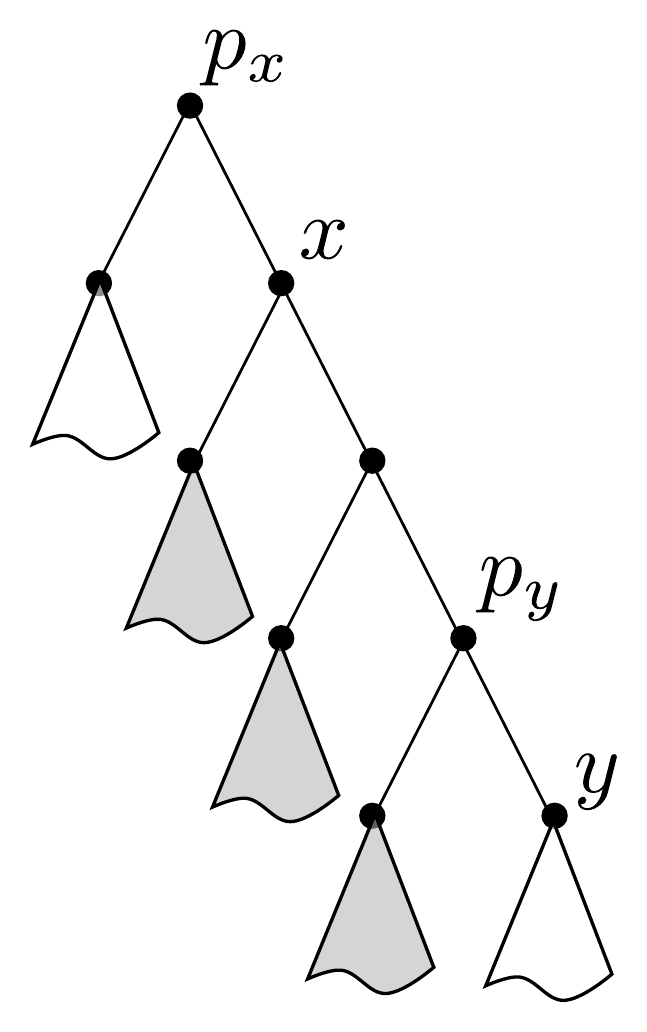}
	\end{subfigure}%
	\begin{subfigure}{0.1\textwidth}
		\centering
		$\to$
	\end{subfigure}
	\begin{subfigure}{0.3\textwidth}
		\centering
		\includegraphics[height=4.5cm]{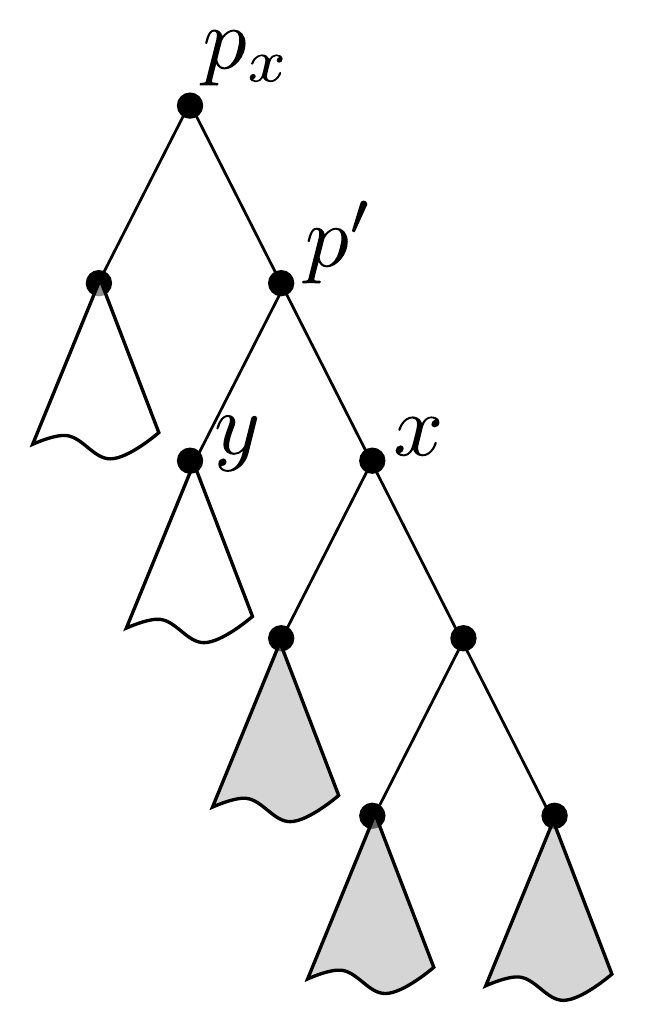}
	\end{subfigure}
	\caption{A surgery on the rank decomposition for a~context factor $C = \lparts(\Tc')[x] \setminus \lparts(\Tc')[y]$. The subtrees comprising $C$ are marked gray.}
	\label{fig:small-closures-replacement}
\end{figure}

We observe that $C$ becomes a tree factor in $\Tc''$, but no other factors change.
This change affects only the widths of edges $ab \in E(T')$ that were on the path from $p_y$ to $p_x$.
Such edges $ab$ were unprocessed, but the corresponding edges $a'b'$ in $T''$ become processed as $C$ becomes a tree factor.
Suppose $b$ is the parent of $a$.
We have that $\lparts(\Tc')[y] \subseteq \lparts(\Tc')[\oab]$, $\cutrk(\lparts(\Tc')[\oab]) \le k$, and $\cutrk(\lparts(\Tc')[y]) \le k$.
The width of the new edge $a'b'$ corresponding to $ab$ will be $\cutrk(\lparts(\Tc'')[\vec{a'b'}]) = \cutrk(\lparts(\Tc')[\oab] \setminus \lparts(\Tc')[y])$, which by symmetry and submodularity of the $\cutrk$ function is at most $\cutrk(\lparts(\Tc')[\oab]) + \cutrk(\lparts(\Tc')[y]) \le 2k$.

Therefore, the process decreases the number of context factors and maintains the invariant, and in the end we obtain a rooted rank decomposition of $G$ of width at most $2k$ so that all parts of $\prt$ are tree factors in the decomposition.
Such decomposition can be easily turned into a rank decomposition of $(G[\prt], \prt)$ of width at most $2k$.
\end{proof}

\paragraph{Closure linkedness.}
Let $A \subseteq B \subseteq V(G)$ be two sets of vertices.
We say that $A$ is \emph{linked} into $B$ if for all sets $S$ with $A \subseteq S \subseteq B$ it holds that $\cutrk(A) \le \cutrk(S)$.
We say that a set $C \subseteq V(G)$ \emph{cuts} a node $t \in V(T)$ if both $\lparts(\Tc)[t] \cap C$ and $\lparts(\Tc)[t] \setminus C$ are non-empty.
Then we say that $k$-closure $\prt$ of $\Tpref$ is \emph{linked} if for every $C \in \prt$ with $C \subseteq \lparts(\Tc)[a]$ for $a \in \App_T(\Tpref)$ it holds that
\begin{enumerate}
\item $C$ is linked into $\lparts(\Tc)[a]$, and
\item if $C$ cuts a descendant $t$ of $a$, then $\cutrk(C \cup \lparts(\Tc)[t]) > \cutrk(C)$.
\end{enumerate}

We say that a $k$-closure $\prt$ \emph{cuts} a node $t \in V(T)$ if there is $C \in \prt$ so that $C$ cuts $t$, or equivalently, if more than one part in $\prt$ intersects $\lparts(\Tc)[t]$.
Note that any $k$-closure of $\Tpref$ cuts all nodes in $\Tpref$.

In our algorithm we will use closures that are linked.
We will need to guarantee the existence of such closures and to give a method for finding them.
For this, the following definition will be useful.
We say that a $c$-small $k$-closure $\prt$ of $\Tpref$ is \emph{minimal} if among all $c$-small $k$-closures it
\begin{itemize}
\item primarily minimizes $\sum_{C \in \prt} \cutrk(C)$, and
\item secondarily minimizes the number of nodes of $T$ that it cuts.
\end{itemize}

Then, the following lemma guarantees the existence of linked $c$-small $k$-closures and provides a method for finding them.

\begin{lemma}
\label{lem:closlink}
Any minimal $c$-small $k$-closure of $\Tpref$ is linked.
\end{lemma}
\begin{proof}
Suppose $\prt$ is a minimal $c$-small $k$-closure of $\Tpref$ that is not linked.
Let $C \in \prt$ be a part that violates the linkedness condition, in particular, with $C \subseteq \lparts(\Tc)[a]$ for some $a \in \App_T(\Tpref)$ so that there is a set $S$ with $C \subseteq S \subseteq \lparts(\Tc)[a]$ and either
\begin{enumerate}
\item\label{lem:closlink:case1} $\cutrk(S) < \cutrk(C)$, or
\item\label{lem:closlink:case2} $\cutrk(S) = \cutrk(C)$ and $S = C \cup \lparts(\Tc)[t]$ for some descendant $t$ of $a$ so that $C$ cuts $t$.
\end{enumerate}
Let us moreover fix such set $S$ that minimizes $\cutrk(S)$.
We will use the set $S$ to construct a new $c$-small $k$-closure $\prt'$ that will contradict the minimality of $\prt$.

We let $\prt' = \{S\} \cup \{D \setminus S \mid D \in \prt \text{ and } D \not\subseteq S\}$.
Let us first show that if $\prt'$ is a $k$-closure then it contradicts the minimality of $\prt$, and then show that it indeed is a $k$-closure.
First, the facts that $C \subseteq S$ and this construction changes only parts that are subsets of $\lparts(\Tc)[a]$ implies that $\prt'$ is $c$-small.
In order to bound $\sum_{C' \in \prt'} \cutrk(C')$ we show the following.
\begin{claim}
\label{lem:closlink:claimrk}
For all $D \in \prt \setminus \{C\}$ it holds that $\cutrk(D \setminus S) \le \cutrk(D)$.
\end{claim}
\begin{claimproof}
Note that if $D$ is a not subset of $\lparts(\Tc)[a]$, then $D$ is disjoint from $S$ and this holds trivially, so we can assume that $D \subseteq \lparts(\Tc)[a]$.
Recall that for a set $X$ we denote $\compl{X} = V(G) \setminus X$.
First we observe that
\begin{equation}
\label{lem:closlink:claimrk:theeq}
\cutrk(S) \le \cutrk(\compl{D} \cap S)
\end{equation}
because $C \subseteq \compl{D} \cap S \subseteq \lparts(\Tc)[a]$, but $S$ minimizes $\cutrk(S)$ among such sets.
Then, 
\begin{align*}
\cutrk(D \setminus S) &= \cutrk(D \cap \compl{S}) = \cutrk(\compl{D} \cup S) && \text{(symmetry of $\cutrk$)}\\
&\le \cutrk(\compl{D}) + \cutrk(S) - \cutrk(\compl{D} \cap S) && \text{(submodularity of $\cutrk$)}\\
&\le \cutrk(\compl{D}) = \cutrk(D). && \text{(\Cref{lem:closlink:claimrk:theeq} and symmetry)}
\end{align*}
\end{claimproof}

\Cref{lem:closlink:claimrk} and the fact that $\cutrk(S) \le \cutrk(C)$ imply that $\sum_{C' \in \prt'} \cutrk(C') \le \sum_{C \in \prt} \cutrk(C)$.
Moreover, if $\cutrk(S) < \cutrk(C)$ then in fact $\sum_{C' \in \prt'} \cutrk(C') < \sum_{C \in \prt} \cutrk(C)$, so in the case of \Cref{lem:closlink:case1} we have already contradicted the minimality of $\prt$ and do not need to consider the secondary minimization.

Then suppose we are in the case of \Cref{lem:closlink:case2}.
First we show that if $\prt'$ cuts some node $x \in V(T)$, then also $\prt$ cuts $x$.
If $x$ is a descendant of $t$, then $\lparts(\Tc)[x] \subseteq S$, so $\prt'$ does not cut $x$.
If $\lparts(\Tc)[x]$ is disjoint from $\lparts(\Tc)[t]$, then $D \cap \lparts(\Tc)[x] = (D \setminus S) \cap \lparts(\Tc)[x]$ for all $D \in \prt \setminus \{C\}$ and $C \cap \lparts(\Tc)[x] = S \cap \lparts(\Tc)[x]$, so $\prt'$ cuts $x$ if and only if $\prt$ cuts $x$.
If $x$ is an ancestor of $t$, then $\prt$ cuts $x$ because $C$ cuts $t$.
Then, the fact that $\prt$ cuts $t$ but $\prt'$ does not cut $t$ implies that $\prt'$ cuts fewer nodes of $T$ than $\prt$.

Next we show that $\prt'$ is a $k$-closure of $\Tpref$.
Because $S \subseteq \lparts(\Tc)[a]$, it holds that for all $C' \in \prt'$ there exists $a' \in \App_T(\Tpref)$ with $C' \subseteq \lparts(\Tc)[a']$.
It remains to bound the rankwidth of $(G[\prt'], \prt')$.

\begin{claim}
\label{lem:closlink:claimrw}
The rankwidth of $(G[\prt'], \prt')$ is at most the rankwidth of $(G[\prt], \prt)$.
\end{claim}
\begin{claimproof}
Let $\Tc^* = (T^*,\lmap^*)$ be an optimum-width rank decomposition of $(G[\prt], \prt)$.
We modify $\Tc^*$ into a rank decomposition $\Tc' = (T',\lmap')$ of $(G[\prt'], \prt')$ by simply mapping $S \in \prt'$ to the leaf to which $C$ was mapped, and for each $D \setminus S \in \prt'$ mapping $D \setminus S$ to the leaf to which $D$ was mapped.
This could create some leaves to which no parts of $\prt'$ are mapped, so finally we iteratively remove leaves with no mapped parts and contract edges of degree $2$.

Consider an edge $x'y' \in E(T')$, and suppose w.l.o.g. that $S \subseteq \lparts(\Tc')[\vec{x'y'}]$.
Then there exists an oriented edge $\oxy \in \oE(T)$ so that $C \subseteq \lparts(\Tc^*)[\oxy]$ and $\lparts(\Tc')[\vec{x'y'}] = \lparts(\Tc^*)[\oxy] \cup S$.
Therefore it suffices to show that $\cutrk(\lparts(\Tc^*)[\oxy] \cup S) \le \cutrk(\lparts(\Tc^*)[\oxy])$.
First, we note that
\begin{equation}
\label{lem:closlink:claimrw:theeq2}
\cutrk(S) \le \cutrk(\lparts(\Tc^*)[\oxy] \cap S)
\end{equation}
because $C \subseteq \lparts(\Tc^*)[\oxy] \cap S \subseteq \lparts(\Tc^*)[a]$, but $S$ minimizes $\cutrk(S)$ among such sets.
Then,
\begin{align*}
\cutrk(\lparts(\Tc^*)[\oxy] \cup S) &\le \cutrk(\lparts(\Tc^*)[\oxy]) + \cutrk(S) - \cutrk(\lparts(\Tc^*)[\oxy] \cap S) && \text{(submodularity)}\\
&\le \cutrk(\lparts(\Tc^*)[\oxy]). && \text{(\Cref{lem:closlink:claimrw:theeq2})}
\end{align*}
\end{claimproof}
This finishes the proof that $\prt'$ is a $c$-small $k$-closure that contradicts the minimality of $\prt$.
\end{proof}

We then observe the main consequence of closure linkedness.

\begin{lemma}
\label{lem:mincloswidthbound}
Let $\prt$ be a $k$-closure of $\Tpref$ that is linked.
If $C \in \prt$ and $C$ cuts a node $t \in V(T) \setminus \Tpref$, then it holds that $\cutrk(C \cap \lparts(\Tc)[t]) < \cutrk(\lparts(\Tc)[t])$.
\end{lemma}
\begin{proof}
Suppose $\cutrk(C \cap \lparts(\Tc)[t]) \ge \cutrk(\lparts(\Tc)[t])$.
Then from submodularity it follows that $\cutrk(C \cup \lparts(\Tc)[t]) \le \cutrk(C)$, which contradicts that $\prt$ is linked.
\end{proof}

\paragraph{Computing closures.}
For a $k$-closure $\prt$ of $\Tpref$, we denote by $\cut_T(\prt)$ the set of nodes of $T$ that are cut by $\prt$.
Note that $\cut_T(\prt)$ is a prefix of $T$ and $\Tpref \subseteq \cut_T(\prt)$.
We wish to manipulate $k$-closures in time proportional to $|\cut_T(\prt)|$.
Let $C \in \prt$.
The \emph{appendix edge set} $\aes_T(C)$ of $C$ is the set $\aes_T(C) = \{\oap \in \oApp_T(\cut_T(\prt)) \mid \lparts(\Tc)[\oap] \subseteq C\} \subseteq \oApp_T(\cut_T(\prt))$ of appendix edges of $\cut_T(\prt)$ that correspond to $C$.
Then, we define the \emph{appendix edge partition} $\aep_T(\prt)$ of $\prt$ to be the partition $\aep_T(\prt) = \{\aes_T(C) \mid C \in \prt\}$ of $\oApp_T(\cut_T(\prt))$.
Note that $|\oApp_T(\cut_T(\prt))| = |\cut_T(\prt)|+1$, so the appendix edge partition can be represented in space $\Oh{|\cut_T(\prt)|}$.


We will use the following prefix-rebuilding data structure for computing closures.
We defer the proof to \Cref{sec:computing-closures}, but the idea will be to adapt the dynamic programming of~\cite{DBLP:journals/siamdm/JeongKO21} for computing optimal rank decompositions to our setting.

\begin{restatable}{lemma}{efficientclosureprds}
\label{lem:closureprds}
There is an $\ell$-prefix-rebuilding data structure that takes integer parameters $c \ge 1$ and $k \le \ell$ at initialization, has overhead $\Oh[c,\ell]{1}$, maintains a rooted annotated rank decomposition $\Tc$, and additionally supports the following query:
\begin{itemize}
\item $\mathsf{Closure}(\Tpref)$: Given a prefix $\Tpref$ of $\Tc$, either in time $\Oh[\ell]{|\Tpref|}$ returns that no $c$-small $k$-closure of $\Tpref$ exists, or for a minimal $c$-small $k$-closure $\prt$ of $\Tpref$ in time $\Oh[\ell]{|\cut_T(\prt)|}$ returns
\begin{itemize}
\item the sets $\cut_T(\prt)$ and $\aep_T(\prt)$, and
\item a rooted rank decomposition $(T^*,\lmap^*)$ of $(G[\prt],\prt)$ of width at most $2k$, where $\lmap^*$ is represented as a function $\lmap \colon \aep_T(\prt) \rightarrow \leafe(T^\star)$.
\end{itemize}
\end{itemize}
\end{restatable}

%% file: automata.tex
\section{Automata}
\label{sec:rwautom}
In this section we define rank decomposition automata in order to formalize and unify dynamic programming working on rank decompositions.
We give a prefix-rebuilding data structure to maintain the runs of rank decomposition automata, give a construction of rank decomposition automata from $\CMSO_1$ sentences (using the construction for cliquewidth by~\cite{CourcelleMR00} as a black-box), and finally give our framework for performing edge updates using $\CMSO_1$.

\subsection{Rank decomposition automata}
\label{sec:rdautom}
We will define a \emph{rank decomposition automaton}, which is an automaton that processes annotated rank decompositions.
Our definitions will be for unrooted annotated rank decompositions, in particular, so that they are suited for computing dynamic programming tables directed in both directions on edges.
While these definitions allow annotated rank decompositions that encode partitioned graphs with non-trivial partitions, they are usually used with annotated rank decompositions that encode graphs.
Let us start with some auxiliary definitions.

We say that a \emph{transition signature} of width $\ell$ is a tuple $\tau = (S_\tau, U_\tau, \reps_\tau, \repse_\tau, \dmap_\tau)$,
where 
\begin{itemize}
\item $S_\tau$ is a tree with three leaf nodes and one non-leaf node,
\item $U_\tau$ is a set of size at most $6 \cdot 2^{\ell}$,
\item $\reps_\tau$ is a function that maps each oriented edge $\oxy \in \oE(S_\tau)$ to a non-empty set $\reps_\tau(\oxy) \subseteq U_\tau$,
\item $\repse_\tau$ is a function that maps each edge $xy \in E(S_\tau)$ to a bipartite graph $\repse_\tau(xy)$ with bipartition $(\reps_\tau(\oxy), \reps_\tau(\oyx))$, with no twins over this bipartition, and with $\cutrk_{\repse_\tau(xy)}(\reps_\tau(\oxy)) \le \ell$, and
\item $\dmap_\tau$ is a function that maps each path of length three $xyz \in \PT(S_\tau)$ in $S_\tau$ to a function $\dmap_\tau(xyz) \colon \reps_\tau(\oxy) \rightarrow \reps_\tau(\oyz)$.
\end{itemize}

Let $\Tc = (T,U,\reps,\repse,\dmap)$ be an annotated rank decomposition and $\vec{tp} \in \oE(T) \setminus \leafe(T)$ a non-leaf oriented edge of $T$ with children $\vec{c_1 t}$ and $\vec{c_2 t}$.
The transition signature of $\Tc$ at $\vec{tp}$, denoted by $\tau(\Tc,\vec{tp})$, is the transition signature obtained by setting $S_\tau = T[\{t,p,c_1,c_2\}]$, $\reps_\tau = \funrestriction{\reps}{\oE(S_\tau)}$, $\repse_\tau = \funrestriction{\repse}{E(S_\tau)}$, $\dmap_\tau = \funrestriction{\dmap}{\PT(S_\tau)}$, and $U_\tau = \bigcup_{\oed \in \oE(S_\tau)} \reps_\tau(\oed)$.
We observe that the width of $\tau(\Tc,\vec{tp})$ is at most the width of $\Tc$.

Then we say that an \emph{edge signature} of width $\ell$ is a tuple $\sigma = (\reps^a_\sigma, \reps^b_\sigma, \repse_\sigma)$,
where
\begin{itemize}
\item $\reps^a_\sigma$ and $\reps^b_\sigma$ are sets of size at most $2^{\ell}$ and
\item $\repse_\sigma$ is a bipartite graph with bipartition $(\reps^a_\sigma,\reps^b_\sigma)$, with no twins over this bipartition, and with $\cutrk_{\repse_\sigma}(\reps^a_\sigma) \le \ell$.
\end{itemize}

Let $\oab \in \oE(T)$.
The edge signature of $\Tc$ at $\oab$ is $\sigma(\Tc,\oab) = (\reps(\oab), \reps(\oba), \repse(ab))$.
Again, the width of $\sigma(\Tc,\oab)$ is at most the width of $\Tc$.

A rank decomposition automaton of width $\ell$ is a tuple $\autom = (Q,\Gamma, \iota, \delta, \varepsilon)$ that consists of
\begin{itemize}
\item a state set $Q$,
\item a vertex label set $\Gamma$,
\item an initial mapping $\iota$ that maps every pair of form $(\sigma,\gamma)$, where $\sigma = (\reps^a_\sigma, \reps^b_\sigma, \repse_\sigma)$ is an edge signature of width $\ell$ and $\gamma$ is a function $\gamma \colon \reps^a_\sigma \rightarrow \Gamma$, to a state $\iota(\sigma,\gamma) \in Q$,
\item a transition mapping $\delta$ that maps every triple of form $(\tau, q_1, q_2)$, where $\tau$ is a transition signature of width $\ell$ and $q_1,q_2 \in Q$, to a state $\delta(\tau, q_1, q_2) \in Q$, and
\item a final mapping $\varepsilon$ that maps every triple of form $(\sigma, q_1, q_2)$, where $\sigma$ is an edge signature of width $\ell$ and $q_1,q_2 \in Q$, to a state $\varepsilon(\sigma, q_1, q_2) \in Q$.
\end{itemize}

The state set $Q$ is allowed to be infinite.
The \emph{evaluation time} of a rank decomposition automaton is the maximum running time to compute the functions $\iota(\sigma,\gamma)$, $\delta(\tau,q_1,q_2)$, or $\varepsilon(\sigma,q_1,q_2)$ given their arguments.

Let $\Tc = (T,V(G),\reps,\repse,\dmap)$ be an annotated rank decomposition of width at most $\ell$ that encodes a partitioned graph $(G,\prt)$, $\oxy \in \oE(T)$ an oriented edge of $T$, and $\alpha \colon V(G) \rightarrow \Gamma$ a vertex-labeling of $G$ with $\Gamma$.
Recall that $\pred_T(\oxy)$ denotes the set of predecessor of $\oxy$.
The \emph{run} of $\autom$ on the triple $(\Tc, \oxy, \alpha)$ is the unique mapping $\rho \colon \pred_T(\oxy) \rightarrow Q$ so that
\begin{itemize}
\item for each leaf edge $\olp \in \pred_T(\oxy) \cap \leafe(T)$ it holds that $\rho(\olp) = \iota(\sigma(\Tc,\olp), \funrestriction{\alpha}{\reps(\olp)})$, and
\item for each non-leaf edge $\vec{tp} \in \pred_T(\oxy) \setminus \leafe(T)$ with children $\vec{c_1 t}$, $\vec{c_2 t}$, where $c_1 < c_2$, it holds that $\rho(\vec{tp}) = \delta(\tau(\Tc,\vec{tp}), \rho(\vec{c_1 t}), \rho(\vec{c_2 t}))$.
\end{itemize}

Then let $a,b \in V(T)$ be two adjacent nodes of $T$.
The run of $\autom$ on the 4-tuple $(\Tc,a,b,\alpha)$ is the unique mapping $\rho \colon \pred_T(\oab) \cup \pred_T(\oba) \cup \vartheta \rightarrow Q$ so that 
\begin{itemize}
\item $\funrestriction{\rho}{\pred_T(\oab)}$ is the run of $\autom$ on $(\Tc,\oab,\alpha)$,
\item $\funrestriction{\rho}{\pred_T(\oba)}$ is the run of $\autom$ on $(\Tc,\oba,\alpha)$, and
\item $\rho(\vartheta) = \varepsilon(\sigma(\Tc,\oab), \rho(\oab), \rho(\oba))$.
\end{itemize}

The \emph{valuation} of $\autom$ on $(\Tc,a,b,\alpha)$ is $\rho(\vartheta)$ and on $(\Tc, \oxy, \alpha)$ is $\rho(\oxy)$.
These definitions are adapted to a rooted annotated rank decompositions with root $r$ whose children are $c_1, c_2$ by setting $\rho(\vec{r c_2}) \coloneqq \rho(\vec{c_1 r})$ and $\rho(\vec{r c_1}) \coloneqq \rho(\vec{c_2 r})$.
Additionally, the run (resp.\ valuation) of $\autom$ on $(\Tc, \alpha)$ is defined as the run (resp.\ valuation) of $\autom$ on $(\Tc, c_1, r, \alpha)$, where $c_1 < c_2$.

If the valuation of $\autom$ on $(\Tc,a,b,\alpha)$ depends only on the partitioned graph $(G,\prt)$ encoded by $\Tc$ and the labeling $\alpha$, then we say that $\autom$ is \emph{decomposition-oblivious}, and refer to this valuation as the valuation of $\autom$ on $(G,\prt,\alpha)$.
When $\Tc$ encodes a graph $G$, we refer to this as the valuation of $\autom$ on $(G,\alpha)$.

Next, if all runs of $\autom$ on $(\Tc, a, b, \alpha)$ are independent on the labeling $\alpha$ (in particular, the value of the initial mapping $\iota$ only depends on the edge signature $\sigma$ and not the function $\gamma \colon \reps^a_\sigma \rightarrow \Gamma$), then we say that $\autom$ is \emph{label-oblivious}.
When defining label-oblivious automata, we will for convenience drop the vertex label set $\Gamma$ from the description of the automaton and consider $\iota$ to be a~mapping from an~edge signature $\sigma = (\reps^a_\sigma, \reps^b_\sigma, \repse_\sigma)$ to a~state $\iota(\sigma)$.
We also define the runs on $\autom$ on pairs $(\Tc, \vec{xy})$ and on triples $(\Tc, a, b)$ in a~natural way.
If $\Tc$ is rooted, we also define the run of $\autom$ on $\Tc$ naturally.

Then we give a prefix-rebuilding data structure for maintaining runs of rank decomposition automata.

\begin{lemma}
\label{lem:automatonprds}
Let $\ell \in \N$ and $\autom = (Q,\Gamma,\iota,\delta,\varepsilon)$ a rank decomposition automaton of width $\ell$ with evaluation time $\beta$.
There exists an $\ell$-prefix-rebuilding data structure with overhead $\Oh[\ell]{1} + \Oh{\beta}$ that maintains a rooted annotated rank decomposition $\Tc = (T,V(G),\reps,\repse,\dmap)$ that encodes a dynamic graph $G$, and a vertex-labeling $\alpha \colon V(G) \rightarrow \Gamma$ whose initial values $\alpha_{\mathsf{init}}$ are given at the initialization, and additionally supports the following operations:
\begin{itemize}
\item $\mathsf{Run}(\oxy)$: Given an oriented edge $\oxy \in \oE(T)$ that is directed towards the root, in time $\Oh{1}$ returns $\rho(\oxy)$, where $\rho$ is the run of $\autom$ on $(\Tc,\oxy,\alpha)$.
\item $\mathsf{Valuation}()$: In time $\Oh{1}$ returns the valuation of $\autom$ on $(\Tc,\alpha)$.
\item $\mathsf{SetLabel}(v, \gamma)$: Given a vertex $v \in V(G)$ and a label $\gamma \in \Gamma$, in time $\Oh{\height(T) \cdot \beta}$ updates $\alpha(v) \coloneqq \gamma$.
\end{itemize}
\end{lemma}
\begin{proof}
We maintain a representation of $\Tc$ with \Cref{lem:basicprds}.
We also maintain the vertex labeling $\alpha$ explicitly, and the runs of $\autom$ on $(\Tc, \vec{c_1 r}, \alpha)$ and $(\Tc, \vec{c_2 r}, \alpha)$, where $r$ is the root and $c_1 < c_2$ are the children of $r$.
Note that this stores exactly one state $\rho(\oxy)$ for each oriented edge $\oxy$ of $T$ directed towards the root.
We also maintain the valuation of $\autom$ on $(\Tc, \alpha)$, which is $\varepsilon(\sigma(\Tc,\vec{c_1 r}), \rho(\vec{c_1 r}), \rho(\vec{c_2 r}))$.

At initialization, we can compute the runs and the valuations in $\Oh{|\Tc| \cdot \beta}$ time.
Then, consider a prefix-rebuilding update that turns $\Tc$ into $\Tc' = (T',V(G),\reps',\repse',\dmap')$, where the prefix of $\Tc$ associated with the update is $\Tpref$ and the prefix of $\Tc'$ is $\Tpref'$.
We observe that all edge signatures and transition signatures at edges directed towards the root in $\oE(T) \setminus \oE(T[\Tpref \cup \App_T(\Tpref)])$ stay the same in $\Tc'$.
Therefore, to recompute the runs and valuations, it suffices to recompute this information only for edges directed towards the root in $\oE(T'[\Tpref' \cup \App_{T'}(\Tpref')])$, which takes $\Oh{|\Tpref| \cdot \beta}$ time.

Then consider the $\mathsf{SetLabel}$ operation.
We observe that it can change the run on $(\Tc, \oxy, \alpha)$ only if $v \in \lparts(\Tc)[\oxy]$.
There are at most $\height(T)$ such edges $\oxy$ directed towards the root, so we recompute the runs on them in $\Oh{\height(T) \cdot \beta}$ time.

We explicitly maintain all information required to answer the $\mathsf{Run}$ and $\mathsf{Valuation}$ queries, so they can be answered in $\Oh{1}$ time.
\end{proof}

\subsection{\texorpdfstring{$\CMSO_1$}{CMSO1}}
\label{subsec:maincmso}
Monadic second-order logic ($\MSO$) is the fragment of second-order logic where quantification is allowed only over single elements of the universe and subsets of the universe.
In logic of graphs, $\MSO_1$ refers to $\MSO$ on the representation of graphs as a relational structure where the universe is the vertices and there is a binary relation describing the vertex adjacencies.
In particular, in $\MSO_1$ we can quantify over sets of vertices, but not over sets of edges.
The extension of $\MSO_1$ with predicates that allow counting the cardinality of a set modulo some given constant is called $\CMSO_1$.
We refer the reader to~\cite{CEbook} for more precise definitions.

For simplicity, we assume in this paper that all free variables of a $\CMSO_1$ sentence are set variables (note that free single-element variables can be expressed as free set variables).
The length of a $\CMSO_1$ sentence $\varphi$ is the number of symbols appearing in it, and denoted by $|\varphi|$.
We note that the length of $\varphi$ is at least the number of free variables of $\varphi$, and use the convention that the free variables are indexed by consecutive integers $1,\ldots,p$.

Let $\varphi$ be a $\CMSO_1$ sentence with $p$ free variables and $G$ a graph.
A tuple $(G,X_1,\ldots,X_p)$, where $X_i \subseteq V(G)$, \emph{satisfies} $\varphi$, written as $(G,X_1,\ldots,X_p) \models \varphi$, if $G$ together with the interpretations of the free variables as $X_1,\ldots,X_p$ satisfies $\varphi$.
Let $\alpha \colon V(G) \rightarrow 2^{[p]}$ be a vertex-labeling of $G$.
We define that $(G,\alpha)$ satisfies $\varphi$ if $(G,X_1,\ldots,X_p)$, where $X_i = \{v \in V(G) \mid i \in \alpha(v)\}$ satisfies $\varphi$.

We prove the following lemma in \Cref{sec:cliquewidth} by translating automata working on a cliquewidth expressions given by Courcelle, Makowsky, and Rotics~\cite{CourcelleMR00} (see also~\cite[Section 6]{CEbook}) to rank decomposition automata.

\begin{restatable}{lemma}{cmsordaut}
\label{lem:cmsordaut}
There is an algorithm that given a $\CMSO_1$ sentence $\varphi$ with $p$ free set variables and $\ell \in \N$, in time $\Oh[\varphi,\ell]{1}$ constructs a decomposition-oblivious rank decomposition automaton $\autom = (Q,\Gamma,\iota,\delta,\varepsilon)$ of width $\ell$ so that $\Gamma = 2^{[p]}$, the valuation of $\autom$ on $(G,\alpha)$ is $\top \in Q$ if and only if $(G,\alpha) \models \varphi$, the number of states is $|Q| \le \Oh[\varphi,\ell]{1}$, and the evaluation time is $\Oh[\varphi,\ell]{1}$.
\end{restatable}

In order to express optimization problems in the language of $\MSO_1$, Courcelle, Makowsky, and Rotics~\cite{CourcelleMR00} defined an extension of $\MSO_1$ they called ``LinEMSOL''.
Similar extension of $\CMSO_1$ was also discussed by Courcelle and Engelfriet~\cite[Section 6]{CEbook}.
Based on~\cite{CourcelleMR00,CEbook}, we define an extension of $\CMSO_1$ that we call $\LinCMSO_1$.
A $\LinCMSO_1$ sentence with $p$ free variables is a pair $(\varphi, f)$, where $\varphi$ is a $\CMSO_1$ sentence with $p+q$ free variables for $q\ge 0$, and $f \colon \Z^q \rightarrow \Z$ a linear integer function defined by $q+1$ integers $c_0,\ldots,c_q$ so that $f(x_1,\ldots,x_q) = c_0 + c_1 x_1 + \ldots + c_q x_q$.
Then, the value of $(\varphi,f)$ on a tuple $(G,X_1,\ldots,X_p)$ is the maximum value of $f(|X_{p+1}|, \ldots, |X_{p+q}|)$, where $X_{p+1},\ldots,X_{p+q} \subseteq V(G)$ and $(G,X_1,\ldots,X_{p+q}) \models \varphi$.
If no such sets $X_{p+1},\ldots,X_{p+q}$ exist, then the value is $\bot$.
We note that even though this naturally defines only maximization problems, we can define minimization problems by using negative coefficients.
We define the length $|(\varphi,f)|$ of $(\varphi,f)$ to be $|\varphi|+\sum_{i=0}^p |c_i|$.

Then, \Cref{lem:cmsordaut} extends to the following lemma.
The proof is in \Cref{sec:cliquewidth}.

\begin{restatable}{lemma}{lincmsordaut}
\label{lem:lincmsordaut}
There is an algorithm that given a $\LinCMSO_1$ sentence $\varphi$ with $p$ free set variables and $\ell \in \N$, in time $\Oh[\varphi,\ell]{1}$ constructs a decomposition-oblivious rank decomposition automaton $\autom = (Q,\Gamma,\iota,\delta,\varepsilon)$ of width $\ell$ so that $\Gamma = 2^{[p]}$, the valuation of $\autom$ on $(G,\alpha)$ is equal to the value of $\varphi$ on $(G,\alpha)$, and the evaluation time is $\Oh[\varphi,\ell]{1}$.
\end{restatable}

We note that the reason for having \Cref{lem:cmsordaut} and \Cref{lem:lincmsordaut} as separate lemmas is that we will use the fact that the number of states in the automaton constructed in \Cref{lem:cmsordaut} is $\Oh[\varphi,\ell]{1}$.
We also note that in both \Cref{lem:cmsordaut,lem:lincmsordaut} the constructed automaton works only on decompositions encoding graphs, not partitioned graphs.

By putting together \Cref{lem:automatonprds,lem:lincmsordaut}, we obtain the following.
\begin{lemma}
\label{lem:lincmsoprds}
Let $w, \ell \in \N$.
There exists an $\ell$-prefix-rebuilding data structure with overhead $\Oh[\ell,w]{1}$ that maintains a rooted annotated rank decomposition $\Tc$ that encodes a dynamic graph $G$, and additionally supports the following query:
\begin{itemize}
\item $\LinCMSO_1(\varphi,X_1,\ldots,X_p)$: Given a $\LinCMSO_1$ sentence $\varphi$ of length at most $w$ with $p$ free variables and $p$ vertex subsets $X_1,\ldots,X_p \subseteq V(G)$, returns the value of $\varphi$ on $(G,X_1,\ldots,X_p)$. Runs in time $\Oh[\varphi]{1}$ if the sets are empty, and in $\Oh[\ell,\varphi]{\sum_{i=1}^p |X_i| \cdot \height(\Tc)}$ time otherwise.
\end{itemize}
\end{lemma}
\begin{proof}
We enumerate all $\LinCMSO_1$ sentences $\varphi$ of length at most $w$, and for each of them construct an auxiliary $\ell$-prefix-rebuilding structure $\D^{\varphi}$ as follows.
Let $p$ be the number of free variables in $\varphi$.
We apply \Cref{lem:lincmsordaut} to obtain a rank decomposition automaton $\autom = (Q,\Gamma,\iota,\delta,\varepsilon)$ of width $\ell$ so that $\Gamma = 2^{[p]}$, the valuation of $\autom$ on $(G,\alpha)$ is equal to the value of $\varphi$ on $(G,\alpha)$, and the evaluation time of $\autom$ is $\Oh[\varphi,\ell]{1}$.
Then we initialize an $\ell$-prefix-rebuilding data structure $\D^{\varphi}$ of \Cref{lem:automatonprds} with $\autom$.
The overhead of $\D^{\varphi}$ is $\Oh[\varphi,\ell]{1}$.
We initialize the labeling $\alpha$ held by $\D^{\varphi}$ to be $\alpha(v) = \emptyset$ for all $v \in V(G)$.

Note that there are at most $\Oh[w]{1}$ $\LinCMSO_1$ sentences of length at most $w$, so the initialization works in $\Oh[\ell,w]{1}$ time.
Then, all prefix-rebuilding updates to our data structures are relayed to all of the auxiliary data structures $\D^{\varphi}$ so that they also hold the decomposition $\Tc$ at all times, resulting in the overhead $\Oh[\ell,w]{1}$.

The $\LinCMSO_1(\varphi,X_1,\ldots,X_p)$ query is implemented as follows.
We maintain that between the queries, the labeling $\alpha$ held by $\D^{\varphi}$ is $\alpha(v) = \emptyset$ for all $v \in V(G)$.
Therefore, if the given sets $X_1,\ldots,X_p$ are empty, we can simply return the value given by the query $\mathsf{Valuation}()$ of $\D^{\varphi}$.
This runs in $\Oh[\varphi]{1}$ time.
If some of the sets $X_1,\ldots,X_p$ is non-empty, we compute $X = X_1 \cup \ldots \cup X_p$, use the $\mathsf{SetLabel}$ query of $\D^{\varphi}$ to set $\alpha(v) = \{i \mid v \in X_i\}$ for all $v \in X$, and return the value given by the query $\mathsf{Valuation}()$ of $\D^{\varphi}$.
Then, we reset the labels $\alpha(v)$ of all $v \in X$ to be $\emptyset$.
This takes $\Oh[\ell,\varphi]{|X| \cdot \height(\Tc)} = \Oh[\ell,\varphi]{\sum_{i=1}^p |X_i| \cdot \height(\Tc)}$ time.
\end{proof}

\subsection{Edge update sentences}
Let $G$ be a graph.
An \emph{edge update sentence} on $G$ is a tuple $\eudesc = (\varphi, X, X_1, \ldots, X_p)$, where $\varphi$ is a $\CMSO_1$ sentence with $p+1$ free set variables, $X \subseteq V(G)$, and $X_i \subseteq X$ for all $i \in [p]$.
The graph resulting from applying $\eudesc$ to $G$ is the graph $G'$ with $V(G') = V(G)$, and with $uv \in E(G')$ for $u \neq v$ if and only if either 
\begin{itemize}
\item $|\{u,v\} \cap X| \le 1$ and $uv \in E(G)$, or
\item $u,v \in X$ and $(G,\{u,v\},X_1,\ldots,X_p) \models \varphi$.
\end{itemize}

In other words, the edges inside $G[X]$ are defined by $\eudesc$, while other edges remain unchanged.
We define that \emph{size} of $\eudesc$ as $|\eudesc| = |X|$ and that the \emph{length} of $\eudesc$ is the length of $\varphi$, i.e., $|\varphi|$.

Next we give our data structure to turn edge update sentences to edge update descriptions.
We note that while it is not immediately obvious that a rank decomposition of $G$ of width $\ell$ would also be a rank decomposition of $G'$ whose width is bounded by $\Oh[\ell,|\varphi|]{1}$, our proof implies this because the resulting edge update description has width $\Oh[\ell,|\varphi|]{1}$.

\begin{lemma}
\label{lem:mso_edge_update}
Let $d,\ell \in \N$.
There exists an $\ell$-prefix-rebuilding data structure with overhead $\Oh[\ell,d]{1}$ that maintains a rooted annotated rank decomposition $\Tc$ that encodes a dynamic graph $G$ and additionally supports the following query:
\begin{itemize}
\item $\mathsf{EdgeUpdate}(\eudesc)$: Given an edge update sentence $\eudesc$ on $G$ of length at most $d$, returns an edge update description of width $\Oh[\ell,d]{1}$ that describes the graph $G'$ that results from applying $\eudesc$ to $G$. Runs in time $\Oh[\ell,d]{\height(\Tc) \cdot |\eudesc|}$.
\end{itemize}
\end{lemma}
\begin{proof}
In the initialization we construct a set of auxiliary automata and prefix-rebuilding data structures as follows.
We enumerate all $\CMSO_1$ sentences of length at most $d$ and at least one free set variable, i.e., all $\CMSO_1$ sentences that could be in the edge update sentence given in $\mathsf{EdgeUpdate}(\eudesc)$.
Let $\varphi$ be such sentence with $p+1$ free variables $Y, X_1, \ldots, X_p$, where $Y$ is the free variable that is supposed to hold the endpoints of the potential edge.
We construct a $\CMSO_1$ sentence $\varphi'$ with $p+2$ free variables $Y, X, X_1, \ldots, X_p$, so that $(G, Y, X, X_1, \ldots, X_p) \models \varphi'$ if and only if either 
\begin{itemize}
\item $Y \subseteq X$, $|Y| = 2$, and $(G, Y, X_1, \ldots, X_p)$ satisfies $\varphi$, or 
\item $Y \not\subseteq X$ and $Y = \{u,v\}$ with $uv \in E(G)$.
\end{itemize}
In particular, $(G, Y, X, X_1, \ldots, X_p) \models \varphi'$ if and only if $Y = \{u,v\}$ corresponds to an edge in the graph $G'$ resulting from applying the edge update sentence $(\varphi,X,X_1, \ldots, X_p)$.
Such $\varphi'$ with $|\varphi'| \le \Oh{|\varphi|}$ can be constructed in time $\Oh{|\varphi|}$.

Then we use \Cref{lem:cmsordaut} to construct a rank decomposition automaton $\autom_{\varphi'} = (Q,\Gamma,\iota,\delta,\epsilon)$ of width $\ell$ so that $\Gamma = 2^{[p+2]}$, the valuation of $\autom_{\varphi'}$ on $(G,\alpha)$ is $\top$ if and only if $(G,\alpha) \models \varphi'$, $|Q| \le \Oh[\varphi,\ell]{1}$, and the evaluation time is $\Oh[\varphi,\ell]{1}$.
We say that a labeling $\alpha \colon V(G) \rightarrow 2^{[p+2]}$ corresponds to an edge update sentence $(\varphi,X,X_1,\ldots,X_p)$ if $2 \in \alpha(v)$ if and only if $v \in X$, and $2+i \in \alpha(v)$ if and only if $v \in X_i$.

Let $\Tc = (T,V(G),\reps,\repse,\dmap)$ be an annotated rank decomposition that encodes $G$, and let $\alpha \colon V(G) \rightarrow 2^{[p+2]}$ be a labeling of $G$ with $1 \notin \alpha(v)$ for all $v \in V(G)$.
Let us also denote by $\alpha_{v}$ the labeling so that $\alpha_{v}(u) \setminus \{1\} = \alpha(u)$ for all $u \in V(G)$ and $1 \in \alpha_{v}(u)$ if and only if $u = v$.
With an oriented edge $\oxy \in \oE(T)$ we associate a 4-tuple $(q_{\oxy},f_{\oxy},g_{\oxy}, h_{\oxy})$ so that 
\begin{itemize}
\item $q_{\oxy}$ is the valuation of $\autom_{\varphi'}$ on $(\Tc,\oxy,\alpha)$,
\item $f_{\oxy} \colon Q \rightarrow V(G) \cup \{\bot\}$ is the function so that for every $q \in Q$ the value $f_{\oxy}(q)$ is the vertex $v \in \lparts(\Tc)[\oxy]$ with the smallest index so that the valuation of $\autom_{\varphi'}$ on $(\Tc,\oxy,\alpha_{v})$ is $q$, or $\bot$ if no such vertex $v$ exists,
\item $g_{\oxy} \colon \reps(\oxy) \rightarrow V(G)$ is the function so that for every $r \in \reps(\oxy)$ the value $g_{\oxy}(r)$ is the smallest-index vertex $v \in \lparts(\Tc)[\oxy]$ so that $N_G(v) \cap \reps(\oyx) = N_G(r) \cap \reps(\oyx)$, and
\item $h_{\oxy} \colon \reps(\oxy) \rightarrow Q$ is the function so that $h_{\oxy}(r)$ is the valuation of $\autom_{\varphi'}$ on $(\Tc,\oxy,\alpha_{g_{\oxy}(r)})$.
\end{itemize}

We construct a rank decomposition automaton $\autom_{\varphi'}'$ of width $\ell$ so that the valuation of $\autom_{\varphi'}'$ on $(\Tc,\oxy,\alpha)$ is the 4-tuple $(q_{\oxy},f_{\oxy},g_{\oxy},h_{\oxy})$.
Such automaton with evaluation time $\Oh[\ell,\varphi]{1}$ can be constructed as follows:
First, the state $q_{\oxy}$ can be maintained simply by simulating $\autom_{\varphi'}$.
Then, we observe that $f_{\oxy}$ can be computed from $f_{\ocax}$, $f_{\ocbx}$, $q_{\ocax}$, and $q_{\ocbx}$, where $\ocax$ and $\ocbx$ are the child edges of $\oxy$, in particular
\[f_{\oxy}(q) = \min\left\{\min_{q_1 \in Q \mid \delta(\tau(\Tc, \oxy), q_1, q_{\ocbx})=q} f_{\ocax}(q_1), \min_{q_2 \in Q \mid \delta(\tau(\Tc, \oxy), q_{\ocax}, q_2)=q} f_{\ocbx}(q_2)\right\},\]
where $\bot$ is regarded as larger than any vertex.
For $g_{\oxy}$ and $h_{\oxy}$, we first observe that if $g_{\oxy}(r) = v$, then there exists either $r' \in \reps(\ocax)$ with $g_{\ocax}(r') = v$ or $r' \in \reps(\ocbx)$ with $g_{\ocbx}(r') = v$.
With this observation, $g_{\oxy}$ can be computed from $g_{\ocax}$ and $g_{\ocbx}$ by using $\dmap(c_1 x y)$, $\dmap(c_2 x y)$, and $\repse(xy)$, which are stored in $\tau(\Tc, \oxy)$.
Then, if $g_{\oxy}(r) = v$ so that there exists $r' \in \reps(\ocax)$ with $g_{\ocax}(r') = v$, we have $h_{\oxy}(r) = \delta(\tau(\Tc,\oxy), h_{\ocax}(r'), q_{\ocbx})$; and the other case is similar.
This completes the construction of $\autom_{\varphi'}'$.

Then, we construct an $\ell$-prefix-rebuilding data structure $\D^{\varphi}$ by invoking \Cref{lem:automatonprds} with $\autom_{\varphi'}'$.
All prefix-rebuilding updates are relayed to $\D^{\varphi}$ so that it always holds the same annotated rank decomposition as the main prefix-rebuilding data structure of the lemma.
The vertex labeling $\alpha_{\varphi}$ that $\D^{\varphi}$ holds will always be $\alpha_{\varphi}(v) = \emptyset$ for all $v \in V(G)$, except when we are processing the $\mathsf{EdgeUpdate}(\eudesc)$ query.
Note that because $|\varphi| \le d$, the number of such prefix-rebuilding data structures $\D^{\varphi}$ we maintain is $\Oh[d]{1}$.

This completes the description of the initialization and the handling of prefix-rebuilding updates.
It remains to describe how $\mathsf{EdgeUpdate}(\eudesc)$ is implemented.

Let $\eudesc = (\varphi,X,X_1,\ldots,X_p)$.
We first use the $\mathsf{SetLabel}(v,\gamma)$ query of $\D^{\varphi}$ for all $v \in X$ to set the labeling $\alpha$ to correspond to $\eudesc$.
This takes $\Oh[\ell,\varphi]{\height(T) \cdot |\eudesc|}$ time.
Then, let $\Tpref$ be the unique smallest prefix of $T$ that contains all leaves $l \in \leafs(T)$ with $\reps(\olp) \subseteq X$.
We have that $|\Tpref| \le \height(T) \cdot |\eudesc|$.
The prefix $\Tpref$ will be the prefix of the edge update description we output.
With the help of $\D^{\varphi}$ we compute the triples $(q_{\oxy},f_{\oxy},g_{\oxy})$ for all oriented edges $\oxy \in \oE(T[\Tpref \cup \App_T(\Tpref)])$ in $\Oh[\ell,\varphi]{|\Tpref|} = \Oh[\ell,\varphi]{\height(T) \cdot |\eudesc|}$ time.
In particular, such triples are directly given by $\D^{\varphi}$ for all oriented edges directed towards the root, and for oriented edges directed towards the leaves we can compute them with $\autom_{\varphi'}'$ in a top-down manner.

Then, the purpose of the definition of $f_{\oxy}$ is to make the following hold.

\begin{claim}
\label{lem:mso_edge_update:claimf}
Let $\oxy \in \oE(T)$ and let $G'$ be the graph resulting from applying $\eudesc$ to $G$.
The set $R_{\oxy} = \bigcup_{q \in Q} \{f_{\oxy}(q)\} \setminus \{\bot\}$ is a representative of $\lparts(\Tc)[\oxy]$ in $G'$, and given $f_{\oxy}$ and $f_{\oyx}$ the graph $G'[R_{\oxy},R_{\oyx}]$ can be determined in $\Oh[\varphi,\ell]{1}$ time.
\end{claim}
\begin{claimproof}
Let $v \in \lparts(\Tc)[\oxy]$ and $u \in \lparts(\Tc)[\oyx]$.
We observe that $uv \in E(G')$ if and only if the valuation of $\autom_{\varphi'}$ on $(\Tc, \oxy, \alpha_{v})$ is $q_1$, the valuation of $\autom_{\varphi'}$ on $(\Tc, \oyx, \alpha_{v})$ is $q_2$, and $\varepsilon(\sigma(\Tc, \oxy), q_1, q_2) = \top$.
Therefore if $r \in R_{\oxy}$ and the valuations of $\autom_{\varphi'}$ on $(\Tc, \oxy, \alpha_{v})$ and $(\Tc, \oxy, \alpha_{r})$ are the same, then $N_{G'}(v) \cap \lparts(\Tc)[\oyx] = N_{G'}(r) \cap \lparts(\Tc)[\oyx]$.
Because for every $v \in \lparts(\Tc)[\oxy]$ there exists such $r \in R_{\oxy}$, we have that $R_{\oxy}$ is a representative of $\lparts(\Tc)[\oxy]$ in $G'$.
Then the graph $G'[R_{\oxy},R_{\oyx}]$ can be determined by verifying whether $\varepsilon(\sigma(\Tc, \oxy), q_1, q_2) = \top$ for all $q_1,q_2 \in Q$.
\end{claimproof}

In particular, by \Cref{lem:mso_edge_update:claimf} in the edge update description we can set $\reps^\star(\oxy) = R_{\oxy}$ for all $\oxy \in \oE(T[\Tpref])$.
It also gives a way to compute the graphs $\repse^\star(xy) = G'[R_{\oxy},R_{\oyx}]$ for $xy \in E(T[\Tpref])$.
For $xyz \in \PT(T[\Tpref])$, the graphs $\repse^\star(xyz) = G'[R_{\oxy}, R_{\ozy}]$ can be computed as follows.
Let $v \in R_{\oxy}$ and $u \in R_{\ozy}$, and let $w$ be the neighbor of $y$ that is not $x$ or $z$.
From $f_{\oxy}$ we know the valuation of $\autom_{\varphi'}$ on $(\Tc, \oxy, \alpha_{v})$, from $f_{\ozy}$ we know the valuation of $\autom_{\varphi'}$ on $(\Tc,\ozy,\alpha_{u})$, and from $q_{\vec{wy}}$ we know the valuation of $\autom_{\varphi'}$ on $(\Tc,\vec{wy},\alpha)$.
By combining these with $\Oh{1}$ transitions of $\autom_{\varphi'}$ we find whether $uv \in E(G')$.
This takes $\Oh[\varphi,\ell]{1}$ time for each $xyz \in \PT(T[\Tpref])$, i.e., $\Oh[\varphi,\ell]{\height(T) \cdot |\eudesc|}$ time in total.

It remains to compute for $xyz \in \PT(T)$ with $x \in \App_T(\Tpref)$ and $y,z \in \Tpref$ the graphs $\repse^\star(xyz) = G'[\reps(\oxy), \reps^\star(\ozy)]$.
For this, we recall that $g_{\oxy}$ stores for each $r \in \reps(\oxy)$ the smallest-index vertex $v \in \lparts(\Tc)[\oxy]$ so that $N_G(v) \cap \lparts(\Tc)[\oyx] = N_G(r) \cap \lparts(\Tc)[\oyx]$, and $h_{\oxy}$ stores for each $r \in \reps(\oxy)$ the valuation of $\autom_{\varphi'}$ on $(\Tc,\oxy,\alpha_{g_{\oxy}(r)})$.
Now, because $\lparts(\Tc)[\oxy]$ is disjoint from $X$, we have that $N_{G'}(v) \cap \lparts(\Tc)[\oyx] = N_{G'}(r) \cap \lparts(\Tc)[\oyx]$.
Therefore, it suffices to find the adjacencies of such vertices $v$ to $\reps^\star(\ozy)$ in $G'$.
Because we know the valuation of $\autom_{\varphi'}$ on $(\Tc,\oxy,\alpha_v)$, we can do this in a similar manner as in the previous paragraph.

This completes the description of the implementation of $\mathsf{EdgeUpdate}(\eudesc)$.
All of the steps took $\Oh[\varphi,\ell]{\height(T) \cdot |\eudesc|} = \Oh[d,\ell]{\height(T) \cdot |\eudesc|}$ time.
\end{proof}

%% file: together.tex
\section{Dynamic rankwidth}
\label{sec:dynrw}
In this section we put together the material from the previous sections to give the final proof of our dynamic data structure for rankwidth.

Let us first bound how much a prefix-rebuilding update resulting from an edge update description can increase the potential of a rank decomposition.

\begin{lemma}
\label{lem:ubpoteudesc}
Let $\Tc$ be a rooted annotated rank decomposition that encodes a graph $G$, $\prdesc$ an edge update description that describes a graph $G'$, $\Tc'$ a rooted annotated rank decomposition that results from applying to $\Tc$ a prefix-rebuilding update that corresponds to $\prdesc$, and $\ell$ an integer so that the widths of both $\Tc$ and $\Tc'$ are at most $\ell$.
Then it holds that
\[\Phi_{\ell,G'}(\Tc') \le \Phi_{\ell,G}(\Tc) + \Oh[\ell]{|\prdesc| \cdot \height(\Tc)}\]
\end{lemma}
\begin{proof}
Recall that both graphs $G$ and $G'$ share the same set of vertices and for both decompositions $\Tc$ and $\Tc'$ the tree $T$ and the sets $\reps(\olp)$ on leaf edges $\olp$ are the same.
Let $\Tpref$ be the prefix of $T$ given in the edge update description.
We have that $|\Tpref| = |\prdesc|$ and the width of an edge can change only if it is in $T[\Tpref]$.
Then, the conclusion follows directly from the definition of $\Phi$.
\end{proof}

Then we state a lemma about computing optimum-width rank decompositions by dynamic programming on annotated rank decompositions, which will be proved in \Cref{ssec:rankwidth-automaton}.

\begin{restatable}{lemma}{nearlineardecompositionrecovery}
\label{lem:near-linear-annotated-decomposition-recovery}
Let $k, \ell \geq 0$ be integers.
There exists an~algorithm that, given as input an~annotated rank decomposition $\Tc$ of width $\ell$ that encodes a partitioned graph $(G,\prt)$, in time $\Oh[\ell]{|\Tc| \log |\Tc|}$ either:
\begin{itemize}
\item correctly determines that $(G, \prt)$ has rankwidth larger than $k$; or
\item outputs an annotated rank decomposition that encodes $(G, \prt)$ and has width at most $k$.
\end{itemize}   
\end{restatable}

Next we give the main lemma giving the basic version of our data structure.
In the statement it is important that the decomposition $\Tc$ is maintained by prefix-rebuilding updates, as this implies that any feature of $\Tc$ that can be maintained by a prefix-rebuilding data structure can be plugged in to the data structure.

\begin{lemma}
\label{lem:finaldynrwds}
Let $k,d,n \in \N$.
There is a data structure that using prefix-rebuilding updates maintains a rooted annotated rank decomposition $\Tc$ that encodes a dynamic $n$-vertex graph $G$ and has width at most $4k$, under the promise that $G$ has rankwidth at most $k$ at all times, under the following operations:
\begin{itemize}
\item $\mathsf{Init}(\tilde{\Tc})$: Given a rooted annotated rank decomposition $\tilde{\Tc}$ that encodes a graph $G$ and has width at most $4k$, initializes the data structure to hold $\Tc \coloneqq \tilde{\Tc}$. Runs in amortized $\Oh[k,d]{n \log^2 n}$ time.
\item $\mathsf{Update}(\eudesc)$: Given an edge update sentence $\eudesc$ of length at most $d$, either returns that the graph resulting from applying $\eudesc$ to $G$ would have rankwidth more than $k$, or applies $\eudesc$ to update $G$. Runs in amortized $|\eudesc| \cdot 2^{\Oh[k,d]{\sqrt{\log n \log \log n}}}$ time.
\end{itemize}
Moreover, it is guaranteed that after each operation, the height of $\Tc$ is at most $2^{\Oh[k,d]{\sqrt{\log n \log \log n}}}$, even though during the implementations of the operations the height of $\Tc$ can be greater.
\end{lemma}
\begin{proof}
We choose $\ell$ to be the smallest positive integer so that $\ell \ge 4k+1$, $\ell \ge d$, and $\ell$ is at least the largest width of an edge update description that is returned by the $\mathsf{EdgeUpdate}(\eudesc)$ query of the $4k$-prefix-rebuilding data structure of \Cref{lem:mso_edge_update} with the parameter $d$.
Note that $\ell \le \Oh[k,d]{1}$.

Then, the $\mathsf{Init}(\tilde{\Tc})$ query is implemented as follows.
Given the decomposition $\tilde{\Tc}$ that encodes $G$, we first use \Cref{lem:near-linear-annotated-decomposition-recovery} to compute a rank decomposition $\tilde{\Tc'}$ of $G$ of width at most $k$, then use \Cref{lem:logdepthdecomp} to turn $\tilde{\Tc'}$ into a rank decomposition $\tilde{\Tc''}$ of height $\Oh{\log n}$ and width at most $2k$, and then use \Cref{lem:rearangmain} with $\tilde{\Tc}$ and $\tilde{\Tc''}$ to compute an annotated rank decomposition $\Tc$ that encodes $G$ and corresponds to $\tilde{\Tc''}$.
This runs in $\Oh[k,d]{n \log n} = \Oh[\ell]{n \log n}$ time in total, and because the resulting decomposition $\Tc$ has width at most $2k$ and height at most $\Oh{\log n}$, its $\ell$-potential is $\Phi_{\ell,G}(\Tc) \le \Oh[\ell]{n \log n}$.
The first prefix-rebuilding update is to update $\tilde{\Tc}$ into $\Tc$.
Note that we can set its description to fully contain $\Tc$ in $\Oh[\ell]{n}$ time.

We then initialize the $\ell$-prefix-rebuilding data structures $\D^{\mathsf{improve}}$ of \Cref{lem:heightreduprds}, $\D^{\mathsf{translate}}$ of \Cref{lem:transleudescprdesc}, and $\D^{\mathsf{prdsutil}}$ of \Cref{lem:prdscompose} with $\Tc$, and the $4k$-prefix-rebuilding data structure $\D^{\mathsf{upd}}$ of \Cref{lem:mso_edge_update} with $\Tc$.
Usually, these four data structures will hold the same current annotated rank decomposition $\Tc$ of width at most $4k$, but during the $\mathsf{Update}$ query the data structure $\D^{\mathsf{improve}}$ may hold an annotated rank decomposition $\Tc'$ of width up to $\ell$.
The initialization of these data structures takes $\Oh[\ell]{n}$ time.

Let $h = 2^{\Oh[\ell]{\sqrt{\log n \log \log n}}}$ be so that the maximum height of $\Tc$ after applying the $\mathsf{ImproveHeight}()$ operation of $\D^{\mathsf{improve}}$ is at most $h$.
We will maintain the invariant that between the $\mathsf{Update}$ queries, the height of $\Tc$ is at most $h$.
During the $\mathsf{Update}$ query the height may grow unboundedly.

Then, the $\mathsf{Update}(\eudesc)$ query is implemented as follows.
Let $G'$ be the graph resulting from applying $\eudesc$ to $G$.
We first use the data structure $\D^{\mathsf{upd}}$ to compute an edge update description $\prdesc$ corresponding to $\eudesc$.
This runs in time $\Oh[k,d]{\height(\Tc) \cdot |\eudesc|} \le \Oh[\ell]{h \cdot |\eudesc|}$, which is also an upper bound for $|\prdesc|$.
By the choice of $\ell$, the width of $\prdesc$ is at most $\ell$, which is also an upper bound for the width of the decomposition resulting from applying $\prdesc$ to $\Tc$.
Then we use the data structure $\D^{\mathsf{translate}}$ to translate $\prdesc$ into a description $\prdesc_1$ of a prefix-rebuilding update.
This runs in $\Oh[\ell]{|\prdesc|} = \Oh[\ell]{h \cdot |\eudesc|}$ time, which is also an upper bound for $|\prdesc_1|$.
Then, we apply $\prdesc_1$ to $\D^{\mathsf{improve}}$ (but not the other prefix-rebuilding data structures).
Let $\Tc' = (T, V(G), \reps', \repse', \dmap')$ be the decomposition resulting from applying $\prdesc_1$ to $\Tc$.
We have that $\Tc'$ encodes $G'$ and by \Cref{lem:ubpoteudesc} the $\ell$-potential of $\Tc'$ is at most \[\Phi_{\ell,G'}(\Tc') \le \Phi_{\ell,G}(\Tc) + \Oh[\ell]{h \cdot |\eudesc|}.\]
Let $\Tpref$ be the prefix of $T'$ associated with $\prdesc_1$.
We note that all nodes of $\Tc'$ of width larger than $4k$ are in $\Tpref$, and apply the $\mathsf{Refine}(\Tpref)$ operation of $\D^{\mathsf{improve}}$.
If it returns that the rankwidth of $G'$ is greater than $k$, we use the $\mathsf{Reverse}$ operation of $\D^{\mathsf{prdsutil}}$ to compute a description of a prefix-rebuilding operation that turns $\Tc'$ back to $\Tc$, apply it to $\D^{\mathsf{improve}}$, and then return.
In this case the time complexity is $\Oh[\ell]{h \cdot |\eudesc|}$.
The other case is that the $\mathsf{Refine}(\Tpref)$ operation returns a description $\prdesc_2$ of a prefix-rebuilding update that turns $\Tc'$ into a decomposition $\Tc''$ that encodes $G'$, has width at most $4k$, and satisfies
\begin{align*}
\Phi_{\ell,G'}(\Tc'') &\le \Phi_{\ell,G'}(\Tc') - \height_{T'}(\Tpref) + \log n \cdot \Oh[\ell]{|\Tpref| + \height_{T'}(\App_{T'}(\Tpref))}\\
&\le \Phi_{\ell,G'}(\Tc') + \log n \cdot \Oh[\ell]{h \cdot |\eudesc| + h^2 \cdot |\eudesc|}\\
&\le \Phi_{\ell,G}(\Tc) + \Oh[\ell]{h^2 \cdot |\eudesc| \cdot \log n}.
\end{align*}
The running time of the operation and therefore also $|\prdesc_2|$ is
\begin{align*}
&\log n \cdot \Oh[\ell]{\Phi_{\ell,G'}(\Tc')-\Phi_{\ell,G'}(\Tc'') + \log n \cdot (|\Tpref| + \height_{T'}(\App_{T'}(\Tpref)))}\\
\le& \log n \cdot \Oh[\ell]{\Phi_{\ell,G'}(\Tc')-\Phi_{\ell,G'}(\Tc'')} + \Oh[\ell]{h^2 \cdot |\eudesc| \cdot \log^2 n}\\
\le& \log n \cdot \Oh[\ell]{\Phi_{\ell,G'}(\Tc) - \Phi_{\ell,G'}(\Tc'')} + \Oh[\ell]{h^2 \cdot |\eudesc| \cdot \log^2 n}.
\end{align*}
Then we use $\D^{\mathsf{compose}}$ to compute from $\prdesc_1$ and $\prdesc_2$ a description $\prdesc_{\circ}$ of a prefix-rebuilding update that turns $\Tc$ into $\Tc''$.
We apply $\prdesc_{\circ}$ to $\D^{\mathsf{translate}}$, $\D^{\mathsf{compose}}$, and $\D^{\mathsf{upd}}$, and then apply $\prdesc_2$ to $\D^{\mathsf{improve}}$.
Now, all of these data structures hold the same decomposition $\Tc''$.
This takes time $\Oh[\ell]{|\prdesc_1|+|\prdesc_2|} \le \log n \cdot \Oh[\ell]{\Phi_{\ell,G'}(\Tc) - \Phi_{\ell,G'}(\Tc'')} + \Oh[\ell]{h^2 \cdot |\eudesc| \cdot \log^2 n}$.


Then, we call the $\mathsf{ImproveHeight}()$ operation of $\D^{\mathsf{improve}}$.
This updates $\Tc''$ through a series of prefix-rebuilding updates into a decomposition $\Tc'''$ that has height at most $h$ and width at most $4k$, and returns the corresponding sequence of descriptions of prefix-rebuilding updates.
We also apply the same sequence of prefix-rebuilding updates to $\D^{\mathsf{translate}}$, $\D^{\mathsf{compose}}$, and $\D^{\mathsf{upd}}$, noting that also the intermediate decompositions in this sequence have width at most $4k$.
It holds that $\Phi_{\ell,G'}(\Tc''') \le \Phi_{\ell,G'}(\Tc'')$ and the running time of this is
\begin{align}
&\log n \cdot \Oh[\ell]{\Phi_{\ell,G'}(\Tc'')-\Phi_{\ell,G'}(\Tc''')} \nonumber\\
\le&\log n \cdot \Oh[\ell]{\Phi_{\ell,G'}(\Tc)-\Phi_{\ell,G'}(\Tc''')} + \Oh[\ell]{h^2 \cdot |\eudesc| \cdot \log^2 n}\label{lem:finaldynrwds:eqtime}.
\end{align}
Finally, $\Tc'''$ is the decomposition that our data structure will hold after the $\mathsf{Update}$ operation.
Note that we updated $\Tc$ into $\Tc'''$ by prefix-rebuilding operations so that all intermediate decompositions had width at most $4k$.
As $\Phi_{\ell,G'}(\Tc''') \le \Phi_{\ell,G'}(\Tc'')$, the total time complexity of the operation is bounded by $\log n \cdot \Oh[\ell]{\Phi_{\ell,G'}(\Tc)-\Phi_{\ell,G'}(\Tc''')} + \Oh[\ell]{h^2 \cdot |\eudesc| \cdot \log^2 n}$.
We also have that $\Phi_{\ell,G'}(\Tc''') \le \Phi_{\ell,G'}(\Tc'') \le \Phi_{\ell,G}(\Tc) + \Oh[\ell]{h^2 \cdot |\eudesc| \cdot \log n}$.

Then we analyze the amortized time complexity.
Let us consider the sequence of $t$ first $\mathsf{Update}$ operations applied to the data structure, and let us denote by $\eudesc_1, \ldots, \eudesc_t$ the edge update sentences given in them and by $\Tc_1, \ldots, \Tc_t$ the decompositions after each of the updates, and by $\Tc_0$ the initial decomposition.
By \Cref{lem:finaldynrwds:eqtime}, the total time used in the first $t$ $\mathsf{Update}$ operations is at most
\[\sum_{i=1}^t \left(\Oh[\ell]{h^2 \cdot |\eudesc_i| \cdot \log^2 n} + \log n \cdot \Oh[\ell]{\Phi_{\ell}(\Tc_{i-1})-\Phi_{\ell}(\Tc_i)}\right).\]
Now, because $\Phi_{\ell}(\Tc_i)$ is always non-negative, $\Phi_{\ell}(\Tc_0) \le \Oh[\ell]{n \log n}$, and $\Phi_{\ell}(\Tc_i) \le \Phi_{\ell}(\Tc_{i-1})+\Oh[\ell]{h^2 \cdot |\eudesc_i| \cdot \log n}$, we have that
\[\sum_{i=1}^t \Oh[\ell]{\Phi_{\ell}(\Tc_{i-1})-\Phi_{\ell}(\Tc_i)} \le \Oh[\ell]{n \log n} + \sum_{i=1}^t \Oh[\ell]{h^2 \cdot |\eudesc_i| \cdot \log n}.\]
This implies that the total running time of the first $t$ operations is bounded by
\[\Oh[\ell]{n \log^2 n} + \sum_{i=1}^t \Oh[\ell]{h^2 \cdot |\eudesc_i| \cdot \log^2 n}.\]
We conclude the claimed amortized running time by charging the $\Oh[\ell]{n \log^2 n}$ term from the $\mathsf{Init}$ operation and for each $i \in [t]$ the $\Oh[\ell]{h^2 \cdot |\eudesc_i| \cdot \log^2 n}$ term from the $i$:th $\mathsf{Update}$ operation.
Note that $\Oh[\ell]{h^2 \cdot |\eudesc_i| \cdot \log^2 n} \le |\eudesc_i| \cdot 2^{\Oh[k,d]{\sqrt{\log n \log \log n}}}$.
\end{proof}

Then we add a couple of more features to the data structure of \Cref{lem:finaldynrwds}.

\begin{lemma}
\label{lem:dynfull}
Let $k,d,n \in \N$.
The data structure of \Cref{lem:finaldynrwds} can furthermore support the following operations:
\begin{itemize}
\item $\mathsf{InitEmpty}()$: Initializes the data structure to hold the $n$-vertex edgeless graph $G$. Runs in amortized $\Oh[k,d]{n \log^2 n}$ time.
\item $\LinCMSO_1(\varphi,X_1,\ldots,X_p)$: Given a $\LinCMSO_1$ sentence $\varphi$ of length at most $d$ with $p$ free set variables and $p$ vertex subsets $X_1,\ldots,X_p \subseteq V(G)$, returns the value of $\varphi$ on $(G,X_1,\ldots,X_p)$. Runs in time $\Oh[d]{1}$ if the sets $X_1,\ldots,X_p$ are empty, and in time $\sum_{i=1}^p |X_i| \cdot 2^{\Oh[k,d]{\sqrt{\log n} \log \log n}}$ otherwise.
\end{itemize}
\end{lemma}
\begin{proof}
First, the $\mathsf{InitEmpty}()$ operation can implemented by the $\mathsf{Init}(\tilde{\Tc})$ operation of the data structure of \Cref{lem:finaldynrwds}, as it is straightforward to construct an annotated rank decomposition of an $n$-vertex edgeless graph in $\Oh{n}$ time.
Then, to support the $\LinCMSO_1(\varphi,X_1,\ldots,X_p)$ queries, we maintain the $4k$-prefix-rebuilding data structure of \Cref{lem:lincmsoprds} for $w = d$.
\end{proof}

It is easy to see that \Cref{the:dynfull} is a special case of \Cref{lem:dynfull}: The operations to insert and delete edges can be simulated by edge update sentences of constant length and size.

%% file: comp.tex
\section{Almost-linear time algorithm for rankwidth}
\label{sec:comp}
In this section we prove \Cref{the:altrw} by using \Cref{lem:finaldynrwds}.
We prove in fact a bit more general statement, showing that if the $2^{\Oh[k]{\sqrt{\log n \log \log n}}}$ factor in \Cref{lem:finaldynrwds} could be improved to $\Oh[k]{\log^{\Oh{1}} n}$, then the $2^{\sqrt{\log n} \log \log n}$ factor in \Cref{the:altrw} could be improved to $\log^{\Oh{1}} n$.

\subsection{The twin flipping problem}
When $G$ is a graph and $F$ is a set of unordered pairs of vertices of $G$, we denote by $G \symd F$ the graph obtained from $G$ by ``flipping'' adjacencies between every pair in $F$.
In other words, $V(G \symd F) = V(G)$ and $E(G \symd F) = E(G) \symd F$.
Recall that a vertex $v$ is a twin of a vertex $u$ if $N(v) = N(u)$.
Our interface between \Cref{lem:finaldynrwds} and \Cref{the:altrw} will be the following problem.

\begin{problem}[Twin Flipping]
\label{prob:rwcomp}
Given an annotated rank decomposition of width at most $k$ that encodes an $n$-vertex bipartite graph $G$ with bipartition $(A,B)$, two disjoint vertex sets $X,Y \subseteq A$ so that every vertex in $X$ has a twin in $Y$, and a set $F \subseteq X \times B$ of size $|F| \le \Oh[k]{n}$, either determine that the rankwidth of $G \symd F$ is more than $k$, or return an annotated rank decomposition that encodes $G \symd F$ and has width at most $k$.
\end{problem}

In this section we will show that algorithms for \Cref{prob:rwcomp} can be translated to algorithms for computing rankwidth.
Before showing that, let us give an algorithm for Twin Flipping by using \Cref{lem:finaldynrwds}.
The following basic observation is useful in this algorithm and later in this section.

\begin{observation}
\label{obs:addtwin}
Let $G$ be a graph that contains twins $u,v \in V(G)$.
The rankwidth of $G$ is at most the rankwidth of $G - \{v\}$.
\end{observation}
\begin{proof}
Observe that if $A \subseteq V(G) \setminus \{v\}$ and $u \in A$, then $\cutrk_{G - \{v\}}(A) = \cutrk_{G}(A \cup \{v\})$.
Therefore, we can construct a rank decomposition of $G$ of equal width from a rank decomposition of $G - \{v\}$ by adding two children $c_1,c_2$ to the leaf corresponding to $u$, and mapping $u$ to $c_1$ and $v$ to $c_2$.
\end{proof}

Then we give the algorithm for Twin Flipping.

\begin{lemma}
\label{lem:algoforprob}
There is a $n \cdot 2^{\Oh[k]{\sqrt{\log n \log \log n}}}$ time algorithm for \Cref{prob:rwcomp}.
\end{lemma}
\begin{proof}
Denote the vertices in $X$ as $X = \{v_1, \ldots, v_{|X|}\}$.
Let $G_0 = G$, and for each $i \in [|X|]$ let $G_i$ be the bipartite graph with bipartition $(A,B)$, so that for $j\le i$ it holds that $N_{G_i}(v_j) = N_{G \symd F}(v_j)$, for $j>i$ it holds that $N_{G_i}(v_j) = N_{G}(v_j)$, and for $u \in A \setminus X$ it holds that $N_{G_i}(u) = N_{G}(u) = N_{G \symd F}(u)$.
We have that $G_{|X|} = G \symd F$ and because for each $v \in X$ there exists $u \in Y$ so that $N_G(v) = N_G(u) = N_{G \symd F}(u)$, each $G_i$ can be obtained from $G \symd F$ by adding twins and deleting vertices, which by \Cref{obs:addtwin} implies that if $G \symd F$ has rankwidth at most $k$ then also $G_i$ for each $i \in [|X|]$ has rankwidth at most $k$.

Now, for each vertex $v_i \in X$, let $F_i$ be the set of vertices $F_i = \{u \in B \mid v_i u \in F\}$.
We can write an edge update sentence $\eudesc_i$ of size $|\eudesc_i| = |F_i|+1$ and constant length that turns $G_{i-1}$ into $G_i$.
Let $\Tc$ be the given annotated rank decomposition that encodes the graph $G$.
We initialize the data structure of \Cref{lem:finaldynrwds} with $\Tc$ and $k$, and the length bound $d = \Oh{1}$ of these edge update sentences, which takes $\Oh[k]{n \log^2 n}$ amortized time.
We then apply the edge update sentences $\eudesc_i$ one by one to $\Tc$.
If the data structure at any point returns that the rankwidth would become larger than $k$, we can return that the rankwidth of $G \symd F$ is more than $k$.
This takes $|F| \cdot 2^{\Oh[k]{\sqrt{\log n \log \log n}}}$ amortized time in total.

Finally, we obtain an annotated rank decomposition $\Tc'$ that encodes $G \symd F$ and has width at most $4k$.
We then use \Cref{lem:near-linear-annotated-decomposition-recovery} to obtain in time $\Oh[k]{n \log n}$ an annotated rank decomposition $\Tc''$ that encodes $G \symd F$ and has width at most $k$ or determine that $G \symd F$ has rankwidth more than $k$, and then return $\Tc''$.

The running time is $\Oh[k]{n \log^2 n} + |F| \cdot 2^{\Oh[k]{\sqrt{\log n \log \log n}}} = n \cdot 2^{\Oh[k]{\sqrt{\log n \log \log n}}}$.
\end{proof}

Then, the rest of this section will be devoted to showing that algorithms for Twin Flipping imply algorithms for computing rankwidth, in particular, to proving the following lemma.

\begin{restatable}{lemma}{compmain}
\label{lem:compmain}
Let $T \colon \N \rightarrow \N$ be a function so that there is a $\Oh[k]{T(n)}$ time algorithm for \Cref{prob:rwcomp}.
Then there is an algorithm that given an $n$-vertex $m$-edge graph $G$ and an integer $k$, in time $\Oh[k]{T(n) \log^2 n} + \Oh{m}$ either returns that the rankwidth of $G$ is more than $k$, or returns an annotated rank decomposition that encodes $G$ and has width at most $k$.
\end{restatable}

Putting \Cref{lem:algoforprob,lem:compmain} together implies the first part of \Cref{the:altrw}.
In particular, as $n \cdot 2^{\Oh[k]{\sqrt{\log n \log \log n}}} \le \Oh[k]{n \cdot 2^{\sqrt{\log n} \log \log n}/\log^2 n}$, we can set $T(n) = n \cdot 2^{\sqrt{\log n} \log \log n}/\log^2 n$ to obtain an algorithm with a running time of $\Oh[k]{n \cdot 2^{\sqrt{\log n} \log \log n}} + \Oh{m}$.
Then, we prove in \Cref{sec:cliquewidth} (\Cref{lem:rdtocwexp}) that given an annotated rank decomposition of width $k$ that encodes $G$, we can in $\Oh[k]{n}$ time output a $(2^{k+1}-1)$-expression for cliquewidth of $G$.
This gives the second part of \Cref{the:altrw}.

We remark that in the proof of \Cref{lem:compmain} we make the natural assumptions that $T(n) \ge \Omega(n)$ and $T(n)$ is increasing and convex.

\subsection{Reduction to bipartite graphs}
\label{ssec:bipartite-reduction}
We will work on bipartite graphs in our algorithm, so the first step is to reduce the task of computing the rankwidth of a graph to bipartite graphs.
For this, we will use a reduction given by Courcelle~\cite{DBLP:journals/japll/Courcelle06} and further analyzed by Oum~\cite[Section 4.1]{Oum08}.

Let $G$ be a graph.
We define $B(G)$ to be the bipartite graph whose vertex set is $V(B(G)) = V(G) \times [4]$, and edge set is defined so that
\begin{enumerate}
\item if $v \in V(G)$ and $i \in [3]$, then $(v,i)$ is adjacent to $(v,i+1)$ in $B(G)$ and
\item if $uv \in E(G)$, then $(u,1)$ is adjacent to $(v,4)$ in $B(G)$.
\end{enumerate}

We observe that given an $n$-vertex $m$-edge graph $G$, we can compute $B(G)$ in $\Oh{n+m}$ time.
Oum showed that the rankwidths of $G$ and $B(G)$ are tied to each other.
\begin{lemma}[\cite{Oum08}]
\label{lem:rwbg}
If the rankwidth of $G$ is $k$, then the rankwidth of $B(G)$ is at least $k/4$ and at most $\max(2k,1)$.
\end{lemma}

Even though Oum gives an explicit construction of a rank decomposition of $G$ given a rank decomposition of $G$, it seems complicated to adapt to work in linear time with annotated rank decompositions.
We use an alternative approach by using edge update sentences.

\begin{lemma}
\label{lem:trrdbgtoardg}
Let $G$ be an $n$-vertex graph.
There is an algorithm that given an annotated rank decomposition $\Tc$ that encodes $B(G)$ and has width $k$, in time $\Oh[k]{n \log n}$ returns an annotated rank decomposition that encodes $G$ and has optimum width.
\end{lemma}
\begin{proof}
Consider an edge update sentence $\eudesc = (\varphi,X,X_1,X_2,X_3,X_4)$ that has $X = V(B(G))$, $X_i = V(G) \times \{i\}$, and $\varphi(Y,X_1,X_2,X_3,X_4) =$
\begin{align*}
\exists u \in Y, v \in Y . (u \neq v \wedge \forall w \in Y . (u=w \vee v=w))
\wedge u \in X_1 \wedge v \in X_1\\
\wedge (\exists u_2 \in X_2, u_3 \in X_3, u_4 \in X_4 . (E(u,u_2) \wedge E(u_2,u_3) \wedge E(u_3,u_4) \wedge E(u_4,v)).
\end{align*}
Let $G'$ be the graph resulting from applying $\eudesc$ to $B(G)$.
We observe that the subgraph of $G'$ induced by $V(G) \times \{1\}$ is equal to $G$, after renaming every vertex of form $(v,1)$ to $v$.

Therefore we use our machinery built in previous sections as follows.
First, we use \Cref{lem:logdepthdecomp} with $\Tc$ to compute a rank decomposition $\Tc^1$ of $B(G)$ of width at most $2k$ and height $\Oh{\log n}$.
Then we use \Cref{lem:rearangmain} with $\Tc$ and $\Tc^1$ to obtain an annotated rank decomposition $\Tc^2$ that encodes $B(G)$, has width at most $2k$, and height $\Oh{\log n}$.
These steps take $\Oh[k]{n \log n}$ time.
Then we initialize the $2k$-prefix-rebuilding data structure of \Cref{lem:mso_edge_update} with $\Tc^2$ and the parameter $d$ (the bound on the length of an~edge update sentence) equal to the length of $\varphi$ (which is constant), and then apply the $\mathsf{EdgeUpdate}(\eudesc)$ query to obtain an edge update description $\prdesc$ of width $\ell = \Oh[k]{1}$ that describes $G'$.
This takes $\Oh[k]{n \log n}$ time as the height of $\Tc^2$ is $\Oh{\log n}$.
Then, we initialize the $2k$-prefix-rebuilding data structure of \Cref{lem:transleudescprdesc} with $\Tc^2$, and translate $\prdesc$ to a description $\prdesc'$ of a prefix-rebuilding update.
This takes $\Oh[k,\ell]{n} = \Oh[k]{n}$ time.
Then, we use \Cref{lem:basicprds} to apply $\prdesc'$ to $\Tc^2$, turning $\Tc^2$ into an annotated rank decomposition $\Tc^3$ that encodes $G'$ and has width at most $\max(k,\ell) = \Oh[k]{1}$.
Then we use \Cref{lem:dropleafs} to turn $\Tc^3$ into an annotated rank decomposition of the subgraph of $G'$ induced by $V(G) \times \{1\}$, and then by renaming vertices turn it into an annotated rank decomposition $\Tc^4$ of $G$.
These steps take $\Oh[k]{n}$ time.
Finally we use \Cref{lem:near-linear-annotated-decomposition-recovery} with $\Tc^4$ to compute an optimum-width rank decomposition that encodes $G$, and return it.
This runs in time $\Oh[k]{n \log n}$.
\end{proof}

\Cref{lem:rwbg,lem:trrdbgtoardg} and the fact that $B(G)$ can be computed from $G$ in $\Oh{n+m}$ time imply that we can now focus on bipartite graphs.

\subsection{Twins and near-twins}
In this subsection we prove lemmas about finding twins and near-twins in graphs of small rankwidth.
The following lemma will be our main tool.
Recall here from \Cref{sec:refi} that for a~rooted rank decomposition $\Tc = (T, \lambda)$ of a~graph $G$, a set $F \subseteq V(G)$ is a \emph{tree factor} whenever $F = \lparts(\Tc)[x]$ for some $x \in V(T)$, and a \emph{context factor} whenever $F$ is not a tree factor but $F = F_1 \setminus F_2$ for tree factors $F_1, F_2$.
$F$ is a factor if $F$ is a tree factor or a context factor.

\begin{lemma}
\label{lem:twinfactors}
There is an algorithm that given a rooted rank decomposition $\Tc$ of an $n$-vertex graph $G$, an integer $\ell \geq 1$, and a set $W \subseteq V(G)$ with $|W| \geq 16\ell$, in time $\Oh{n}$ outputs a set of at least $|W|/(16 \ell)$ disjoint factors of $\Tc$ so that each of them contains at least $\ell$ vertices in $W$.
The outputted tree factors are represented by single nodes of $\Tc$ and context factors by pairs of nodes of $\Tc$.
\end{lemma}
\begin{proof}
Let $\Tc = (T,\lmap)$.
We say that a node $x$ of $T$ is \emph{important} if $|\lparts(\Tc)[x] \cap W| \ge \ell$.
Let us denote the set of important nodes of $T$ by $I \subseteq V(T)$.
If a node is important, then also its parent is, so $I$ is a prefix of $T$.
Note that the root of $T$ is important.
Let us furthermore say that a node is a \emph{junction} if it is important, and also either has degree $1$ or $3$ in $T[I]$ or is the root of $T$.
We denote the set of junctions by $J \subseteq I$.
Note that if $x,y \in J$, then the lowest common ancestor of $x$ and $y$ is also in $J$.

Then we define a rooted tree $T'$ so that $V(T') = J$, there is an edge between $x,y \in J$ if there is a path between $x$ and $y$ in $T$ that avoids other nodes in $J$, and the root of $T'$ is the root of $T$.
Observe that $T'$ is a rooted tree where each node except the root has either $0$ or $2$ children, and the root has $1$ or $2$ children.
Now, $V(G)$ can be partitioned into a disjoint union of factors of $\Tc$ as follows:
\begin{itemize}
\item for each leaf $l$ of $T'$ there is a tree factor $\lparts(\Tc)[l]$,
\item for each edge $xp$ of $T'$, where $p$ is the parent of $x$ in $T'$ and $c$ is the child of $p$ on the path from $p$ to $x$ in $T$ there is a context factor $\lparts(\Tc)[c] \setminus \lparts(\Tc)[x]$, and
\item if $c$ is a child of the root and is not in $I$, then there is a tree factor $\lparts(\Tc)[c]$.
\end{itemize}

We consider cases based on $|V(T')|$.
First, suppose that $|V(T')| \ge |W|/(8 \ell)$.
This implies that $T'$ has at least $|W|/(16 \ell)$ leaves, so by outputting the leaves of $T'$ we output at least $|W|/(16 \ell)$ tree factors that each contains at least $\ell$ vertices in $W$.

Then, suppose $|V(T')| \le |W|/(8 \ell)$.
We note that each tree factor corresponding to a leaf of $T'$ contains at most $2\ell-1$ vertices in $W$, and the possible single tree factor corresponding to a child of the root not in $I$ contains at most $\ell-1$ vertices in $W$, so therefore the context factors corresponding to the edges of $T'$ contain at least 
\[|W|-|V(T')|\cdot (2\ell-1)-(\ell-1) \ge |W|-\frac{|W|\cdot (2\ell-1)}{8 \ell} - \frac{|W|}{16} \ge |W|/2\] vertices in $W$.
Now, consider an edge $xp$ of $T'$, where $p$ is the parent of $x$ in $T'$.
This corresponds to a path $x, y_1, \ldots, y_t, p$ in $T$.
Then, for each $i \in [t]$ let $z_i$ be the child of $y_i$ that is not on this path.
We observe that the context factor associated with $xp$ is equal to $\bigcup_{i=1}^t \lparts(\Tc)[z_i]$, and that for each $i$ it holds that $|\lparts(\Tc)[z_i] \cap W| < \ell$.
This implies that if this context factor contains $w$ vertices in $W$, then it can be further partitioned into at least $\lfloor \frac{w}{2\ell} \rfloor \geq \frac{w}{2\ell} - 1$ context factors that each contain at least $\ell$ vertices in $W$, plus at most one context factor that contains less than $\ell$ vertices in $W$.
By performing this partitioning to all $|E(T')| \le |W|/(8\ell)$ such context factors that in total contain at least $|W|/2$ vertices in $W$, we obtain at least
\[\frac{|W|/2}{2\ell}-|E(T')| \ge |W|/(8\ell)\]
context factors that each contain at least $\ell$ vertices in $W$.
This procedure clearly can be implemented in $\Oh{n}$ time given $\Tc$.
\end{proof}

Then we apply \Cref{lem:twinfactors} to prove that bipartite graphs with small rankwidth and unbalanced bipartition contain a lot of twins.

\begin{lemma}
\label{lem:bipartitetwins}
There is a function $f(k) \in 2^{\Oh{k}}$, so that if $G$ is a bipartite graph with bipartition $(A,B)$ and rankwidth $k$, and $|A| \ge f(k) \cdot |B|$, then there exist at least $|A|/f(k)$ disjoint pairs of twins in $A$.
\end{lemma}
\begin{proof}
We will prove the lemma for $f(k) = 32 \cdot (2^{2k}+1)$, so assume that $|A| \ge f(k) \cdot |B|$.
Let $\Tc$ be a rank decomposition of $G$ of width at most $k$, and let us apply \Cref{lem:twinfactors} with $\ell = f(k)/32$ and $W = A$.
This outputs at least $\frac{2|A|}{f(k)}$ disjoint factors of $\Tc$ so that each of them contains at least $f(k)/32 = 2^{2k}+1$ vertices in $A$.
Among them, there are at least $\frac{2|A|}{f(k)} - |B| \ge |A|/f(k)$ factors that contain no vertices in $B$.
It suffices to prove that each of these contains a pair of twins in $A$.

Consider a factor $F$ of $\Tc$ with $|F| \ge 2^{2k}+1$ and $F \subseteq A$.
If $F$ is a tree factor, then $\cutrk_G(F) \le k$ by definition, and if $F$ is a context factor, we can prove by symmetry and submodularity of $\cutrk_G$ that $\cutrk_G(F) \le 2k$.
Now, \Cref{lem:repsbound} implies that $F$ has a representative $R$ of size $|R| \le 2^{2k}$.
Because $|F| > 2^{2k}$, there exists a vertex $v \in F \setminus R$, and because $R$ is a representative of $F$, there exists $u \in R$ so that $N(v) \setminus F = N(u) \setminus F$.
Because $F \subseteq A$, the vertices $u$ and $v$ are twins.
\end{proof}

We say that two vertices $u$ and $v$ of a graph $G$ are \emph{$q$-near-twins} if $|N(u) \symd N(v)| \le q$.
Next we use \Cref{lem:twinfactors} to give an algorithm for finding many near-twins in graphs of small rankwidth.

\begin{lemma}
\label{lem:neartwins}
There exists a function $f \in 2^{\Oh{k}}$ so that there is an algorithm that given an annotated rank decomposition $\Tc$ of width $k$ that encodes an $n$-vertex graph $G$ and a set $W \subseteq V(G)$ such that $|W| \geq f(k)$, in time $\Oh[k]{n}$ returns $|W|/f(k)$ disjoint pairs of vertices $(u_1,v_1),\ldots,(u_t,v_t)$ in $W$, so that $u_i$ and $v_i$ are $(f(k) \cdot n/|W|)$-near-twins.
The algorithm furthermore returns the sets $N(u_i) \symd N(v_i)$ for all $i \in [t]$.
\end{lemma}
\begin{proof}
The proof will use similar ideas to the proof of \Cref{lem:bipartitetwins}.
We will prove the the lemma for $f(k) = 32 \cdot (2^{2k}+1)$.
Let us root $\Tc = (T,V(G),\reps,\repse,\dmap)$ arbitrarily and apply \Cref{lem:twinfactors} with $\Tc$, $\ell = f(k)/32$, and the set $W$.
This outputs at least $\frac{2|W|}{f(k)}$ disjoint factors of $\Tc$ so that each of them contains at least $f(k)/32 = 2^{2k}+1$ vertices in $W$.
Let us say that a factor $F \subseteq V(G)$ is \emph{big} if $|F| > \frac{f(k) n}{|W|}$ and \emph{small} otherwise.
Because the factors are disjoint, there are at most $\frac{|W|}{f(k)}$ big factors, implying that there are at least $\frac{|W|}{f(k)}$ small factors.

Now it suffices to output a single such pair $(u_i,v_i)$ from each small factor.
We observe that if $F$ is a small factor and $u,v \in F \cap W$ are two vertices with $N(u) \setminus F = N(v) \setminus F$, then $|N(u) \symd N(v)| \le |F|$, implying that they are $(f(k) \cdot n/|W|)$-near-twins.
It remains to argue that we can find such $u$ and $v$ in $\Oh[k]{|F|}$ time for each small factor $F$.

First suppose that $F$ is a tree factor, given as $F = \lparts(\Tc)[x]$ for some $x \in V(T)$, and let $p$ be the parent of $x$ in $T$.
In this case, the subtree below $x$ in $T$ has $\Oh{|F|}$ nodes.
For each vertex $v \in F$ there exists a vertex $w \in \reps(\vec{xp})$ so that $N(v) \setminus F = N(w) \setminus F$, and given $v$ we can find such vertex $w$ in time $\Oh{|F|}$ by following the mapping $\dmap$ of $\Tc$.
We iterate through vertices in $F \cap W$ until we find two vertices $u,v \in F \cap W$ with the same such vertex $w$.
This implies that $N(u) \setminus F = N(v) \setminus F$, so we can return the pair $(u,v)$.
As $|\reps(\vec{xp})| \le 2^k$, finding such $u$ and $v$ takes at most $2^k+1$ iterations of finding such $w$, resulting in $\Oh[k]{|F|}$ time, and we are guaranteed to find such $u$ and $v$ because $|F \cap W| \ge 2^{2k}+1$.
To compute $N(u) \symd N(v)$, we first compute $N(u) \cap F$ and $N(v) \cap F$ in $\Oh[k]{|F|}$ time by modifying the method of \Cref{lem:getnbs} so that we follow the mapping $\dmap$ only inside the subtree below $x$.
Then, we can output $(N(u) \cap F) \symd (N(v) \cap F) = N(u) \symd N(v)$.

Then suppose $F$ is a context factor, given as $F = \lparts(\Tc)[x] \setminus \lparts(\Tc)[y]$ for some nodes $x,y \in V(T)$, so that $y$ is a descendant of $x$.
Let $p^x$ be the parent of $x$ and $p^y$ the parent of $y$.
We have that the subtree of $T$ consisting of the descendants of $x$ minus the descendants of $y$ has $\Oh{|F|}$ nodes.
Again, for each vertex $v \in F$ there exists a vertex $w^x \in \reps(\vec{x p^x})$ so that $N(v) \setminus \lparts(\Tc)[x] = N(w^x) \setminus \lparts(\Tc)[x]$ and a vertex $w^y \in \reps(\vec{p^y y})$ so that $N(v) \cap \lparts(\Tc)[y] = N(w^y) \cap \lparts(\Tc)[y]$, and we can find such $w^x$ and $w^y$ in $\Oh{|F|}$ time given $v$ by following the mapping $\dmap$ of $\Tc$.
Now, if we find two vertices $u,v \in F \cap W$ with the same such pair $(w^x,w^y)$, then $N(u) \setminus F = N(v) \setminus F$.
Because $|\reps(\vec{x p^x})|,|\reps(\vec{p^y y})| \le 2^k$, there are at most $2^{2k}$ such pairs, so we find such $u,v$ within the first $2^{2k}+1$ iterations, resulting in $\Oh[k]{|F|}$ time.
The set $N(u) \symd N(v)$ can be computed in $\Oh[k]{|F|}$ time by similar arguments as in the previous case.
\end{proof}

Then we give a data structure for finding twins guaranteed by \Cref{lem:bipartitetwins} efficiently in a certain setting where we consider induced subgraphs defined by an interval.
For a graph $G$ and a vertex set $X \subseteq V(G)$, the \emph{twin-equivalence classes} of $X$ in $G$ are the maximal sets $X' \subseteq X$ so that any two vertices in $X'$ are twins in $G$.

\begin{lemma}
\label{lem:twinds}
There is a data structure that is initialized with an $n$-vertex $m$-edge bipartite graph $G$ given with a bipartition $(A,B)$, where $B$ is indexed as $B = \{v_1, \ldots, v_{|B|}\}$, and supports the following query:
\begin{itemize}
\item $\mathsf{Twins}(X,\ell,r)$: Given a set $X \subseteq A$ and two integers $\ell,r$ with $1 \le \ell \le r \le |B|$, in time $\Oh{|X| \log n}$ returns the twin-equivalence classes of $X$ in the graph $G[X, \{v_\ell,\ldots,v_r\}]$.
\end{itemize}
The initialization time of the data structure is $\Oh{n+m}$.
\end{lemma}
\begin{proof}
We will use tools from the theory of string algorithms: the \emph{suffix array} and the \emph{LCP array}.
For a string $S = s_1, \ldots, s_t$ of length $t$, the suffix array of $S$ is the array $\mathsf{SA}$ of length $t$ that at position $i \in [t]$ stores the index $\mathsf{SA}[i] \in [t]$ so that the $i$th lexicographically smallest suffix of $S$ starts at index $\mathsf{SA}[i]$ of $S$.
The LCP array associated with $S$ and $\mathsf{SA}$ is the array $\mathsf{LCP}$ of length $t-1$ that at position $i \in [t-1]$ stores the length $\mathsf{LCP}[i]$ of the longest common prefix of the suffix of $S$ starting at $\mathsf{SA}[i]$ and the suffix of $S$ starting at $\mathsf{SA}[i+1]$.
It is known that both the suffix array and the LCP array of a given string can be computed in linear time~\cite{DBLP:journals/jacm/KarkkainenSB06}.

The initialization of our data structure works as follows.
We consider the total order of the vertices $B = \{v_1, \ldots, v_{|B|}\}$ so that $v_i < v_j$ whenever $i<j$.
First we use bucket sort to sort the neighborhoods $N(a)$ of each vertex $a \in A$ into an ordered list, in total time $\Oh{n+m}$.
Then we concatenate these lists into a string $S$ of length $m$, so that for each vertex $a \in A$, the neighborhood of $a$ corresponds to a substring $S[L_a,R_a] = s_{L_a},\ldots,s_{R_a}$ of $S$, in which the neighbors of $a$ occur in the sorted order.
We store the indices $L_a$ and $R_a$ of each $a \in A$.
We then compute the suffix array $\mathsf{SA}$ and the LCP array $\mathsf{LCP}$ of $S$ by using the algorithm of~\cite{DBLP:journals/jacm/KarkkainenSB06} in $\Oh{m}$ time.
We also compute the inverse array of $\mathsf{SA}$, in particular, the array $\mathsf{InvSA}$ so that for each $i \in [m]$ it holds that $\mathsf{SA}[\mathsf{InvSA}[i]] = i$.
Finally, we compute a range minimum query data structure on the LCP array, in particular, a data structure that can answer queries that given indices $\ell,r \in [m]$, report $\min_{i \in [\ell,r]} \mathsf{LCP}[i]$.
Such data structure that answers queries in $\Oh{\log m} = \Oh{\log n}$ time can be computed by folklore techniques with binary trees in $\Oh{m}$ time.
All together, the initialization works in $\Oh{n+m}$ time.

Then the $\mathsf{Twins}(X,\ell,r)$ query is implemented as follows.
Let us denote $Y = \{v_{\ell},\ldots,v_r\}$.
First, for each $a \in X$, we use binary search to compute the indices $L'_a,R'_a$ so that the neighborhood $N(a) \cap Y$ of $a$ into $Y$ corresponds to the substring $S[L'_a,R'_a]$, or decide that the neighborhood of $a$ into $Y$ is empty.
This takes $\Oh{|X| \log n}$ time.
The first equivalence class is the vertices in $X$ whose neighborhood into $Y$ is empty.
Then, based on the computed indices $L'_a$ and $R'_a$, we know for each $a \in X$ the size $|N(a) \cap Y|$.
We group the remaining vertices in $X$ based on $|N(a) \cap Y|$, which can be done in $\Oh{|X| \log n}$ time.
It remains to consider the problem where given $X' \subseteq X$ so that each $a \in X'$ has exactly $p \ge 1$ neighbors in $Y$, we have to compute the twin-equivalence classes of $X'$ in $G[X',Y]$.

Consider two vertices $a,b \in X'$ and assume $\mathsf{InvSA}[L'_a] < \mathsf{InvSA}[L'_b]$; so the suffix of $S$ starting at index $L'_a$ is lexicographically smaller than the suffix starting at index $L'_b$.
Then $N(a) \cap Y = N(b) \cap Y$ holds if and only if these suffixes share a~common prefix of length $p$, or equivalently $p \le \min_{i \in [\mathsf{InvSA}[L'_a],\mathsf{InvSA}[L'_b]-1]} \mathsf{LCP}[i]$.
Therefore, to compute the twin-equivalence classes of $X'$, we first sort $X'$ based on the integers $\mathsf{InvSA}[L'_a]$ in time $\Oh{|X'| \log n}$, then assuming this sorted order of $X'$ is $a_1,\ldots,a_{|X'|}$, we compute for each $i \in [|X'|-1]$ the integer $z_i = \min_{j \in [\mathsf{InvSA}[L'_{a_i}], \mathsf{InvSA}[L'_{a_{i+1}}]-1]} \mathsf{LCP}[j]$ by using the range minimum query data structure in $\Oh{|X'| \log n}$ time.
Now we have that $a_i$ and $a_j$ with $i < j$ have $N(a_i) \cap Y = N(a_j) \cap Y$ if and only if $p \le \min_{k \in [i,j-1]} z_k$, so with this information we can output the twin-equivalence classes of $X'$ in $\Oh{|X'|}$ time.
Therefore, the total time to answer the query is $\Oh{|X| \log n}$.
\end{proof}

Then we show that the method of adding twins to a rank decomposition discussed in \Cref{obs:addtwin} can be efficiently implemented on annotated rank decompositions.

\begin{lemma}
\label{lem:addtwin}
Let $G$ be a graph with twins $u,v \in V(G)$.
Suppose a representation of an annotated rank decomposition $\Tc'$ that encodes $G - \{v\}$ and has width $k$ is already stored.
Then, given $u$ and $v$, the representation of $\Tc'$ can in time $\Oh{1}$ be turned into a representation of an annotated rank decomposition $\Tc$ that encodes $G$ and has width $k$.
\end{lemma}
\begin{proof}
We implement the construction discussed in the proof of \Cref{obs:addtwin}.
Denote the stored decomposition by $\Tc' = (T',V(G) \setminus \{v\},\reps',\repse',\dmap')$ and let $\olp \in \leafe(T')$ so that $\reps'(\olp) = \{u\}$.
We construct $\Tc = (T,V(G),\reps,\repse,\dmap)$ as follows.
The tree $T$ is created by adding two children $c_1$ and $c_2$ for the leaf $l$ of $T'$.
The annotations for edges of $T$ that exist in $T'$ are directly copied from $\Tc'$ to $\Tc$.
Then we set $\reps(\vec{c_1 l}) \coloneqq \{u\}$, $\reps(\vec{c_2 l}) \coloneqq \{v\}$, and $\reps(\vec{l c_1}) \coloneqq \reps(\vec{l c_2}) \coloneqq \reps'(\opl)$.
We also set $\repse(c_1 l) \coloneqq \repse'(lp)$ and obtain $\repse(c_2 l)$ by replacing $u$ by $v$ in $\repse'(lp)$.
The functions $\dmap(c_1 l p)$ and $\dmap(c_2 l p)$ both map to the single vertex $u \in \reps(\olp)$.

We can verify that $\Tc$ is indeed an annotated rank decomposition that encodes $G$, and whose width is at most the width of $\Tc'$.
The construction can be implemented in $\Oh{1}$ time because $|\reps'(\olp)| = 1$ and $|\reps'(\opl)| \le 2$.
\end{proof}

\subsection{Proof of \texorpdfstring{\Cref{lem:compmain}}{Lemma \ref{lem:compmain}}}
Before finally proving \Cref{lem:compmain}, let us give the crucial subroutine for which the algorithm for \Cref{prob:rwcomp} is used.

\begin{lemma}
\label{lem:combine}
Let $T \colon \N \rightarrow \N$ be a function so that there is a $\Oh[k]{T(n)}$ time algorithm for \Cref{prob:rwcomp}.
Let also $G$ be a bipartite graph with bipartition $(X,Y_1 \cup Y_2)$, where $Y_1$ and $Y_2$ are disjoint.
There is an algorithm that given an annotated rank decomposition $\Tc_1$ of width at most $k$ that encodes $G[X,Y_1]$ and an annotated rank decomposition $\Tc_2$ of width at most $k$ that encodes $G[X,Y_2]$, either returns that the rankwidth of $G$ is more than $k$, or returns an annotated rank decomposition that encodes $G$ and has width at most $k$.
The algorithm runs in time $\Oh[k]{T(n) \log n}$, where $n = |X|+|Y_1|+|Y_2|$.
\end{lemma}
\begin{proof}
The algorithm is recursive.
Let $f$ be the function from \Cref{lem:neartwins}.

We first consider the base case that $|Y_2| \le f(k)$.
If $Y_1$ is empty, we can simply return $\Tc_2$.
Otherwise, let $v$ be an arbitrary vertex in $Y_1$.
We use \Cref{lem:addtwin} to add to $\Tc_1$ for each vertex $u \in Y_2$ two new vertices $u'$ and $u''$ as twins of $v$, and denote by $Y_2'$ the set of such vertices $u'$ and by $Y_2''$ such vertices $u''$.
Let $G'$ denote the resulting graph.
The rankwidth of $G'$ is at most $k$ because it is created from $G[X,Y_1]$ by adding twins.

We use \Cref{lem:getnbs} with $\Tc_2$ to compute for each $u \in Y_2$ the neighborhood $N(u)$, and with $\Tc_1$ to compute $N(v)$.
Then, we compute $F = \{u'w \mid u' \in Y_2', w \in N(u) \symd N(v)\}$.
As $|Y_2| \le f(k)$, this takes $\Oh[k]{n}$ time, which is also an upper bound for $|F|$.
We observe that the graph $G' \symd F$ is isomorphic to a graph created from $G[X, Y_1 \cup Y_2]$ by adding a twin for each vertex in $Y_2$.
Then we apply the algorithm for Twin Flipping (\Cref{prob:rwcomp}) with the sets $Y_2'$, $Y_2''$, and $F$, and the decomposition $\Tc_1$ to obtain either that the rankwidth of $G' \symd F$ is more than $k$, in which case we can return that the rankwidth of $G[X, Y_1 \cup Y_2]$ is more than $k$, or an annotated rank decomposition $\Tc'$ of $G' \symd F$ of width at most $k$.
This takes $\Oh[k]{T(n)}$ time.
Then, $\Tc'$ can be turned into an annotated rank decomposition $\Tc$ of $G[X,Y_1 \cup Y_2]$ by using \Cref{lem:dropleafs} to delete $Y_2''$ and renaming all vertices $u' \in Y_2'$ to $u \in Y_2$.
This takes $\Oh[k]{n}$ time.
This finishes the description of the base case.
The total running time in this case is $\Oh[k]{n} + \Oh[k]{T(n)} = \Oh[k]{T(n)}$ (we assume $T(n) \ge \Omega(n)$).

Then consider the case that $|Y_2| > f(k)$.
We first apply \Cref{lem:neartwins} with $\Tc_2$ and $Y_2$ to find $|Y_2|/f(k)$ disjoint pairs of vertices $(u_1,v_1), \ldots, (u_t,v_t)$ so that $u_i$ and $v_i$ are $(f(k) \cdot n/|Y_2|)$-near-twins, and the sets $N(u_i) \symd N(v_i)$.
We let $F = \bigcup_{i=1}^t \{v_i w \mid w \in N(u_i) \symd N(v_i)\}$.
This runs in $\Oh[k]{n}$ time, which is also an upper bound for $|F|$.
Let $Y_2' = \{u_1,\ldots,u_t\}$ and $Y_2'' = \{v_1,\ldots,v_t\}$.
We use \Cref{lem:dropleafs} to obtain an annotated rank decomposition $\Tc_2'$ that encodes $G[X,Y_2 \setminus Y_2'']$, and call the algorithm recursively with $\Tc_1$ and $\Tc_2'$.
If it returns that the rankwidth of $G[X,Y_1 \cup Y_2 \setminus Y_2'']$ is more than $k$, then we can return that the rankwidth of $G[X,Y_1 \cup Y_2]$ is more than $k$.
Otherwise, let $\Tc$ be the returned annotated rank decomposition that encodes $G[X,Y_1 \cup Y_2 \setminus Y_2'']$ and has width at most $k$.
We insert the vertices $Y_2'' = \{v_1,\ldots,v_t\}$ into $\Tc$ with \Cref{lem:addtwin} so that $v_i$ is inserted as a twin of $u_i$.
Let $G'$ be the graph that the resulting decomposition encodes.
We have that $G' \symd F = G[X, Y_1 \cup Y_2]$, and we apply the algorithm for \Cref{prob:rwcomp} with this decomposition and the sets $Y_2'$, $Y_2''$, and $F$.
This either returns that the rankwidth of $G[X, Y_1 \cup Y_2]$ is more than $k$ or an annotated rank decomposition of $G[X, Y_1 \cup Y_2]$ of width at most $k$.
This finishes the description of the recursive case.
The total running time of also this case, not counting the time spent in the recursive call, is also $\Oh[k]{T(n)}$.

At each level of recursion the size of $Y_2$ decreases by at least $|Y_2|/f(k)$, so the  depth of the recursion is $\Oh[k]{\log |Y_2|}$.
At each level the running time is $\Oh[k]{T(n)}$, so the total running time is $\Oh[k]{T(n) \log |Y_2|} = \Oh[k]{T(n) \log n}$.
\end{proof}

Then we prove \Cref{lem:compmain}, which we restate here.

\compmain*
\begin{proof}
By \Cref{lem:rwbg,lem:trrdbgtoardg}, proving the lemma under the assumption that $G$ is bipartite implies the lemma for general $G$.
Therefore, we then assume that $G$ is bipartite.
Let us fix a bipartition $(A,B)$ of $G$ and an indexing $B = \{v_1, \ldots, v_{|B|}\}$ of $B$, and initialize the data structure of \Cref{lem:twinds} with these.
This takes $\Oh{n + m}$ time.

We will describe a recursive algorithm that takes as input
\begin{itemize}
\item a subset $X \subseteq A$ and two integers $\ell,r$ with $1 \le \ell \le r \le |B|$,
\end{itemize}
and outputs
\begin{itemize}
\item either an annotated rank decomposition of $G[X,\{v_\ell,\ldots,v_r\}]$ of width at most $k$, or that $G[X,\{v_\ell,\ldots,v_r\}]$ has rankwidth more than $k$.
\end{itemize}

We denote $Y = \{v_\ell,\ldots,v_r\}$.
If $|Y| = 1$, we compute an annotated rank decomposition of $G[X,Y]$ of width at most $1$ in time $\Oh{|X|+|Y|}$ and return it.
Then assume $|Y| \ge 2$.

We first use the data structure of \Cref{lem:twinds} to find the twin-equivalence classes of $X$ in $G[X,Y]$, and then compute a set $X' \subseteq X$ that contains exactly one vertex from each of the equivalence classes.
We also store for each vertex $u \in X \setminus X'$ a vertex $u_x \in X'$ so that $N_{G[X,Y]}(u) = N_{G[X,Y]}(u_x)$.
This step takes in total $\Oh{|X| \log n}$ time.
Let $f(k)$ be the function from \Cref{lem:bipartitetwins}.
The lemma implies that if $|X'| \ge f(k) \cdot |Y|$, then the rankwidth of $G[X',Y]$ (and thus also of $G[X,Y]$) is more than $k$.
In this case we can return immediately,
Then assume $|X'| \le \Oh[k]{|Y|}$.

We select $t \in [\ell,r-1]$ so that both $Y_1 = \{v_\ell,\ldots,v_t\}$ and $Y_2 = \{v_{t+1},\ldots,v_r\}$ have size either $\lfloor |Y|/2 \rfloor$ or $\lceil |Y|/2 \rceil$.
Then we make two recursive calls of the algorithm, one with $X'$ and $\ell,t$, and another with $X'$ and $t+1,r$.
If either of the calls returns that the graph has rankwidth more than $k$, we can return that the rankwidth of $G[X,Y]$ is more than $k$.
Otherwise, let $\Tc_1$ be the decomposition returned by the first call and $\Tc_2$ the decomposition returned by the second call.
We apply the algorithm of \Cref{lem:combine} with these decompositions to either conclude that the rankwidth of $G[X,Y]$ is more than $k$, or to obtain an annotated rank decomposition $\Tc$ of $G[X',Y]$ of width at most $k$.
This runs in $\Oh[k]{T(|X'|+|Y|) \log (|X'|+|Y|)}$ time.
Finally, we insert the vertices $X \setminus X'$ to the decomposition in $\Oh{|X \setminus X'|}$ time by using \Cref{lem:addtwin}, and return the resulting decomposition.
This completes the description of the algorithm.

We observe that the running time of each recursive call, not counting the time spent in the subcalls, is $\Oh[k]{T(|X|+|Y|) \log n}$.
The sum of the sizes of the sets $Y$ over all such calls is $\Oh{n \log n}$.
On all calls except the first, it is guaranteed that $|X| \le \Oh[k]{|Y|}$, so the sum of sizes of the sets $X$ over all such calls is $\Oh[k]{n \log n}$.
Then, the facts that $|X|+|Y| \le n$ in each call and the function $T$ is convex imply that the total running time past the initialization of the data structure of \Cref{lem:twinds} is $\Oh[k]{T(n) \log^2 n}$.
This concludes the proof since the data structure is initialized in $\Oh{n + m}$ time.
\end{proof}

%% file: dealternation.tex
\section{Dealternation Lemma}
\label{sec:dealternation}

In this section, we prove the Dealternation Lemma announced in \cref{lem:dealt}:

\dealternationlemma*

We will actually prove a~slightly more general result, showing an~analog of the Dealternation Lemma for \emph{subspace arrangements} -- structures described by families of linear spaces that generalize the notions of graphs, hypergraphs and linear matroids.

We begin by introducing the concepts and notation used throughout the proof.

\subsection{Section-specific preliminaries}
\label{ssec:dealternation-prelims}

\paragraph*{Linear spaces.}
Let $\F$ be a~fixed finite field; in this work we assume $\F = \GF(2)$.
The linear space over $\F$ of dimension $d$ is denoted by $\F^d$.
Given two linear subspaces $V_1, V_2$ of $\F^d$, we denote by $V_1 + V_2$ their sum and by $V_1 \cap V_2$ their intersection.
By $\dim(V)$ we denote the dimension of the subspace $V$ of $\F^d$.

The following facts are standard.

\begin{lemma}
  For any two linear subspaces $V_1, V_2$ of $\F^d$, we have that
  \[ \dim(V_1) + \dim(V_2) = \dim(V_1 + V_2) + \dim(V_1 \cap V_2). \]
\end{lemma}

\begin{lemma}[{\cite[Lemma 25]{DBLP:journals/tit/JeongKO17}}]
  \label{lem:linear-dimension-switching}
  For any four linear subspaces $U_1, U_2, V_1, V_2$ of $\F^d$, we have that
  \[
  \begin{split}
  &\hphantom{=} \dim((U_1 + U_2) \cap (V_1 + V_2)) + \dim(U_1 \cap U_2) + \dim(V_1 \cap V_2)\\
  &= \dim((U_1 + V_1) \cap (U_2 + V_2)) + \dim(U_1 \cap V_1) + \dim(U_2 \cap V_2).
  \end{split}
  \]
\end{lemma}

For any set of vectors $A \subseteq \F^d$, we denote by $\sumof{A}$ the subspace of $\F^d$ spanned by the vectors of $A$.
If $\dim(\sumof{A}) = |A|$, then we say that $A$ is a~\emph{basis} of $\sumof{A}$.
Then any permutation $\BB$ of elements of $A$ is called an~\emph{ordered basis} of $\sumof{A}$; for convenience, we define that $\sumof{\BB} = \sumof{A}$.
Letting $\BB = (\lin{v}_1, \lin{v}_2, \dots, \lin{v}_c)$, we have that every vector $\lin{u} \in \sumof{\BB}$ can be uniquely represented as a~linear combination $\lin{u} = \sum_{i=1}^c \alpha_i \lin{v}_i$.
In this work, whenever the ordered basis $\BB$ of a~vector space $V$ is known from context, all vectors $\lin{u} \in V$ will be implicitly represented as such a~linear combination.
Similarly, subspaces of $V$ are then implicitly represented as $\sumof{\{\lin{u}_1, \dots, \lin{u}_d\}}$, where $\lin{u}_1, \dots, \lin{u}_d \in V$ are implicitly represented as linear combinations of vectors of $\BB$.
Such a~representation can be then stored using $\Oh{cd}$ elements of $\F$.

\paragraph*{Subspace arrangements and rank decompositions.}
Let $d \in \N$ and consider the linear space $\F^d$.
Any family $\Vc = \{V_1, V_2, \dots, V_n\}$ of linear subspaces of $\F^d$ is called a~\emph{subspace arrangement}.
For visual clarity, let $\sumof{\Vc} = \sum_{i=1}^n V_i$ be the sum of all subspaces in the arrangement.

A~\emph{rank decomposition} of a~subspace arrangement $\Vc$ is a~pair $\Tc = (T, \lambda)$, where $T$ is a~cubic tree and $\lambda$ is a~bijection $\lambda\,\colon\,\Vc \to \leafe(T)$.
For an~oriented edge $\ouv \in \oE(T)$, we denote by $\lparts(\Tc)[\ouv] = \{\lmap^{-1}(\olp) \,\mid\, \olp \in \leafe(T)[\ouv]\}$ the subfamily of $\Vc$ comprising all linear subspaces that are mapped to leaf edges that are closer to $u$ than $v$.
The \emph{boundary space} of an~edge $uv$ is defined as $B_{uv} = \sumof{\lparts(\Tc)[\ouv]} \cap \sumof{\lparts(\Tc)[\ovu]}$.

A~\emph{rooted rank decomposition} is defined analogously to a~rank decomposition, only that $T$ is a~binary tree.
Recall that a~rank decomposition can be rooted by subdividing a~single edge $uv$ once -- replacing it with a~path $urv$ -- and rooting the tree at $r$.
The boundary space of a~non-root node $v$ with parent $p$ is $B_v = B_{vp}$ and the boundary space of the root $r$ is $B_r = \{\lin{0}\}$.
Also, we set $\lparts(\Tc)[v] = \lparts(\Tc)[\vec{vp}]$ for $v \neq r$ and $\lparts(\Tc)[r] = \Vc$.

The width of an~edge $uv \in E(T)$ is defined as $\dim(B_{uv})$.
The width of a~rank decomposition is the maximum width of any edge of the decomposition.
Thus, the width of a~rooted rank decomposition is equivalently the maximum value of $\dim(B_v)$ ranging over non-root nodes $v$.

Rank decompositions of (partitioned) graphs can be transformed to equivalent rank decompositions of subspace arrangements; the reduction is shown below, but it is also present in~\cite{DBLP:journals/siamdm/JeongKO21}.

Suppose $G$ is a~graph; for simplicity, assume $V(G) = \{1, \dots, |V(G)|\}$.
Consider the vector space $\GF(2)^{|V(G)|}$ and its canonical basis $\{\lin{e}_1, \lin{e}_2, \dots, \lin{e}_{|V(G)|}\}$.
To each vertex $v \in V(G)$ assign the vector space $A_v$ spanned by the vectors $\lin{e}_v$ and $\sum_{u \in N(v)} \lin{e}_u$, which we will call the \emph{canonical subspace} of $v$.
Similarly, for a~set $S \subseteq V(G)$, we assign to it the canonical subspace $A_S \coloneqq \sum_{v \in S} A_v$. 
It is then straightforward to verify that:
\begin{lemma}[{\cite[Lemma 52]{DBLP:journals/tit/JeongKO17}}]
For any set $S \subseteq V(G)$, we have
\[ \dim\left( A_S \cap A_{V(G) \setminus S} \right) = 2 \cdot \cutrk(S). \]
\end{lemma} 

We then immediately have that:

\begin{lemma}
  \label{lem:subspace-arrangement-from-graph}
  Let $(G, \Cc)$ be a~partitioned graph with $\Cc = \{S_1, S_2, \dots, S_n\}$.
  Let $\Vc = \{V_1, V_2, \dots, V_n\}$ be a~subspace arrangement over $\GF(2)^{|V(G)|}$, where $V_i = A_{S_i}$ for each $i \in [n]$. 
  Then $\Vc$ satisfies the following property.
  
  Let $T$ be a~cubic tree with leaves $\ell_1, \dots, \ell_n$.
  Define bijections $\lambda_1\,\colon\,\Cc \to \leafe(T)$ and $\lambda_2\,\colon\,\Vc \to \leafe(T)$ so that for every $i \in [n]$, both $\lambda_1(S_i)$ and $\lambda_2(V_i)$ are assigned to the oriented edge incident to $\ell_i$.
  Note that $\Tc = (T, \lambda_1)$ is a~rank decomposition of $(G, \Cc)$ and $\Tc' = (T, \lambda_2)$ is an~(isomorphic) rank decomposition of $\Vc$.
  Then the width of $\Tc'$ is equal to twice the width of $\Tc$.
\end{lemma}


The statement of the Dealternation Lemma (\cref{lem:dealt}) can be thus generalized to the rank decompositions of subspace arrangements.
Mimicking the concepts defined for graphs, we say that a~set $F \subseteq \Vc$ is a~tree factor of $\Tc = (T, \lambda)$ if $F = \lparts(\Tc)[t]$ for some $t \in V(T)$; and a~context factor if it is not a~tree factor, but a set of the form $F = F_1 \setminus F_2$, where $F_1$ and $F_2$ are tree factors of $\Tc$.
$F$ is a~factor of $\Tc$ if it is either a~tree factor or a~context factor of $\Tc$.
Then:

\begin{lemma}[Dealternation Lemma for subspace arrangements]
  \label{lem:subspace-dealternation}
  There exists a~function $f_{\ref{lem:subspace-dealternation}}\,\colon\,\N\to\N$ so that if $\Vc$ is a~subspace arrangement and $\Tc^b = (T^b, \lambda^b)$ is a~rooted rank decomposition of $\Vc$ of width $\ell \geq 0$, then there exists a~rooted rank decomposition $\Tc$ of $\Vc$ of optimum width so that for every node $t \in V(T^b)$, the set $\lparts(\Tc^b)[t]$ can be partitioned into a~disjoint union of at most $f_{\ref{lem:subspace-dealternation}}(\ell)$ factors of $\Tc$.
\end{lemma}

Note that \cref{lem:subspace-dealternation} directly implies the Dealternation Lemma through \cref{lem:subspace-arrangement-from-graph}.
Hence, the rest of this section will be devoted to the proof of \cref{lem:subspace-dealternation}.

\paragraph*{Fullness, emptiness and mixedness of edges and nodes.}
Let $\Vc = \{\Vc_1, \Vc_2, \dots, \Vc_n\}$ be a~subspace arrangement and $\Tc^b = (T^b, \lambda^b)$ be a~rooted rank decomposition of $\Vc$ (possibly of unoptimal width).
We introduce the ancestor-descendant relationship on the nodes of $T^b$: we say $x \leq y$ whenever $x = y$ or $x$ is a~descendant of $y$, and by $x < y$ we mean $x \leq y$ and $x \neq y$.
Moreover, define $\Vc_x = \lparts(\Tc^b)[x]$ as the subfamily of $\Vc$ comprising those subspaces $\Vc_i$ that are mapped to the leaf edges $\olp$ with $x \geq l$.
Note that $\Vc_r = \Vc$ if and only if $r$ is the root of $T^b$, and $|\Vc_l| = 1$ if and only if $l$ is a~leaf of $T^b$.
We will then say that each $V \in \Vc_x$ is \emph{in the subtree of $\Tc^b$ rooted at $x$}.
We remark that if $x, y \in V(T^b)$ with $x \leq y$, then $\Vc_x \subseteq \Vc_y$; and whenever $x, y$ are not in the ancestor-descendant relationship in $T^b$, then $\Vc_x \cap \Vc_y = \emptyset$.

For the following description, consider a~node $x$ of $T^b$.
Let $\Tc = (T, \lambda)$ be a~rank decomposition of $\Vc$ (rooted or unrooted).
Define $\lparts_x(\Tc)[\ouv] = \lparts(\Tc)[\ouv] \cap \Vc_x$ as the family of linear spaces containing exactly those linear spaces $V \in \Vc$ that:
\begin{itemize}
  \item are in the subtree of $(T^b, \lambda^b)$ rooted at $x$; and
  \item in $(T, \lambda)$, are mapped to a~leaf edge closer to $u$ than $v$.
\end{itemize}
Similarly, we set $\lparts_{\bar{x}}(\Tc)[\ouv] = \lparts(\Tc)[\ouv] \setminus \Vc_x = \lparts(\Tc)[\ouv] \setminus \lparts_x(\Tc)[\ouv]$.
Note that if an~edge $\vec{v_1 v_2}$ is a~predecessor of an~edge $\vec{v_3 v_4}$ in $T$, then $\lparts_x(\Tc)[\vec{v_1 v_2}] \subseteq \lparts_x(\Tc)[\vec{v_3 v_4}]$ and $\lparts_{\bar{x}}(\Tc)[\vec{v_1 v_2}] \subseteq \lparts_{\bar{x}}(\Tc)[\vec{v_3 v_4}]$.

We also say that a directed edge $\ouv$ of $T$ is:
\begin{itemize}
  \item \emph{$x$-full} if $\lparts(\Tc)[\ouv] \subseteq \Vc_x$; that is, for every leaf edge $e$ of $(T, \lambda)$ closer to $u$ than $v$, $e$ is mapped to a~space $V \in \Vc$ in the subtree of $\Tc^b$ rooted at $x$;
  \item \emph{$x$-empty} if $\lparts(\Tc)[\ouv] \cap \Vc_x = \emptyset$, or equivalently, $\lparts_x(\Tc)[\ouv] = \emptyset$;
  \item \emph{$x$-mixed} otherwise.
\end{itemize}

Similarly, if $\Tc$ is rooted, then we additionally say that a~node $v \in V(T)$ is $x$-full (resp.~$x$-empty or $x$-mixed) if $\lparts(\Tc)[v] \subseteq \Vc_x$ (resp.~$\lparts(\Tc)[v] \cap \Vc_x = \emptyset$ or $\lparts(\Tc)[v] \cap \Vc_x \notin \{\emptyset, \lparts(\Tc)[v] \}$).
Equivalently for non-root nodes $v$, $v$ is $x$-full (resp.~$x$-empty, $x$-mixed) if and only if the directed edge $\vec{vp}$ is $x$-full (resp.~$x$-empty, $x$-mixed), where $p$ is the parent of $v$ in $T$.

The following observation shows how the notions of fullness, emptiness and mixedness of edges of $T$ are related for pairs of nodes of $T^b$:

\begin{observation}
  \label{obs:mixedness-relations}
  Let $x, y \in V(T^b)$ and $\ouv \in \oE(T)$.
  \begin{itemize}
    \item If $\ouv$ is $x$-empty and $y \leq x$, then $\ouv$ is $y$-empty.
    \item If $\ouv$ is $x$-mixed and $y \ngtr x$, then $\ouv$ is $y$-empty or $y$-mixed.
    \item If $\ouv$ is $x$-mixed and $y \geq x$, then $\ouv$ is $y$-mixed or $y$-full.
    \item If $\ouv$ is $x$-full and $y \geq x$, then $\ouv$ is $y$-full.
    \item If $\ouv$ is $x$-full and $x, y$ are not in the ancestor-descendant relationship, then $\ouv$ is $y$-empty.
  \end{itemize}
\end{observation}

Naturally, \cref{obs:mixedness-relations} directly translates to the fullness, emptiness and mixedness of nodes of $T$ whenever $T$ is rooted.

\paragraph*{Well-structured rank decompositions.}
In the proof we will use the result of Jeong, Kim and Oum~\cite{DBLP:journals/siamdm/JeongKO21} asserting the existence of well-structured rank decompositions of subspace arrangements of optimum width, called \emph{totally pure rank decompositions}.
We defer the formal definition to \cref{sec:totally-pure-decomp}, but intuitively, a~rank decomposition $\Tc$ of a~subspace arrangement $\Vc$ is totally pure with respect to another rank decomposition $\Tc^b$ if, for every $x \in V(\Tc^b)$, $\Tc$ excludes some small local patterns defined in terms of subspaces $\lparts_x(\Tc)[\ouv]$ for $\ouv \in \oE(\Tc)$.

\begin{restatable}[{\cite[Proposition 4.6]{DBLP:journals/siamdm/JeongKO21}}]{lemma}{koreanrankstructuretheorem}
  \label{lem:korean-rank-structure-theorem}
  Let $\Tc^b$ be a~rooted rank decomposition of a~subspace arrangement $\Vc$.
  Then there exists a~rooted rank decomposition $\Tc$ of the same subspace arrangement $\Vc$ of optimum width that is totally pure with respect to $\Tc^b$.
\end{restatable}

\subsection{Mixed skeletons}
\label{ssec:dealternation-mixed-skeletons}
Suppose again that $\Vc$ is a~subspace arrangement, $\Tc^b = (T^b, \lambda^b)$ is a~rooted rank decomposition of $\Vc$, and $\Tc = (T, \lambda)$ is a~rooted rank decomposition of $\Vc$.
Let $x \in V(T^b)$ be a~node of $T^b$.
We define the \emph{$x$-mixed skeleton} of $\Tc$ as a~(possibly empty) rooted tree $T^\mix$ with $V(T^\mix) \subseteq V(T)$ constructed as follows.
For $v \in V(T)$, we put $v$ in $V(T^\mix)$ if $v$ has two children and one of the following cases holds:
\begin{itemize}
  \item one child is $x$-empty and the other is $x$-full; or
  \item both children are $x$-mixed.
\end{itemize}
In the first case we will say that $v$ is an~\emph{$x$-leaf point}, and in the second -- that $v$ is an~\emph{$x$-branch point}.
Then two vertices $u, v \in V(T^\mix)$ are connected by an~edge in $T^\mix$ if the path between $u$ and $v$ in $T$ is internally disjoint from $V(T^\mix)$ (\cref{fig:skeleton-def}).

\begin{figure}
  \centering
  \begin{subfigure}[b]{0.4\textwidth}
    \centering
    \includegraphics[width=\textwidth]{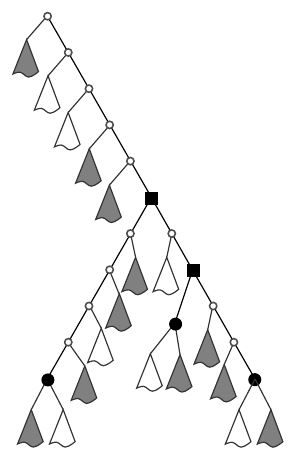}
    \caption{}
  \end{subfigure}
  \hspace{0.1\textwidth}
  \begin{subfigure}[b]{0.15\textwidth}
    \centering
    \includegraphics[width=\textwidth]{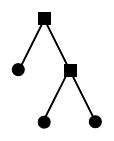}
    \caption{}
  \end{subfigure}
  
  \caption{(a) An~example rooted rank decomposition $\Tc$. All nodes of $\Tc$ that are not $x$-mixed are contracted to rooted subtrees; white subtrees have $x$-empty roots, while the dark subtrees have $x$-full roots. The $x$-leaf points of $\Tc$ are marked by $\bullet$, while the $x$-branch points of $\Tc$ are marked by~$\blacksquare$. The nodes that are $x$-mixed in $\Tc$, but neither $x$-leaf points nor $x$-branch points, are marked by~$\circ$. (b) The $x$-mixed skeleton of $\Tc$.}
  \label{fig:skeleton-def}
\end{figure}

We will now show the correctness and the properties of this construction in a~series of claims.

\begin{lemma}
  \label{lem:mixed-skeleton-mixed-ancestors}
  Suppose $v \in V(T^\mix)$.
  Then every ancestor of $v$ (including $v$) is $x$-mixed.
\end{lemma}
\begin{proof}
  Follows from the straightforward verification with the definitions.
\end{proof}

It is also easily verified that a~``converse'' statement also holds:
\begin{lemma}
  \label{lem:mixed-skeleton-mixed-descendant}
  Suppose $v \in V(T)$ is $x$-mixed.
  Then some descendant of $v$ in $T$ is an $x$-leaf point in $\Tc$.
\end{lemma}

From the following lemma it follows directly that $T^\mix$ indeed forms a~rooted tree; in particular, $uv \in E(T^\mix)$ implies that $u$ and $v$ are in the ancestor-descendant relationship in $T$:

\begin{lemma}
  \label{lem:mixed-skeleton-lca}
  Suppose $u, v \in V(T^\mix)$.
  Then the lowest common ancestor of $u$ and $v$ belongs to $T^\mix$.
\end{lemma}
\begin{proof}
  Let $w$ be the lowest common ancestor of $u$ and $v$.
  if $w \in \{u, v\}$, then the lemma is trivial.
  Otherwise, let $w_u$ and $w_v$ be the two children of $w$ that are ancestors of $u$ and $w$, respectively.
  By \cref{lem:mixed-skeleton-mixed-ancestors}, both $w_u$ and $w_v$ are $x$-mixed.
  Thus $w$ is an~$x$-branch point.
\end{proof}

We continue with several properties of mixed skeletons:

\begin{lemma}
  \label{lem:mixed-skeleton-mixed-path}
  Suppose $p, q \in V(T^\mix)$ and let $uv \in E(T)$ be an~edge on the path between $p$ and $q$ in $T$.
  Then both $\ouv$ and $\ovu$ are $x$-mixed.
\end{lemma}
\begin{proof}
  Suppose not. Without loss of generality assume that: $pq \in E(T^\mix)$, and in particular that $p$ is an~ancestor of $q$ in $T$; and that in $T$, $q$ is closer to $u$ than $v$.
  Let $q_1, q_2$ be the two children of $q$ in $T$ and $p_1, p_2$ be the two children of $p$ in $T$; without loss of generality, assume $p_1$ is an~ancestor of $q$.
  Note that by \cref{lem:mixed-skeleton-mixed-ancestors}, $p_1$ is $x$-mixed; therefore, since $p \in V(T^\mix)$, we have that $p_2$ is $x$-mixed as well.
  
  First suppose that $\ouv$ is $x$-full.
  Then it follows immediately that both $q_1$ and $q_2$ are $x$-full as well (since both $\vec{q_1q}$ and $\vec{q_2q}$ are predecessors of $\ouv$), contradicting that $q \in V(T^\mix)$.
  A~similar contradiction follows when $\ouv$ is $x$-empty.
  In the same way, observe that if $\ovu$ is $x$-full (resp.~$x$-empty), then $p_2$ is $x$-full (resp.~$x$-empty) as well since $\vec{p_2p}$ is a~predecessor of $\ovu$.
  Therefore, both $\ouv$ and $\ovu$ must be $x$-mixed.
\end{proof}

The following lemma implies that the $x$-mixed skeleton is a~full binary tree. 

\begin{lemma}
  \label{lem:mixed-skeleton-binary}
  Every $x$-leaf point is a~leaf of $T^\mix$, and every $x$-branch point is an~internal node of $T^\mix$ with two children.
\end{lemma}
\begin{proof}
  If $v$ is an~$x$-leaf point, then naturally every strict descendant of $v$ in $T$ is either $x$-full or $x$-empty.
  Thus by \cref{lem:mixed-skeleton-mixed-ancestors}, no strict descendant of $v$ is in $T^\mix$ and therefore $v$ is a~leaf in $T^\mix$.
  
  Then let $v$ be a~$x$-branch point.
  Let $v_1, v_2$ be the children of $v$ in $T$; by definition, both $v_1$ and $v_2$ are $x$-mixed.
  By \cref{lem:mixed-skeleton-mixed-descendant}, there exist $x$-leaf points $u_1, u_2 \in V(T^\mix)$ that are descendants of $v_1$ and $v_2$ in $T$, respectively, which implies that $v$ has at least two children in $T^\mix$.
  The lemma follows by observing from \cref{lem:mixed-skeleton-lca} that for each $i \in \{1, 2\}$, at most one vertex of $V(T^\mix)$ in the subtree of $T$ rooted at $v_i$ can be connected to $v$ by a~path internally disjoint from $V(T^\mix)$.
\end{proof}

The main product of this subsection is the following statement asserting that there exists an~optimum-width rooted decomposition of $\Vc$ admitting small $x$-mixed skeletons for \emph{all} $x \in V(T^b)$.

\begin{restatable}{lemma}{smallmixedskeletonlemma}
  \label{lem:small-mixed-skeleton}
  There exists a~function $f_{\ref{lem:small-mixed-skeleton}}\,\colon\,\N \to \N$ such that the following holds.
  Let $\Tc^b = (T^b, \lambda^b)$ be a~rooted rank decomposition of $\Vc$ of width $\ell \geq 0$.
  Then there exists a~rooted rank decomposition $\Tc$ of $\Vc$ of optimum width such that, for every $x \in V(T^b)$, the $x$-mixed skeleton of $\Tc$ contains at most $f_{\ref{lem:small-mixed-skeleton}}(\ell)$ nodes.
\end{restatable}

The proof of \cref{lem:small-mixed-skeleton} is delayed to \cref{sec:totally-pure-decomp}.
There, we will show that any decomposition that is totally pure with respect to $\Tc^b$ fulfills the requirements of \cref{lem:small-mixed-skeleton}; this is done by a~straightforward (though careful) analysis of the definition of a totally pure decomposition.
Hence, the lemma is correct thanks to \cref{lem:korean-rank-structure-theorem}.


\subsection{Statement of the Local Dealternation Lemma}

The strategy of the proof of the Dealternation Lemma for subspace arrangements (\cref{lem:subspace-dealternation}) will be similar to that in the work of Bojańczyk and Pilipczuk~\cite{DBLP:journals/lmcs/BojanczykP22}: Given as input a~decomposition $\Tc^b$ of width $\ell \geq 0$, we first create a~decomposition $\Tc$ satisfying some strong structural properties and then update $\Tc$ in a~sequence of local improvement steps so as to produce the decomposition satisfying the Dealternation Lemma, preserving the structural properties throughout the process.
In our case of rank decompositions of subspace arrangements, the property maintained throughout the process is precisely admitting small $x$-mixed skeletons for all $x \in V(T^b)$.
Now we define the local improvement step in the form of the Local Dealternation Lemma.

Reusing the notation from the previous sections, assume that $x \in V(T^b)$.
We say that a~set $F \subseteq \Vc$ is an~$x$-factor (resp.~$x$-tree factor, $x$-context factor) in $\Tc$ if it is a~factor (resp.~tree factor, context factor) in $\Tc$ and moreover $F \subseteq \Vc_x$.

\begin{lemma}[Local Dealternation Lemma]
  \label{lem:local-dealternation}
  There exists a~function $f_{\ref{lem:local-dealternation}}\,\colon\,\N\to\N$ so that the following holds.
  Suppose $\Tc^b$ is a~rooted rank decomposition of $\Vc$ of width $\ell \geq 0$, and $\Tc$is a~rooted rank decomposition of $\Vc$ of optimum width. Moreover, let $x \in V(T^b)$ be such that the $x$-mixed skeleton of $\Tc$ has at most $f_{\ref{lem:small-mixed-skeleton}}(\ell)$ nodes. Then there exists a~rooted rank decomposition $\Tc'$ of $\Vc$ of optimum width such that:
  \begin{itemize}
    \item the set $\lparts(\Tc^b)[x]$ is a~disjoint union of at most $f_{\ref{lem:local-dealternation}}(\ell)$ $x$-factors of $\Tc'$;
    \item for every $y \in V(T^b)$, the $y$-mixed skeletons of $\Tc$ and $\Tc'$ are equal; and
    \item for every $y \in V(T^b)$ with $y \ngtr x$, every $y$-factor of $\Tc$ is also a~$y$-factor of $\Tc'$.
  \end{itemize}
\end{lemma}

We proceed to show how the ``global variant'' of the Dealternation Lemma for subspace arrangements (\cref{lem:subspace-dealternation}) follows from \cref{lem:local-dealternation}.

\begin{proof}[Proof of \cref{lem:subspace-dealternation} from the Local Dealternation Lemma]
  Create an~ordering $x_1, x_2, x_3, \dots, x_n$ of the nodes of $T^b$ consistent with the descendant-ancestor relationship $<$; that is, choose any~ordering of the nodes in which for every pair of nodes $x, y$ such that $x$ is a~descendant of $y$, $x$ precedes $y$ in the ordering.
  Throughout the proof, we will inductively create a~sequence of rooted rank decompositions of $\Vc$ of optimum width: $\Tc_0, \Tc_1, \dots, \Tc_n$, such that for each $t \in [0, n]$, the decomposition $\Tc_t$ satisfies the following properties:
  \begin{itemize}
    \item for every $i \in [t]$, the set $\lparts(\Tc^b)[x_i]$ is a~disjoint union of at most $f_{\ref{lem:local-dealternation}}(\ell)$ $x_i$-factors of $\Tc_t$; and
    \item for every $i \in [n]$, the $x_i$-mixed skeleton of $\Tc_t$ contains at most $f_{\ref{lem:small-mixed-skeleton}}(\ell)$ nodes.
  \end{itemize}
  Then the decomposition $\Tc_n$ will witness the Dealternation Lemma for the subspace arrangement $\Vc$, with $f_{\ref{lem:subspace-dealternation}} = f_{\ref{lem:local-dealternation}}$.

  By \cref{lem:small-mixed-skeleton}, there exists a~rank decomposition $\Tc_0$ of $\Vc$ of optimum width such that for every $x \in V(T^b)$, the $x$-mixed skeleton of $\Tc_0$ contains at most $f_{\ref{lem:small-mixed-skeleton}}(\ell)$ nodes.
  This verifies the inductive assumption about $\Tc_0$.
  
  Now assume that $t \in [n]$, we are given a~rank decomposition $\Tc_{t-1}$ of optimum width satisfying the inductive assumption, and we want to produce a~rank decomposition $\Tc_t$.
  Let us apply \cref{lem:local-dealternation} with the decomposition $\Tc_{t-1}$ and $x = x_t$, yielding the decomposition $\Tc_t$.
  We are left to verify that $\Tc_t$ satisfies the inductive assumptions.
  
  First, for every $i \in [t - 1]$, the set $\lparts(\Tc^b)[x_i]$ is a~disjoint union of at most $f_{\ref{lem:local-dealternation}}(\ell)$ $x_i$-factors of $\Tc_{t-1}$.
  Observe that $x_i \ngtr x_t$ by the construction of the order $x_1, \dots, x_n$.
  Thus by \cref{lem:local-dealternation}, each such factor is also an~$x_i$-factor of $\Tc_t$.
  Also, directly by \cref{lem:local-dealternation} we have that $\lparts(\Tc^b)[x_t]$ is a~disjoint union of at most $f_{\ref{lem:local-dealternation}}(\ell)$ $x_t$-factors of $\Tc_t$.
  
  Finally, let $i \in [n]$ and recall that the $x_i$-mixed skeleton of $\Tc_{t-1}$ contains at most $f_{\ref{lem:small-mixed-skeleton}}(\ell)$ nodes.
  By \cref{lem:local-dealternation}, the $x_i$-mixed skeletons of $\Tc_t$ and $\Tc_{t-1}$ are equal, so the bound on the number of nodes applies also to the $x_i$-mixed skeleton of $\Tc_t$.  Thus the inductive step is correct and thus the sought decomposition $\Tc_n$ exists.
\end{proof}

The following sections will introduce operations implementing ``local rearrangements'' of rank decompositions that will be used in the proof of the Local Dealternation Lemma: tree swaps and block shuffles.

\subsection{Tree swaps}


Again assume that $\Tc^b = (T^b, \lambda^b)$ and $\Tc = (T, \lambda)$ are rooted rank decompositions of $\Vc$.
Let $T$ contain a~vertical path $v_0v_1v_2v_3$.
We define a~\emph{swap} of $\Tc$ along the vertical path $v_0v_1v_2v_3$ as an~update of the decomposition replacing the path $v_0v_1v_2v_3$ with the (vertical) path $v_0v_2v_1v_3$ (\cref{fig:swap-def}).
It is easy to see that after the swap, the resulting tree remains binary.
Note also that swaps are invertible: whenever the swap of $\Tc$ along $v_0v_1v_2v_3$ produces a~tree $\Tc_\swap$, the original decomposition $\Tc$ is a~result of a~swap of $\Tc_\swap$ along $v_0v_2v_1v_3$.
Finally, we say that a~swap of $\Tc$ along the vertical path $v_0v_1v_2v_3$ is an~\emph{$x$-swap} for some $x \in V(T^b)$ if the following preconditions are met:
\begin{itemize}
  \item $v_3$ is $x$-mixed; and
  \item if $v'_1$ and $v'_2$ are the (unique) children of $v_1$ and $v_2$, respectively, outside of the path, then exactly one of the nodes $v'_1, v'_2$ is $x$-empty and the other is $x$-full.
\end{itemize}

Observe that whenever $\Tc'$ is an~$x$-swap of $\Tc$ along $v_0v_1v_2v_3$, then also $\Tc$ is an~$x$-swap of $\Tc'$ along $v_0v_2v_1v_3$.

\begin{figure}
  \centering
  \includegraphics[width=0.25\textwidth]{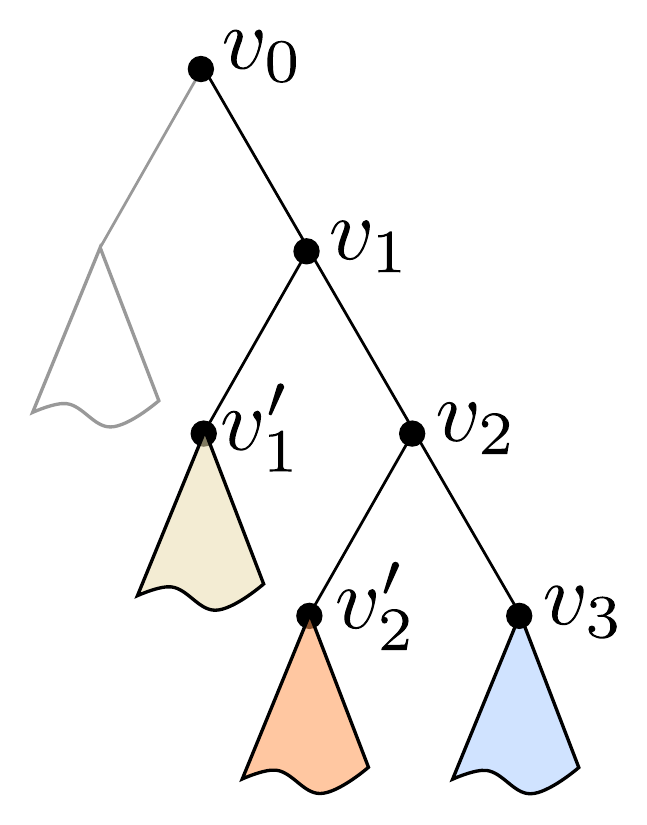}
  \hspace{0.1\textwidth}
  \includegraphics[width=0.25\textwidth]{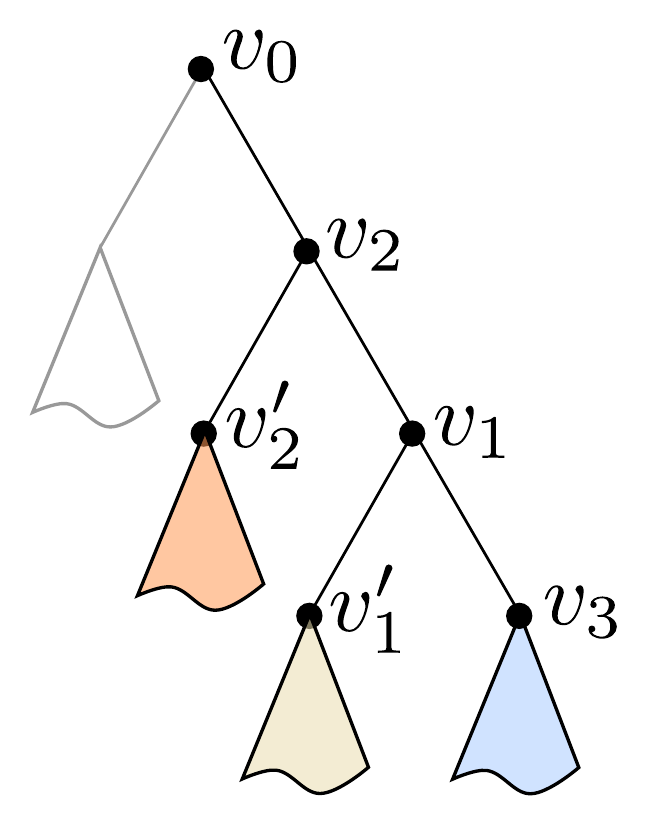}
  
  \caption{An~example swap. The right decomposition is a~swap of the left decomposition along $v_0v_1v_2v_3$.}
  \label{fig:swap-def}
\end{figure}


The main product of this subsection is the following lemma, asserting that for any $x \in V(T^b)$, any $x$-swap of $\Tc$ preserves the $y$-mixed skeletons of $\Tc$ for all $y \in V(T^b)$:
\begin{lemma}
  \label{lem:swap-preserves-mixed-skeletons}
  Let $x, y \in V(T^b)$. Suppose $\Tc_\swap$ is created from $\Tc$ by performing an~$x$-swap along the path $v_0v_1v_2v_3$.
  Then the $y$-mixed skeletons of $\Tc$ and $\Tc_\swap$ are equal.
\end{lemma}

The rest of this section is dedicated to the proof of \cref{lem:swap-preserves-mixed-skeletons}.
The proof proceeds in two steps.
First, we phrase, in terms of $y$-emptiness, $y$-mixedness and $y$-fullness of nodes only, the structural properties of a~vertical path $v_0v_1v_2v_3$ in $\Tc$ which, when fulfilled by the path, implies the perseverance of the $y$-mixed skeleton of $\Tc$ after the swap along $v_0v_1v_2v_3$.
Then we show that whenever a~swap of $\Tc$ along a~path $P$ happens to be an~$x$-swap for any $x \in V(T^b)$, then $P$ fulfills this structural property for every $y \in V(T^b)$; hence, such a~swap will preserve all $y$-mixed skeletons for all $y \in V(T^b)$.

\smallskip

Let $v_0v_1v_2v_3$ be a~vertical path in $T$, and $v'_1, v'_2$ be the neighbors of $v_1$ and $v_2$, respectively, outside of the path.
Let also $y \in V(T^b)$.
We then say that the path satisfies:
\begin{itemize}
  \item the \emph{$y$-empty property} if at least one of $v'_1$ and $v'_2$ is $y$-empty, and $v_3$ is either $y$-empty or $y$-mixed; and
  \item the \emph{$y$-full property} if at least one of $v'_1$ and $v'_2$ is $y$-full, and $v_3$ is either $y$-full or $y$-mixed.
\end{itemize}

\begin{lemma}
  \label{lem:skeleton-perserverance-aux}
  Let $y \in V(T^b)$ and $v_0v_1v_2v_3$ be a~vertical path in $T$ satisfying either the $y$-empty property or the $y$-full property.
  Suppose $\Tc_\swap$ is created from $\Tc$ by performing a~swap along $v_0v_1v_2v_3$.
  Then the $y$-mixed skeletons of $\Tc$ and $\Tc_\swap$ are equal.
\end{lemma}
\begin{proof}
  In the proof, we assume the $y$-empty property; the proof for the $y$-full property is analogous (with the roles of the $y$-emptiness and the $y$-fullness of nodes exchanged).
  For the course of the proof, let $\Tc_\swap = (T_\swap, \lambda_\swap)$, let $T^\mix$ be a~$y$-mixed skeleton of $\Tc$, and let $T^\mix_\swap$ be a~$y$-mixed skeleton of $\Tc_\swap$.
  Let also $v'_1, v'_2$ be the children of $v_1$ and $v_2$, respectively, outside of the path $v_0v_1v_2v_3$ in $T$.
  
  Our proof crucially relies on the following helper claim:
  \begin{claim}
    \label{clm:skeleton-vtx-equality-implies-all-equality}
    Suppose that $V(T^\mix_\swap) = V(T^\mix)$ and $\{v_1, v_2\} \not\subseteq V(T^\mix)$.
    Then $T^\mix_\swap = T^\mix$.
  \end{claim}
  \begin{claimproof}
    By \cref{lem:mixed-skeleton-lca}, we find that $T^\mix_\swap = T^\mix$ if and only if $V(T^\mix_\swap) = V(T^\mix)$ and the ancestor-descendant relationship is preserved on the pairs of vertices of $V(T^\mix)$ (i.e., $u_1 \leq u_2$ holds in $T$ for some $u_1, u_2 \in V(T^\mix)$ if and only if $u_1 \leq u_2$ holds in $T_\swap$).

	So suppose there exist $u_1, u_2 \in V(T^\mix)$ such that the relation $u_1 \leq u_2$ holds in exactly one of the trees $T$, $T_\swap$.
	By the construction of $T_\swap$, one of these two vertices (say, $u_1$) either is equal to $v_1$ or is a~descendant of $v'_1$; and the other (say, $u_2$) either is equal to $v_2$ or is a~descendant of $v'_2$.
	The lowest common ancestor of $u_1$ and $u_2$ is then $v_1$ in $T$ and $v_2$ in $T_\swap$.
	By \cref{lem:mixed-skeleton-lca} and $V(T^\mix_\swap) = V(T^\mix)$, we have $\{v_1, v_2\} \subseteq V(T^\mix)$ -- a~contradiction.
  \end{claimproof}
  
  It is immediate that for every non-leaf node $w \notin \{v_1, v_2\}$, both subtrees rooted at the children of $w$ in $T$ contain the same set of nodes before and after the $x$-swap.
  Hence,
  \[ \{\lparts(T_\swap)[w']\,\mid\, w'\text{ is a child of }w\text{ in }T_\swap\} \ =\ 
    \{\lparts(T)[w']\,\mid\, w'\text{ is a child of }w\text{ in }T\}. \]
  Thus, each $w \notin \{v_1, v_2\}$ is a~$y$-branch point (resp.~a~$y$-leaf point) in $T_\swap$ if and only if $w$ is a~$y$-branch point (resp.~a~$y$-leaf point) in $T$.
  Moreover, it is easy to see that for each $w \notin \{v_1, v_2\}$, $w$ is $y$-empty (resp.~$y$-mixed, $y$-full) in $T$ if and only if $w$ is $y$-empty (resp.~$y$-mixed, $y$-full) in $T_\swap$. 
  
  Therefore, by \cref{clm:skeleton-vtx-equality-implies-all-equality}, for the equality of the $y$-mixed skeletons of $\Tc$ and $\Tc_\swap$ it is enough to prove that:
  \begin{itemize}
    \item for each $w \in \{v_1, v_2\}$, $w \in V(T^\mix_\swap)$ if and only if $w \in V(T^\mix)$; and
    \item $v_1, v_2$ do not both belong to the $y$-mixed skeleton of $\Tc$.
  \end{itemize}
  
  These conditions will follow immediately from the following series of claims.
  \begin{claim}
    \label{cl:skeleton-top-vtx-imply}
    Suppose $v_1 \in V(T^\mix)$. Then $v_1 \in V(T^\mix_\swap)$.
  \end{claim}
  \begin{claimproof}
    If $v'_1$ is $y$-empty in $\Tc$, then $v_2$ must be $y$-full in $\Tc$ (otherwise we would have $v_1 \notin V(T^\mix)$); but this contradicts the assumption that $v_3$ is $y$-empty or $y$-mixed.
    Therefore, it is $v'_2$ that is $y$-empty in $\Tc$.
    We now consider cases depending on the type of $v_3$ in $\Tc$:
    \begin{itemize}
      \item If $v_3$ is $y$-empty in $\Tc$, then it follows that $v_2$ is $y$-empty in $\Tc$.
        Since $v_1 \in V(T^\mix)$, we infer that $v'_1$ is $y$-full in $\Tc$ and $v_1$ is a~$y$-leaf point in $\Tc$.
        Then, in $\Tc_\swap$, the two children of $v_1$ (that is, $v'_1$ and $v_3$) are $y$-full and $y$-empty, respectively.
        Thus $v_1$ is also a~$y$-leaf point in $\Tc_\swap$.
      \item If $v_3$ is $y$-mixed in $\Tc$, then so is $v_2$.
        Since $v_1 \in V(T^\mix)$, it must be the case that $v'_1$ is also $y$-mixed in $\Tc$ and $v_1$ is a~$y$-branch point in $\Tc$.
        Hence in $\Tc_\swap$, both children of $v_1$ (again, $v'_1$ and $v_3$) are $y$-mixed, witnessing that $v_1$ is a~$y$-branch point also in $\Tc_\swap$. \qedhere
    \end{itemize}
  \end{claimproof}
  
  \begin{claim}
    \label{cl:skeleton-bottom-vtx-imply}
    Suppose $v_2 \in V(T^\mix)$. Then $v_2 \in V(T^\mix_\swap)$.
  \end{claim}
  \begin{claimproof}
    If $v'_2$ is $y$-empty in $\Tc$, then $v_3$ must be $y$-full in $\Tc$ (otherwise $v_2 \notin V(T^\mix)$) -- a~contradiction with the $y$-empty property of $v_0v_1v_2v_3$ in $\Tc$.
    So it is $v'_1$ that is $y$-empty in $\Tc$.
    Again, consider cases depending on the type of $v_3$ in $T$:
    \begin{itemize}
      \item If $v_3$ is $y$-empty in $\Tc$, then $v_2 \in V(T^\mix)$ implies that $v'_2$ is $y$-full in $\Tc$.
        Then, in $\Tc_\swap$, $v_1$ is $y$-empty (since both children $v'_1, v_3$ are $y$-empty) and so $v_2 \in V(T^\mix_\swap)$ (since one child $v'_2$ is $y$-full and the other child $v_1$ is $y$-empty).
      \item If $v_3$ is $y$-mixed in $\Tc$, then $v_2 \in V(T^\mix)$ implies that $v'_2$ is also $y$-mixed in $\Tc$.
        Hence, in $\Tc_\swap$, $v_1$ is $y$-mixed (since a~child $v_3$ is $y$-mixed), and so $v_2 \in V(T^\mix_\swap)$ (since both children $v'_2, v_1$ are $y$-mixed). \qedhere
    \end{itemize}
  \end{claimproof}

  \begin{claim}
    \label{cl:skeleton-reverse-vtx-imply}
    If $v_1 \in V(T^\mix_\swap)$, then $v_1 \in V(T^\mix)$. Similarly, if $v_2 \in V(T^\mix_\swap)$, then $v_2 \in V(T^\mix)$.
  \end{claim}
  \begin{claimproof}
    Observe that the vertical path $v_0v_2v_1v_3$ satisfies the $y$-empty property in $\Tc_\swap$; moreover, the swap of $\Tc_\swap$ along this path produces the original decomposition $\Tc$.
    Thus, by \cref{cl:skeleton-top-vtx-imply}, $v_2 \in V(T^\mix_\swap)$ implies that $v_2 \in V(T^\mix)$.
    Similarly, by \cref{cl:skeleton-bottom-vtx-imply}, $v_1 \in V(T^\mix_\swap)$ implies $v_1 \in V(T^\mix)$.
  \end{claimproof}
  
  \begin{claim}
    \label{cl:skeleton-no-both-vertices}
    It cannot happen that $v_1, v_2 \in V(T^\mix)$.
  \end{claim}
  \begin{claimproof}
    If $v_3$ is $y$-empty in $\Tc$, then $v'_2$ must be $y$-full (otherwise $v_2 \notin V(T^\mix)$), and so $v_2$ must be $y$-mixed.
    But then from the $y$-empty property of $v_0v_1v_2v_3$, the node $v'_1$ must be $y$-empty and thus $v_1 \notin V(T^\mix)$ -- a~contradiction.
    
    If $v_3$ is $y$-mixed in $\Tc$, then so is $v'_2$ (or else $v_2 \notin V(T^\mix)$), and $v_2$ is $y$-mixed, too.
    But then again, $v'_1$ must be $y$-empty from the $y$-empty property of $v_0v_1v_2v_3$, which contradicts that $v_1 \in V(T^\mix)$.
  \end{claimproof}
  
  \Cref{cl:skeleton-top-vtx-imply,cl:skeleton-bottom-vtx-imply,cl:skeleton-reverse-vtx-imply,cl:skeleton-no-both-vertices} conclude the proof of the lemma.
\end{proof}

We are now ready to give a proof of \cref{lem:swap-preserves-mixed-skeletons}.

\begin{proof}[Proof of \cref{lem:swap-preserves-mixed-skeletons}]
  We only show the proof in the case where $v'_1$ is $x$-empty and $v'_2$ is $x$-full in $\Tc$; the proof for the symmetric case is analogous.
  Recall that $v_3$ is $x$-mixed in $\Tc$.
  We consider three cases, depending on how $x$ and $y$ are related with respect to the ancestor-descendant relationship in~$T^b$.
  
  \smallskip
  
  \emph{Case 1: $y \geq x$ (i.e., $y$ is an~ancestor of $x$ in $T^b$).}
  Then by \cref{obs:mixedness-relations}, we have that $v'_2$ is $y$-full in $\Tc$; and $v_3$ is $y$-mixed or $y$-full.
  So $v_0v_1v_2v_3$ satisfies the $y$-full property, hence \cref{lem:skeleton-perserverance-aux} applies.
  
  \smallskip
  
  \emph{Case 2: $y \leq x$ (i.e., $y$ is a~descendant of $x$ in $T^b$).}
  Then by \cref{obs:mixedness-relations}, we have that in $\Tc$, $v'_1$ is $y$-empty and $v_3$ is $y$-empty or $y$-mixed.
  Therefore, $v_0v_1v_2v_3$ satisfies the $y$-empty property and \cref{lem:skeleton-perserverance-aux} applies.
  
  \smallskip
  
  \emph{Case 3: $y$ is not in the ancestor-descendant relationship with $x$ in $T^b$.}
  Again by \cref{obs:mixedness-relations}, we have that in $\Tc$, $v'_2$ is $y$-empty and $v_3$ is $y$-empty or $y$-mixed.
  Hence we can apply \cref{lem:skeleton-perserverance-aux} as the path $v_0v_1v_2v_3$ satisfies the $y$-empty property.
\end{proof}

\subsection{Block shuffles}

While the operation of swaps is quite strong in the sense that any $x$-swap preserves the $y$-mixed skeleton for any $x, y \in V(T^b)$, this unfortunately is not the case for $y$-factors: it could happen that a~$y$-factor of $\Tc$ could cease to exist after performing an~$x$-swap.
We will resolve this issue by introducing a~more structured counterpart of a~swap: a~(boundary-preserving) block-shuffle.

Suppose that $T$ contains a~long vertical path $v_0v_1 \dots v_pv_{p+1}$, $p \geq 0$.
For each $i \in [p]$, let $v'_i$ be the (unique) child of $v_i$ not on the path.
Let also $x \in V(T^b)$ and consider the case that for each $i \in [p]$, the vertex $v'_i$ is either $x$-empty or $x$-full in $\Tc$; and that $v_{p+1}$ is $x$-mixed in $\Tc$.
(This is equivalently the case where the $x$-mixed skeleton of $\Tc$ contains a~vertex in the subtree rooted at $v_{p+1}$, but none of the vertices $v_1, \dots, v_p$ are vertices of this skeleton.)
Any such path will be called \emph{$x$-shuffleable} from now on.

Now we say that an~integer interval $I = [\ell, r] \subseteq [1, p]$ is an~\emph{$x$-empty block} if all the vertices $v'_i$ for $i \in I$ are $x$-empty, and the interval cannot be extended from either side so as to preserve this property.
We similarly define \emph{$x$-full blocks}.
Then an~\emph{$x$-block} is either an~$x$-empty block or an~$x$-full block.
Naturally, $x$-blocks form a~partitioning of $[1, p]$ into intervals, and in this partitioning, $x$-empty blocks and $x$-full blocks alternate.
In the following description, we will sometimes identify $x$-blocks $[\ell, r]$ with the sequences of vertices $(v'_\ell, \dots, v'_r)$ and $(v_\ell, \dots, v_r)$.

For a~permutation $\sigma$ of $\{1, 2, \dots, p\}$, we say that the replacement of the vertical path $v_0v_1 \dots v_pv_{p+1}$ with the path $v_0v_{\sigma(1)}v_{\sigma(2)} \dots v_{\sigma(p)}v_{p+1}$ is an~\emph{$x$-block shuffle along $v_0\dots v_{p+1}$ using $\sigma$} if all the following conditions hold:
\begin{itemize}
  \item If $i$ and $i + 1$ belong to the same $x$-block, then $\sigma^{-1}(i + 1) = \sigma^{-1}(i) + 1$ (i.e., the value $i+1$ appears in the permutation immediately after $i$); and
  \item If $1 \leq i < j \leq p$ and both $v'_i, v'_j$ are $x$-empty (or both are $x$-full), then $\sigma^{-1}(i) < \sigma^{-1}(j)$ (i.e., the value $j$ appears in the permutation later than $i$).
\end{itemize}

For convenience, we say that the permutation $\sigma$ is the \emph{recipe} of the block shuffle.

Intuitively, an~$x$-block shuffle can be pictured as an~arbitrary shuffle of vertices along the vertical path that preserves the $x$-blocks of vertices along the path and never swaps two $x$-blocks of the same kind.
For our convenience, we extend $\sigma$ to be a~permutation of $\{0, \dots, p+1\}$ by setting $\sigma(0) = 0$ and $\sigma(p+1) = p+1$.
If additionally it holds that $\sigma(1) = 1$ and $\sigma(p) = p$, then we say that an~$x$-block shuffle is \emph{boundary-preserving}; equivalently, the first and the last $x$-blocks are preserved intact by the shuffle (\cref{fig:shuffleable-path}).

\begin{figure}
  \centering
  \begin{subfigure}[b]{0.9\textwidth}
    \centering
    \includegraphics[width=\textwidth]{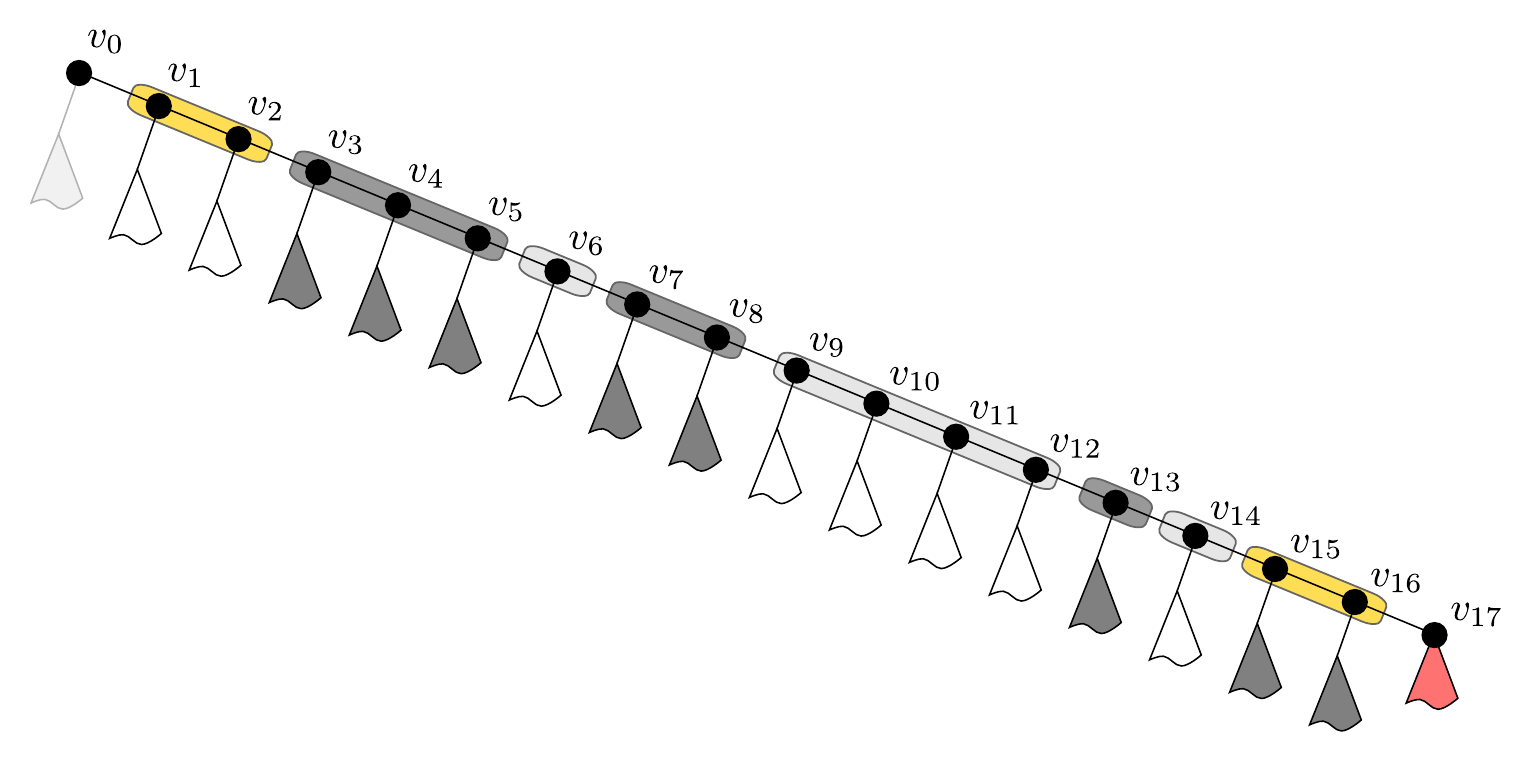}
    \caption{}
  \end{subfigure}
  
  \begin{subfigure}[b]{0.9\textwidth}
    \centering
    \includegraphics[width=\textwidth]{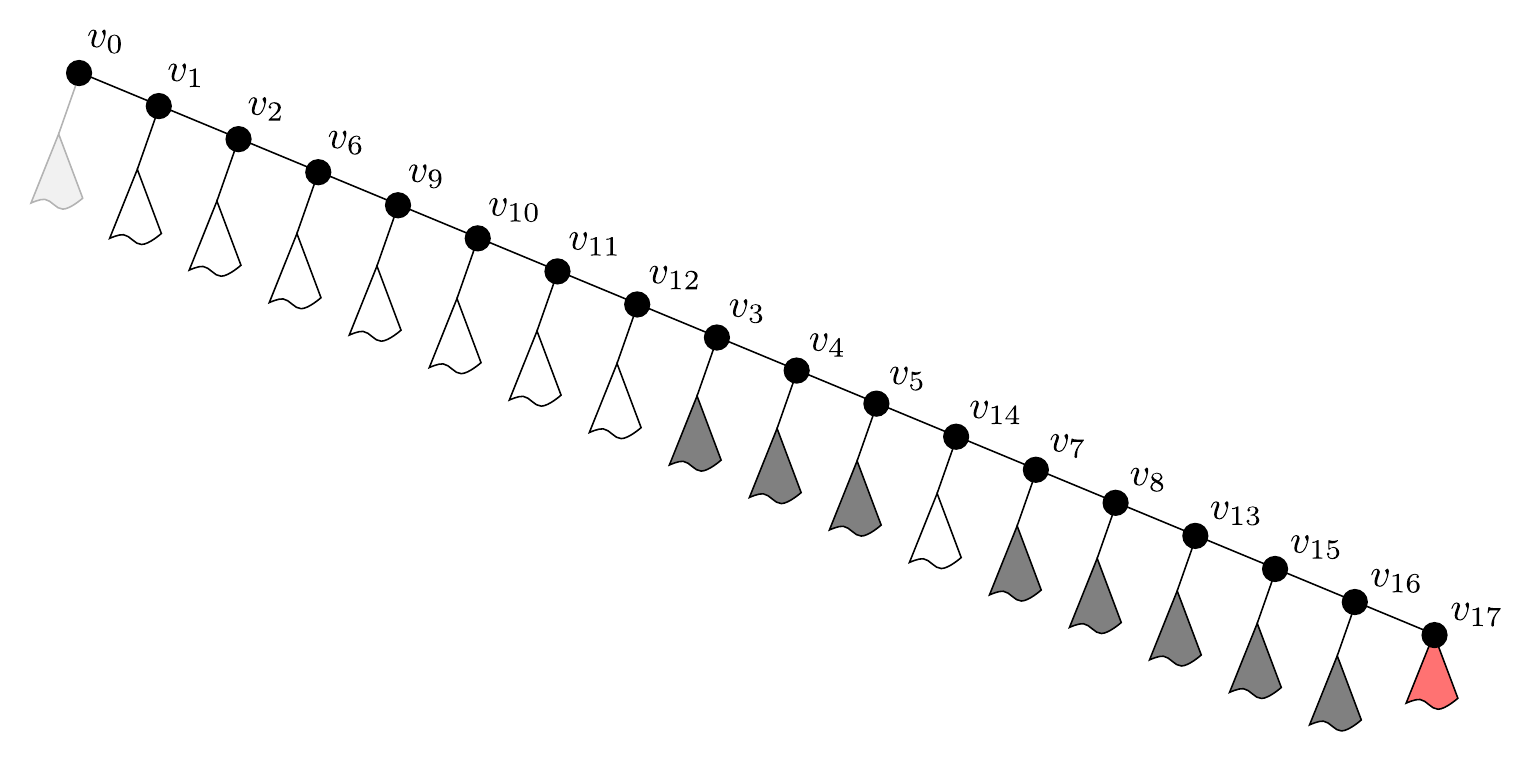}
    \caption{}
  \end{subfigure}
  
  \caption{(a) A sample $x$-shuffleable path. White subtrees have $x$-empty roots and dark subtrees have $x$-full roots; the root $v_{17}$ of the red subtree is $x$-mixed.
  The blocks of the path are indicated by boxes; the two boundary blocks are colored yellow. \\
  (b) An~example boundary-preserving $x$-block shuffle of the path. The recipe of the block shuffle is $\sigma = (1, 2, 6, 9, 10, 11, 12, 3, 4, 5, 14, 7, 8, 13, 15, 16)$.}
  \label{fig:shuffleable-path}
\end{figure}


The following fact is straightforward.

\begin{lemma}
  An~$x$-block shuffle of a~rank decomposition is equivalent to a~composition of $x$-swaps.
  In other words, if $\Tc'$ is a~result of an~$x$-block shuffle along a~vertical path of $\Tc$, then $\Tc'$ can also be produced from $\Tc$ by applying a~sequence of $x$-swaps.
\end{lemma}

Together with \cref{lem:swap-preserves-mixed-skeletons}, this immediately implies the following:

\begin{lemma}
  Let $x, y \in V(T^b)$. Suppose $\Tc'$ is created from $\Tc$ by performing an~$x$-block shuffle along a~vertical path. Then the $y$-mixed skeletons of $\Tc$ and $\Tc'$ are equal.
\end{lemma}

However, the structure introduced to $x$-block shuffles atop the $x$-swaps now allows us to reason about the perseverance of $y$-factors in the modified rank decomposition:

\begin{lemma}
  \label{lem:factor-swap-perseverance}
  Let $x, y \in V(T^b)$ with $y \ngtr x$.
  Suppose $\Tc'$ is created from $\Tc$ by performing a~boundary-preserving $x$-block shuffle along $v_0 \dots v_{p+1}$.
  Then every $y$-factor of $\Tc$ is also a~$y$-factor of $\Tc'$.
\end{lemma}
\begin{proof}
  Assume that $\Tc' \neq \Tc$, i.e., the performed block shuffle was non-trivial.
  Then $v_0v_1 \dots v_{p+1}$ comprises at least four $x$-blocks; let $[\ell_1, r_1], [\ell_2, r_2], \dots, [\ell_t, r_t]$ be the partitioning of $[1, p]$ into $x$-blocks, with $1 = \ell_1 \leq r_1 < \ell_2 \leq r_2 < \dots < \ell_t \leq r_t = p$ and $\ell_{i+1} = r_i + 1$ for all $i \in [p-1]$.
  Let $\sigma$ be the recipe of the block shuffle.
  Since the block shuffle is boundary-preserving, we have $\sigma(i) = i$ for $i \leq r_1$ and $i \geq \ell_t$.
  Note that by the construction, $\lparts(\Tc)[v] = \lparts(\Tc')[v]$ for every $v \in V(\Tc) \setminus \{v_{\ell_2}, v_{\ell_2 + 1}, \dots, v_{r_{t-1}}\}$.
  Moreover, $\lparts(\Tc)[v_{\ell_2}] = \lparts(\Tc')[v_{\sigma(\ell_2)}]$.
  
  Let $F \subseteq \Vc_y$ be a~$y$-factor of $\Tc$.
  The following claim captures the essential property of $y$-factors for $y \ngtr x$ that will be used in the current proof.
  \begin{claim}
    \label{cl:factor-subset-or-disjoint}
    $F \subseteq \Vc_x$ or $F$ is disjoint from $\Vc_x$.
  \end{claim}
  \begin{claimproof}
    If $y \leq x$, then $\Vc_y \subseteq \Vc_x$ and thus $F \subseteq \Vc_x$.
    On the other hand, if $y$ is incomparable with $x$ with respect to the ancestor-descendant relationship in $T^b$, then $\Vc_y$ is disjoint from $\Vc_x$, so also $F$ is disjoint from $\Vc_x$.
  \end{claimproof}
  
  First suppose that $F$ is a~$y$-tree factor, i.e., $F = \lparts(\Tc)[w]$ for some $w \in V(T)$.
  Note that if $w$ is an~ancestor of $v_{r_{t-1}}$, then $w$ is also an ancestor of both $v'_{r_{t-1}}$ and $v'_{\ell_t}$.
  But exactly one of the vertices $v'_{r_{t-1}}, v'_{\ell_t}$ is $x$-empty and the other is $x$-full.
  In other words, we have $\lparts(\Tc)[v'_{r_{t-1}}] \cup \lparts(\Tc)[v'_{\ell_t}] \subseteq F$, but exactly one of the sets $\lparts(\Tc)[v'_{r_{t-1}}]$, $\lparts(\Tc)[v'_{\ell_t}]$ is a~subset of $\Vc_x$ and the other is disjoint from $\Vc_x$.
  This, however, contradicts \cref{cl:factor-subset-or-disjoint}.
  Hence, $w$ is not an~ancestor of $v_{r_{t-1}}$.
  But then $w \notin \{v_{\ell_2}, v_{\ell_2 + 1}, \dots, v_{r_{t-1}}\}$, so $\lparts(\Tc)[w] = \lparts(\Tc')[w]$ and thus $F$ is also a~$y$-factor of $\Tc'$.
  
  Now consider the case where $F$ is a~$y$-context factor in $\Tc$, that is, $F = \lparts(\Tc)[w_1] \setminus \lparts(\Tc)[w_2]$ and $w_1$ is a~strict ancestor of $w_2$ in $T$.
  \begin{claim}
    \label{cl:subtract-close-from-add}
    It cannot happen that, for some $i \in [t-1]$, $w_1$ is an~ancestor of $v_{r_i}$ and $w_2$ is not an~ancestor of $v_{\ell_{i+1}} = v_{r_i + 1}$.
  \end{claim}
  \begin{claimproof}
    Proof by contradiction.
	First suppose that $w_2$ is not in the ancestor-descendant relationship with $v_{r_i + 2}$ in $T$.
	Since $r_i < \ell_{i+1} \le p$, we get that $\lparts(\Tc)[v_{p+1}]$ is disjoint from $\lparts(\Tc)[w_2]$ and thus $\lparts(\Tc)[v_{p+1}] \subseteq F$.
	But $v_{p+1}$ is $x$-mixed in $\Tc$, so $\lparts(\Tc)[v_{p+1}]$ is neither a subset of $\Vc_x$ nor disjoint from $\Vc_x$.
	Hence contradiction with \cref{cl:factor-subset-or-disjoint}.
	
	Since $w_2$ is not an~ancestor of $v_{r_i + 1}$, it means that $w_2$ is a~descendant of $v_{r_i + 2}$ and so $\lparts(\Tc)[v'_{r_i}] \cup \lparts(\Tc)[v'_{r_i + 1}] \subseteq F$.
%
%
    	However exactly one of $\lparts(\Tc)[v'_{r_i}]$ and $\lparts(\Tc)[v'_{r_i+1}] = \lparts(\Tc)[v'_{\ell_{i+1}}]$ is $x$-empty in $T$ and the other is $x$-full in $T$.
    	So again $\lparts(\Tc)[v'_{r_i}] \cup \lparts(\Tc)[v'_{r_i + 1}]$ is neither a~subset of $\Vc_x$ nor disjoint from $\Vc_x$ -- a~contradiction.
  \end{claimproof}
  
  If $w_1, w_2 \notin \{v_{\ell_2}, v_{\ell_2 + 1}, \dots, v_{r_t - 1}\}$, then $\lparts(\Tc)[w_1] = \lparts(\Tc')[w_1]$ and $\lparts(\Tc)[w_2] = \lparts(\Tc')[w_2]$, so $F$ is also a~$y$-context factor in $\Tc'$.
  Now suppose that at least one of $w_1, w_2$ is in $\{v_{\ell_2}, v_{\ell_2 + 1}, \dots, v_{r_t - 1}\}$.
  Since $w_1$ is a~(strict) ancestor of $w_2$, we must have that $w_1$ is also an~ancestor of $v_{r_{t-1}}$ and $w_2$ is a~descendant of $v_{\ell_2}$.
  Let then $j \in [t-1]$ be the smallest positive integer such that $w_1$ is an~ancestor of $v_{r_j}$.
  So by \cref{cl:subtract-close-from-add}, $w_2$ is an~ancestor of $v_{\ell_{j+1}}$.
  If $j = 1$, then $w_2 = v_{\ell_2}$ and $w_1 \notin \{v_{\ell_2}, v_{\ell_2 + 1}, \dots, v_{r_t - 1}\}$.
  Hence $F = \lparts(\Tc)[w_1] \setminus \lparts(\Tc)[v_{\ell_2}] = \lparts(\Tc')[w_1] \setminus \lparts(\Tc')[v_{\sigma(\ell_2)}]$ and $F$ is a~$y$-context factor in $\Tc'$.
  On the other hand, assume $j \geq 2$.
  In this case, $w_2$ is an~ancestor of $v_{\ell_{j+1}}$ and $w_1$ is an~ancestor of $w_2$, but a~descendant of $v_{\ell_j}$ (by the definition of $j$). Let $i_1, i_2$ (with $\ell_j \leq i_1 < i_2 \leq \ell_{j+1}$) be such that $w_1 = v_{i_1}$ and $w_2 = v_{i_2}$.
  Then, $F = \bigcup_{i = i_1}^{i_2 - 1} \lparts(\Tc)[v'_i]$.
  Since $[i_1, i_2 - 1]$ is a~part of an~$x$-block of the path $v_0 v_1 \dots v_{p+1}$, there exists some $q \in \N$ such that $\sigma(q + i) = i_1 + i$ for all $i \in [0, i_2 - i_1 - 1]$.
  We conclude that $F = \bigcup_{i=0}^{i_2 - i_1 - 1} \lparts(\Tc')[v'_{\sigma(q + i)}] = \lparts(\Tc')[v_{\sigma(q)}] \setminus \lparts(\Tc')[v_{\sigma(q + i_2 - i_1)}]$.
  Hence also in this case, $F$ is a~$y$-context factor of $T'$.
  As all cases have been exhausted, this finishes the proof.
\end{proof}

Observe that an~$x$-block shuffle will never increase the number of $x$-blocks along the shuffled path; on the other hand, the number of such $x$-blocks might decrease significantly if many $x$-blocks of the same kind are placed one after another.
We will now prove that it is indeed possible to perform such a~shuffle so as to decrease the number of $x$-blocks to a~constant (depending only on the width of $\Tc^b$) \emph{without} increasing the width of $\Tc$:

\begin{lemma}
  \label{lem:block-shuffle-few-blocks}
  There exists a~function $f_{\ref{lem:block-shuffle-few-blocks}}\,\colon\,\N\to\N$ such that the following holds.
  Assume that the width of $\Tc$ and $\Tc^b$ is bounded by $\ell \geq 0$ and let $x \in V(T^b)$.
  Suppose $v_0v_1\dots v_{p+1}$ is an~$x$-shuffleable path in $\Tc$.
  Then there exists a~boundary-preserving $x$-block shuffle of the path using a~permutation $\sigma$ such that:
  \begin{itemize}
    \item the decomposition $\Tc'$ after the shuffle has width not greater than the width of $\Tc$; and
    \item in $\Tc'$, the vertical path $v_{\sigma(1)}\dots v_{\sigma(p)}$ contains at most $f_{\ref{lem:block-shuffle-few-blocks}}(\ell)$ $x$-blocks.
  \end{itemize}
\end{lemma}

In the remaining part of this section we will cover the proof of \cref{lem:block-shuffle-few-blocks}.
We will call an~$x$-shuffleable vertical path $v_0v_1\dots v_{p+1}$:
\begin{itemize}
  \item \emph{$x$-static} if all of the following subspace equalities hold:
    \[
    \begin{split}
    \sumof{\lparts_x(T)[\vec{v_1v_0}]} \cap B_x &= \sumof{\lparts_x(T)[\vec{v_{p+1}v_p}]} \cap B_x, \\
    \sumof{\lparts_{\bar{x}}(T)[\vec{v_1v_0}]} \cap B_x &= \sumof{\lparts_{\bar{x}}(T)[\vec{v_{p+1}v_p}]} \cap B_x, \\
    \sumof{\lparts_x(T)[\vec{v_0v_1}]} \cap B_x &= \sumof{\lparts_x(T)[\vec{v_p v_{p+1}}]} \cap B_x, \\
    \sumof{\lparts_{\bar{x}}(T)[\vec{v_0v_1}]} \cap B_x &= \sumof{\lparts_{\bar{x}}(T)[\vec{v_p v_{p+1}}]} \cap B_x;
    \end{split}
  \]
  \item \emph{$x$-separable} if there exist integers $c_0, c_1, \dots, c_p \in \Z$ such that the following holds.
Suppose $\Tc'$ is formed from $\Tc$ by performing a~boundary-preserving $x$-block shuffle along $v_0v_1\dots v_{p+1}$ using $\sigma$.
Then, for every $i \in [0, p]$, the width of the edge $v_{\sigma(i)} v_{\sigma(i + 1)}$ in $\Tc'$ is equal to
$c_{\sigma(0)} + c_{\sigma(1)} + \ldots + c_{\sigma(i)}$.
\end{itemize}

The following lemma relates these notions:

\begin{lemma}
  \label{lem:static-implies-separable}
  Every $x$-static path is $x$-separable.
%
\end{lemma}
\begin{proof}
  Let $v_0v_1\dots v_{p+1}$ be an~$x$-static path, and for $i \in [p]$, let $v'_i$ be the unique child of $v_i$ outside of the path.
  Let us partition the sequence of nodes $v'_1, v'_2, \dots, v'_p$ into those that are $x$-full and those that are $x$-empty.
  Formally, let $q$ be the number of $x$-full nodes among $\{v'_1, \dots, v'_p\}$ and let $1 \leq a^+_1 < a^+_2 < \dots < a^+_q \leq p$ denote the sequence of indices of $x$-full nodes $v'_{a^+_1}, \dots, v'_{a^+_q}$.
  Similarly define $r$ as the number of $x$-empty nodes among $\{v'_1, \dots, v'_r\}$ and let $1 \leq a^-_1 < a^-_2 < \dots < a^-_r \leq p$ denote the complementary sequence of indices of $x$-empty nodes $v'_{a^-_1}, \dots, v'_{a^-_r}$.
  
  Recall that $x$-block shuffles do not exchange the order of $x$-full nodes or the order of $x$-empty nodes; that is, in every decomposition formed by an~$x$-block shuffle, the order of the nodes $v_{a^+_1}, \dots, v_{a^+_q}$ along the shuffled path is preserved, and so is the order of the nodes $v_{a^-_1}, \dots, v_{a^-_r}$.
  Therefore, if we assume that a~rank decomposition $\Tc'$ is formed by performing an~$x$-block shuffle using a~permutation $\sigma$ on $\Tc$, then for any $i \in [0, p]$, the sets $\lparts(\Tc')[\vec{v_{\sigma(i)}v_{\sigma(i+1)}}], \lparts(\Tc')[\vec{v_{\sigma(i+1)}v_{\sigma(i)}}]$ of vector spaces on either side of the edge of the edge $v_{\sigma(i)} v_{\sigma(i+1)}$ only depend on:
  \begin{itemize}
    \item the number $i^+ \in [0, q]$ of $x$-full nodes in the prefix $v'_{\sigma(1)}, \dots, v'_{\sigma(i)}$; and
    \item the number $i^- = i - i^+ \in [0, r]$ of $x$-empty nodes in the prefix $v'_{\sigma(1)}, \dots, v'_{\sigma(i)}$.
  \end{itemize}
  Note that $\{v'_{\sigma(1)}, v'_{\sigma(2)}, \dots, v'_{\sigma(i)}\} = \{v'_{a^+_1}, \dots, v_{a^+_{i^+}}, v'_{a^-_1}, \dots, v_{a^-_{i^-}}\}$.
  Next, define the following vector spaces:
  \begin{alignat*}{3}
    X_L &= \sumof{\lparts_x(\Tc)[\vec{v_0 v_1}]}, \qquad && Y_L &&= \sumof{\lparts_{\bar{x}}(\Tc)[\vec{v_0 v_1}]}, \\
    X_R &= \sumof{\lparts_x(\Tc)[\vec{v_{p+1} v_p}]}, \qquad && Y_R &&= \sumof{\lparts_{\bar{x}}(\Tc)[\vec{v_{p+1} v_p}]}, \\
    X_i &= \sumof{\lparts_x(\Tc)[\vec{v'_{a^+_i} v_{a^+_i}}]} = \sumof{\lparts(\Tc)[\vec{v'_{a^+_i} v_{a^+_i}}]}, \qquad && Y_j &&= \sumof{\lparts_{\bar{x}}(\Tc)[\vec{v'_{a^-_j} v_{a^-_j}}]} = \sumof{\lparts(\Tc)[\vec{v'_{a^-_j} v_{a^-_j}}]}, \\
    \hspace{-0.5em}\text{where }&\text{$i \in [q]$ and $j \in [r]$, and}\\[0.3em]
    X_{\leq i} &= X_L + X_1 + \ldots + X_i, \qquad && Y_{\leq j} &&= Y_L + Y_1 + \ldots + Y_j, \\
    X_{> i} &= X_{i+1} + \ldots + X_q + X_R, \qquad && Y_{> j} &&= Y_{j+1} + \ldots + Y_r + Y_R,
  \end{alignat*}
  where $i \in [0, q]$ and $j \in [0, r]$.
  Then
  \[
    \begin{split}
    \sumof{\lparts(T')[\vec{v_{\sigma(i)} v_{\sigma(i+1)}}]} &= X_{\leq i^+} + Y_{\leq i^-}, \\
    \sumof{\lparts(T')[\vec{v_{\sigma(i+1)}v_{\sigma(i)}}]} &= X_{> i^+} + Y_{> i^-}.
    \end{split}
  \]
  Moreover, the property of the path being $x$-static can be equivalently restated as follows:
  \begin{alignat*}{4}
    X_L \cap B_x &= X_{\leq q} \cap B_x, \qquad &&Y_L \cap B_x &&= Y_{\leq r} \cap B_x, \\
    X_R \cap B_x &= X_{> 0} \cap B_x, \qquad &&Y_R \cap B_x &&= Y_{> 0} \cap B_x.
  \end{alignat*}
  Note also that $B_x = \sumof{\Vc_x} \cap \sumof{\Vc \setminus \Vc_x} = (X_L + X_1 + \ldots + X_q + X_R) \cap (Y_L + Y_1 + \ldots + Y_r + Y_R)$.
  
  We are interested in the width of the edge $v_{\sigma(i)} v_{\sigma(i + 1)}$, that is, the dimension $d_i$ of the subspace $\lparts(\Tc')[\vec{v_{\sigma(i)} v_{\sigma(i+1)}}] \cap \lparts(\Tc')[\vec{v_{\sigma(i+1)}v_{\sigma(i)}}] = (X_{\leq i^+} + Y_{\leq i^-}) \cap (X_{> i^+} + Y_{> i^-})$.
  Applying \cref{lem:linear-dimension-switching} with $U_1 = X_{\leq i^+}$, $U_2 = Y_{\leq i^-}$, $V_1 = X_{> i^+}$, $V_2 = Y_{> i^-}$, we find that
  \begin{equation}
    \label{eq:shuffling-switching}
    \begin{split}
    \dim(&(X_{\leq i^+} + Y_{\leq i^-}) \cap (X_{> i^+} + Y_{> i^-})) + \dim(X_{\leq i^+} \cap Y_{\leq i^-}) + \dim(X_{>i^+} \cap Y_{>i^-}) = \\
    &= \dim((X_{\leq i^+} + X_{> i^+}) \cap (Y_{\leq i^-} + Y_{> i^-})) + \dim(X_{\leq i^+} \cap X_{> i^+}) + \dim(Y_{\leq i^-} \cap Y_{> i^-}).
    \end{split}
  \end{equation}
  Since $X_{\leq i^+} + X_{> i^+} = \sumof{\Vc_x}$ and $Y_{\leq i^-} + Y_{> i^-} = \sumof{\Vc \setminus \Vc_x}$, we have by definition
  \begin{equation}
    \label{eq:shuffling-e1}
    (X_{\leq i^+} + X_{> i^+}) \cap (Y_{\leq i^-} + Y_{> i^-}) = B_x.
  \end{equation}
  Now, $X_{\leq i^+} \subseteq \sumof{\Vc_x}$ and $Y_{\leq i^-} \subseteq \sumof{\Vc \setminus \Vc_x}$; since $\sumof{\Vc_x} \cap \sumof{\Vc \setminus \Vc_x} = B_x$, we see that $X_{\leq i^+} \cap Y_{\leq i^-} \subseteq B_x$.
  Therefore,
  \[ X_{\leq i^+} \cap Y_{\leq i^-} = (X_{\leq i^+} \cap B_x) \cap (Y_{\leq i^-} \cap B_x). \]
  But now, using the fact that the path $v_0v_1 \dots v_{p+1}$ is $x$-static, we have
  \[ X_L \cap B_x \,\subseteq\, X_{\leq i^+} \cap B_x \,\subseteq\, X_{\leq q} \cap B_x \,=\, X_L \cap B_x, \]
  so $X_{\leq i^+} \cap B_x = X_L \cap B_x$; similarly, we compute that $Y_{\leq i^+} \cap B_x = Y_L \cap B_x$.
  Hence,
  \begin{equation}
    \label{eq:shuffling-e2}
    X_{\leq i^+} \cap Y_{\leq i^-} = (X_L \cap B_x) \cap (Y_L \cap B_x) = X_L \cap Y_L \cap B_x = X_L \cap Y_L,
  \end{equation}
  since once again, $X_L \cap Y_L \subseteq B_x$.
  By an~analogous argument, we also deduce that
  \begin{equation}
    \label{eq:shuffling-e3}
    X_{> i^+} \cap Y_{> i^-} = X_R \cap Y_R.
  \end{equation}
  Plugging in \cref{eq:shuffling-e1,eq:shuffling-e2,eq:shuffling-e3} into \cref{eq:shuffling-switching}, we conclude that
  \[
    \begin{split}
    d_i & = \dim((X_{\leq i^+} + Y_{\leq i^-}) \cap (X_{> i^+} + Y_{> i^-})) = \\
      &= [\dim(B_x) - \dim(X_L \cap Y_L) - \dim(X_R \cap Y_R)] + \dim(X_{\leq i^+} \cap X_{> i^+}) + \dim(Y_{\leq i^-} \cap Y_{> i^-}).
    \end{split}
  \]
  That is, setting $\alpha = \dim(B_x) - \dim(X_L \cap Y_L) - \dim(X_R \cap Y_R)$ (a~constant independent on $i$ and $\sigma$), $\beta_{i^+} = \dim(X_{\leq i^+} \cap X_{> i^+})$ (a~constant dependent only on $i^+$, but not on $i$ or $\sigma$), and $\gamma_{i^-} = \dim(Y_{\leq i^-} \cap Y_{> i^-})$ (a~constant dependent only on $i^-$ and not on $i$ or $\sigma$), we have that
  \[ d_i = \alpha + \beta_{i^+} + \gamma_{i^-}. \]
  Now, set
  \[ \begin{split}
    c_0 &= \alpha + \beta_0 + \gamma_0, \\
    c_{a^+_i} &= \beta_i - \beta_{i-1} \qquad\quad\text{for }i \in [q], \\
    c_{a^-_j} &= \gamma_j - \gamma_{j-1} \qquad\quad\text{for }j \in [r].
  \end{split} \]
  It is now easy to verify that for every $i \in [0, p]$, the width of the edge $v_{\sigma(i)} v_{\sigma(i + 1)}$ in $\Tc'$ is
  \[ \begin{split}
    \dim(\lparts(\Tc')[\vec{v_{\sigma(i)} v_{\sigma(i+1)}}&] \cap \lparts(\Tc')[\vec{v_{\sigma(i+1)}v_{\sigma(i)}}]) = \\
      &= \dim((X_{\leq i^+} + Y_{\leq i^-}) \cap (X_{> i^+} + Y_{> i^-})) = \\
      &= \alpha + \beta_{i^+} + \gamma_{i^-} = \\
      &= c_0 + (c_{a^+_1} + c_{a^+_2} + \ldots + c_{a^+_{i^+}}) + (c_{a^-_1} + c_{a^-_2} + \ldots + c_{a^-_{i^-}}) = \\
      &= c_{\sigma(0)} + c_{\sigma(1)} + \ldots + c_{\sigma(i)},
    \end{split}
  \]
  since $\sigma(0) = 0$ and $\{\sigma(1), \ldots, \sigma(i)\} = \{a^+_1, \ldots, a^+_{i^+}, a^-_1, \ldots, a^-_{i^-}\}$.
\end{proof}

We now show that \cref{lem:block-shuffle-few-blocks} holds for $x$-separable paths (so, in turn, also for $x$-static paths).

\begin{lemma}
  \label{lem:block-shuffle-preserving-few-blocks}
  There exists a~function $f_{\ref{lem:block-shuffle-preserving-few-blocks}}\,\colon\,\N\to\N$ such that the following holds.
  Let $x \in V(T^b)$ and assume that the width of $\Tc$ and $\Tc^b$ is bounded by $\ell \geq 0$.
  Suppose $v_0v_1\dots v_{p+1}$ is an~$x$-separable path in $\Tc$.
  Then there exists a~boundary-preserving $x$-block shuffle of the path using a~permutation $\sigma$ such that:
  \begin{itemize}
    \item the decomposition $\Tc'$ after the shuffle has width not greater than the width of $\Tc$; and
    \item in $\Tc'$, the vertical path $v_{\sigma(1)}\dots v_{\sigma(p)}$ contains at most $f_{\ref{lem:block-shuffle-preserving-few-blocks}}(\ell)$ $x$-blocks.
  \end{itemize}
\end{lemma}
\begin{proof}
  The lemma is a~consequence of a~similar statement from the work of Bojańczyk and Pilipczuk \cite{DBLP:journals/lmcs/BojanczykP22}, formulated for bichromatic words, which in turn captures the understanding of \emph{typical sequences} from the work of Bodlaender and Kloks \cite{DBLP:journals/jal/BodlaenderK96}.
  Before we provide the statement of their lemma, we need to define block shuffles for words.
  We mostly follow the exposition from \cite{DBLP:journals/lmcs/BojanczykP22}, with the difference that their proof concerns words over alphabet $\{-, +\}$, excluding $0$ from the alphabet.
  However, it can be readily seen that their proof also works in the setting below.
  
  Fix the alphabet $\Sigma = \{0, -, +\}$.
  Given a~word $w \in \Sigma^*$, define:
  \begin{itemize}
    \item $\mathsf{sum}(w)$, the \emph{sum} of $w$, as the number of occurrences of $+$ in $w$, minus the number of occurrences of $-$ in $w$;
    \item $\mathsf{pmax}(w)$, the \emph{prefix maximum} of $w$, as the maximum sum of any prefix of $w$; and
    \item $\mathsf{pmin}(w)$, the \emph{prefix minimum} of $w$, as the minimum sum of any prefix of $w$.
  \end{itemize}
  Suppose the characters in a~word $w$ are colored with one of two colors, say red and blue; in such an~instance we say that $w$ is a~\emph{bichromatic word}.
  A~\emph{block} in such a~word is a~maximal subword comprising consecutive letters of $w$ of the same color.
  Then a~\emph{block shuffle} of $w$ is any word $w'$ created from $w$ by permuting the blocks of $w$ such that within each color, the order of the characters remains the same as in $w$.
  Then the Dealternation Lemma for bichromatic words reads as follows:
  \begin{claim}[{\cite[Lemma 7.1]{DBLP:journals/lmcs/BojanczykP22}}]
    \label{cl:word-dealternation}
    Let $w \in \Sigma^*$ be a~bichromatic word.
    Let $a, b \geq 0$ be two integers with the following properties: $\mathsf{pmax}(w) \leq a$, and if $u$ is a~word created from $w$ by restricting it to all letters of the same color, then $\mathsf{pmin}(u) \geq -b$.
    Then there exists a~block shuffle $w'$ of $w$ such that $\mathsf{pmax}(w') \leq \mathsf{pmax}(w)$ and $w'$ has at most $a + 4b + 2$ blocks in total.
  \end{claim}
  Let $f_{\ref{lem:block-shuffle-preserving-few-blocks}}(\ell) = 5\ell + 4$.
  Consider an~$x$-separable path $v_0v_1 \dots v_{p+1}$ and let $c_0, c_1, \dots, c_p \in \Z$ be the constants associated with the path.
  If the path comprises at most $2$ blocks, the lemma follows trivially -- we can choose $\Tc' = \Tc$ and $\sigma$ to be the identity permutation.
  Suppose now the path contains $t \geq 3$ blocks: $[\ell_1, r_1], [\ell_2, r_2], \dots, [\ell_t, r_t]$, where $1 = \ell_1 < r_1 < \ell_2 < r_2 < \dots < \ell_t < r_t = p$ and $\ell_{i+1} = r_i + 1$ for $i \in [t-1]$.
  Aiming to apply \cref{cl:word-dealternation}, we construct a~bichromatic word $w$ as follows:
  \begin{itemize}
    \item For every $i \in [\ell_2, r_{t-1}]$, define the word $w_i$ as follows:
    \[ w_i = \begin{cases}
       +^{c_i} & \text{if } c_i > 0; \\
       -^{|c_i|} & \text{if } c_i < 0; \\
       0 & \text{if } c_i = 0.
    \end{cases} \]
      Then we set $w = w_{\ell_2} w_{\ell_2 + 1} \dots w_{r_{t-1}}$.
    \item For every $i \in [\ell_2, r_{t-1}]$, color the letters of $w_i$ in $w$ red if $i \in [\ell_j, r_j]$ for even $j$, and blue otherwise. (That is, we color the subwords of $w$ corresponding to different $x$-blocks of the vertical path alternately; in other words, one color is allocated to the subwords corresponding to the $x$-empty nodes $v'_i$, and the other to the $x$-full nodes $v'_i$).
  \end{itemize}
  It is easy to see that block shuffles $w'$ of $w$ are in a~natural bijection with boundary-preserving $x$-block shuffles along $v_0v_1 \dots v_{p+1}$: Any reordering of the blocks in $w$ can be directly translated to a~reordering of the $x$-blocks of the path preserving the first and the last $x$-block, and vice versa.
  Such a~boundary preserving $x$-shuffle is said to be \emph{prescribed} by $w'$.
  
  The following claim about the prefix maximum of $w$ follows straight from the definition.
  \begin{claim}
    \label{cl:original-pmax}
    $\mathsf{pmax}(w) = \max_{i \in [\ell_2-1, r_{t-1}]} (c_{\ell_2} + c_{\ell_2 + 1} + \ldots + c_i)$.
  \end{claim}
  Note that the width of the edge $v_{\ell_2 - 1} v_{\ell_2}$ in the original decomposition $\Tc$ is (trivially) at least $0$; and for every $i \in [\ell_2-1, r_{t-1}]$, the width of the edge $v_i v_{i+1}$ is (by our assumption) at most $\ell$ (i.e., it exceeds the width of $v_{\ell_2 - 1} v_{\ell_2}$ by at most $\ell$).
  Thus, by the $x$-separability of the path $v_0v_1 \dots v_{p+1}$ for the trivial block shuffle using the identity permutation $\sigma$, we have $c_{\ell_2} + c_{\ell_2 + 1} + \ldots + c_i \leq \ell$ for every $i$ and hence $\mathsf{pmax}(w) \leq \ell$ by \cref{cl:original-pmax}.
  
  Now, suppose we found a~block shuffle $w'$ of $w$ with a~smaller or equal prefix maximum.
  Then the block shuffle can be naturally translated to an~$x$-block shuffle of $\Tc$ of width not greater than the width of $\Tc'$:
  \begin{claim}
    \label{cl:word-gives-tree-shuffling}
    Suppose $w'$ is a~block shuffle of $w$ with $\mathsf{pmax}(w') \leq \mathsf{pmax}(w)$, and let $\Tc'$ be the decomposition formed from $\Tc$ by performing a~boundary-preserving $x$-block shuffle prescribed by $w'$.
    Then the width of $\Tc'$ is at most the width of $\Tc$.
  \end{claim}
  \begin{claimproof}
    Let $\sigma$ be the recipe of the block shuffle prescribed by $w'$; note that for every $i \notin [\ell_2, r_{t-1}]$, it holds that $\sigma(i) = i$.
    By the same argument as in \cref{cl:original-pmax}, we have $\mathsf{pmax}(w') = \max_{i \in [\ell_2-1, r_{t-1}]} (c_{\sigma(\ell_2)} + c_{\sigma(\ell_2 + 1)} + \dots + c_{\sigma(i)})$.
    
    None of the edges outside of the path $v_{\ell_2 - 1}v_{\ell_2} \dots v_{r_{t-1} + 1}$ are affected by a~boundary-preserving $x$-block shuffle; that is, for any edge $e$ outside of this path, the $x$-block shuffle preserves the partitioning of the leaves of $T$ on either side of $e$. In particular, for every such edge $e$, the width of $e$ remains unchanged.
    Hence it is enough to verify the widths of each of the edges $v_{\ell_2 - 1}v_{\sigma(\ell_2)} = v_{\sigma(\ell_2 - 1)} v_{\sigma(\ell_2)},\, v_{\sigma(\ell_2)} v_{\sigma(\ell_2 + 1)},\, \ldots,\, v_{\sigma(r_{t-1})} v_{\sigma(r_{t-1}+1)} = v_{\sigma(r_{t-1})} v_{r_{t-1} + 1}$.
    
    Consider an~edge $e = v_{\sigma(i)} v_{\sigma(i+1)}$ for $i \in [\ell_2 - 1, r_{t-1}]$.
    By the $x$-separability of $v_0 v_1 \dots v_{p+1}$, the width of $e$ is
    \[ c_{\sigma(0)} + c_{\sigma(1)} + \ldots + c_{\sigma(i)} = \sum_{j=0}^{\ell_2 - 1} c_j + \sum_{j=\ell_2}^{i} c_{\sigma(i)} \leq \sum_{j=0}^{\ell_2 - 1} c_j + \mathsf{pmax}(w') \leq \sum_{j=0}^{\ell_2 - 1} c_j + \mathsf{pmax}(w). \]
    Let $i_{\max} \in [\ell_2 - 1, r_{t-1}]$ be such that $\mathsf{pmax}(w) = c_{\ell_2} + c_{\ell_2 + 1} + \ldots + c_{i_{\max}}$.
    Then
    \[ c_{\sigma(0)} + c_{\sigma(1)} + \ldots + c_{\sigma(i)} \leq c_0 + c_1 + \ldots + c_{i_{\max}}; \]
    that is, by the $x$-separability, the width of $e$ is upper-bounded by the width of the edge $v_{i_{\max}} v_{i_{\max} + 1}$ in the original decomposition $\Tc$, so in particular by the width of $\Tc$.
  \end{claimproof}
  
  It remains to bound the prefix minima of the restrictions of $w$ to all letters of a~single color.
  \begin{claim}
    \label{cl:restricted-pmin}
    Let $u$ be a~word created by restricting $w$ to all letters of the same color. Then $\mathsf{pmin}(u) \geq -\ell$.
  \end{claim}
  \begin{claimproof}
    Assume that $u = w_{i_1} w_{i_2} \dots w_{i_z}$ for $i_1 < i_2 < \ldots < i_z \in [\ell_2, r_{t-1}]$ such that the subwords $w_{i_1}, w_{i_2}, \ldots, w_{i_z}$ all have the same color in $w$; equivalently, $\{v'_{i_1}, v'_{i_2}, \ldots, v'_{i_z}\}$ is the subset of $\{v'_{\ell_2}, v'_{\ell_2 + 1}, \dots, v'_{r_{t-1}}\}$ comprising exactly the set of $x$-empty nodes or exactly the set of $x$-full nodes in $\Tc$.
    
    By the construction of $w$, we have
    \[ \mathsf{pmin}(u) = \min_{j \in [0, z]} (c_{i_1} + c_{i_2} + \ldots + c_{i_j}). \]
    
    Now construct:
    \begin{itemize}
      \item a~block shuffle $w'$ of $w$ by placing $u$ at the front of $w'$ and all the blocks of the opposite color at the back of $w'$, in the same order as in $w$;
      \item a~boundary-preserving $x$-block shuffle $\Tc'$ of $\Tc$ along $v_0 v_1 \dots v_{p+1}$ prescribed by $w'$; let also $\sigma$ be the recipe of this shuffle.
      (In other words, $\Tc'$ is constructed by placing all non-boundary blocks comprising $x$-empty (resp.~$x$-full) nodes next to each other.)
    \end{itemize}
    By construction, we have $\sigma(j) = j$ for all $j \in [0, \ell_2 - 1]$ and $\sigma(\ell_2 + j - 1) = i_j$ for all $j \in [z]$.
    
    Obviously, the width of the edge $v_{\ell_2 - 1} v_{\ell_2}$ in $\Tc$ is not larger than $\ell$; and by the $x$-separability of the vertical path $v_0 v_1 \dots v_{p+1}$ for the trivial block shuffle, it is equal to $c_0 + c_1 + \ldots + c_{\ell_2 - 1}$.
    On the other hand, for every $j \in [0, z]$, the width of the edge $v_{\sigma(\ell_2 + j - 1)} v_{\sigma(\ell_2 + j)}$ in $\Tc'$ is trivially at least $0$; and by the $x$-separability applied to the block shuffle along $\sigma$, it is equal to
    \[
      \sum_{q = 0}^{\ell_2 + j - 1} c_{\sigma(q)} = \sum_{q=0}^{\ell_2 - 1} c_q + \sum_{q=1}^j c_{i_q} \leq \ell + \sum_{q=1}^j c_{i_q}.
    \]
  Thus $c_{i_1} + c_{i_2} + \ldots + c_{i_j} \geq -\ell$ for every $j \in [0, z]$.
  Hence, $\mathsf{pmin}(u) \geq -\ell$.
  \end{claimproof}
  
  The proof of the lemma follows now in a~straightforward way: from \cref{cl:original-pmax,cl:restricted-pmin} it follows that $\mathsf{pmax}(w) \leq \ell$ and $\mathsf{pmin}(u) \geq -\ell$, where $u$ is the restriction of $w$ to the letters of any chosen color.
  Hence by \cref{cl:word-dealternation}, there exists a~block shuffle $w'$ of $w$ such that $\mathsf{pmax}(w') \leq \mathsf{pmax}(w)$ and $w'$ has at most $5\ell + 2$ blocks in total.
  Then by \cref{cl:word-gives-tree-shuffling}, the boundary-preserving $x$-block shuffle of $\Tc$ prescribed by $w'$ produces a~decomposition $\Tc'$ of width upper-bounded by the width of $\Tc$.
  Moreover, the path $v_{\sigma(0)} v_{\sigma(1)} \dots v_{\sigma(p+1)}$ in $\Tc'$ after the shuffle has at most $5\ell + 4$ $x$-blocks: the two boundary $x$-blocks and one additional $x$-block for each block of $w'$.
\end{proof}

It remains to lift the result of \cref{lem:block-shuffle-preserving-few-blocks} to general $x$-shuffleable paths:

\begin{proof}[Proof of \cref{lem:block-shuffle-few-blocks}]
  We will show that every $x$-shuffleable path can be partitioned into a~small number of $x$-static paths.
  Then the proof will follow from \cref{lem:static-implies-separable,lem:block-shuffle-preserving-few-blocks}.
  
  Recall that in our setting, $\Tc$ and $\Tc^b$ are rooted rank decompositions of $\Vc$ of width at most $\ell$ ($\ell \ge 0$) and $v_0 v_1 \dots v_{p+1}$ is an~$x$-shuffleable path in $\Tc$.
  
  \begin{claim}
    \label{cl:static-partitioning}
    There exists a~partitioning of the interval $[1, p]$ into $t \leq 8\ell + 1$ subintervals $[\ell_1, r_1],\allowbreak [\ell_2, r_2],\allowbreak \dots,\allowbreak [\ell_t, r_t]$ such that, for every $i \in [t]$, either $\ell_i = r_i$ or the vertical path $v_{\ell_i - 1} v_{\ell_i} v_{\ell_i + 1} \dots v_{r_i} v_{r_i + 1}$ is $x$-static.
    Moreover, any two subintervals with $\ell_i \neq r_i$ are separated by a~one-element subinterval.
  \end{claim}
  \begin{claimproof}
    For every $i \in [0, p]$, define the \emph{profile} of the edge $v_i v_{i+1}$ in $\Tc$ as the quadruple of integers $(\alpha_i, \beta_i, \gamma_i, \delta_i)$, where
    \[ \begin{split}
    \alpha_i &= \dim(\sumof{\lparts_x(T)[\vec{v_i v_{i+1}}]} \cap B_x), \\
    \beta_i &= \dim(\sumof{\lparts_{\bar{x}}(T)[\vec{v_i v_{i+1}}]} \cap B_x), \\
    \gamma_i &= \dim(\sumof{\lparts_x(T)[\vec{v_{i+1} v_i}]} \cap B_x), \\
    \delta_i &= \dim(\sumof{\lparts_{\bar{x}}(T)[\vec{v_{i+1} v_i}]} \cap B_x).
    \end{split} \]
    By construction, the sequences $(\alpha_i)_{i=0}^p$ and $(\beta_i)_{i=0}^p$ are non-decreasing, while the sequences $(\gamma_i)_{i=0}^p$ and $(\delta_i)_{i=0}^p$ are non-increasing; moreover, each of the values $\alpha_i, \beta_i, \gamma_i, \delta_i$ range from $0$ from $\ell$ since each value describes the dimension of a~subspace of $B_x$ (and $\dim(B_x) \leq \ell$ as $\Tc^b$ has width at most $\ell$).
    Therefore, there exist at most $4\ell + 1$ different profiles among all the edges $v_i v_{i+1}$.
    
    Say a~vertex $v_i$ ($i \in [p]$) is a~\emph{milestone} if the edges $v_{i-1} v_i$ and $v_i v_{i+1}$ have different profiles; observe that on the path $v_0 v_1 \dots v_{p+1}$, there are at most $4\ell$ milestones.
    We construct a~partitioning of $[1, p]$ into subintervals by:
    \begin{itemize}
      \item creating, for each milestone $v_i$, a~one-element subinterval $[i, i]$; and
      \item adding to the partitioning all maximal subintervals of $[1, p]$ not containing any milestones.
    \end{itemize}
    It is obvious that the partitioning contains at most $8\ell + 1$ subintervals.
    Now, let $[\ell, r]$ be some maximal subinterval of $[1, p]$ without any milestones.
    We claim that the path $v_{\ell-1} v_\ell \dots v_r v_{r+1}$ is $x$-static.
    Since none of the vertices $v_\ell, \dots, v_r$ are milestones, the profiles of the edges $v_{\ell-1} v_\ell$ and $v_r v_{r+1}$ are equal.
    Since $\alpha_{\ell-1} = \alpha_r$ and $\sumof{\lparts_x(T)[\vec{v_r v_{r+1}}]} \cap B_x \subseteq  \sumof{\lparts_x(T)[\vec{v_{\ell-1} v_\ell}]} \cap B_x$, we conclude that $\sumof{\lparts_x(T)[\vec{v_r v_{r+1}}]} \cap B_x = \sumof{\lparts_x(T)[\vec{v_{\ell-1} v_\ell}]} \cap B_x$; this verifies one of the equalities required by the definition of $x$-static paths.
    The remaining three equalities are proved analogously by analyzing the equalities $\beta_{\ell-1} = \beta_r$, $\gamma_{\ell-1} = \gamma_r$ and $\delta_{\ell-1} = \delta_r$.
  \end{claimproof}
  
  Let $[\ell_1, r_1], \dots, [\ell_t, r_t]$ be the partitioning of $[1, p]$ given by \cref{cl:static-partitioning}, and suppose that $\ell_1 \leq r_1 < \ell_2 \leq r_2 < \dots < \ell_t \leq r_t$.
  We inductively construct a~sequence of rank decompositions $\Tc_0 = \Tc, \Tc_1, \dots, \Tc_t$ with the following invariants:
  \begin{itemize}
    \item for every $1 \leq j \leq i \leq t$, vertices $v_{\ell_j}, \dots, v_{r_j}$ form -- in some order -- a~path in $\Tc_i$ with at most $f_{\ref{lem:block-shuffle-preserving-few-blocks}}(\ell)$ $x$-blocks; and
    \item for every $0 \leq i < j \leq t$ with $\ell_j \neq r_j$, the vertical path $v_{\ell_j - 1} v_{\ell_j} \dots v_{r_j} v_{r_j + 1}$ is $x$-static in $\Tc_i$.
  \end{itemize}
  We construct this sequence as follows: iterate the integers $i = 1, \dots, t$.
  If $\ell_i = r_i$, then set $\Tc_i = \Tc_{i-1}$.
  Otherwise, since the vertical path $v_{\ell_i - 1} v_{\ell_i} \dots v_{r_i} v_{r_i + 1}$ is $x$-static in $\Tc_{i-1}$, it is also $x$-separable (\cref{lem:static-implies-separable}); hence, we apply \cref{lem:block-shuffle-preserving-few-blocks} to produce a~boundary-preserving $x$-block shuffle $\Tc_i$ from $\Tc_{i-1}$, where the vertices $v_{\ell_i}, v_{\ell_i + 1}, \dots, v_{r_i}$ form a~vertical path with at most $f_{\ref{lem:block-shuffle-preserving-few-blocks}}(\ell)$ $x$-blocks.
  Since each boundary-preserving $x$-block only modifies the decomposition locally along the path $v_{\ell_i - 1}v_{\ell_i}\dots v_{r_i} v_{r_i+1}$, it can be easily verified that all the invariants are preserved by the update.
  Also, note that boundary-preserving $x$-block shuffles of $v_{\ell_i - 1} v_{\ell_i} \dots v_{r_i + 1}$ are also boundary-preserving $x$-block shuffles of $v_0 v_1 \dots v_{p+1}$ and a~composition of (boundary-preserving) $x$-block shuffles is also a~(boundary-preserving) $x$-block shuffle.
  We thus conclude that $\Tc_t$ is a~boundary-preserving $x$-block shuffle of $\Tc$ along the path $v_0 v_1 \dots v_{p+1}$.
  Moreover, by the invariants, for every $j \in [t]$, the vertices $v_{\ell_j}, \dots, v_{r_j}$ form a~vertical path in $\Tc_t$ with at most $f_{\ref{lem:block-shuffle-preserving-few-blocks}}(\ell)$ $x$-blocks.
  Since $t \leq 8\ell + 1$, we find that the produced decomposition contains at most $f_{\ref{lem:block-shuffle-few-blocks}}(\ell) \coloneqq (8\ell + 1) f_{\ref{lem:block-shuffle-preserving-few-blocks}}(\ell)$ $x$-blocks along the vertical path $v_{\sigma(1)} \dots v_{\sigma(p)}$.
  Also the width of $\Tc_t$ is upper-bounded by the width of $\Tc$.
\end{proof}

\subsection{Proof of the Local Dealternation Lemma}
\label{ssec:final-local-dealternation-lemma}

With all the required tools at hand, we can prove the Local Dealternation Lemma (\cref{lem:local-dealternation}).

\begin{proof}[Proof of \cref{lem:local-dealternation}]
  Let $T^\mix$ be the $x$-mixed skeleton of $\Tc$; by our assumption, it has at most $f_{\ref{lem:small-mixed-skeleton}}(\ell)$ nodes.
  If the skeleton is an~empty tree, then by \cref{lem:mixed-skeleton-mixed-descendant}, the root $r$ of $\Tc$ either $x$-full or $x$-empty.
  In either case, the lemma follows by setting $\Tc' = \Tc$, in which case the set $\lparts(\Tc')[x]$ is a~disjoint union of at most one $x$-factor of $\Tc'$.
  From now on assume that the skeleton is non-empty.

  The following observations are straightforward:
  \begin{claim}
    \label{cl:edge-path-shuffleable}
    Let $uv \in E(T^\mix)$, where $u$ is a~parent of $v$ in $T^\mix$ (i.e., $u$ is an~ancestor of $v$ in $T$).
    Then the simple vertical path between $u$ and $v$ in $T$ is $x$-shuffleable.
  \end{claim}
  \begin{claimproof}
    Let $w$ be an~internal node of the path (so $w \notin V(T^\mix)$), $w^+$ be the child of $w$ on the path and $w'$ be the child of $w$ not on the path.
    By \cref{lem:mixed-skeleton-mixed-path}, the node $w^+$ is $x$-mixed.
    Hence it cannot be that $w'$ is $x$-mixed -- otherwise, by definition, $w$ would be an~$x$-branch point.
  \end{claimproof}
  
  For the following observation, let $r$ be the root of $T$ and $r^\mix$ be the root of $T^\mix$.
  \begin{claim}
    If $r \neq r^\mix$, then the simple vertical path between $r$ and $r^\mix$ in $T$ is $x$-shuffleable.
  \end{claim}
  \begin{claimproof}
    Let $w$, $w^+$ and $w'$ be defined as in \cref{cl:edge-path-shuffleable}.
    Since $w^+$ is an~ancestor of $r^\mix$, the node $w^+$ is $x$-mixed.
    As in the previous claim, we conclude that $w'$ cannot be $x$-mixed.
  \end{claimproof}
  
  We now create a~new rank decomposition $\Tc'$ of $\Vc$ as follows: for every $uv \in E(T^\mix)$ in arbitrary order, perform on $\Tc$ a~boundary-preserving $x$-block shuffle of the vertical path between $u$ and $v$ in $T$ compliant with the statement of \cref{lem:block-shuffle-few-blocks}.
  Also, when $r \neq r^\mix$, apply an~analogous boundary-preserving $x$-block shuffle of the vertical path between $r$ and $r^\mix$ in $\Tc$.
  Let then $\Tc'$ be the decomposition after applying all the $x$-block shuffles.
  We will now verify that $\Tc'$ satisfies all the requirements of the lemma.
  
  First, by \cref{lem:swap-preserves-mixed-skeletons} and the fact that each $x$-block shuffle is a~composition of $x$-swaps, it follows that, for every $y \in V(T^b)$, the $y$-mixed skeletons of $\Tc$ and $\Tc'$ are equal; in particular, $T^\mix$ is the $x$-mixed skeleton of $\Tc'$.
  Next, assume that $y \in V(T^b)$ with $y \ngtr x$. That every $y$-factor of $\Tc$ is also a~$y$-factor of $\Tc'$ follows immediately from \cref{lem:factor-swap-perseverance}.
  It remains to show that the set $\lparts(\Tc^b)[x]$ can be decomposed into $f_{\ref{lem:local-dealternation}}(\ell)$ $x$-factors of $\Tc'$, for some function $f_{\ref{lem:local-dealternation}}$ yet to be defined.
  To this end, we will use the following simple claim:
  \begin{claim}
    \label{cl:shuffling-blocks-into-factors}
    Let $v_0v_1 \dots v_{p+1}$ be an~$x$-shuffleable path in $\Tc'$.
    Assume that the path $v_1 \dots v_p$ is comprised of $n \geq 1$ $x$-blocks.
    Then the set $\lparts_x(\Tc')[v_1] \setminus \lparts_x(\Tc')[v_{p+1}]$ can be decomposed into at most $n$ $x$-context factors of $\Tc'$.
  \end{claim}
  \begin{claimproof}
    For every $i \in [p]$, let $v'_i$ be the child of $v_i$ not on the path.
    Since $v_0v_1 \dots v_{p+1}$ is $x$-shuffleable, every node $v'_i$ is either $x$-empty or $x$-full.
    Moreover, each $x$-block $[\ell, r] \subseteq [1, p]$ is either an $x$-empty block (and then $\lparts(\Tc)[v_\ell] \setminus \lparts(\Tc)[v_{r+1}]$ is disjoint from $\Vc_x$) or an $x$-full block (and then $\lparts(\Tc)[v_\ell] \setminus \lparts(\Tc)[v_{r+1}] \subseteq \Vc_x$; so in particular, $\lparts(\Tc)[v_\ell] \setminus \lparts(\Tc)[v_{r+1}]$ is an~$x$-context factor of $\Tc'$).
    Therefore, $\lparts_x(\Tc)[v_1] \setminus \lparts_x(\Tc)[v_{p+1}] = (\lparts(\Tc)[v_1] \setminus \lparts(\Tc)[v_{p+1}]) \cap \Vc_x$ is a~disjoint union of $x$-context factors of the form $\lparts(\Tc)[v_\ell] \setminus \lparts(\Tc)[v_{r+1}]$, ranging over all $x$-full blocks $[\ell, r] \subseteq [1, p]$.
  \end{claimproof}
  
  Let $r$ be the root of $T'$, and $r^\mix$ be the root of $T^\mix$.
  Observe that every leaf $l$ of $T'$ can be uniquely assigned to one of the following groups:
  \begin{itemize}
    \item the group of leaves that are not descendants of $r^\mix$ (if $r \neq r^\mix$);
    \item for every $x$-leaf point $v \in V(T^\mix)$, the group of leaves that are descendants of $v$;
    \item for every edge $uv \in E(T^\mix)$, where $u$ is an~ancestor of $v$ in $T$, the group of leaves that are descendants of $w$ but not $v$, where $w$ is the child of $u$ on the path between $u$ and $v$ in $T$.
  \end{itemize}
  Thus, $\lparts(\Tc^b)[x] = \lparts_x(\Tc')[r]$ is the disjoint union of the following sets:
  \begin{itemize}
    \item $\lparts_x(\Tc')[r] \setminus \lparts_x(\Tc')[r^\mix]$ (if $r \neq r^\mix$);
    \item for every $x$-leaf point $v \in V(T^\mix)$, the set $\lparts_x(\Tc')[v]$;
    \item for every edge $uv \in E(T^\mix)$, where $u$ is an~ancestor of $v$ in $T$, the set $\lparts_x(\Tc')[w] \setminus \lparts_x(\Tc')[v]$, where $w$ is the child of $u$ on the path between $u$ and $v$ in $T$.
  \end{itemize}
  
  If $r \neq r^\mix$, then let $r_0r_1 \dots r_t$ be the vertical path between $r$ and $r^\mix$ ($r_0 = r$, $r_t = r^\mix$, $t \geq 1$).
  Since we performed a~boundary-preserving $x$-block shuffle along the vertical path $r_0 r_1 \dots r_t$, we get that the vertical path $r_1 \dots r_{t-1}$ comprises at most $f_{\ref{lem:block-shuffle-few-blocks}}(\ell)$ $x$-blocks; hence, by \cref{cl:shuffling-blocks-into-factors}, the set $\lparts_x(\Tc')[r_1] \setminus \lparts_x(\Tc')[r^\mix]$ can be partitioned into at most $f_{\ref{lem:block-shuffle-few-blocks}}(\ell)$ $x$-context factors of $\Tc'$.
  Let $r'$ be the child of $r$ not on the path from $r$ to $r^\mix$; then $r'$ is $x$-empty or $x$-full.
  If $r'$ is $x$-full, we add one additional tree factor $\lparts(\Tc')[r']$ to the partitioning.
  So $\lparts_x(\Tc')[r] \setminus \lparts_x(\Tc')[r^\mix] = \lparts_x(\Tc')[r'] \cup \left(\lparts_x(\Tc')[r_1] \setminus \lparts_x(\Tc')[r^\mix]\right)$ can be partitioned into at most $f_{\ref{lem:block-shuffle-few-blocks}}(\ell) + 1$ $x$-factors of $\Tc'$.
  
  Next, for every edge $uv \in V(T^\mix)$, where $u$ is an~ancestor of $v$ in $T$, let $w$ be the child of $u$ on the path from $u$ to $v$ in $T'$.
  Applying \cref{cl:shuffling-blocks-into-factors}, we get that the set $\lparts_x(\Tc')[w] \setminus \lparts_x(\Tc')[v]$ can be partitioned into $f_{\ref{lem:block-shuffle-few-blocks}}(\ell)$ $x$-context factors of $\Tc'$.
  Finally, for every $x$-leaf point $v$ in $T'$, one child $v^+$ of $v$ is $x$-full and the other child $v^-$ is $x$-empty. 
  So $\lparts_x(\Tc')[v] = \lparts(\Tc')[v^+]$ and the set $\lparts_x(\Tc')[v]$ is exactly an~$x$-tree factor of $\Tc'$.
  
  Summing up, we can partition the set $\lparts(\Tc^b)[x] = \lparts_x(\Tc')[r]$ into at most
  \[ f_{\ref{lem:block-shuffle-few-blocks}}(\ell) \cdot (|E(T^\mix)| + 1) + |V(T^\mix)| + 1 \leq
    (f_{\ref{lem:block-shuffle-few-blocks}}(\ell) + 1) f_{\ref{lem:small-mixed-skeleton}}(\ell) + 1 \]
  $x$-factors of $T'$.
  This finishes the proof of the Local Dealternation Lemma and it is enough to set $f_{\ref{lem:local-dealternation}}(\ell) = (f_{\ref{lem:block-shuffle-few-blocks}}(\ell) + 1) f_{\ref{lem:small-mixed-skeleton}}(\ell) + 1$.
\end{proof}

%% file: computing-closures.tex
\section{Using rank decomposition automata to compute closures}
\label{sec:computing-closures}
This section is dedicated to the proofs of \Cref{lem:closureprds,lem:near-linear-annotated-decomposition-recovery}.
Along the way, we produce two rank decomposition automata that will be used by us heavily throughout the proof:
\begin{itemize}
  \item the \emph{exact rankwidth automaton} (\cref{ssec:rankwidth-automaton}) that for two fixed integers $k, \ell$ verifies, given an~annotated rank decomposition of width $\ell$ encoding a~partitioned graph $(G, \prt)$, whether $(G, \prt)$ has rankwidth at most $k$; and
  \item the \emph{closure automaton} (\cref{ssec:closure-automaton}) that, roughly speaking, for an~annotated tree decomposition $\Tc$ encoding a~graph $G$ and a~prefix $\Tpref$ of $\Tc$, represents how a~$c$-small $k$-closure of $\Tpref$ can look like in each subtree rooted at an~edge $\opx \in \oApp_T(\Tpref)$.
\end{itemize}

The exact rankwidth automaton given in \cref{ssec:rankwidth-automaton} will imply \Cref{lem:near-linear-annotated-decomposition-recovery}.
Finally, in \cref{ssec:ds-closures}, we will use both automata to produce a~data structure for minimal closures of \cref{lem:closureprds}.


In this section, we rely on the concepts and notation defined in \cref{ssec:dealternation-prelims}, in particular the subspace arrangements of linear spaces and rank decompositions thereof.


\input{exact-rankwidth-automaton.tex}

\input{closure-automaton.tex}

\subsection{State optimization problem for rank decomposition automata}

In this subsection, we introduce an~optimization problem for rank decomposition automata that will be used in the proof of \cref{lem:closureprds}.
We will also show that this problem can be solved efficiently under the reasonable assumptions on the automaton.

Let $(S, +, \leq)$ be a~totally ordered commutative semigroup, i.e., a~commutative semigroup $(S, +)$ with a~total order $\leq$ with the property that, for any $x, y, z \in S$ with $x \leq y$, we have $x + z \leq y + z$.
Assume that $+$ can be evaluated in time $\beta$.

Let also $\autom = (Q,\iota,\delta,\varepsilon)$ be a~label-oblivious rank decomposition automaton of width $\ell$ with evaluation time $\beta$ and a~finite set of states.
Suppose $\Tc = (T, U, \reps, \repse, \dmap)$ is an~unrooted annotated rank decomposition of width at most $\ell$.
We will call any function $\kappa \,\colon\, \leafe(T) \to Q$ a~\emph{leaf edge state mapping}.
Given a~leaf edge state mapping $\kappa$ and an~edge $\vec{ab} \in \oE(T)$, we define the \emph{$\kappa$-run} of $(\autom, a, b)$ as the function $\rho_\kappa \,\colon\, \pred_T(\vec{ab}) \cup \pred_T(\vec{ba}) \cup \{\vartheta\} \to Q$ defined as follows:
\begin{itemize}
  \item for each leaf edge $\vec{lp} \in \leafe(T)$ it holds that $\rho_\kappa(\vec{lp}) = \kappa(\vec{lp})$;
  \item for each non-leaf edge $\vec{tp}$ of $T$ with children $\vec{c_1 t}, \vec{c_2t}$, where $c_1 < c_2$, it holds that $\rho_\kappa(\vec{tp}) = \delta(\tau(\Tc, \vec{tp}), \rho_\kappa(\vec{c_1t}), \rho_\kappa(\vec{c_2t}))$;
  \item $\rho_\kappa(\vartheta) = \varepsilon(\delta(\Tc, \vec{ab}), \rho_\kappa(\vec{ab}), \rho_\kappa(\vec{ba}))$.
\end{itemize}
So, in other words, a~$\kappa$-run of an~automaton is defined similarly to a~run of an~automaton, only that the initial mapping $\iota$ of the automaton is ignored, and instead we fix the state $\rho_\kappa(\vec{lp})$ of each leaf edge $\vec{lp}$ to $\kappa(\vec{lp})$.

Moreover, let ${\bf c} \,\colon\, \leafe(T) \times Q \to S$ be a~\emph{cost function}.
Then the cost of a~leaf edge state mapping $\kappa$ is defined as ${\bf c}(\kappa) \coloneqq \sum_{e \in \leafe(T)} {\bf c}(e, \kappa(e))$.

We now show that the optimization problem where, given a~set $F \subseteq Q$ of states, we are to find a~leaf edge state mapping $\kappa$ of minimum cost for which $\rho_\kappa(\vartheta) \in F$, can be solved efficiently.
The proof is a~standard application of the dynamic programming technique.
\ms{low prio: maybe we can cite something?}
\tk{\Cref{lem:automaton-dp} seems to be closely related to what happens in the proof of \Cref{lem:lincmsordaut} (in \Cref{subsec:appendixcmso}). At some point (maybe after the deadline), we could look into merging these lemmas a bit}

\begin{lemma}
  \label{lem:automaton-dp}
  Given:
  \begin{itemize}
    \item a~totally ordered commutative semigroup $(S, +, \leq)$ with evaluation time $\beta$,
    \item a~label-oblivious rank decomposition automaton $\autom = (Q, \iota, \delta, \varepsilon)$ of width $\ell$ with evaluation time $\beta$ and a~finite set $Q$ of states,
    \item an~annotated rank decomposition $\Tc = (T, U, \reps, \repse, \dmap)$ of width at most $\ell$ with $n$ nodes,
    \item an~edge $\vec{ab} \in \oE(T)$,
    \item a~cost function ${\bf c} \,\colon\, \leafe(T) \times Q \to S$, and
    \item a~set $F \subseteq Q$ of \emph{accepting states},
  \end{itemize}
  it is possible to determine in time $\Oh{|Q|^2 n \beta}$, whether there exists a~leaf edge state mapping $\kappa$ such that $\rho_\kappa(\vartheta) \in F$.
  If such a~mapping exists, then it is also possible to determine any such mapping minimizing the value of ${\bf c}(\kappa)$.
\end{lemma}
\begin{proof}
  Note that in the definition of $\rho_\kappa$ before, the value $\rho_\kappa(\vec{uv})$ for an~edge $\vec{uv} \in \oE(T)$ only depends on the values of $\kappa$ for $\vec{lp} \in \pred_T(\vec{uv})$.
  Therefore, without confusion we will write $\rho_{\kappa'}(\vec{uv})$ whenever $\kappa'$ is a~partial function defined on $\pred_T(\vec{uv}) \cap \leafe(T)$.  
  
  \newcommand{\bestdp}{\mathrm{best}}
    We want to compute, for every oriented edge $\vec{uv} \in \pred_T(\vec{ab}) \cup \pred_T(\vec{ba})$, the function $\bestdp_{\vec{uv}} \,\colon\, Q \to S \cup \{\bot\}$ with the following property for every $f \in Q$: suppose $K_{\vec{uv}, f}$ is the set of all partial valuations $\kappa' \,\colon\, \pred_T(\vec{uv}) \cap \leafe(T) \to Q$ such that $\rho_{\kappa'}(\vec{uv}) = f$.
    Then $\bestdp_{\vec{uv}}(f) = \bot$ if $K_{\vec{uv}, f} = \emptyset$; otherwise, $\bestdp_{\vec{uv}}(f)$ is equal to the minimum value of $\sum_{e \in \pred_T(\vec{uv}) \cap \leafe(T)} {\bf c}(e, \kappa'(e))$ over all $\kappa' \in K_{\vec{uv}, f}$.
    It is easy to observe that:
    \begin{itemize}
      \item for a~leaf edge $\vec{lp} \in \leafe(T)$, we have $\bestdp_{\vec{lp}}(q) = {\bf c}(\vec{lp}, f)$ for each $f \in Q$;
      \item for a~non-leaf edge $\vec{tp} \in \oE(T)$ with children $\vec{c_1t}$ and $\vec{c_2t}$ with $c_1 < c_2$, we have, for every $f \in Q$,
      \begin{equation}
        \label{eq:treedp}
        \begin{split}
        \bestdp_{\vec{tp}}(f) &= \min \{ \bestdp_{\vec{c_1t}}(f_1) + \bestdp_{\vec{c_2t}}(f_2) \,\mid \\ &\, f_1, f_2 \in Q,\, \bestdp_{\vec{c_1 t}}(f_1) \neq \bot,\, \bestdp_{\vec{c_2 t}}(f_2) \neq \bot,\, f = \delta(\tau(\Tc, \vec{tp}), f_1, f_2) \};
        \end{split}
      \end{equation}
      where we set $\bestdp_{\vec{tp}}(q) = \bot$ if the set on the right-hand side of \cref{eq:treedp} is empty.
      So given $\bestdp_{\vec{c_1t}}$ and $\bestdp_{\vec{c_2t}}$, we can compute $\bestdp_{\vec{tp}}$ in time $\Oh{|Q|^2\beta}$.
    \end{itemize}
    Therefore, all functions $\bestdp_{\vec{uv}}$ can be computed in time $\Oh{|Q|^2 n \beta}$ by a~simple bottom-up dynamic programming on trees with a~depth-first search on $T$.
    Similarly we define $\bestdp \,\colon\, Q \to S \cup \{\bot\}$ with the following property for all $f \in Q$: let $K_f$ be the set of valuations $\kappa \,\colon\, \leafe(T) \to Q$ such that $\rho_{\kappa}(\vartheta) = f$.
    Then $\bestdp(f) = \bot$ if $K_f = \emptyset$, and otherwise $\bestdp(f)$ is the minimum value of $\sum_{e \in \leafe(T)} {\bf c}(e, \kappa(e))$ over all $\kappa \in K_f$.
    As in \cref{eq:treedp}, we get that
    \begin{equation}
      \label{eq:endtreedp}
      \begin{split}
      \bestdp(f) = \min \{ &\bestdp_{\vec{ab}}(f_1) + \bestdp_{\vec{ba}}(f_2) \,\mid \\ &\, f_1, f_2 \in Q,\, \bestdp_{\vec{ab}}(f_1) \neq \bot,\, \bestdp_{\vec{ba}}(f_2) \neq \bot,\, f = \varepsilon(\delta(\Tc, \vec{ab}), f_1, f_2) \};
      \end{split}
    \end{equation}
    where $\bestdp(f) = \bot$ is set if the set on the right-hand side of \cref{eq:endtreedp} is empty.
    Then $\bestdp$ can be computed in time $\Oh{|Q|^2\beta}$ given $\bestdp_{\vec{ab}}$ and $\bestdp_{\vec{ba}}$.
    Now, if $\bestdp(f) = \bot$ for all $f \in F$, then we return that no mapping $\kappa$ with $\rho_\kappa(\vartheta) = q_0$ exists.
    Otherwise, such a~mapping exists.
    Let $f_0 \in F$ be the argument minimizing $\bestdp(f_0)$ among all $f \in F$ with $\bestdp(f) \neq \bot$.
    By retracing the optimum choices done by the dynamic programming scheme using the top-bottom depth-first search on $T$, we fully recover a~run $\rho_\kappa$ for some $\kappa \,\colon\, \leafe(T) \to Q$ such that $\rho_\kappa(\vartheta) = f_0$ and ${\bf c}(\kappa)$ is minimum possible; and we recover $\kappa$ by observing that for every $\vec{lp} \in \leafe(T)$, it holds that $\kappa(\vec{lp}) = \rho_\kappa(\vec{lp})$.
\end{proof}

Note that \cref{lem:automaton-dp} can be easily generalized to the case where $Q$ is an~infinite set, but there exists a~bound $q \in \N_{\ge 1}$ on the size of the set
\[ \{ \rho_\kappa(x) \,\mid\, \kappa\,\colon\,\leafe(T) \to Q \} \]
for all $x$.
Then it can be verified that the optimization problem stated above can be solved in time $\Oh{q^2 n \beta}$.

\subsection{Prefix-rebuilding data structure for minimal closures}
\label{ssec:ds-closures}

In this subsection, we finally give a~proof of \cref{lem:closureprds}.
Before we begin, we describe an~operation of \emph{gluing} rank decompositions; a~similar notion appears in the proof of \cref{lem:rearangmain}.

Suppose we have two disjoint sets of vertices $A, B$ and that $R_A \subseteq A$, $R_B \subseteq B$; we also have two partitioned graphs $(G_A, \prt_A)$, $(G_B, \prt_B)$ with vertex sets $A \cup R_B$ and $B \cup R_A$, respectively, such that $\{R_B\} \in \prt_A$ and $\{R_A\} \in \prt_B$.
Suppose also $\Tc_A = (T_A, U_A, \reps_A, \repse_A, \dmap_A)$ and $\Tc_B = (T_B, U_B, \reps_B, \repse_B, \dmap_B)$ are annotated rank decompositions encoding $(G_A, \prt_A)$ and $(G_B, \prt_B)$, with the following properties: $V(T_A) \cap V(T_B) = \{x, y\}$ and there exists a~leaf edge $\oxy \in \leafe(T_A)$ and a~leaf edge $\oyx \in \leafe(T_B)$ such that:
\begin{itemize}
  \item $\lparts(\Tc_A)[\oxy] = R_B$ and $\lparts(\Tc_B)[\oyx] = R_A$,
  \item $\reps_A(\oxy) = \reps_B(\oxy) = R_B$ and $\reps_B(\oyx) = \reps_A(\oyx) = R_A$, and
  \item $\repse_A(xy) = \repse_B(xy)$.
\end{itemize}
We then define the \emph{gluing of $\Tc_A$ along $xy$ with $\Tc_B$} as the annotated rank decomposition $\Tc = (T, U, \reps, \repse, \dmap)$ as follows:

\begin{itemize}
  \item $V(T) = V(T_A) \cup V(T_B)$ and $E(T) = E(T_A) \cup E(T_B)$,
  \item $U = U_A \cup U_B = A \cup B$,
  \item $\funrestriction{\reps}{\oE(T_A)} = \reps_A$ and $\funrestriction{\reps}{\oE(T_B)} = \reps_B$,
  \item $\funrestriction{\repse}{E(T_A)} = \repse_A$ and $\funrestriction{\repse}{E(T_B)} = \repse_B$, and
  \item $\funrestriction{\dmap}{\PT(T_A)} = \dmap_A$ and $\funrestriction{\dmap}{\PT(T_B)} = \dmap_B$.
\end{itemize}

It can be verified that $\Tc$ is an~annotated rank decomposition encoding a~partitioned graph $(G, \prt)$, where $V(G) = A \cup B$, $\prt = (\prt_A \setminus \{R_B\}) \cup (\prt_B \setminus \{R_A\})$, $G[A] = G_A$, $G[B] = G_B$ and $uv \in E(G)$ for $u \in A$, $v \in B$ if and only if $u'v' \in \repse(xy)$, where $u' \in R_A$ is the (unique) vertex such that $N_{G_A}(u) \cap R_B = N_{G_A}(u') \cap R_B$, and $v' \in R_B$ is the unique vertex such that $N_{G_B}(v) \cap R_A = N_{G_B}(v') \cap R_A$.
Moreover, the width of $\Tc$ is trivially the maximum of the widths of $\Tc_A$ and $\Tc_B$.

\newcommand{\prtpref}{\prt_{\mathrm{pref}}}
\newcommand{\skel}{\mathrm{skel}}
\newcommand{\tcskel}{\Tc_{\mathrm{skel}}}

We are now ready to prove \cref{lem:closureprds}, which we restate below for convenience.

\efficientclosureprds*
%
\begin{proof}
  For the course of the proof, fix $s \coloneqq 2^{2k}$ and the following label-oblivious rank decomposition automata:
  \begin{itemize}
    \item the exact rankwidth automaton $\ExactRankAutom = \ExactRankAutom_{2k,\, cs\ell+\ell}$, given by \cref{lem:exact-rank-automaton}; and
    \item the closure automaton $\BoundRepr = \BoundRepr_{c,\, s,\, \ell}$ of width $\ell$, given by \cref{lem:bound-repr-automaton}.
  \end{itemize}
  Note that both $\ExactRankAutom$ and $\BoundRepr$ have evaluation time $\Oh[c,\ell]{1}$.
  Our data structure consists simply of an~instance of $\BoundRepr$, maintained dynamically by the data structure of \cref{lem:automatonprds}.
  Thus the initialization time of the data structure on a~rooted annotated rank decomposition $\Tc$ is $\Oh[c,\ell]{|\Tc|}$, each prefix-rebuilding update $\prdesc$ is applied to the decomposition and the automaton in time $\Oh[c,\ell]{|\prdesc|}$, and each operation $\mathsf{Run}$ and $\mathsf{Valuation}$ runs in time $\Oh{1}$.
  
  It remains to implement $\mathsf{Closure}(\Tpref)$.
  So suppose we are given as a~query a~leafless prefix $\Tpref$ of $\Tc$.
  Let $A = \oApp_T(\Tpref)$ be the set of appendix edges of $\Tpref$ and let $\prtpref \coloneqq \{\reps(\oxp) \mid \oxp \in A\}$.
  We first perform a~clean-up of the prefix $\Tpref$ of $\Tc$ by replacing all representatives on the annotations in $\Tpref$ with elements of $\boldcup \prt_{\mathrm{pref}}$:
  \begin{claim}
    In time $\Oh[\ell]{|\Tpref|}$, one can produce a~rooted annotated rank decomposition $\tcskel = (T_\skel, U_\skel, \reps_\skel, \repse_\skel, \dmap_\skel)$ encoding the partitioned graph $(G[\prtpref], \prtpref)$ such that: (i) $T_\skel = T[\Tpref \cup \App_T(\Tpref)]$, (ii) for every $\oxp \in A$, we have $\lparts(\tcskel)[\oxp] = \reps(\oxp)$ and $\reps_\skel(\oxp) = \reps(\oxp)$.
  \end{claim}
  \begin{claimproof}
    Follows immediately from \cref{lem:annotdecomppart} and its proof.
  \end{claimproof}
  Note that for each $\oxp \in A$, $\reps_\skel(\opx)$ is a~(minimal) representative of $\lparts(\Tc)[\opx]$ in~$G$.
  
  \paragraph*{Auxiliary objects and definitions.}
%
  For $(q, \H) \in \AutomataReps^{c,s}(\Tc, \oxp)$, let the \emph{cut-rank cost} of $\H$ with respect to $\oxp$, denoted $\mathsf{ccost}(\H, \oxp)$, be the value computed as follows.
  Let $(H, \Dc)$ be the partitioned graph derived from $\H$.
  Let also $\overline{\Dc} = \Dc \cup \{\reps_\skel(\opx)\}$ and $\overline{H} = G[\overline{\Dc}]$.
  Then $\mathsf{ccost}(\H, \oxp) = \sum_{C \in \Dc} \cutrk_{\overline{H}}(C)$.

  Consider $\Lambda$ -- the set of all mappings $\lambda$ assigning to each edge $\oxp \in A$ a~member of $\AutomataReps^{c,s}(\Tc, \oxp)$.
  For every $\oxp \in A$, define $(q_\lambda(\oxp), \H_\lambda(\oxp)) \coloneqq \lambda(\oxp)$, i.e., $q_\lambda(\oxp)$ and $\H_\lambda(\oxp)$ are the first and the second coordinate of $\lambda(\oxp)$.
  Let also $(H_\lambda(\oxp), \Dc_\lambda(\oxp))$ denote the partitioned graph derived from $\H_\lambda(\oxp)$.
  Also, let $\overline{\Dc}_\lambda(\oxp) = \Dc_\lambda(\oxp) \cup \{\reps_\skel(\opx)\}$ and $\overline{H}_\lambda(\oxp) = G[\overline{\Dc}_\lambda(\oxp)]$.
  Next, set $r_\lambda(\oxp) = \sum_{C \in \Dc_\lambda(\oxp)} \cutrk_{\overline{H}_\lambda(\oxp)}(C) = \mathsf{ccost}(\H_\lambda(\oxp), \oxp)$.
  Then, for any $\lambda \in \Lambda$, define:
  \begin{itemize}
    \item $\Dc_\lambda \coloneqq \bigcup_{\oxp \in A} \Dc_\lambda(\oxp)$; equivalently, $\Dc_\lambda$ is the union of all nonempty parts in all indexed graphs $\H_\lambda(\oxp)$ for $\oxp \in A$;
    
    \item $G_\lambda \coloneqq G[\Dc_\lambda]$;
    
    \item $q_\lambda \coloneqq \sum_{\oxp \in A} q_\lambda(\oxp)$;
    
    \item $r_\lambda \coloneqq \sum_{\oxp \in A} r_\lambda(\oxp)$.
  \end{itemize}
  
  \paragraph*{Reduction from finding minimal closures to the optimization of $\lambda$.}
  For any $k$-closure $\prt$ of $\Tpref$, we shall say that it is \emph{represented} by a~family $\Dc$ of nonempty disjoint sets of $V(G)$ if $|\Dc| = |\prt|$ and for every set $C \in \prt$, $\Dc$ contains a~representative $D$ of $C$ in $G$.
  We will now prove the following claim, implying that a~representation of a~minimal $k$-closure can be found by examining only families~$\prt_\lambda$:
  
  \begin{claim}
    \label{cl:lambdas-represent-closures}
    Let $\lambda \in \Lambda$ be such that the rankwidth of $(G_\lambda, \Dc_\lambda)$ is at most $2k$ and, among all such mappings $\lambda$, the value $r_\lambda$ is minimum; and among those, $q_\lambda$ is minimum.
    Then for every partition $\prt$ of $V(G)$ defined as $\prt = \bigcup_{\oxp \in A} \prt_{\oxp}$, where $\prt_{\oxp}$ is a~partition of $\lparts(\Tc)[\oxp]$ into at most $c$ sets encoded by $\H_\lambda(\oxp)$ and of cost $q_\lambda(\oxp)$, $\prt$ is a~minimal $c$-small $k$-closure of $\Tpref$ represented by $\Dc_\lambda$.
    In particular, $\Dc_\lambda$ represents some minimal $c$-small $k$-closure of $\Tpref$.
    Moreover, if no $\lambda$ with the property above exists, then no $c$-small $k$-closure of $\Tpref$ exists.
  \end{claim}
  \begin{claimproof}
    Fix $\lambda \in \Lambda$ with the property that the rankwidth of $(G_\lambda, \Dc_\lambda)$ is at most $2k$.
    For every $\oxp \in A$, let $\prt_{\oxp}$ be a~partition of $\lparts(\Tc)[\oxp]$ into at most $c$ sets encoded by $\H_\lambda(\oxp)$ of cost $q_\lambda(\oxp)$.
    Then let $\prt$ be the partition of $V(G) = \bigcup_{\oxp \in A} \lparts(\Tc)[\oxp]$ defined as $\prt = \bigcup_{\oxp \in A} \prt_{\oxp}$.
    We claim that $\prt$ is a~$c$-small $k$-closure of $\Tpref$ such that $\sum_{C \in \prt} \cutrk_G(C) = r_\lambda$ and the number of nodes of $T$ cut by $\prt$ is exactly $q_\lambda + |\Tpref|$.
    
    \begin{itemize}
      \item \emph{$\prt$ is a~$k$-closure of $\Tpref$:} Let $\oxp \in A$. Since $\H_\lambda(\oxp)$ encodes $\prt_{\oxp}$, there exists a~bijection $\chi_{\oxp} \,\colon\, \prt_{\oxp} \to \Dc_\lambda(\oxp)$ such that for every $C \in \prt_{\oxp}$, $\chi_{\oxp}(C)$ is a~minimal representative of $C$ in $G$.
      Thus there exists a~bijection $\chi\,\colon\, \prt \to \Dc_\lambda$ such that for every $C \in \prt$, $\chi(C)$ is a~minimal representative of $C$ in $G$.
      Hence, the rankwidth of $(G[\prt], \prt)$ is equal to the rankwidth of $(G_\lambda, \Dc_\lambda) = (G[\Dc_\lambda], \Dc_\lambda)$, which is bounded from above by $2k$.
      Moreover, by construction, for every $C \in \prt$ we have $C \subseteq \lparts(\Tc)[\oxp]$ for some $\oxp \in A$.
      Therefore, $\prt$ is a~$k$-closure of $\Tpref$.
      
      \item \emph{$\prt$ is $c$-small}: for every $\oxp \in A$, $\prt_{\oxp}$ is the subfamily of $\prt$ forming a~partitioning of $\lparts(\Tc)[\oxp]$.
      By construction, $|\prt_{\oxp}| \leq c$.
      
      \item \emph{$\sum_{C \in \prt} \cutrk_G(C) = r_\lambda$}: choose $C \in \prt$ and let $\oxp \in A$ be such that $C \in \prt_{\oxp}$. As noted before, the bijection $\chi_{\oxp}\,\colon\,\prt_{\oxp} \to \Dc_\lambda(\oxp)$ is such that for every $C \in \prt_{\oxp}$, $\chi_{\oxp}(C)$ is a~minimal representative of $C$ in $G$.
      Let $R_C \coloneqq \chi_{\oxp}(C)$.
      Also, $\reps_\skel(\opx)$ is a~minimal representative of $\lparts(\Tc)[\opx]$ in $G$.
      Since $\boldcup \prt_{\oxp} \cup \lparts(\Tc)[\opx] = V(G)$, we find that
      \[ \cutrk_G(C) = \cutrk_{G[\boldcup \Dc_\lambda(\oxp) \cup \{\reps_\skel(\opx)\}]}(R_C) = \cutrk_{\overline{H}_\lambda(\oxp)}(R_C). \]
      The statement now follows by summing the equation above for all $C \in \prt$.
      
      \item \emph{$\prt$ cuts exactly $q_\lambda + |\Tpref|$ nodes of $T$}:
      Each node of $\Tpref$ must obviously be cut by every closure of $\Tpref$.
      Then, for every $\oxp \in A$, the value $q_\lambda(\oxp)$ denotes the cost of the partitioning $\prt_{\oxp}$ of $\lparts(\Tc)[\oxp]$, i.e., the number of nodes cut by $\prt_{\oxp}$ (equivalently, $\prt$) in the subtree of $T$ rooted at~$x$.
      Therefore, $\prt$ cuts $|\Tpref| + \sum_{\oxp \in A} q_\lambda(\oxp) = q_\lambda + |\Tpref|$ nodes of $T$.      
    \end{itemize}
    
    Conversely, let $\prt$ be a~$c$-small $k$-closure of $\Tpref$ and suppose that $\sum_{C \in \prt} \cutrk_G(C) = r$ and that $\prt$ cuts $q$ nodes of $T$.
    Our goal is to find a~mapping $\lambda \in \Lambda$ such that the rankwidth of $(G_\lambda, \Dc_\lambda)$ is at most $2k$, and $r_\lambda = r$ and $q_\lambda \leq q - |\Tpref|$.
    It is easy to see that the verification of this claim will finish the proof.
    
    For every $\oxp \in A$, let $\prt_{\oxp} \subseteq \prt$ comprise the parts of $\prt$ that are subsets of $\lparts(\Tc)[\oxp]$.
    Since $\prt$ is a~$c$-small closure of $\Tpref$, we have $\prt = \bigcup_{\oxp \in A} \prt_{\oxp}$ and $|\prt_{\oxp}| \leq c$ for all $\oxp \in A$.
    Let $q'_{\oxp}$ be the cost of $\prt_{\oxp}$, i.e., the number of the nodes in the subtree rooted at $\oxp$ that are cut by $\prt_{\oxp}$.
    As discussed earlier in the course of the proof, we have $q = |\Tpref| + \sum_{\oxp \in A} q'_{\oxp}$.
    
    For every $C \in \prt_{\oxp}$, we have $\cutrk_G(C) \leq 2k$ as $\prt$ is a~$k$-closure of $\Tpref$.
    So for every $C \in \prt_{\oxp}$, we can find a~minimal representative $R_C$ of $C$ in $G$ of cardinality at most $2^{2k} = s$.
    Thus, we can define an~$s$-small $(c, \reps(\oxp))$-indexed graph $\H'_{\oxp} = ((V^{\oxp}_1, \dots, V^{\oxp}_c), H_{\oxp}, \eta_{\oxp})$ encoding $\prt_{\oxp}$ by setting $\{V^{\oxp}_1, \dots, V^{\oxp}_{|\prt_{\oxp}|}\} = \{R_C \mid C \in \prt_{\oxp}\}$, $V^{\oxp}_{|\prt_{\oxp}| + 1} = \dots = V^{\oxp}_c = \emptyset$, and choosing $H_{\oxp}$ and $\eta_{\oxp}$ so as to ensure that $\H'_{\oxp}$ respects $\Tc$ along $\oxp$ (as discussed before, such a~choice is unique as soon as the sets $V^{\oxp}_1, \dots, V^{\oxp}_c$ are determined).
    Now by definition of $\AutomataReps^{c,s}(\Tc, \oxp)$, there exists a~pair $(q_{\oxp}, \H_{\oxp}) \in \AutomataReps^{c,s}(\Tc, \reps(\oxp))$ such that $\H_{\oxp} \sim^{c,\reps(\oxp)} \H'_{\oxp}$ and $q_{\oxp} \leq q'_{\oxp}$.
    Define then the mapping $\lambda \in \Lambda$ by setting $\lambda(\oxp) = (q_{\oxp}, \H_{\oxp})$ for each $\oxp \in A$.
    We claim that $\lambda$ satisfies the required conditions.
    
    In the following arguments, let $\pi_{\oxp}\,\colon\,V(\H_{\oxp}) \to V(\H'_{\oxp})$ be any isomorphism from $\H_{\oxp}$ to $\H'_{\oxp}$.
    Let also $\pi\,\colon\, \bigcup_{\oxp \in A} V(\H_{\oxp}) \to \bigcup_{\oxp \in A} V(\H'_{\oxp})$ be defined by $\funrestriction{\pi}{V(\H_{\oxp})} = \pi_{\oxp}$ for each $\oxp \in A$.
    By the properties of the isomorphism of indexed graphs, for every $\oxp \in A$ and $v \in V(\H_{\oxp})$, it holds that $N_G(v) \cap \lparts(\Tc)[\opx] = N_G(\pi(v)) \cap \lparts(\Tc)[\opx]$.
    
    Define $\Dc'_\lambda \coloneqq \pi(\Dc_\lambda) = \{\pi(C) \mid C \in \Dc_\lambda\} = \{R_C \mid C \in \prt\}$.
    We claim that $G[\Dc'_\lambda]$ is isomorphic to $G_\lambda$, with the isomorphism given by $\pi$.
    So let $u, v \in V(G_\lambda)$, aiming to show that $uv \in E(G_\lambda)$ if and only if $\pi(u) \pi(v) \in E(G[\Dc'_\lambda])$.
    
    \begin{itemize}
      \item Naturally, if $u$ and $v$ belong to the same part of $\Dc_\lambda$, then $\pi(u)$ and $\pi(v)$ belong to the same part of $\Dc'_\lambda$ and so $uv \notin E(G_\lambda)$ and $\pi(u)\pi(v) \notin E(G[\Dc'_\lambda])$.
      
      \item Otherwise, if $u, v \in V(\H_{\oxp})$ for some $\oxp \in A$ (but $u, v$ belong to different parts), then $uv \in E(\H_{\oxp})$ if and only if $\pi(u)\pi(v) \in E(\H'_{\oxp})$, since $\pi$ is an~isomorphism from $\H_{\oxp}$ to $\H'_{\oxp}$.
    As both $\H_{\oxp}$ and $\H'_{\oxp}$ respect $\Tc$ along $\oxp$, we have that $uv \in E(\H_{\oxp})$ if and only if $uv \in E(G)$; and that $\pi(u)\pi(v) \in E(\H'_{\oxp})$ if and only if $\pi(u)\pi(v) \in E(G)$.
    This settles this case.
      
      \item Finally, suppose $u \in V(\H_{\vec{x_1 p_1}})$ and $v \in V(\H_{\vec{x_2 p_2}})$ for $x_1 \neq x_2$.
    Then $u, \pi(u) \in \lparts(\Tc)[\vec{p_2 x_2}]$ and $v, \pi(v) \in \lparts(\Tc)[\vec{p_1 x_1}]$.
    From $N_G(u) \cap \lparts(\Tc)[\vec{p_1 x_1}] = N_G(\pi(u)) \cap \lparts(\Tc)[\vec{p_1 x_1}]$ we find that $uv \in E(G)$ if and only if $\pi(u)v \in E(G)$.
    And from $N_G(v) \cap \lparts(\Tc)[\vec{p_2 x_2}] = N_G(\pi(v)) \cap \lparts(\Tc)[\vec{p_2 x_2}]$ we get that $\pi(u)v \in E(G)$ if and only if $\pi(u)\pi(v) \in E(G)$ and we are done.
     \end{itemize}
    
    So $G[\Dc'_\lambda]$ is isomorphic to $G_\lambda$.
    We now verify the conditions required from $\lambda$.
    
    \begin{itemize}
      \item \emph{$(G_\lambda, \Dc_\lambda)$ has rankwidth at most $2k$:}
      For each $\oxp \in A$, by the construction of $\H'_{\oxp}$, each part of $\H'_{\oxp}$ is a~subset (in fact, a~minimal representative) of a~unique set in $\prt_{\oxp}$.
      Hence $\Dc'_\lambda$ is formed from $\prt$ by replacing each part $C \in \prt$ with some minimal representative of $C$ in $G$.
      Thus obviously, since $(G[\prt], \prt)$ has rankwidth at most $2k$, then so does $(G[\Dc'_\lambda], \Dc'_\lambda)$.
      As $G[\Dc'_\lambda] = \pi(G[\Dc_\lambda])$ and $\Dc'_\lambda = \pi(\Dc_\lambda)$, also $(G[\Dc_\lambda], \Dc_\lambda)$ has rankwidth at most $2k$.
      
      \item \emph{$r_\lambda = r$:} let $\oxp \in A$.
      Let $\overline{H}'_\lambda(\oxp) = G[\overline{\Dc}'_\lambda(\oxp)]$, where $\overline{\Dc}'_\lambda(\oxp) = \{R_C \mid C \in \prt_{\oxp}\} \cup \{\reps_\skel(\opx)\}$.
      Since $R_C$ is a~representative of $C$ in $G$ for each $C \in \prt_{\oxp}$ and $\reps_\skel(\opx)$ is a~representative of $\lparts(\Tc)[\opx]$ in $G$ and $\boldcup \prt_{\oxp} \cup \lparts(\Tc)[\opx] = V(G)$, we have, for every $C \in \prt_{\oxp}$,
      \[ \cutrk_G(C) = \cutrk_{\overline{H}'_\lambda(\oxp)}(R_C). \]
      But now observe that the partitioned graphs $(\overline{H}_\lambda(\oxp), \overline{\Dc}_\lambda(\oxp))$ and $(\overline{H}'_\lambda(\oxp), \overline{\Dc}'_\lambda(\oxp))$ are isomorphic, with the isomorphism preserving $\reps_\skel(\opx)$ and mapping each vertex $v \in \boldcup \Dc_\lambda(\oxp) = V(\H_{\oxp})$ to $\pi(v)$.
      Therefore, for every $C \in \prt_{\oxp}$,
      \[ \cutrk_{\overline{H}'_\lambda(\oxp)}(R_C) = \cutrk_{\overline{H}_\lambda(\oxp)}(\pi^{-1}(R_C)). \]
      Since $\Dc_\lambda(\oxp) = \{\pi^{-1}(R_C) \mid C \in \prt_{\oxp}\}$, we conclude that
      \begin{equation}
      \label{eq:cutrk-bashing}
      \begin{split}
      \sum_{C \in \prt_{\oxp}} \cutrk_G(C) &= \sum_{C \in \prt_{\oxp}} \cutrk_{\overline{H}_\lambda(\oxp)}(\pi^{-1}(R_C)) \\&= \sum_{C \in \Dc_\lambda(\oxp)} \cutrk_{\overline{H}_\lambda(\oxp)}(C) = \mathsf{ccost}(\H_\lambda(\oxp), \oxp).
      \end{split}
      \end{equation}
      We get the required equality by summing \cref{eq:cutrk-bashing} for all $\oxp \in A$ and recalling that $r = \sum_{C \in \prt} \cutrk_G(C)$ and $r_\lambda = \sum_{\oxp \in A} \mathsf{ccost}(\H_\lambda(\oxp), \oxp)$.
      
      \item \emph{$q_\lambda \leq q - |\Tpref|$:} this follows immediately from the facts that $q = |\Tpref| + \sum_{\oxp \in A} q'_{\oxp}$ and that $q_{\oxp} \leq q'_{\oxp}$ for each $\oxp \in A$.
    \end{itemize}
    Therefore, the proof is complete.
  \end{claimproof}
  
  
  \paragraph*{Rank decompositions of partitioned graphs $(G_\lambda, \Dc_\lambda)$.}
  We now show how, for any mapping $\lambda \in \Lambda$, we produce a~rank decomposition of the partitioned graph $(G_\lambda, \Dc_\lambda)$.
  
  Consider an~edge $\oxp \in A$ and a~pair $(q, \H) \in \AutomataReps^{c,s}(\Tc, \oxp)$.
  For technical reasons, we will now rename vertices of $\H$ so as to ensure that $\H$ contains all vertices of $\reps(\oxp)$.
  We construct a~graph $\H^\star$ from $\H$ as follows: For every vertex $u \in \reps(\oxp)$ such that $u \notin V(\H)$, choose any vertex $v \in V(\H)$ such that $\eta(\H)(v) = u$ (such a~vertex exists since $\H$ encodes some partition of $\lparts(\Tc)[\oxp]$ and $\reps(\oxp)$ is a~minimum representative of $\lparts(\Tc)[\oxp]$), and rename $v$ to $u$.
  Let also $\pi_\H$ be the isomorphism from $\H^\star$ to $\H$ prescribed by the procedure above.
  Naturally, this construction ensures that $\H^\star$ is isomorphic to $\H$ (but we stress that there could be $u \in V_i(\H^\star)$ and $v \in V_j(\H^\star)$ with $i \neq j$ such that $uv \in E(\H^\star) \not\Leftrightarrow uv \in E(G)$).
  By the properties of $\eta(\H)$, we have that, for every $u \in \reps(\oxp)$,
  \begin{equation}
    \label{eq:stupid-technical-mapping}
    N_G(u) \cap \lparts(\Tc)[\opx] = N_G(\pi_\H(u)) \cap \lparts(\Tc)[\opx].
  \end{equation}
  
  Given an~edge $\oxp \in A$ and a~pair $(q, \H) \in \AutomataReps^{c,s}(\Tc, \oxp)$, define now an~annotated rank decomposition derived from $\H^\star$, denoted $\Tc_{\H^\star}$, as follows.  
  Recall that $(H, \Dc)$ is the partitioned graph derived from $\H$ and $\overline{\Dc} = \Dc \cup \{\reps_\skel(\opx)\}$, and $\overline{H} = G[\overline{\Dc}]$.
  Then define $(\overline{H}^\star, \overline{\Dc}^\star)$ as the partitioned graph created from $(\overline{H}, \overline{\Dc})$ by renaming each vertex $v \in V(H)$ to $\pi_\H^{-1}(v)$.
  Note that by the construction of $\overline{H}^\star$, we have that $\reps_\skel(\oxp) \cup \reps_\skel(\opx) \subseteq V(\overline{H}^\star)$ and moreover $\reps_\skel(\opx) \in \overline{\Dc}^\star$.
  Then let $\Tc_{\H^\star} = (T_{\H^\star}, U_{\H^\star}, \reps_{\H^\star}, \repse_{\H^\star}, \dmap_{\H^\star})$ be an~arbitrary annotated rank decomposition of $(\overline{H}^\star, \overline{\Dc}^\star)$ with the following properties:
  \begin{itemize}
    \item $V(T_{\H^\star}) \cap V(T_\skel) = \{x, p\}$ and $\opx$ is a~leaf edge of $T_{\H^\star}$;
    \item $\lparts(\Tc_{\H^\star})[\opx] = \reps_{\H^\star}(\opx) = \reps_\skel(\opx)$ and $\reps_{\H^\star}(\oxp) = \reps_\skel(\oxp) = \reps(\oxp)$.
  \end{itemize}
  It can be easily seen that such a~decomposition exists and can be constructed from $\H$ and the annotations on the edge $xp$ of $\Tc$ in time $\Oh[c,\ell]{1}$.
  Observe also that $\repse_{\H^\star}(xp) = \repse_\skel(xp)$: For any pair of vertices $u \in \reps_\skel(\oxp)$, $v \in \reps_\skel(\opx)$ we have $uv \in E(\repse_{\H^\star}(xp))$ if and only if $\pi_\H(u) v \in E(G)$ by the definition of $\H^\star$.
  But by \cref{eq:stupid-technical-mapping}, $\pi_\H(u) v \in E(G)$ if and only if $uv \in E(G)$, which holds if and only if $uv \in E(\repse_\skel(xp))$.
  Next, since $|\boldcup \overline{\Dc}^\star| = |V(\H^\star)| + |\reps_\skel(\opx)| \leq cs\ell + \ell$, the width of $\Tc_{\H^\star}$ is bounded by $cs\ell + \ell$.
  Let also $\Tc_{\H}$ be the decomposition formed from $\Tc_{\H^\star}$ by renaming all vertices $v \in V(\H^\star)$ of the graph encoded by the decomposition back to $\pi_\H(v)$.
  Naturally, $\Tc_{\H}$ encodes $(\overline{H}, \overline{\Dc}) = (G[\overline{\Dc}], \overline{\Dc})$.
  
  Next, for any $\lambda \in \Lambda$, define the following rank decompositions:
  \begin{itemize}
    \item $\Tc_\lambda^\star$ -- the decomposition formed by gluing $\tcskel$ along $xp$ with each decomposition $\Tc_{\H^\star_{\lambda(\oxp)}} = (T_{\H^\star_{\lambda(\oxp)}},\allowbreak U_{\H^\star_{\lambda(\oxp)}},\allowbreak \reps_{\H^\star_{\lambda(\oxp)}},\allowbreak \repse_{\H^\star_{\lambda(\oxp)}},\allowbreak \dmap_{\H^\star_{\lambda(\oxp)}})$ for $\oxp \in A$ in arbitrary order; this gluing is possible since for every $\oxp \in A$, we have $\oxp \in \leafe(T_\skel)$, $\opx \in \leafe(T_{\H^\star})$, $\reps_{\H^\star_{\lambda(\oxp)}}(\opx) = \reps_\skel(\opx)$, $\reps_{\H^\star_{\lambda(\oxp)}}(\oxp) = \reps_\skel(\oxp)$ and $\repse_{\H^\star_{\lambda(\oxp)}}(xp) = \reps_\skel(xp)$.
    It is easy to see that $\Tc_\lambda^\star$ encodes some partitioned graph with vertex set $\bigcup_{\oxp \in A} V(H^\star_{\lambda(\oxp)})$.
    Moreover, its width is bounded by $cs\ell + \ell$ as discussed at the introduction of the notion of gluing decompositions.
    
    \item $\Tc_\lambda$ -- the decomposition formed from $\Tc_\lambda^\star$ by renaming every vertex $v$ in the partitioned graph encoded by $\Tc_\lambda^\star$ such that $v \in V(H^\star_{\lambda(\oxp)})$ for $\oxp \in A$ back to $\pi_{\H_{\lambda(\oxp)}}(v)$.
    Of course, the width of $\Tc_\lambda$ is also bounded by $cs\ell + \ell$.
  \end{itemize}
  
  The following observation follows straight from the analysis of the construction of $\Tc_\lambda$ and $\Tc^\star_\lambda$.
  \begin{observation}
    $\Tc_\lambda$ encodes the partitioned graph $(G_\lambda, \Dc_\lambda)$, and $\Tc^\star_\lambda$ encodes a~partitioned graph isomorphic to $(G_\lambda, \Dc_\lambda)$.
  \end{observation}

  \paragraph*{Optimizing $\lambda$.}
  At this point of time, we have reduced the problem to finding a~mapping $\lambda \in \Lambda$ with the rankwidth of $(G_\lambda, \Dc_\lambda)$ bounded by $2k$, such that the pair $(r_\lambda, q_\lambda)$ is lexicographically minimum possible.
  In the sequel, we will show how this can be done using the exact rankwidth automaton $\ExactRankAutom = (Q,\iota,\delta,\varepsilon)$.
  
  We now briefly sketch the idea.
  A~brute-force search for an~optimum $\lambda$ would look as follows: recall that $\Tc_\lambda$ is an~annotated rank decomposition of $(G_\lambda, \Dc_\lambda)$ of width $cs\ell + \ell$.
  Hence running $\ExactRankAutom$ on $\Tc_\lambda$ will correctly determine whether the rankwidth of the encoded partitioned graph $(G_\lambda, \Dc_\lambda)$ is at most $2k$.
  Repeating this procedure for all possible $\lambda \in \Lambda$ yields all viable mappings $\lambda$; for each of these, we can easily compute the values $r_\lambda$ and $q_\lambda$ -- each of these is of the form $q_\lambda = \sum_{\oxp \in A} f_{\oxp}(\lambda(\oxp))$ and $r_\lambda = \sum_{\oxp \in A} g_{\oxp}(\lambda(\oxp))$ for some functions $f_{\oxp}$, $g_{\oxp}$ that can be evaluated efficiently given $\lambda(\oxp)$.
  Thus we can find the optimum mapping $\lambda$.
  
  Note that in the description above, instead of the annotated decomposition $\Tc_\lambda$ encoding $(G_\lambda, \Dc_\lambda)$, we could have used an~annotated decomposition $\Tc^\star_\lambda$ encoding a~partitioned graph isomorphic to $(G_\lambda, \Dc_\lambda)$.
  Then $\ExactRankAutom$, when run on $\Tc^\star_\lambda$, will return that the encoded partitioned graph has rankwidth at most $2k$ if and only if it would do so when run on $\Tc_\lambda$.
  This choice has an~important consequence: All annotated decompositions $\Tc^\star_\lambda$ have \emph{the same} annotated prefix.
  Formally, given two annotated rank decompositions $\Tc_1 = (T_1, U_1, \reps_1, \repse_1, \dmap_1)$ and $\Tc_2 = (T_2, U_2, \reps_2, \repse_2, \dmap_2)$ and a~set $S \subseteq V(T_1) \cap V(T_2)$, we say that $\Tc_1$ and $\Tc_2$ \emph{agree on} $S$ if
  \[
  \begin{split}
  T_1[S] &= T_2[S], \\
  \funrestriction{\reps_1}{\oE(T_1[S])} &= \funrestriction{\reps_2}{\oE(T_2[S])}, \\
  \funrestriction{\repse_1}{E(T_1[S])} &= \funrestriction{\repse_2}{E(T_2[S])}, \\
  \funrestriction{\dmap_1}{\PT(T_1[S])}&= \funrestriction{\dmap_2}{\PT(T_2[S])}.
  \end{split}
  \]
  Then, for any $\lambda_1, \lambda_2 \in \Lambda$, the decompositions $\Tc^\star_{\lambda_1}$ and $\Tc^{\star}_{\lambda_2}$ agree on $\Tpref' \coloneqq \Tpref \cup \App_T(\Tpref)$.
  This observation will allow us to reuse the partial runs of $\ExactRankAutom$, which will enable us to find the optimum mapping $\lambda$ by means of a~dynamic programming on the rooted subtree induced by $\Tpref'$ (precisely, using \cref{lem:automaton-dp}).
  The details can be found below.
  
  Let $\oxp \in A$ and $(q, \H) \in \AutomataReps^{c,s}(\Tc, \oxp)$.
  Define the state $\xi_{\H^\star} \in Q$ of $\ExactRankAutom$ as follows.
  Recall that $\opx$ is the unique leaf edge of $\Tc_{\H^\star}$ such that $\lparts(\Tc_{\H^\star})[\opx] = \reps_\skel(\opx)$ (and so $\reps_{\H^\star}(\opx) = \reps_\skel(\opx)$ and $\reps_{\H^\star}(\oxp) = \reps_\skel(\oxp)$).
  Then let $\rho_{\H^\star}$ be the run of $\ExactRankAutom$ on $(\Tc_{\H^\star}, x, p)$, and set $\xi_{\H^\star} \coloneqq \rho_{\H^\star}(\oxp)$.
  Note that $\xi_{\H^\star}$ can be determined in time $\Oh[c,\ell]{1}$.

  We now claim that in a~run of $\ExactRankAutom$ on $\Tc^\star_\lambda$ for some $\lambda \in \Lambda$, the partial runs on the glued decompositions $\Tc_{\H^\star_{\lambda(\oxp)}}$ are exactly the recorded states $\xi_{\H^\star_{\lambda(\oxp)}}$.

  \begin{claim}
    \label{cl:partialruns}
    Let $\lambda \in \Lambda$ and $\oxp \in A$.
    If $\rho$ is the run of $\ExactRankAutom$ on $\Tc^\star_\lambda$, then $\rho(\oxp) = \xi_{\H^\star_{\lambda(\oxp)}}$.
  \end{claim}
  \begin{claimproof}
    Observe that the set $B \coloneqq V(\Tc_{\H^\star_{\lambda(\oxp)}})$ comprises exactly $p$ and the set of descendants of $x$ in $T^\star_\lambda$.
    Moreover, by the construction of $\Tc^\star_\lambda$ (and the properties of gluing decompositions), we get that the decompositions $\Tc_{\H^\star_{\lambda(\oxp)}}$ and $\Tc^\star_\lambda$ agree on $B$.
    We immediately infer that $\rho(\oxp) = \rho_{\H^\star_{\lambda(\oxp)}}(\oxp) = \xi_{\H^\star_{\lambda(\oxp)}}$.
  \end{claimproof}
  
  Aiming to use \cref{lem:automaton-dp} in our case, let $\overline{\Z} \coloneqq \Z \cup \{+\infty\}$ and define the totally ordered commutative semigroup $(S, +, \leq)$, where $S = \overline{\Z} \times \overline{\Z}$, $+$ is the coordinate-wise sum and $\leq$ is the lexicographic order on $S$. 
  Then define the cost function ${\bf c} \,\colon\, A \times Q \to S$ by setting, for every $\oxp \in A$ and $f \in Q$, the value
  \begin{equation}
    \label{eq:costdef}
    {\bf c}(\oxp, f) = \min\{ (\mathsf{ccost}(\H, \oxp), q) \,\mid\, (\H, q) \in \AutomataReps^{c,s}(\oxp), \,\xi_{\H^\star} = f \};
  \end{equation}
  where we set ${\bf c}(\oxp, f) = (+\infty, +\infty)$ if the set on the right-hand side of \cref{eq:costdef} is empty.
  Let also $F \subseteq Q$ be the set of states of $\ExactRankAutom$ accepting that the input decomposition describes a~partitioned graph of rankwidth at most $2k$; or equivalently, $F$ is the set of states representing non-empty full sets of width $2k$ at a~root of an~input decomposition.
  
  We now show that the results of \cref{lem:automaton-dp} will be enough to determine the existence of $\lambda$ with the rankwidth of $(G_\lambda, \Dc_\lambda)$ bounded by $2k$, and in the case any such $\lambda$ exists -- to determine an~optimum mapping $\lambda$.
  
  \begin{claim}
    \label{cl:lambda-to-automaton}
    Suppose $\lambda \in \Lambda$ is such that $(G_\lambda, \Dc_\lambda)$ has rankwidth at most $2k$.
    Let $\kappa \,\colon\, A \to Q$ be defined as $\kappa(\oxp) = \xi_{\H^\star_{\lambda(\oxp)}}$ for each $\oxp \in A$, and let $\rho_\kappa$ be the $\kappa$-run of $\autom$ on $\Tc_\skel$.
    Then $\rho_\kappa(\vartheta) \in F$ and ${\bf c}(\kappa) \leq (r_\lambda, q_\lambda)$.
  \end{claim}
  \begin{claimproof}
    Let also $\rho$ be the run of $\ExactRankAutom$ on $\Tc^\star_\lambda$.
    By \cref{cl:partialruns}, we have $\rho(\oxp) = \xi_{\H^\star_{\lambda(\oxp)}} = \kappa(\oxp)$.
    Since $\Tc^\star_\lambda$ and $\Tc_\skel$ agree on $V(\Tc_\skel)$, we infer that $\rho_\kappa(\vartheta) = \rho(\vartheta)$.
    Therefore, $\rho_\kappa(\vartheta) \in F$ if and only if $\rho(\vartheta) \in F$, which only holds when the full set of $\Tc^\star_\lambda$ at the root $r$ of width at most $2k$ is nonempty (i.e., $(G_\lambda, \Dc_\lambda)$ has rankwidth at most $2k$).
    So $\rho_\kappa(\vartheta) \in F$.
    Since ${\bf c}(\oxp, \kappa(\oxp)) \leq (r_\lambda(\oxp), q_\lambda(\oxp))$ for each $\oxp \in A$ (by \cref{eq:costdef}), ${\bf c}(\kappa) = \sum_{\oxp \in A} {\bf c}(\oxp, \kappa(\oxp))$, $r_\lambda = \sum_{\oxp \in A} r_\lambda(\oxp)$ and $q_\lambda = \sum_{\oxp \in A} q_\lambda(\oxp)$, we conclude that ${\bf c}(\kappa) \leq (r_\lambda, q_\lambda)$.
  \end{claimproof}
  
  \begin{claim}
    \label{cl:automaton-to-lambda}
    Suppose there exists a~leaf edge state mapping $\kappa \,\colon\, A \to Q$ such that ${\bf c}(\kappa) \neq (+\infty, +\infty)$ and, for the $\kappa$-run $\rho_\kappa$ of $\autom$ on $(\overline{\Tc}_\skel, r_1, r_2)$, we have $\rho_\kappa(\vartheta) \in F$.
    Then there exists $\lambda \in \Lambda$ such that $(G_\lambda, \Dc_\lambda)$ has rankwidth at most $2k$ and $(r_\lambda, q_\lambda) = {\bf c}(\kappa)$.
    Moreover, $\lambda$ can be constructed in time $\Oh[c,\ell]{|\Tpref|}$.
  \end{claim}
  \begin{claimproof}
    Construct a~valuation $\lambda \in \Lambda$ as follows.
    For every $\oxp \in A$, choose $\lambda(\oxp)$ to be such a~pair $(\H, q) \in \AutomataReps^{c,s}(\oxp)$ that $\xi_{\H^\star} = \kappa(\oxp)$ and $(\mathsf{ccost}(\H, \oxp), q) = {\bf c}(\oxp, \kappa(\oxp))$.
    Repeating the same argument involving \cref{cl:partialruns} as before, we find that since $\rho_\kappa(\vartheta) \in F$, we have that $(G_\lambda, \Dc_\lambda)$ has rankwidth at most $2k$.
    We also easily verify that ${\bf c}(\kappa) = (r_\lambda, q_\lambda)$.
  \end{claimproof}

  Apply now \cref{lem:automaton-dp} for the automaton $\ExactRankAutom$, the semigroup $(S, +, \leq)$, the decomposition $\Tc_\skel$, the cost function ${\bf c}$, and the set of accepting states $F$.
  The algorithm of \cref{lem:automaton-dp} runs in time $\Oh[c,\ell]{|T_\skel|} = \Oh[c,\ell]{|\Tpref|}$ and returns one of the following:
  \begin{itemize}
    \item \emph{there is no mapping $\kappa \,\colon\, A \to Q$ such that $\rho_\kappa(\vartheta) \in F$ where $\rho_\kappa$ is the $\kappa$-run on $\Tc_\skel$, or the cost of all such mappings is $(+\infty, +\infty)$}.
      Then by \cref{cl:lambda-to-automaton} there exists no $\lambda \in \Lambda$ with the property that the rankwidth of $(G_\lambda, \Dc_\lambda)$ is at most $2k$; hence we can return that $\Tpref$ has no $c$-small $k$-closure.
      
    \item \emph{$\kappa \,\colon\, A \to Q$ is the minimum-cost mapping such that $\rho_\kappa(\vartheta) \in F$ where $\rho_\kappa$ is the $\kappa$-run on $\Tc_\skel$, and the cost of the mapping is finite}.
      Then we reconstruct the mapping $\lambda \in \Lambda$ in time $\Oh[c,\ell]{|\Tpref|}$ such that $(r_\lambda, q_\lambda) = {\bf c}(\kappa)$ using \cref{cl:automaton-to-lambda}.
      By \cref{cl:lambda-to-automaton}, such a~mapping has the minimum value of $r_\lambda$; and among all such optimal mappings, it also has the minimum possible value of $q_\lambda$.
  \end{itemize}    
  
  Finally, using \cref{cl:lambdas-represent-closures}, we conclude that:
  \begin{corollary}
    \label{cor:found-best-mapping-lambda}
    In time $\Oh[c, \ell]{|\Tpref|}$, we can:
    \begin{itemize}
      \item correctly decide that $\Tpref$ has no $c$-small $k$-closure; or
      \item find a~mapping $\lambda \in \Lambda$ such that $\Dc_\lambda$ represents some minimal $c$-small $k$-closure of $\Tpref$.
      Moreover, for every partition $\prt$ of $V(G)$ defined as $\prt = \bigcup_{\oxp \in A} \prt_{\oxp}$, where $\prt_{\oxp}$ is a~partition of $\lparts(\Tc)[\oxp]$ into at most $c$ sets encoded by $\H_\lambda(\oxp)$ and of cost $q_\lambda(\oxp)$, $\prt$ is a~minimal $c$-small $k$-closure of~$\Tpref$.
    \end{itemize}
  \end{corollary}
  
  
  \paragraph*{Reconstructing the closure.}
  Having found $\lambda$, we want now to reconstruct any minimal $c$-small $k$-closure $\prt$ of $\Tpref$.
  Recall that, since we cannot afford to compute $\prt$ explicitly (since a~closure is essentially an arbitrary partitioning of $V(G)$), we are required to return the closure in a~compact form -- precisely, the sets $\cut_T(\prt)$ and $\aep_T(\prt)$, that is the prefix of $T$ cut by $\prt$ and the appendix edge partition of $\prt$.
  The procedure should work in time $\Oh[c,\ell]{|\cut_T(\prt)|}$.
  
  Let $\oxp \in A$ and recall that $(q_\lambda(\oxp), \H_\lambda(\oxp)) = \lambda(\oxp) \in \AutomataReps^{c,s}(\Tc, \oxp)$.
  We will now present a~subroutine finding a~partition $\prt_{\oxp}$ of $\lparts(\Tc)[\oxp]$, represented implicitly as $\aep_T(\prt_{\oxp})$, so that $\prt_{\oxp}$ is of cost $q_\lambda(\oxp)$ and is encoded by $\H_\lambda(\oxp)$.
  
  At the start of the subroutine, we initialize a~sequence of initially empty pairwise disjoint subsets $(V_1, \dots, V_c)$ of $\lparts(\Tc)[\oxp]$; eventually, $(V_1, \dots, V_c)$ will form an~indexed partition of $\lparts(\Tc)[\oxp]$.
  The sets $V_1, \dots, V_c$ are represented implicitly by sets $E_1, \dots, E_c$ of oriented edges of $T$ with the property that $E_i = \aes_T(V_i)$ for each $i \in [c]$.
  We now implement a~recursive function $\textsc{Populate}(\vec{ab}, q, \H)$ that, under the assumptions that $\vec{ab}$ is a~predecessor of $\oxp$ in $T$ and $(q, \H) \in \AutomataReps^{c,s}(\Tc, \vec{ab})$, adds to each set $V_1, \dots, V_c$ a~subset $V^{\vec{ab}}_1, \dots, V^{\vec{ab}}_c$, respectively, so that $(V^{\vec{ab}}_1, \dots, V^{\vec{ab}}_c)$ is an~indexed partition of $\lparts(\Tc)[\vec{ab}]$ of cost $q$ encoded by $\H$. (Note that such an~indexed partition must exist by the assumptions.)
  In the implementation, we consider two cases.
  \begin{itemize}
    \item If $q = 0$, then no node of the subtree of $T$ rooted at $\vec{ab}$ may be cut by $(V^{\vec{ab}}_1, \dots, V^{\vec{ab}}_c)$.
    That is, the entire subset $\lparts(\Tc)[\vec{ab}]$ belongs to one of the sets $V^{\vec{ab}}_j$.
    Here, the value $j$ can be found in constant time since it is exactly the unique index $j$ such that $V_j(\H) \neq \emptyset$.
    So we add $\vec{ab}$ to $E_j$ and we are done.
    
    \item If $q \geq 1$, then some nodes of the subtree of $T$ rooted at $\vec{ab}$ are cut by $(V^{\vec{ab}}_1, \dots, V^{\vec{ab}}_c)$; in particular, one of these nodes must be $a$, and moreover, $\vec{ab}$ cannot be a~leaf edge of $T$ and so $\vec{ab}$ has two children $\vec{y_1a}, \vec{y_2a}$.
    In constant time (using the dynamic data structure of \cref{lem:automatonprds} maintaining $\BoundRepr$ on $\Tc$ dynamically), we read the value $\rho(\vec{ab})$, where $\rho$ is the run of $\BoundRepr$ on $\Tc$.
    By \cref{lem:bound-repr-automaton}, the value $\rho(\vec{ab})$ contains a~mapping $\Phi$; let $((q_1, \H_1), (q_2, \H_2)) = \Phi((q, \H))$ such that $(q_1, \H_1) \in \AutomataReps^{c,s}(\Tc, \vec{y_1a})$ and $(q_2, \H_2) \in \AutomataReps^{c,s}(\Tc, \vec{y_2a})$.
    We then run $\textsc{Populate}(\vec{y_1 a}, q_1, \H_1)$ and $\textsc{Populate}(\vec{y_2 a}, q_2, \H_2)$ and add $a$ to $\cut_T(\prt)$.
    
    The two recursive calls add to the sets $V_1, \dots, V_c$ the subsets $V^{\vec{y_1a}}_1, \dots, V^{\vec{y_1a}}_c$ and $V^{\vec{y_2a}}_1, \dots, V^{\vec{y_2a}}_c$, respectively, with the property that for each $t \in [2]$, the sequence $(V^{\vec{y_ta}}_1, \dots, V^{\vec{y_ta}}_c)$ is an~indexed partition of $\lparts(\Tc)[\vec{y_t a}]$ of cost $q_t$ encoded by $\H_t$.
    So again by \cref{lem:bound-repr-automaton}, the sequence $V^{\vec{ab}}_1, \dots, V^{\vec{ab}}_c$ given by $V^{\vec{ab}}_j = V^{\vec{y_1a}}_j \cup V^{\vec{y_2a}}_j$ is an~indexed partition of $\lparts(\Tc)[\vec{ab}]$ of cost $q$ encoded by $\H$.
    Since the recursive calls already added each set $V^{\vec{ab}}_j$ to $V_j$, we are done.
  \end{itemize}
  Thus running $\textsc{Populate}(\oxp, q_\lambda(\oxp), \H_\lambda(\oxp))$ will create an~indexed partition $(V_1, \dots, V_c)$ of $\lparts(\Tc)[\oxp]$ of cost $q_\lambda(\oxp)$ encoded by $\H_\lambda(\oxp)$; the partition is stored implicitly as sets $E_1, \dots, E_c$.
  So letting $\prt_{\oxp} = \{V_1, \dots, V_c\} \setminus \{\emptyset\}$, the nonempty sets in $E_1, \dots, E_c$ form $\aep_T(\prt_{\oxp})$.
  Tracing the execution of $\textsc{Populate}$, it is easy to verify that this set $\aep_T(\prt_{\oxp})$ can be computed in time $\Oh[c,\ell]{|\cut_{\oxp}(\prt_{\oxp})| + 1}$, where $\cut_{\oxp}(\prt_{\oxp})$ is the set of nodes of $T$ that are children of $\oxp$ that are cut by $\prt_{\oxp}$.
  
  Now let $\prt \coloneqq \bigcup_{\oxp \in A} \prt_{\oxp}$, so that $\prt$ is encoded by $\aep_T(\prt) \coloneqq \bigcup_{\oxp \in A} \aep_T(\prt_{\oxp})$.
  Then by \cref{cor:found-best-mapping-lambda}, $\prt$ is indeed a~minimal $c$-small $k$-closure of $\Tpref$.
  The set of nodes cut by $\prt$ is exactly $\cut_T(\prt) = \Tpref \cup \bigcup_{\oxp \in A} \cut_{\oxp}(\prt_{\oxp})$.
  The set $\aep_T(\prt)$ can be found by invoking the function $\textsc{Populate}(\oxp, q_\lambda(\oxp), \H_\lambda(\oxp))$ for each $\oxp \in A$ separately and gathering the nonempty sets of edges after each call.
  The time complexity of all recursive calls is bounded by
  \[ \Oh[c,\ell]{\sum_{\oxp \in A} |\cut_{\oxp}(\prt_{\oxp})| + 1} \leq
    \Oh[c,\ell]{|\Tpref| + \sum_{\oxp \in A} |\cut_{\oxp}(\prt_{\oxp})|} =
    \Oh[c,\ell]{|\cut_T(\prt)|},
  \]
  since $|A| = |\Tpref| + 1$.
  This finishes the description of the effective reconstruction of $\cut_T(\prt)$ and $\aep_T(\prt)$.
  
  \paragraph*{Obtaining the decomposition of the closure.}
  The final object we are required to return is a~rank decomposition $(T^\star, \lambda^\star)$ of $(G[\prt], \prt)$ of width at most $2k$.
  Remembering that the partition $\prt$ reconstructed a~moment ago is represented by $\Dc_\lambda$, we observe that the task at hand can be accomplished~by:
  \begin{itemize}
    \item computing a~rank decomposition $(T^\square, \lambda^\square)$ of $(G_\lambda, \Dc_\lambda)$ of width at most $2k$, and
    \item producing a~rank decomposition $(T^\star, \lambda^\star)$ of $(G[\prt], \prt)$ by setting $T^\star \coloneqq T^\square$ and setting $\lambda^\star(C) \coloneqq \lambda^\square(R_C)$ for every $C \in \prt$, where $R_C \in \prt_\lambda$ is a~representative of $C$ in $G$.
  \end{itemize}
  The former step is done by constructing the annotated decomposition $\Tc_\lambda$ of $(G_\lambda, \Dc_\lambda)$ (of width at most $cs\ell + \ell = \Oh[c,\ell]{1}$) explicitly in time $\Oh[c,\ell]{|\Tpref|}$.
  Since the rankwidth of $(G_\lambda, \Dc_\lambda)$ -- equal to the rankwidth of $(G[\prt], \prt)$ -- is at most $2k$, we apply \cref{lem:linear-decomposition-recovery} in time $\Oh[c,\ell]{|\Tpref|}$ and we are done.
  The latter step can then be performed in time $\Oh{|\Tpref|}$ as long as $\prt$ is represented by $\aep_T(\prt)$; or in other words, $\lambda^\star$ is represented as a~function $\lambda\colon \aep_T(\prt) \to \leafe(T^\star)$.
  This concludes the proof of \cref{lem:closureprds}.
\end{proof}

%% file: exact-rankwidth-automaton.tex
\subsection{Exact rankwidth automaton}
\label{ssec:rankwidth-automaton}

In this subsection, we will present an~implementation of the exact rankwidth automaton.
As a~consequence, we will also show that for any pair of integers $k, \ell \in \N$, one can determine -- in linear time with respect to the size of the graph -- whether a~partitioned graph, encoded by an~annotated rank decomposition of width at most $\ell$, has rankwidth at most $k$.
Moreover, in the positive case, in linear time we can recover a~rank decomposition of the partitioned graph of width at most $k$ (or in near-linear time if we require the output to be an~annotated decomposition).
The construction of the automaton crucially relies on the understanding of the cubic-time algorithm of Jeong, Kim and Oum~\cite{DBLP:journals/siamdm/JeongKO21} computing optimum-width rank decompositions of graphs and, more generally, subspace arrangements.
We proceed to give a~summary of this algorithm below.

\paragraph*{Summary of the algorithm of Jeong, Kim and Oum~\cite{DBLP:journals/siamdm/JeongKO21}.}
The $\Oh[\ell]{n^3}$ algorithm of~\cite{DBLP:journals/siamdm/JeongKO21} for rankwidth of subspace arrangements uses at its core the following subroutine: given two integers $k$, $\ell$ with $\ell \geq k$, a~subspace arrangement $\Vc$ of subspaces of $\F^d$, and a~rank decomposition of $\Vc$ of width $\ell$, determine whether a~rank decomposition of $\Vc$ of width $k$ exists; and if so, construct any such decomposition.
This subroutine is an~analog of a~similar linear-time algorithm for tree decompositions of graphs by Bodlaender and Kloks~\cite{DBLP:journals/jal/BodlaenderK96}.
Here, we provide a~brief description of the subroutine in~\cite{DBLP:journals/siamdm/JeongKO21}.

Suppose $\Tc^b = (T^b, \lambda^b)$ is a~rooted rank decomposition of $\Vc$ of width $\ell$.
Ideally, we would wish to compute, for every node $x \in V(\Tc^b)$, the set of all possible (unrooted) rank decompositions of $\Vc_x$ of width at most $k$.
Such sets would be computed using a~bottom-up dynamic programming scheme on $\Tc^b$ -- the only slightly non-trivial part is understanding, for a~node $x \in V(\Tc^b)$ with two children $c_1$, $c_2$, how to find the set of all rank decompositions of $\Vc_x$ of small width, given the corresponding sets of decompositions of $\Vc_{c_1}$ and $\Vc_{c_2}$.
Obviously, this idea, while correct, is doomed to fail since a~graph can (and usually will) have an~exponential number of valid rank decompositions of small width.

Thus, \cite{DBLP:journals/siamdm/JeongKO21} mimics the insight of Bodlaender and Kloks that, for each rank decomposition of $\Vc_x$ of small width, we can record just essential information about it, which. Roughly speaking, this information is a~heavily compressed version of the rank decomposition, with the details irrelevant to the subspaces in $\Vc \setminus \Vc_x$ stripped off. This information is named a~\emph{compact $B$-namu} in~\cite{DBLP:journals/siamdm/JeongKO21}.\footnote{In~\cite{DBLP:journals/jal/BodlaenderK96}, an~analogous piece of information is called a~\emph{characteristic}.}
Precisely, given a~subspace $B$ of $\F^d$, we define a~$B$-namu as a~tuple $(T, \alpha, {\bf w}, U)$, where: (i) $T$ is a~subcubic tree, possibly with some degree-$2$ nodes, (ii) $U$ is a~subspace of $B$, (iii) every oriented edge $\ouv \in \oE(T)$ is decorated with a~subspace $\alpha(\ouv) \subseteq U$, (iv) every edge $uv \in E(T)$ is decorated with an~integer ${\bf w}(uv) \geq 0$.
(There are a~couple of additional restrictions on the values of $\alpha(\cdot)$ and ${\bf w}(\cdot)$ -- that is, we have $\alpha(\vec{v_1v_2}) \subseteq \alpha(\vec{v_3v_4})$ whenever $\vec{v_1v_2}$ is a~predecessor of $\vec{v_3v_4}$ in $T$, and we have ${\bf w}(uv) \geq \dim(\alpha(\ouv) \cap \alpha(\ovu))$ for all $uv \in E(T)$ -- but these will be unimportant for our purposes. Similarly, their definition of a~$B$-namu includes additional objects that can be uniquely deduced from $(T, \alpha, {\bf w}, U)$.)
The width of a~$B$-namu is the maximum value of ${\bf w}(uv)$, or $0$ if $T$ is edgeless.

Note that there exists a~natural way of turning a~rank decomposition $\Tc = (T, \lambda)$ of $\Vc$ into a~$B$-namu $(T, \alpha, {\bf w}, U)$: we define $\alpha(\ouv) = B \cap \sumof{\lparts(\Tc)[\ouv]}$, ${\bf w}(uv) = \dim(\sumof{\lparts(\Tc)[\ouv]} \cap \sumof{\lparts(\Tc)[\ovu]})$, and $U = B \cap \sumof{\Vc}$.
It is a~straightforward exercise to verify that such a~construction indeed produces a~valid $B$-namu of width equal to the width of $\Tc$.

For the purposes of this summary we do not describe how to compress a~$B$-namu into the equivalent compact $B$-namu of equal width~\cite[Section 3.2]{DBLP:journals/siamdm/JeongKO21}.
Here, we only present the most essential takeaway of this process: Assuming bounded $\dim(B)$, compact $B$-namus of bounded width are small and all such $B$-namus can be generated quickly.
Henceforth, let $U_k(B)$ denote the set of all compact $B$-namus of width at most $k$.
Next, for any ordered basis $\BB = (\lin{v}_1, \dots, \lin{v}_{|\BB|})$ of $B$, let $U_k(\BB)$ denote the same set of compact $B$-namus, but where each subspace of $B$ is represented in the basis $\BB$. (So in a $B$-namu represented in the ordered basis $\BB$ of $B$, every subspace $A \subseteq B$ is encoded by a~sequence of $\dim(A) \cdot |\BB|$ bits $c_{ij}$ for $i \in [\dim(A)]$, $j \in [|\BB|]$ as the subspace spanned by vectors $\sum_{j=1}^{|\BB|} c_{ij}\lin{v}_j$ for $i \in [\dim(A)]$.)
For $t \in \N$, let $U_k^t$ denote the set of all possible encodings of compact $B$-namus of width at most $k$ in an~ordered basis of size at most $t$.
(So $U_k(\BB) \subseteq U_k^t$ for every ordered basis $\BB$ with $|\BB| \leq t$.)
Then we have that:


\begin{lemma}[{\cite[Lemma 5.4]{DBLP:journals/siamdm/JeongKO21}}]
  \label{lem:compact-namus}
  There exist functions $f_{\ref{lem:compact-namus}}\,\colon\, \N^2 \to \N$ and $g_{\ref{lem:compact-namus}}, h_{\ref{lem:compact-namus}}\,\colon\, \N^3 \to \N$ such that the following holds.
  Assume $\dim(B) = \ell \geq 0$.
  If $\Gamma = (T, \alpha, {\bf w}, U)$ is a~compact $B$-namu of width at most $k$, then $|E(T)| \leq f_{\ref{lem:compact-namus}}(k, \ell)$.
  Moreover, we have that $|U_k(B)| \leq g_{\ref{lem:compact-namus}}(k, \ell, |\F|)$ and moreover, the entire set $U_k^{\ell}$ can be generated in time $h_{\ref{lem:compact-namus}}(k, \ell, |\F|)$.
\end{lemma}


Now, given $\Tc^b$, Jeong, Kim and Oum aim to compute for each $x \in V(T^b)$ the \emph{full set} at $x$ of width $k$ with respect to $\Tc^b$: essentially, the set of compact $B_x$-namus of all possible \emph{totally pure} unrooted rank decompositions of $\Vc_x$ of width at most $k$.\footnote{The notion of totally pure decompositions is announced in \cref{ssec:dealternation-prelims} and formally introduced in \cref{sec:totally-pure-decomp}, however the exact definition is not relevant here. For this recap, it is enough to remember that a~rank decomposition of $\Vc$ of width at most $k$ exists if and only if a~totally pure rank decomposition of $\Vc$ of width at most $k$ also exists (\cref{lem:korean-rank-structure-theorem}).}
This full set is denoted $\fullset_k(x)$.
Note that $\fullset_k(x) \subseteq U_k(B_x)$.
Since $\F = \GF(2)$ in our work and $|B_x| \leq \ell$ (since $\Tc^b$ is a~decomposition of width at most $\ell$), we see that $|\fullset_k(x)| \leq |U_k(B_x)| \leq g_{\ref{lem:compact-namus}}(k, \ell, 2)$.
Similarly, for an~ordered basis $\BB_x$ of $B_x$, let $\fullset_k^{\BB_x}(x) \subseteq U_k(\BB_x) \subseteq U_k^{|\BB_x|}$ denote the set $\fullset_k(x)$, but where all subspaces of $B_x$ are represented in the ordered basis $\BB_x$ as described above.

In order to facilitate the efficient computation of the full sets, \cite{DBLP:journals/siamdm/JeongKO21} introduces the notion of a~\emph{transcript} of $\Tc^b$.
Recall that the boundary space of $x$ is defined as $B_x = \sumof{\Vc_x} \cap \sumof{\Vc \setminus \Vc_x}$.
We now also define the space $B'_x$ as follows:
\[
  B'_x = 
  \begin{cases}
    B_x & \text{if $x$ is a~leaf of $T^b$}, \\
    B_x + B_{c_1} + B_{c_2} & \text{if $x$ is an~internal node of $T^b$ with children $c_1$ and $c_2$}.
  \end{cases}
\]
(\cite{DBLP:journals/siamdm/JeongKO21} equivalently uses $B_{c_1} + B_{c_2}$ in the second case, after having proved the inclusion $B_x \subseteq B_{c_1} + B_{c_2}$.)
Then a~\emph{transcript} of $\Tc^b$ is formed from two sets of ordered bases $\{\BB_x\}_{x \in V(T^b)}$ and $\{\BB'_x\}_{x \in V(T^b)}$ under the following conditions for all $x \in V(T^b)$:
\begin{itemize}
  \item $\BB_x$ is an~ordered basis of $B_x$ and $\BB'_x$ is an~ordered basis of $B'_x$;
  \item $\BB_x$ is a~prefix of $\BB'_x$.
\end{itemize}
Then, for any non-root node $x$ with parent $y$, we define the \emph{transition matrix} of $x$ as the unique $|\BB'_y| \times |\BB_x|$ matrix $M_{\vec{xy}}$ over $\F$ with the following property: Suppose ${\bf v}$ is a~vector in $B_x$ and that ${\bf v}' \in \F^{|\BB_x|}$ is the (unique) representation of ${\bf v}$ in the ordered basis $\BB_x$, that is, ${\bf v} = \sum_{i=1}^{|\BB_x|} {\bf v}'_i \cdot (\BB_x)_i$.
Then $M_{\vec{xy}} {\bf v}'$ is the unique representation of ${\bf v}$ in the ordered basis $\BB'_y$.
Intuitively, $M_{\vec{xy}}$ describes how the space $B_x$ embeds as a~subspace in $B'_y$.
Notably, this description has bitsize bounded by $\Oh{\ell^2}$ since $|\BB_x| = \dim(B_x) \leq \ell$ and $|\BB'_y| = \dim(B'_y) \leq \Oh{\ell}$, even though $B_x$ and $B'_y$ are subspaces of the highly-dimensional space $\F^d$.
This should be contrasted with the actual ordered bases $\BB_x$, $\BB'_y$: The representation of each ordered basis requires $\Omega(d)$ bits, and in our setting we will have $d = n$.
Therefore, even storing the transcript $\{\BB_x\}_{x \in V(T^b)}$, $\{\BB'_x\}_{x \in V(T^b)}$ requires $\Omega(n^2)$ bits of storage, so we cannot hope to compute it in subquadratic time.


It is then proved that:
\begin{lemma}[informal statement of {\cite[Theorem 7.8]{DBLP:journals/siamdm/JeongKO21}}]
  \label{lem:informal-transcript}
  Suppose that the subspace arrangement $\Vc$ is suitably preprocessed and let $n = |\Vc|$.
  Moreover, assume that each subspace in $\Vc$ has dimension at most $\ell$.
  Then given a~rooted rank decomposition $\Tc^b$ of width at most $\ell$, we can compute a~transcript $(\{\BB_x\}_{x \in V(T^b)},\ \{\BB'_x\}_{x \in V(T^b)})$ and the set of transition matrices $M_{\vec{xy}}$ in time $\Oh[\ell]{n^2}$.
\end{lemma}

The quadratic dependency on the size of the subspace arrangement in \cref{lem:informal-transcript} is a~bottleneck of the algorithm in \cite{DBLP:journals/siamdm/JeongKO21}.
The reason the algorithm in \cref{lem:informal-transcript} is inefficient is that it \emph{does} determine the transcript explicitly; it is, however, not clear at all how to avoid this step when processing general subspace arrangements.
Our contribution is to show that in the setting of rank decompositions of graphs, the transition matrices of some fixed transcript of $\Tc^b$ can be efficiently inferred from an~annotated rank decomposition that encodes a~graph.


Finally, Jeong et al.\ prove the following claim.
Note that this statement is not present explicitly in their work, but it follows immediately from the analysis of their Algorithm~3.1 and their discussion of Proposition~7.10 in Section~7.5.
\begin{lemma}[{\cite{DBLP:journals/siamdm/JeongKO21}}]
  \label{lem:full-set-ingredients}
  Suppose $(\{\BB_x\}_{x \in V(T^b)},\ \{\BB'_x\}_{x \in V(T^b)})$ is a~transcript of $\Tc^b$ and for every $x \in V(T^b)$ with parent $y$, $M_{\vec{xy}}$ is the transition matrix of $x$ with respect to the transcript.
  Then, for any $x \in V(T^b)$:
  \begin{itemize}
    \item If $x$ is a~leaf of $T^b$, then $\fullset^{\BB_x}_k(x)$ contains exactly one $B_x$-namu that can be computed knowing only the cardinality of $\BB_x$ in time $\Oh[\ell]{1}$.
    \item If $x$ is a~non-leaf node of $T^b$ with two children $c_1, c_2$, then $\fullset^{\BB_x}_k(x)$ can be computed from $\fullset^{\BB_{c_1}}_k(c_1)$, $\fullset^{\BB_{c_2}}_k(c_2)$, $M_{\vec{c_1x}}$, $M_{\vec{c_2x}}$ and $|\BB_x|$ in time $\Oh[\ell]{1}$.
  \end{itemize}
\end{lemma}

Finally, the authors show how to construct a~rank decomposition of small width, having computed all full sets:
\begin{lemma}[{\cite[Proposition~7.12]{DBLP:journals/siamdm/JeongKO21}}]
  \label{lem:korean-decomposition-recovery}
  Let $r$ be the root of $T^b$ and let $n = |\Vc|$.
  Then $\Vc$ admits a~rank decomposition of width at most $k$ if and only if $\fullset_k(r) \neq \emptyset$ (equivalently, $\fullset_k^{\BB_r}(r) \neq \emptyset$).
  If such a~decomposition exists, then a~(rooted) rank decomposition of $\Vc$ of width at most $k$ can be constructed from the set of transition matrices $M_{\vec{xy}}$, where $x \in V(T^b)$ and $y$ is the parent of $x$, and full sets $\fullset_k^{\BB_x}(x)$ for $x \in V(T^b)$ in time $\Oh[\ell]{n}$.
\end{lemma}


\paragraph*{Transcript and transition matrices in annotated rank decompositions.}

Assume we are given a~rooted annotated rank decomposition $\Tc^b = (T^b, U^b, \reps^b, \repse^b, \dmap^b)$ of width $\ell$ encoding a~partitioned graph $(G, \Cc)$.
Assume for convenience that the vertices of $G$ are assigned integer labels from $1$ to $n \coloneqq |V(G)|$. 
Recall from \cref{ssec:dealternation-prelims} that $\Tc^b$ is isomorphic to a~rank decomposition of the subspace arrangement $\Vc = \{A_C\}_{C \in \prt}$ of width $2\ell$, where for every $C \in \prt$, $A_C \subseteq \GF(2)^n$ is the canonical subspace of $C$, spanned by the vectors $\lin{e}_v$ and $\sum_{u \in N(v)} \lin{e}_u$ for all $v \in C$.
Henceforth, without worrying about confusion, we will simultaneously treat $\Tc^b$ as an~annotated rank decomposition of $(G, \prt)$ and as a~rank decomposition of $\Vc$.
Whenever we consider an edge $xp \in E(\Tc^b)$, where $p$ is the parent of $x$, by the width of $xp$ we mean its width $q \in [0, \ell]$ in the decomposition of $(G, \prt)$; so $\dim(B_x) = 2q$.

Our current aim is to define a~specific transcript of $\Tc^b$ -- which we shall name a~\emph{canonical transcript} of $\Tc^b$ -- and then show that for any $x \in V(T^b)$ with parent $p$, the transition matrix $M_{\vec{xp}}$ with respect to the canonical transcript can be uniquely and efficiently deduced from the annotations around $x$ in $\Tc^b$.

We begin with understanding the boundary space $B_x$ for a~node $x \in V(T^b)$. Recall that $B_x = \sumof{\Vc_x} \cap \sumof{\Vc \setminus \Vc_x}$.
For convenience, we introduce the following shorthand notation: $\lin{e}_S \coloneqq \sum_{u \in S} \lin{e}_u$ for any $S \subseteq V(G)$.

\begin{lemma}
Let $x$ be a non-root node of $T^b$ and $p$ the parent of $x$, and let $q \in [0, \ell]$ be the width of the edge $xp$ in $\Tc^b$.
  Let also $S = \lparts(\Tc^b)[\vec{xp}] \subseteq V(G)$ be the set of vertices of $G$ assigned to leaf edges in the subtree of $x$ in $\Tc^b$.
  Then:
  \begin{itemize}
    \item The subspace $A_{\vec{xp}} \coloneqq \sumof{\{\lin{e}_{N(v) \setminus S} \mid v \in \reps^b(\vec{xp})\}}$ of $\GF(2)^n$ has dimension $q$;
    \item The subspace $A_{\vec{px}} \coloneqq \sumof{\{\lin{e}_{N(v) \cap S} \mid v \in \reps^b(\vec{px})\}}$ of $\GF(2)^n$ has dimension $q$;
    \item $B_x = A_{\vec{xp}} + A_{\vec{px}}$ and $A_{\vec{xp}} \cap A_{\vec{px}} = \{\lin{0}\}$.
  \end{itemize}
\end{lemma}
\begin{proof}
  Recall that the rank $q$ of the edge $xp$ is defined as the rank of the 0-1-matrix $M$ describing adjacencies between vertices in $S$ and vertices in $\compl{S}$ over $\GF(2)$.
  Supposing the rows of $M$ are indexed by $S$ and the columns are indexed by $\compl{S}$, we see that the row rank of $M$ is exactly
  \[ \dim(\sumof{\{\lin{e}_{N(v) \setminus S} \mid v \in S\}}) = \mathrm{rk}(M) = q. \]
  (This is because for any $v \in S$, the $v$th row of $M$ is given exactly by the vector $\lin{e}_{N(v) \setminus S}$, with the 0 entries of the vector corresponding to the elements of $S$ removed.)
  Since $\reps^b(\vec{xp})$ is a~representative of $S$ in $G$, we immediately have that $\{\lin{e}_{N(v) \setminus S} \mid v \in S\} = \{\lin{e}_{N(v) \setminus S} \mid v \in \reps^b(\vec{xp})\}$ and the first statement of the lemma follows.
  The second point is proved analogously, only that we consider the column rank of $M$ instead.

  For the final point, recall that
  \[ \sumof{\Vc_x} = \sumof{\{\lin{e}_v \mid v \in S\} \cup \{ \lin{e}_{N(v)} \mid v \in S \}}. \]
  For every $v \in S$, we subtract from $\lin{e}_{N(v)}$ all vectors $\lin{e}_u$ with $u \in N(v) \cap S$; since such vectors belong to $\sumof{\Vc_x}$, this operation does not change the subspace spanned by vectors and thus
  \[ \begin{split}
  \sumof{\Vc_x} &= \sumof{\{\lin{e}_v \mid v \in S\} \cup \{ \lin{e}_{N(v) \setminus S} \mid v \in \reps^b(\vec{xp}) \}} \\
  &= \sumof{\{\lin{e}_v \mid v \in S\}} + \sumof{\{ \lin{e}_{N(v) \setminus S} \mid v \in \reps^b(\vec{xp}) \}} = \sumof{\{\lin{e}_v \mid v \in S\}} + A_{\vec{xp}}.
  \end{split} \]
  Similarly,
  \[ \sumof{\Vc \setminus \Vc_x} = \sumof{\{\lin{e}_v \mid v \in \compl{S}\}} + \sumof{\{ \lin{e}_{N(v) \cap S} \mid v \in \reps^b(\vec{px}) \}} = \sumof{\{\lin{e}_v \mid v \in \compl{S}\}} + A_{\vec{px}}. \]
  Since $A_{\vec{xp}} \subseteq \sumof{\{\lin{e}_v \mid v \in \compl{S}\}}$, $A_{\vec{px}} \subseteq \sumof{\{\lin{e}_v \mid v \in S\}}$, and $\sumof{\{\lin{e}_v \mid v \in S\}} \cap \sumof{\{\lin{e}_v \mid v \in \compl{S}\}} = \{\lin{0}\}$, we conclude that $A_{\vec{xp}} \cap A_{\vec{px}} = \{\lin{0}\}$ and $B_x = \sumof{\Vc_x} \cap \sumof{\Vc \setminus \Vc_x} = A_{\vec{xp}} + A_{\vec{px}}$.
\end{proof}

Next, given a~sequence of vectors $(\lin{v}_1, \dots, \lin{v}_m)$ of a~linear space, define the \emph{lexicographically earliest basis} as the subsequence $(\lin{v}_{i_1}, \dots, \lin{v}_{i_t})$ of $(\lin{v}_1, \dots, \lin{v}_m)$ that is an~ordered basis of $\sumof{\{\lin{v}_1, \dots, \lin{v}_m\}}$ with the property that the sequence $(i_1, \dots, i_t)$ is lexicographically smallest possible.

We now define the canonical transcript $(\{\BB_x\}_{x \in V(T^b)},\ \{\BB'_x\}_{x \in V(T^b)})$ of $\Tc^b$.
\begin{itemize}
  \item For every $x \in V(T^b)$, define the canonical ordered basis $\BB_x$ of $B_x$ as follows.
  If $x$ is the root of $T^b$, then $\BB_x$ is empty.
  Otherwise, let $p$ be the parent of $x$ in $T^b$ and $q \in [0, \ell]$ be the rank of the edge $xp$.
  Consider the sequence of vectors $(\lin{e}_{N(v) \setminus S})_{v \in \reps^b(\vec{xp})}$ with indexes sorted by $<$; that is, $\lin{e}_{N(u) \setminus S}$ appears before $\lin{e}_{N(v) \setminus S}$ if and only if $u < v$.
  Then let $\BB_{\vec{xp}}$ be the lexicographically earliest basis of this sequence (so $\BB_{\vec{xp}}$ is an~ordered basis of $A_{\vec{xp}}$).
  Also define the ordered basis $\BB_{\vec{px}}$ of $A_{\vec{px}}$ as the lexicographically earliest basis of the analogous sequence of vectors $(\lin{e}_{N(v) \cap S})_{v \in \reps^b(\vec{px})}$. 
  Now define $\BB_x$ as the concatenation of $\BB_{\vec{xp}}$ and $\BB_{\vec{px}}$.
  
  \item Then, for every $x \in V(T^b)$, define the canonical ordered basis $\BB'_x$ of $B'_x$ as follows.
  If $x$ is a~leaf of $T^b$, then $\BB'_x = \BB_x$.
  Otherwise, let $c_1 < c_2$ be the two children of $x$ in $T^b$ and set $\BB'_x$ to the lexicographically earliest basis of the concatenation of the sequences $\BB_x$, $\BB_{c_1}$ and $\BB_{c_2}$.
\end{itemize}

It is easy to verify that $(\{\BB_x\}_{x \in V(T^b)},\ \{\BB'_x\}_{x \in V(T^b)})$ is indeed a~transcript of $\Tc^b$.
For all non-root nodes $x$ of $\Tc^b$ with parent $p$, define $M_{\vec{xp}}$ as the transition matrix of $x$ with respect to the canonical transcript of $\Tc^b$.
Our aim now is to show that each transition matrix $M_{\vec{xp}}$ can be recovered from the annotations around $p$ in $\Tc$.
For the following statement, recall the definitions of the transition signature and the edge signature from \Cref{sec:rwautom}.

\begin{lemma}
  \label{lem:transition-matrices}
  Let $x \in V(T^b)$ be a~non-root node and $p$ be the parent of $x$ in $T^b$.
  Then, in time $\Oh[\ell]{1}$, one can construct $M_{\vec{xp}}$ from:
  \begin{itemize}
    \item the transition signature $\tau(\Tc^b, \vec{pp'})$ if $p$ is a~non-root node of $T^b$ with parent $p'$; or
    \item the edge signatures $\sigma(\Tc^b, \vec{c_1r}), \sigma(\Tc^b, \vec{c_2r})$ if $p=r$ is the root of $T^b$ with children $c_1 < c_2$.
  \end{itemize}
\end{lemma}

In the proof, we will use the following simple observation.
We say that two sequences of vectors of equal length $\lin{v}_1, \dots, \lin{v}_m \in \F^d$ and $\lin{v}'_1, \dots, \lin{v}'_m \in \F^{d'}$ are \emph{linearly equivalent} if for every sequence of coefficients $a_1, \dots, a_m \in \F$, we have that $\sum_{i=1}^m a_i \lin{v}_i = \lin{0}$ if and only if $\sum_{i=1}^m a_i \lin{v}'_i = \lin{0}$.
Note that in this case, $(\lin{v}_{i_1}, \dots, \lin{v}_{i_t})$ is the lexicographically earliest basis of $(\lin{v}_1, \dots, \lin{v}_m)$ if and only if $(\lin{v}'_{i_1}, \dots, \lin{v}'_{i_t})$ is the lexicographically earliest basis of $(\lin{v}'_1, \dots, \lin{v}'_m)$.
Moreover, if $j \in [m]$ and $a_1, \dots, a_t \in F$, then $\lin{v}_j = \sum_{k=1}^t a_k \lin{v}_{i_k}$ if and only if $\lin{v}'_j = \sum_{k=1}^t a_k \lin{v}'_{i_k}$.
Then:
\begin{observation}
  \label{obs:dimension-reduction}
  Let $\lin{v}_1, \dots, \lin{v}_m$ be vectors of the same vector space $\F^d$ and let $a \in [d]$.
  Suppose one of the following conditions holds:
  \begin{itemize}
    \item $(\lin{v}_i)_a = 0$ for every $i \in [m]$, i.e., in all vectors, the $a$th entry is zero; or
    \item there exists a~different index $b \in [d]$ such that $(\lin{v}_i)_a = (\lin{v}_i)_b$ for every $i \in [m]$, i.e., in all vectors, the $a$th entry and the $b$th entry coincide.
  \end{itemize}
  Let $\lin{v}'_1, \lin{v}'_2, \dots, \lin{v}'_m$ be the vectors of $\F^{d-1}$ formed by dropping the $a$th coordinate from each vector $\lin{v}_i$.
  Then the sequences $(\lin{v}_1, \dots, \lin{v}_m)$ and $(\lin{v}'_1, \dots, \lin{v}'_m)$ are linearly equivalent.
\end{observation}

Also we will use the following algorithmic tool which is an~easy application of Gaussian elimination: 
\begin{lemma}
  \label{lem:computing-earliest-basis}
  Let $\lin{v}_1, \lin{v}_2, \dots, \lin{v}_m$ be vectors of the same vector space $\F^d$.
  Then in time $\Oh{md^2}$ one can compute:
  \begin{itemize}
    \item the lexicographically earliest basis $(\lin{v}_{i_1}, \dots, \lin{v}_{i_t})$ of $(v_1, \dots, \lin{v}_m)$, and
    \item for every $j \in [m]$, the unique representation of $\lin{v}_j$ in this basis (i.e., the coefficients $a_{j,1}, \dots, a_{j,t} \in \F$ such that $\lin{v}_j = \sum_{k=1}^t a_{j,k} \lin{v}_{i_k}$).
  \end{itemize}
\end{lemma}

Therefore, we quickly get that:
\begin{lemma}
  \label{lem:canonical-basis-simple}
  Let $x \in V(T^b)$ be a~non-root vertex of $T^b$ with parent $p$, and let $S = \lparts(\Tc^b)[\vec{xp}] \subseteq V(G)$ and $q \in [0, \ell]$ be the width of $xp$.
  Then, given the sets $\reps^b(\vec{xp}), \reps^b(\vec{px})$ and the bipartite graph $\repse^b(xp)$, in time $\Oh[q]{1}$ one can compute the canonical ordered basis $\BB_x$, represented implicitly as two sequences of vertices $v^{\vec{xp}}_1, \dots, v^{\vec{xp}}_{q} \in \reps^b(\vec{xp})$ and $v^{\vec{px}}_1, \dots, v^{\vec{px}}_q \in \reps^b(\vec{px})$ such that
  \[
    \BB_x = (\lin{e}_{N(v^{\vec{xp}}_1) \setminus S},\, \dots,\, \lin{e}_{N(v^{\vec{xp}}_q) \setminus S},\, \lin{e}_{N(v^{\vec{px}}_{1}) \cap S},\, \dots,\, \lin{e}_{N(v^{\vec{px}}_{q}) \cap S}).
  \]
\end{lemma}
\begin{proof}
  Recall that $\BB_x$ is the concatenation of $\BB_{\vec{xp}}$ and $\BB_{\vec{px}}$.
  Here, we only show how to compute $\BB_{\vec{xp}}$; the latter is determined analogously.
  Let $\reps^b(\vec{xp}) = \{v_1 < v_2 < \dots < v_m\}$, where $m = |\reps^b(\vec{xp})| \leq 2^q$.
  Then $\BB_{\vec{xp}}$ is defined as the lexicographically earliest basis of the sequence $(\lin{e}_{N(v_1) \setminus S}, \dots, \lin{e}_{N(v_m) \setminus S})$.
  By a~repeated application of \cref{obs:dimension-reduction}, we produce a~linearly equivalent sequence of vectors $(\lin{u}'_1, \dots, \lin{u}'_m)$ by dropping from each vector of this sequence the coordinates corresponding to the vertices $u \in V(G)$ such that:
  \begin{itemize}
    \item $u \in S$ (since $(\lin{e}_{N(v_i) \setminus S})_u = 0$ for all $i \in [m]$); or
    \item $u \notin S$ and $u \notin \reps^b(\vec{px})$ (since then there exists a~representative $u' \in \reps^b(\vec{px})$ of $u$ such that $N(u) \cap S = N(u') \cap S$, or equivalently, $(\lin{e}_{N(v_i) \setminus S})_u = (\lin{e}_{N(v_i) \setminus S})_{u'}$ for all $i \in [m]$).
  \end{itemize}
  In other words, let $(\lin{u}'_1, \dots, \lin{u}'_m)$ be the sequence of vectors in $\F^{|\reps^b(\vec{px})|}$, where $\lin{u}'_i$ is constructed from $\lin{e}_{N(v_i) \setminus S}$ by dropping all coordinates not corresponding to the vertices of $\reps^b(\vec{px})$.
  This sequence can be constructed explicitly in time $\Oh[q]{1}$ using $\reps^b(\vec{xp})$, $\reps^b(\vec{px})$ and $\repse^b(xp)$.
  Using \cref{lem:computing-earliest-basis}, we find the lexicographically earliest basis $(\lin{u}'_{i_1}, \dots, \lin{u}'_{i_q})$ of $(\lin{u}'_1, \dots, \lin{u}'_m)$.
  Then by \cref{obs:dimension-reduction}, we have that $\BB_{\vec{xp}} = (\lin{e}_{N(v_{i_1}) \setminus S}, \dots, \lin{e}_{N(v_{i_q}) \setminus S})$.
\end{proof}

We are now ready to prove \cref{lem:transition-matrices}.

\begin{proof}[Proof of \cref{lem:transition-matrices}]
  First suppose that $p$ is the parent of $x$ and $p'$ is the parent of $p$ in $T^b$.
  Let also $x^\star$ be the sibling of $x$, i.e., the other child of $p$ in $T^b$.
  We showcase the proof in the case where $x < x^\star$, but the case $x > x^\star$ is analogous.
  
  Recall that $\BB'_p$ is the lexicographically earliest basis of the concatenation of $\BB_p$, $\BB_x$ and $\BB_{x^\star}$ (where $\BB_p$ is the concatenation of $\BB_{\vec{pp'}}$ and $\BB_{\vec{p'p}}$; $\BB_x$ is the concatenation of $\BB_{\vec{xp}}$ and $\BB_{\vec{px}}$; and $\BB_{x^\star}$ is the concatenation of $\BB_{\vec{x^\star p}}$ and $\BB_{\vec{px^\star}}$).
  Note that the transition signature $\tau(\Tc^b, \vec{pp'})$ contains the representative sets $\reps^b(\vec{xp})$, $\reps^b(\vec{px})$, $\reps^b(\vec{x^\star p})$, $\reps^b(\vec{p x^\star})$, $\reps^b(\vec{pp'})$, $\reps^b(\vec{p'p})$ and the bipartite graphs $\repse^b(xp)$, $\repse^b(x^\star p)$, $\repse^b(pp')$, so we can use \cref{lem:canonical-basis-simple} to compute the implicit representations of each $\BB_p, \BB_x, \BB_{x^\star}$ in time $\Oh[\ell]{1}$.
  Let $S$ be the concatenation of these ordered bases.
  Let also $S'$ be the sequence of vectors produced from $S$ by dropping all the coordinates corresponding to vertices outside of $\reps^b(\vec{xp}) \cup \reps^b(\vec{x^\star p}) \cup \reps^b(\vec{p'p})$.
  
  \begin{claim}
    $S$ and $S'$ are linearly equivalent.
  \end{claim}
  \begin{claimproof}
    Note that $V(G)$ is a~disjoint union of $\lparts(\Tc^b)[\vec{xp}]$, $\lparts(\Tc^b)[\vec{x^\star p}]$ and $\lparts(\Tc^b)[\vec{p'p}]$.
    First suppose that $u \in \lparts(\Tc^b)[\vec{xp}]$, but $u \notin \reps^b(\vec{xp})$.
    Then there exists a~representative $u' \in \reps^b(\vec{xp})$ such that $N(u) \cap \lparts(\Tc^b)[\vec{px}] = N(u') \cap \lparts(\Tc^b)[\vec{p'x}]$.
    We now claim that for every vector $\lin{v} \in S$, we have $\lin{v}_u = \lin{v}_{u'}$.
    \begin{itemize}
      \item If $\lin{v}$ belongs to the ordered basis $\BB_{\vec{xp}}$ (i.e., $\lin{v}$ is implicitly represented by a~vertex of $\reps^b(\vec{xp})$), then by definition $\lin{v}_s = 0$ for all $s \in \lparts(\Tc^b)[\vec{xp}]$. Hence $\lin{v}_u = \lin{v}_{u'} = 0$.
      \item Similarly, if $\lin{v}$ belongs to $\BB_{\vec{pp'}}$ (resp.\ $\BB_{\vec{px^\star}}$), then the same argument follows from the fact that $\lparts(\Tc^b)[\vec{xp}]$ is a~subset of $\lparts(\Tc^b)[\vec{pp'}]$ (resp.\ $\lparts(\Tc^b)[\vec{px^\star}]$).
      So $\lin{v}_u = \lin{v}_{u'} = 0$.
      \item If $\lin{v}$ belongs to any of the ordered bases $\BB_{\vec{px}}$, $\BB_{\vec{p'p}}$ or $\BB_{\vec{x^\star p}}$, then $\lin{v}_u = \lin{v}_{u'}$ follows from $N(u) \cap \lparts(\Tc^b)[\vec{px}] = N(u') \cap \lparts(\Tc^b)[\vec{p'x}]$ and the fact that $\lparts(\Tc^b)[\vec{px}]$ is a~superset of both $\lparts(\Tc^b)[\vec{p'p}]$ and $\lparts(\Tc^b)[\vec{x^\star p}]$.
    \end{itemize}
    By case exhaustion we conclude that $\lin{v}_u = \lin{v}_{u'}$ for all vectors $\lin{v} \in S$, and so \cref{obs:dimension-reduction} applies and the coordinate corresponding to the vertex $u$ can be removed from all vectors of $S$ while maintaining the linear equivalence.
    A~symmetric proof for $u \in \lparts(\Tc^b)[\vec{x^\star p}] \setminus \reps^b(\vec{x^\star p})$ and $u \in \lparts(\Tc^b)[\vec{p'p}] \setminus \reps^b(\vec{p'p})$ settles the claim.
  \end{claimproof}

  Now observe that $S'$ can be constructed explicitly in time $\Oh[\ell]{1}$: it is enough to determine, for each $e_1 \in \{\vec{xp}, \vec{px}, \vec{x^\star p}, \vec{px^\star}, \vec{pp'}, \vec{p'p}\}$ and a~vector $\lin{u}$ in $\BB_{e_1}$ (implicitly represented by a~vertex $u \in \reps^b(e_1)$), and for each $e_2 \in \{\vec{xp}, \vec{x^\star p}, \vec{p'p}\}$ and $v \in \reps^b(e_2)$, the value of $\lin{u}_v$.
  It can be easily observed that this value is equal to $1$ if and only if $e_1$ and $e_2$ point towards each other (i.e., $e_1$ is a~predecessor of $e'_2$, where $e'_2$ is the edge $e_2$ with its head and tail swapped), and $uv \in E(G)$.
  Both of these conditions can be easily verified using the transition signature of $\vec{pp'}$ in $\Tc^b$.
  
  We now run the algorithm of \cref{lem:computing-earliest-basis} to find the lexicographically earliest basis $\BB_{S'}$ of $S'$ in time $\Oh[\ell]{1}$; moreover, this algorithm provides, for each vector $\lin{v} \in S'$, the representation of $\lin{v}$ in this basis.
  Since $S$ and $S'$ are linearly equivalent and $\BB'_p$ is the lexicographically earliest basis of $S$, we can easily recover, for each vector $\lin{v} \in \BB_x$, the representation of $\lin{v}$ in $\BB'_p$.
  These representations form the $|\BB'_p| \times |\BB_x|$ transition matrix $M_{\vec{xp}}$.
  
  \smallskip
  We now briefly discuss the case where $p=r$ is the root of $T^b$ with two children $c_1 < c_2$ (so that $x \in \{c_1, c_2\}$).
  Note that $B_r = \{\lin{0}\}$, so $\BB_r$ is empty; moreover, $B_{c_1} = B_{c_2} = \sumof{\Vc_{c_1}} \cap \sumof{\Vc_{c_2}}$ and hence $\BB'_r = \BB_{c_1}$.
  The implicit representations of both $\BB_{c_1}$ and $\BB_{c_2}$ can be deduced from the edge signatures $\sigma(\Tc^b, \vec{c_1r})$, $\sigma(\Tc^b, \vec{c_2r})$ in time $\Oh[\ell]{1}$ using \cref{lem:canonical-basis-simple}.
  Thus both $M_{\vec{c_1r}}$ -- the identity matrix of dimension $|\BB_{c_1}|$ -- and $M_{\vec{c_2r}}$ -- the transition matrix from the basis $\BB_{c_2}$ to $\BB_{c_1}$ -- can be computed using only $\sigma(\Tc^b, \vec{c_1r})$ and $\sigma(\Tc^b, \vec{c_2r})$ in time $\Oh[\ell]{1}$.
\end{proof}

\paragraph*{Construction of the rank decomposition automaton.}

We have now gathered enough tools to prove the following statement.

\newcommand{\ExactRankAutom}{{\cal JKO}}


\begin{lemma}
  \label{lem:exact-rank-automaton}
  Let $k, \ell \geq 0$ be integers with $k \leq \ell$.
  There exists a~label-oblivious rank decomposition automaton $\ExactRankAutom_{k,\ell} = (Q, \iota, \delta, \varepsilon)$ of width $\ell$ with evaluation time $\Oh[\ell]{1}$ and $|Q| = \Oh[k,\ell]{1}$, called the \emph{exact rankwidth automaton}, with the following properties:
  
  Suppose that $\Tc^b$ is a~rooted annotated rank decomposition of width at most $\ell$ that encodes a~partitioned graph $(G, \prt)$. 
  Let $\rho$ be the run of $\autom$ on $\Tc^b$.
  Then, for every $x \in V(\Tc^b)$, the full set $\fullset^{\BB_x}_k(x)$ at $x$ of width $k$ with respect to $\Tc^b$ is equal to:
  \begin{itemize}
    \item $\rho(\vec{xp})$ if $x$ is not the root of $T^b$ and $p$ is the parent of $x$; or
    \item $\rho(\vartheta)$ if $x$ is the root of $T^b$.
  \end{itemize}
\end{lemma}

\begin{proof}
  Let $Q = 2^{U_k^\ell}$, i.e., every state in $Q$ is a~subfamily of the family of all possible encodings of compact $B$-namus of width at most $k$ in an~ordered basis of size at most $\ell$.
  Note that since $\Tc^b$ has width at most $\ell$, we get that for any node $x \in V(T^b)$ and any ordered basis $\BB_x$ of the boundary space $B_x$, we have $\fullset_k^{\BB_x}(x) \subseteq U_k^\ell$ and so $\fullset_k^{\BB_x}(x) \in Q$.
  
  We define the initial mapping $\iota$ so that, for any leaf edge $\olp \in \leafe(T^b)$, we have that $\rho(\vec{lp}) = \fullset_k^{\BB_l}(l)$.
  By \cref{lem:full-set-ingredients}, $\fullset_k^{\BB_l}(l)$ only depends on the cardinality of $\BB_l$, which can be uniquely deduced from the edge signature $\sigma(\Tc, \vec{lp})$.
  Since $\iota$ accepts a~leaf edge signature as an~argument, such an~initial mapping can be constructed.
  
  The transition mapping $\delta$ is constructed as follows.
  Suppose $x$ is not a~leaf nor a~root of $\Tc^b$ and let $p$ be the parent of $x$ in $\Tc^b$.
  Let also $c_1 < c_2$ be the two children of $x$ in $T^b$.
  Then, we compute $\rho(\vec{xp}) = \fullset_k^{\BB_x}(x)$ as follows.
  By \cref{lem:full-set-ingredients}, $\fullset_k^{\BB_x}(x)$ can be deduced uniquely from $\fullset_k^{\BB_{c_1}}(c_1)$, $\fullset_k^{\BB_{c_2}}(c_2)$, $M_{\vec{c_1x}}$, $M_{\vec{c_2x}}$ and $|\BB_x|$.
  From \cref{lem:transition-matrices} it follows that both $M_{\vec{c_1x}}$ and $M_{\vec{c_2x}}$ can be determined from the transition signature $\tau(\Tc^b, \vec{xp})$.
  Also $|\BB_x|$ can be quickly deduced from the transition signature.
  On the other hand, $\fullset_k^{\BB_{c_1}}(c_1)$ is simply $\rho(\vec{c_1x})$ and $\fullset_k^{\BB_{c_2}}(c_2)$ is $\rho(\vec{c_2x})$.
  So we define the transition mapping $\delta$ so that $\fullset_k^{\BB_x}(x) = \delta(\tau(\Tc^b, \vec{xp}),\, \fullset_k^{\BB_{c_1}}(c_1),\, \fullset_k^{\BB_{c_2}}(c_2))$.
  
  For the final mapping $\varepsilon$, let $x = r$ be the root of $T^b$ with children $c_1 < c_2$.
  Our aim is to determine $\rho(\vartheta) = \fullset_k^{\BB_r}(r)$ from $\rho(\vec{c_1r}) = \fullset_k^{\BB_{c_1}}(c_1)$, $\rho(\vec{rc_1}) = \rho(\vec{c_2r}) = \fullset_k^{\BB_{c_2}}(c_2)$ and the edge signature $\sigma(\Tc^b, \vec{c_1r})$.
  By the definitions of runs of automata on rooted trees, the edge signature $\sigma(\Tc^b, \vec{c_2r})$ is uniquely determined by $\sigma(\Tc^b, \vec{c_1r})$.
  Again by \cref{lem:full-set-ingredients}, $\fullset_k^{\BB_r}(r)$ can be deduced uniquely from $\fullset_k^{\BB_{c_1}}(c_1)$, $\fullset_k^{\BB_{c_2}}(c_2)$, $M_{\vec{c_1r}}$, $M_{\vec{c_2r}}$ and $|\BB_r| = 0$.
  And by \cref{lem:transition-matrices}, $M_{\vec{c_1r}}$ and $M_{\vec{c_2r}}$ can be computed in $\Oh[\ell]{1}$ time given $\sigma(\Tc^b, \vec{c_1r})$ and $\sigma(\Tc^b, \vec{c_2r})$.
  Thus we define $\varepsilon$ so that $\fullset_k^{\BB_r}(r) = \varepsilon(\delta(\Tc^b, \vec{c_1r}), \fullset_k^{\BB_{c_1}}(c_1), \fullset_k^{\BB_{c_2}}(c_2))$.
  
  Since $\iota$, $\delta$ and $\varepsilon$ can be computed from its arguments in time $\Oh[k,\ell]{1}$, the proof is complete.
\end{proof}

Combining \cref{lem:exact-rank-automaton} with \cref{lem:korean-decomposition-recovery}, we immediately obtain the following lemma.

\begin{lemma}
\label{lem:linear-decomposition-recovery}
Let $k, \ell \geq 0$ be integers.
There exists an~algorithm that, given as input an~annotated rank decomposition $\Tc$ of width $\ell$ that encodes a partitioned graph $(G,\prt)$, in time $\Oh[\ell]{|\Tc|}$ either:
\begin{itemize}
\item correctly determines that $(G, \prt)$ has rankwidth larger than $k$; or
\item outputs a (non-annotated) rank decomposition of $(G, \prt)$ of width at most $k$.
\end{itemize}   
\end{lemma}

Which then by combining with \Cref{lem:rearangmain} implies \Cref{lem:near-linear-annotated-decomposition-recovery}, which we restate here.

\nearlineardecompositionrecovery*


%% file: closure-automaton.tex
\subsection{Closure automaton}
\label{ssec:closure-automaton}

\newcommand{\hall}{\mathcal{H}}
\newcommand{\hmap}{\beta}
\newcommand{\hcost}{\mathsf{cost}}

We move on to the description of another rank decomposition automaton -- an~automaton computing possible small closures within the subtrees of a~given rank decomposition.
This automaton, together with $\ExactRankAutom$ from \cref{lem:exact-rank-automaton}, will be used by us in the proof of \cref{lem:closureprds}.
The description below should be considered to be an~analog of a~similar closure automaton for treewidth \cite[Appendix A.2]{dyntw}.
However, this construction of the automaton is noticeably more involved here: In \cite{dyntw}, it was enough to maintain, for each subtree $T$ of the decomposition, a~bounded-size family of small subsets of $V(G)$ (so the description of each subtree $T$ simply had bounded size and could be manipulated explicitly).
Here, given an~annotated rank decomposition $\Tc$ of $G$, we will need to store, for each edge $\oxp \in \oE(\Tc)$, a~bounded-size family of \emph{partitions} of $V(\Tc)[\oxp]$ into a~small number of subsets.
Since we cannot store the partitions of $V(\Tc)[\oxp]$ explicitly in an~efficient manner, we first need to roll out a~way of encoding such partitions succinctly.
Intuitively, given a~partition $\prt$ of $V(\Tc)[\oxp]$, we want to select from each set $C \in \prt$ a~minimal representative $R_C$ of $C$ and encode the connections between $R_C$ and $\overline{C}$ in $G$.
The details follow below.

Let $X$ be a~nonempty finite set.
We define an~\emph{indexed partition} of $X$ as any sequence $\C = (X_1, \ldots, X_c)$ of (possibly empty) pairwise disjoint subsets of $X$ with $X_1 \cup \ldots \cup X_c = X$.
Then $\C$ is said to \emph{represent} the (non-indexed) partition $\prt = \{X_1, \ldots, X_c\} \setminus \{\emptyset\}$ of $X$.

Next, fix $c \in \N$.
We say that a~triple $\H = ((V_1, \dots, V_c), H, \eta)$ is a~\emph{$(c,X)$-indexed graph} if:
\begin{itemize}
  \item $H$ is an~undirected graph,
  \item $(V_1, \dots, V_c)$ is an~indexed partition of $V(H)$; 
  \item for every $i \in [c]$, the subgraph $H[V_i]$ is edgeless; and
  \item $\eta \,\colon\,V(H) \to X$ is a~labeling function.
\end{itemize}
Given a~$(c,X)$-indexed graph $\H = ((V_1, \dots, V_c), H, \eta)$, we define the \emph{derived} partitioned graph $(H, \Dc)$ by setting $\Dc = \{V_1, \dots, V_c\} \setminus \{\emptyset\}$.
Also, for convenience, define $V(\H) \coloneqq V(H)$, $E(\H) \coloneqq E(H)$, $G(\H) \coloneqq H$, $V_i(\H) \coloneqq V_i$ and $\eta(\H) \coloneqq \eta$.

Two $(c,X)$-indexed graphs $\H_1 = ((V^1_1, \dots, V^1_c), H_1, \eta_1)$, $\H_2 = ((V^2_1, \dots, V^2_c), H_2, \eta_2)$ are isomorphic (denoted $\H_1 \sim^{c,X} \H_2$) if there exists an~isomorphism $\pi\,\colon\, V(H_1) \to V(H_2)$ from $H_1$ to $H_2$ such that: (i) $\pi(V^1_i) = V^2_i$ for all $i \in [c]$, and (ii) $\eta_1(v) = \eta_2(\pi(v))$ for all $v \in V(H_1)$.

For $s \in \N$, we say that $\H = ((V_1, \dots, V_c), H, \eta)$ is \emph{$s$-small} if for every $i \in [c]$ and $x \in X$, we have $|V_i \cap \eta^{-1}(x)| \leq s$; i.e., each subset $V_i$ contains at most $s$ vertices of any given label.
Thus if $\H$ is an~$s$-small $(c,X)$-indexed graph, then $|V(\H)| \leq c s |X|$.
Note that the property of $s$-smallness of indexed graphs is preserved by isomorphism.
Hence we define $\sim^{c,X}_s$ as the restriction of $\sim^{c,X}$ to only the classes containing $s$-small indexed graphs.
It is easy to see that $\sim^{c,X}_s$ has $\Oh[c, |X|, s]{1}$ distinct equivalence classes.

Now suppose that a~graph $G$ is encoded by an~annotated rank decomposition $\Tc = (T, U, \reps, \repse, \dmap)$ of width $\ell$ and let $\oxp \in \oE(T)$.
Recall that $\lparts(\Tc)[\oxp]$ comprises the vertices of $G$ assigned to the leaf edges of $T$ that are closer to $x$ than $p$, and that $\reps(\oxp)$ is a~minimal representative of $\lparts(\Tc)[\oxp]$ in $G$.
We say that a~$(c,\reps(\oxp))$-indexed graph $\H = ((V_1, \dots, V_c), H, \eta)$ \emph{respects $\Tc$ along $\oxp$} if:
\begin{itemize}
  \item $H = G[\{V_1, \ldots, V_c\}]$; and
  \item for each $v \in V(H)$, the label $\eta(v)$ is the unique vertex in $\reps(\oxp)$ so that $N_G(v) \cap \lparts(\Tc)[\opx] = N_G(\eta(v)) \cap \lparts(\Tc)[\opx]$.
\end{itemize}

Observe that if the graph $G$ and the decomposition $\Tc$ is fixed, then both the graph $H$ and the labeling function $\eta$ of an~indexed graph respecting $\Tc$ along $\oxp$ only depend on the choice of the sets $V_1, \dots, V_c$.

Assuming $\H = ((V_1, \dots, V_c), H, \eta)$ respects $\Tc$ along $\oxp$, we say that it \emph{encodes} an~indexed partition $\C = (X_1, \dots, X_c)$ of $\lparts(\Tc)[\oxp]$ if $V_i$ is a~minimal representative of $X_i$ in $G$ for each $i \in [c]$.
It is straightforward to see that all indexed graphs encoding $\C$ are pairwise isomorphic: For each $i$ the collection of neighborhoods $\{N(v) \setminus X_i\}_{v \in X_i}$ is uniquely determined by $X_i$, so $V_i$ contains one vertex $v \in X_i$ for each distinct neighborhood $N(v) \setminus X_i$; and the resulting indexed graph is the same up to isomorphism regardless of the choice of $v$.
Also, we say that $\H$ encodes a~partition $\prt$ if $\H$ encodes some indexed partition $\C$ representing~$\prt$.



If $\Cc$ is a~partition of $\lparts(\Tc)[\oxp]$, then we define its \emph{cost} to be the number of nodes in the subtree rooted at $\oxp$ that are cut by $\prt$; i.e., the number of oriented edges $\vec{e}$ that are predecessors of $\oxp$ in $T$ such that $\lparts(\Tc)[\vec{e}]$ intersects more than one set of $\Cc$.
We similarly define the cost of indexed partitions of $\lparts(\Tc)[\oxp]$.


Finally, for every equivalence class $\Kc$ of $\sim^{c,\reps(\oxp)}_s$, let $A_\Kc$ be the set of pairs $(q, \H)$, where $\H \in \Kc$ is an~$s$-small $(c,\reps(\oxp))$-indexed graph encoding some partition $\prt$ of $\lparts(\Tc)[\oxp]$ of cost $q$.
Then we say that a~set $F$ is a \emph{set of $(c,s)$-small representatives of $\Tc$ along $\oxp$} if, for every equivalence class $\Kc$ of $\sim^{c,\reps(\oxp)}_s$ with $A_\Kc \neq \emptyset$, $F$ contains a~single pair $(q, \H) \in A_\Kc$ with the minimum cost $q$.
%
%
Note that the cardinality of $F$ is bounded by the number of equivalence classes $\sim^{c,\reps(\oxp)}_s$, which is bounded by $\Oh[c,s,\ell]{1}$.

Our aim is now to prove that a~rank decomposition automaton can compute, for each edge $\oxp$, some set of $(c, s)$-small representatives of $\Tc$ along $\oxp$ -- which we will call $\AutomataReps^{c,s}(\Tc, \oxp)$ from now on -- and additional annotations allowing us to efficiently recover, for each $(q, \H) \in \AutomataReps^{c,s}(\Tc, \oxp)$, an~indexed partition of $\lparts(\Tc)[\oxp]$ of cost $q$ encoded by $\H$.

\newcommand{\BoundRepr}{\mathcal{CR}}


\begin{lemma}
  \label{lem:bound-repr-automaton}
  For every triple of non-negative integers $c, s, \ell$, there exists a~label-oblivious rank decomposition automaton $\BoundRepr = \BoundRepr_{c,s,\ell}$ with evaluation time $\Oh[c, s, \ell]{1}$ with the following property.
  Suppose $G$ is a~graph encoded by an annotated rank decomposition $\Tc = (T, U, \reps, \repse, \dmap)$ of width at most $\ell$.
  Then the run $\rho$ of $\BoundRepr$ on $\Tc$ satisfies that for every $\oxp \in \oE(\Tc)$,
  \[ \rho(\oxp) = (\AutomataReps^{c,s}(\Tc, \oxp), \Phi), \]
  where $\AutomataReps^{c,s}(\Tc, \oxp)$ is a set of $(c,s)$-small representatives of $\Tc$ along $\oxp$, and $\Phi$ is a~mapping from $\AutomataReps^{c,s}(\Tc, \oxp)$ such that:
  \begin{itemize}
    \item if $\oxp$ is a~leaf oriented edge, then $\Phi$ maps each pair in $\AutomataReps^{c,s}(\Tc, \oxp)$ to $\bot$; and
    \item if $\oxp$ is a~non-leaf oriented edge, where $\oxp$ has two children $\vec{y_1x}$ and $\vec{y_2x}$, then for every $(q, \H) \in \AutomataReps^{c,s}(\Tc, \oxp)$, we have $\Phi((q, \H)) = ((q_1, \H_1), (q_2, \H_2))$ such that:
    \begin{itemize}
      \item $(q_t, \H_t) \in \AutomataReps^{c,s}(\Tc, \vec{y_t x})$ for each $t \in [2]$;
      \item for every indexed partition $(X^1_1, \dots, X^1_c)$ of $\lparts(\Tc)[\vec{y_1 x}]$ of cost $q_1$ encoded by $\H_1$, and every indexed partition $(X^2_1, \dots, X^2_c)$ of $\lparts(\Tc)[\vec{y_2 x}]$ of cost $q_2$ encoded by $\H_2$, the indexed partition $(X^1_1 \cup X^2_1, \dots, X^1_c \cup X^2_c)$ of $\lparts(\Tc)[\oxp]$ has cost $q$ and is encoded by~$\H$.
    \end{itemize}
  \end{itemize}
\end{lemma}
\begin{proof}
  We need to implement the following two procedures:
  \begin{itemize}
    \item for a~leaf oriented edge $\olp$ of $T$ with edge signature $\sigma(\Tc, \olp)$, determine $\AutomataReps^{c,s}(\Tc, \olp)$; and
    \item for a~non-leaf oriented edge $\oxp$ of $T$ where $\oxp$ has two children $\vec{y_1x}$, $\vec{y_2x}$, find $\AutomataReps^{c,s}(\Tc, \oxp)$ and the mapping $\Phi$ as in the statement of the lemma, given $\AutomataReps^{c,s}(\Tc, \vec{y_1x})$, $\AutomataReps^{c,s}(\Tc, \vec{y_2x})$ and the transition signature $\tau(\Tc, \oxp)$.
    Here we inductively assume that for $t \in [2]$, $\AutomataReps^{c,s}(\Tc, \vec{y_tx})$ is a~set of $(c,s)$-small representatives of $\Tc$ along $\vec{y_tx}$.
  \end{itemize}
  
  First, for a~leaf edge $\olp$, observe that $|\lparts(\Tc)[\olp]| = 1$ and the only vertex $v \in \lparts(\Tc)[\olp]$ can be read from the edge signature $\sigma(\Tc, \olp)$.
  Thus there exist exactly $c$ non-isomorphic $(c, \reps(\olp))$-indexed graphs $\H_1, \dots, \H_c$ respecting $\Tc$ along $\olp$ and encoding an~indexed partition of $\lparts(\Tc)[\olp]$: For each $i \in [c]$, the indexed graph $\H_i$ is defined by the sequence of sets $(V^i_1, \dots, V^i_c)$, where $V^i_i = \{v\}$ and $V^i_j = \emptyset$ for $j \neq i$. Moreover, $E(\H_i) = \emptyset$ and $\eta(\H_i)(v) = v$.
  Naturally, each $\H_i$ encodes a~partition of $\lparts(\Tc)[\olp]$ of cost $0$.
  Hence $\AutomataReps^{c,s}(\Tc, \olp)$ can be enumerated by brute force in time $\Oh[c,s,\ell]{1}$.
  
  Now assume $\oxp$ is a~non-leaf oriented edge and let $\vec{y_1x}$ and $\vec{y_2x}$ be the two children of $\oxp$.
  For convenience, define $S_1 = \lparts(\Tc)[\vec{y_1x}]$, $S_2 = \lparts(\Tc)[\vec{y_2x}]$, and $S = \lparts(\Tc)[\oxp]$; we have that $S_1 \cap S_2 = \emptyset$ and $S = S_1 \cup S_2$.	
  
  \newcommand{\Combine}{\mathsf{Combine}}
  We now define a~function $\Combine(\H_1, \H_2)$, taking as arguments a~$(c,\reps(\vec{y_1x}))$-indexed graph $\H_1$ respecting $\Tc$ along $\vec{y_1x}$, and a~$(c,\reps(\vec{y_2x}))$-indexed graph $\H_2$ respecting $\Tc$ along $\vec{y_2x}$ and returning a~$(c,\reps(\oxp))$-indexed graph $\H$ respecting $\Tc$ along $\oxp$ as follows.
  Let us denote the input graphs by $\H_1 = ((V^1_1, \dots, V^1_c), H_1, \eta_1)$ and $\H_2 = ((V^2_1, \dots, V^2_c), H_2, \eta_2)$.
  Note that $V(H_1) \cap V(H_2) = \emptyset$.
  Define an~auxiliary $(c, \reps(\oxp))$-indexed graph $\H' = ((V'_1, \dots, V'_c), H', \eta')$ as follows:
  
  \begin{itemize}
    \item $V'_i = V^1_i \cup V^2_i$ for each $i \in [c]$;
    \item $V(H') = V(H_1) \cup V(H_2)$;
    \item $H'[V(H_1)] = H_1$ and $H'[V(H_2)] = H_2$;
    \item for $u \in V(H_1)$ and $v \in V(H_2)$, we have $uv \in E(H')$ if and only if $u$, $v$ do not belong to the same set $V'_i$ and moreover $\eta_1(u)\eta_2(v) \in E(G)$; and
    \item for $t \in [2]$ and $v \in V(H_t)$, we have $\eta'(v) = \dmap(y_i x p)(\eta(v))$.
  \end{itemize}
  
  A~verification with the definitions shows that $\H'$ respects $\Tc$ along $\oxp$.
  In particular, whenever $i \in [c]$ and $u, v \in V_i(\H')$ with $\eta'(u) = \eta'(v)$, we have that $N_G(u) \cap \lparts(\Tc)[\opx] = N_G(v) \cap \lparts(\Tc)[\opx]$.
  Also, $\Combine(\H_1, \H_2)$ can be constructed given $\H_1$ and $\H_2$ using only the transition signature $\tau(\Tc, \oxp)$. In particular, for $u \in V(H_1), v \in V(H_2)$, we have $\eta_1(u) \in \reps(\vec{y_1x})$, $\eta_2(v) \in \reps(\vec{y_2x})$, so whether $\eta_1(u)\eta_2(v) \in E(G)$ depends only on $\tau(\Tc, \oxp)$.
  
  Then $\H$ is constructed from $\H'$ as follows.
  We begin with $\H = \H'$. Whenever there is an~index $i \in [c]$ and two vertices $u, v \in V_i(\H)$ such that $N_{H'}(u) = N_{H'}(v)$ and $\eta'(u) = \eta'(v)$, we remove one of the vertices from $V_i(\H)$ (and therefore $\H$).
  
  
  We now prove a string of properties of $\Combine$:
  \begin{claim}
    \label{obs:combine-merges}
    Whenever $\H_1$ encodes an~indexed partition $(X^1_1, \dots, X^1_c)$ of $S_1$ and $\H_2$ encodes an~indexed partition $(X^2_1, \dots, X^2_c)$ of $S_2$, then $\Combine(\H_1, \H_2)$ encodes the indexed partition $(X^1_1 \cup X^2_1, \dots, X^1_c \cup X^2_c)$ of $S$.
  \end{claim}
  \begin{claimproof}
    Take $\H_1 = ((V^1_1, \ldots, V^1_c), H_1, \eta_1)$, $\H_2 = ((V^2_1, \ldots, V^2_c), H_2, \eta_2)$ and $\H \coloneqq \Combine(\H_1, \H_2) = ((V_1, \ldots, V_c), H, \eta)$.
    Let also $\H' = ((V'_1, \ldots, V'_c), H', \eta')$ be the auxiliary graph in the definition of $\Combine$.
    Since $\H'$ respects $\Tc$ along $\oxp$ and $\H$ is an~induced subgraph of $\H'$ (i.e., $V_i \subseteq V'_i$ for all $i \in [c]$), we find that also $\H$ respects $\Tc$ along $\oxp$.
    Finally define $X_i = X^1_i \cup X^2_i$ for $i \in [p]$.
    
    First consider two vertices $u, v \in V(H')$ with $u \in V'_i$, $v \in V'_j$ and $i \neq j$.
    We will show that $uv \in E(H')$ if and only if $uv \in E(G)$.
    If $u \in V^t_i$ and $v \in V^t_j$ for some $t \in [2]$, this follows from the fact that $\H_t$ respects $\Tc$ along $\vec{c_tx}$: We have $H_t = G[\{V^t_1, \ldots, V^t_c\}]$, so $uv \in E(H_t)$ if and only if $uv \in E(G)$. Then the statement follows from $H'[V(H_t)] = H_t$.
    On the other hand, if $u \in V^1_i$ and $v \in V^2_j$, then by construction we have placed an edge $uv \in E(H')$ if and only if $\eta_1(u)\eta_2(v) \in E(G)$.
    Then observe that $\eta_1(u)$ is defined so that $N_G(u) \cap \lparts(\Tc)[\vec{xy_1}] = N_G(\eta_1(u)) \cap \lparts(\Tc)[\vec{xy_1}]$, and $\eta_2(v)$ is defined similarly: $N_G(v) \cap \lparts(\Tc)[\vec{xy_2}] = N_G(\eta_2(v)) \cap \lparts(\Tc)[\vec{xy_2}]$.
    Since $u, \eta_1(u) \in \lparts(\Tc)[\vec{xy_2}]$ and $v, \eta_2(v) \in \lparts(\Tc)[\vec{xy_1}]$, we get that $uv \in E(G)$ if and only if $\eta_1(u)\eta_2(v) \in E(G)$.
    The statement follows.
    
    Then pick $i \in [c]$.
    We ought to show that $V_i$ is a~minimal representative of $X_i$ in $G$.
    Let $t \in [2]$ and $v \in X^t_i$.
    Since $\H_t$ encodes an~indexed partition $(X^t_1, \ldots, X^t_c)$ of $S_t$, there is $u \in V^t_i$ with $N_G(v) \setminus X^t_i = N_G(u) \setminus X^t_i$, and $u \in V'_i$ by construction.
    Also by construction, there exists $u' \in V_i$ such that $\eta'(u) = \eta'(u')$ and $N_{H'}(u) = N_{H'}(u')$.
    We have:
    \begin{itemize}
      \item $N_G(u) \cap \lparts(\Tc)[\vec{px}] = N_G(u') \cap \lparts(\Tc)[\vec{px}]$ (since $\eta'(u) = \eta'(u')$),
      \item $N_G(u) \cap X_j = N_G(u') \cap X_j$ for all $j \neq i$: Let $w \in X_j$.
      Pick $t' \in [2]$ for which $w \in X^{t'}_j$.
      Since $V^{t'}_j$ is a~minimal representative of $X^{t'}_j$ in $G$, there is some $w' \in V^{t'}_j$ such that $N_G(w) \setminus X^{t'}_j = N_G(w') \setminus X^{t'}_j$.
      Since $uw' \in E(H') \Leftrightarrow u'w' \in E(H')$, we have $uw' \in E(G) \Leftrightarrow u'w' \in E(G)$ by the considerations above.
      As $u, u' \notin X^{t'}_j$, we conclude that
      \[ uw \in E(G) \,\Leftrightarrow\, uw' \in E(G) \,\Leftrightarrow\, u'w' \in E(G) \,\Leftrightarrow\, u'w \in E(G). \]
    \end{itemize}
    
    So $N_G(u') \setminus X_i = N_G(u) \setminus X_i = N_G(v) \setminus X_i$, where the last equality follows from $N_G(u) \setminus X^t_i = N_G(v) \setminus X^t_i$.
    It follows that $u' \in V_i$ represents $v$ in $X_i$.
    As $v$ was arbitrary, we conclude that $V_i$ is a~representative of $X_i$.
    
    For minimality, observe that if $u, v \in V'_i$ with $N_G(u) \setminus X_i = N_G(v) \setminus X_i$, then also $\eta'(u) = \eta'(v)$ (since $N_G(u) \cap \lparts(\Tc)[\oxp] = N_G(v) \cap \lparts(\Tc)[\oxp]$) and $N_{H'}(u) = N_{H'}(v)$ (since $N_G(u) \cap X_j = N_G(v) \cap X_j$ for $j \neq i$).
    Thus the construction of $\H$ from $\H'$ would remove either $u$ or $v$ from the graph.
  \end{claimproof}  

  Next, $\Combine$ preserves isomorphism in the following sense:
  \begin{claim}
    \label{obs:combine-isomorphic}
    For each $t \in [2]$, suppose that $\H_t$ and $\H^\star_t$ are $(c, \reps(\vec{y_tx}))$-indexed graphs respecting $\Tc$ along $\vec{y_tx}$ such that $\H_t \sim^{c,\reps(\vec{y_tx})} \H^\star_t$.
    Then $\Combine(\H_1, \H_2) \sim^{c, \reps(\oxp)} \Combine(\H^\star_1, \H^\star_2)$.
  \end{claim}
  \begin{claimproof}
    Let $\H \coloneqq \Combine(\H_1, \H_2)$ and $\H'$ be the auxiliary graph in the construction of~$\H$.
    Likewise, let $\H^\star \coloneqq \Combine(\H^\star_1, \H^\star_2)$ and $(\H^\star)'$ be the auxiliary graph in the construction of~$\H^\star$.
    Also let $\pi_1 \,:\, V(\H_1) \to V(\H^\star_1)$, $\pi_2 \,:\, V(\H_2) \to V(\H^\star_2)$ be the isomorphisms promised by the statement of the claim.
    
    Observe that $\pi' \,:\, V(\H') \to V((\H^\star)')$ given by $\funrestriction{\pi'}{V(\H_1)} = \pi_1$ and $\funrestriction{\pi'}{V(\H_2)} = \pi_2$ is an~isomorphism between $\H'$ and $(\H^\star)'$: This holds since for each $i \in [c]$, we have $\pi(V_i(\H')) = \pi(V_i(\H_1) \cup V_i(\H_2)) = \pi_1(V_i(\H_1)) \cup \pi_2(V_i(\H_2)) = V_i(\H^\star_1) \cup V_i(\H^\star_2) = V_i((\H^\star)')$ and, for each $v \in V(\H_t)$ with $t \in [2]$,
    $\eta(\H')(v) = \dmap(y_t xp)\left(\eta(\H_t)(v)\right) = \dmap(y_t xp)\left(\eta(\H^\star_t)(\pi_t(v))\right) = \eta((\H^\star)')(\pi_t(v))$.
    Also the same arguments as in \cref{obs:combine-merges} show that $\pi'$ gives an isomorphism of the graphs $G(\H')$ and $G((\H^\star)')$.
    
    Since $\H'$ and $(\H^\star)'$ are isomorphic, it can be easily verified that the process of the construction of $\H$ from $\H'$ and $\H^\star$ from $(\H^\star)'$ preserves isomorphism.
    This finishes the proof.
%
%
  \end{claimproof}
  
  The following claim follows from a~simple application of \cref{obs:combine-merges}.
  \begin{claim}
    \label{cl:reverse-combine}
    Suppose $\H$ encodes an~indexed partition $(X_1, \dots, X_c)$ of $S$.
    Then there exists a~$(c, \reps(\vec{y_1 x}))$-indexed graph $\H_1$ and a~$(c, \reps(\vec{y_2 x}))$-indexed graph $\H_2$ such that:
    \begin{itemize}
      \item for each $t \in [2]$, $\H_t$ encodes the indexed partition $(X_1 \cap S_t, \dots, X_c \cap S_t)$ of $S_t$; and
      \item $\Combine(\H_1, \H_2) \sim^{c, \reps(\vec{xp})} \H$.
    \end{itemize}
  \end{claim}
  \begin{claimproof}
    For $t \in [2]$, let $\H_t$ be any $(c, \reps(\vec{y_t x}))$-indexed graph encoding the indexed partition $(X_1 \cap S_t, \ldots, X_c \cap S_t)$.
    Such an~indexed graph must exist since it is enough to take $\H_t = (V^t_1, \ldots, V^t_c), H_t, \eta_t)$, where for $i \in [c]$, $V^t_i$ is any minimal representative of $X_i \cap S_t$ in $G$, and the objects $H_t$, $\eta_t$ are uniquely deduced from $V^t_1, \ldots, V^t_c$.
%
    
    By \cref{obs:combine-merges}, $\Combine(\H_1, \H_2)$ encodes $(X_1, \dots, X_c)$.
    Since indexed graphs encoding the same indexed partition of $S$ are isomorphic, we conclude that $\Combine(\H_1, \H_2) \sim^{c,\reps(\oxp)} \H$.
  \end{claimproof}
  
  We also notice the following claim binding the cost of an~indexed partition of $S$ to the costs of indexed partitions of $S_1, S_2$:
  \begin{claim}
    \label{cl:merging-costs}
    Let $\H$ be a $(c, \reps(\oxp))$-indexed graph encoding an indexed partition $(X_1, \dots, X_c)$ of $S$ of cost $q \in \N$, and for each $t \in [2]$, $(X_1 \cap S_t, \dots, X_c \cap S_t)$ be an~indexed partition of $S_t$ of cost $q_i \in \N$.
    Let $\delta \in \{0, 1\}$ be the indicator equal to $1$ if and only if $\H$ has at least two nonempty parts (equivalently, at least two sets $X_1, \dots, X_c$ are nonempty).
    Then $q = q_1 + q_2 + \delta$.
  \end{claim}
  \begin{claimproof}
    Recall that $q$ is the number of oriented edges $\vec{e}$ that are predecessors of $\oxp$ such that $\lparts(\Tc)[\vec{e}]$ intersects more than one set in $(X_1, \dots, X_c)$.
    Noting that the two edges $\vec{y_1 x}$ and $\vec{y_2 x}$ are the two children of $\oxp$, we see that $q$ is the sum of the following values:
    \begin{itemize}
      \item for each $t \in [2]$, the number of predecessors $\vec{e}$ of $\vec{y_t x}$ such that $\lparts(\Tc)[\vec{e}]$ intersects at least two sets in $(X_1, \dots, X_c)$.
      Since $\lparts(\Tc)[\vec{e}] \subseteq S_t$, this is equivalently the number of predecessors $\vec{e}$ of $\vec{y_t x}$ with $\lparts(\Tc)[\vec{e}]$ intersecting at least two sets in $(X_1 \cap S_t, \dots, X_c \cap S_t)$, or exactly $q_t$; and
      \item $1$ if $\lparts(\Tc)[\opx] = S$ intersects at least two sets in $(X_1, \dots, X_c)$, or $0$ otherwise; equivalently, this indicator is equal to $1$ if and only if at least two sets in $(X_1, \dots, X_c)$ are nonempty.
      Since for each $i \in [c]$, the set $X_i$ is nonempty if and only if $V_i(\H)$ is nonempty, this indicator is equal to exactly $\delta$.
    \end{itemize}
    Therefore, $q = q_1 + q_2 + \delta$.
  \end{claimproof}
  
  Finally, the following claim will enable us to compute $\AutomataReps^{c,s}(\Tc, \oxp)$.
  \begin{claim}
    \label{cl:partition-composition}
    Let $\Kc$ be an~equivalence class of $\sim^{c,\reps(\oxp)}_s$ and suppose $(q, \H) \in A_\Kc$ has the minimum possible cost $q$ among all pairs in $A_\Kc$.
    Let $\delta \in \{0,1\}$ be an~indicator equal to $0$ if $\H$ has at most one nonempty part, and $1$ otherwise.
    Then there exist pairs $(q_1, \H_1) \in \AutomataReps^{c,s}(\Tc, \vec{y_1 x})$ and $(q_2, \H_2) \in \AutomataReps^{c,s}(\Tc, \vec{y_2 x})$ such that
    \[ \begin{split}
      q &= q_1 + q_2 + \delta, \\
      \H &\sim^{c,\reps(\oxp)} \Combine(\H_1, \H_2).
    \end{split} \]
  \end{claim}  
  \begin{claimproof}
    Let $(q, \H)$ and $\delta \in \{0, 1\}$ be as in the statement of the claim.
    By definition, $\H$ is an~$s$-small $(c, \reps(\oxp))$-indexed graph and there exists an~indexed partition $(X_1, \dots, X_c)$ of $S$ of cost $q$ encoded by $\H$.
    For $t \in [2]$ and $j \in [c]$, define $X^t_j = X_j \cap S_t$, so that $(X^t_1, \dots, X^t_c)$ is an~indexed partition of $S_t$; let then $q_t \in \N$ be the cost of this partition.
    Then $q = q_1 + q_2 + \delta$ by \cref{cl:merging-costs}.
    
    Let $\H_1, \H_2$ be $(c, \reps(\vec{y_1 x}))$-indexed and $(c, \reps(\vec{y_2 x}))$-indexed, respectively, graphs with the properties that $\Combine(\H_1, \H_2) \sim^{c,\reps(\oxp)} \H$ and for each $t \in [2]$, $\H_t$ encodes $(X^t_1, \dots, X^t_c)$.
    Note that such indexed graphs exist by \cref{cl:reverse-combine}.
    For each $t \in [2]$, we claim that $\H_t$ is $s$-small.
    Suppose otherwise; let $\H_t = ((V^t_1, \ldots, V^t_c), H_t, \eta_t)$ so that $V^t_j$ is a~minimal representative of $X^t_j$ for all $j \in [c]$.
    Then there is some index $j \in [c]$ and $s+1$ vertices $v_1, \dots, v_{s+1} \in V^t_j$ such that:
    \begin{itemize}
      \item $\eta_t(v_1) = \ldots = \eta_t(v_{s+1})$,
      \item the neighborhoods $N_G(v_1) \cap \compl{X^t_j}, \dots, N_G(v_{s+1}) \cap \compl{X^t_j}$ are pairwise different.
    \end{itemize}
    From $\eta_t(v_1) = \ldots = \eta_t(v_{s+1})$ and $\H_t$ respecting $\Tc$ along $\vec{y_t x}$, we also have $N_G(v_1) \cap \compl{S_t} = \ldots = N_G(v_{s+1}) \cap \compl{S_t}$.
    Since $X^t_j \subseteq S_t$, we infer that all the neighborhoods $N_G(v_1) \cap (S_t \setminus X^t_j), \dots, N_G(v_{s+1}) \cap (S_t \setminus X^t_j)$ are pairwise different.
    So we have that:
    \begin{itemize}
      \item $v_1, \dots, v_{s+1} \in X_j$,
      \item $N_G(v_1) \cap \compl{X_j}, \dots, N_G(v_{s+1}) \cap \compl{X_j}$ are pairwise different (since $S_t \setminus X^t_j \subseteq \compl{X_j}$), and
      \item $\eta(v_1) = \ldots = \eta(v_{s+1})$ (since $N_G(v_1) \cap \compl{S} = \ldots = N_G(v_{s+1}) \cap \compl{S}$).
    \end{itemize}
    Therefore, any minimal representative of $X_j$ in $G$ must contain at least $s + 1$ vertices with the same neighborhood $N_G(v_1) \cap \compl{S}$ in $\compl{S}$.
    This implies that $V_j(\H)$ must contain at least $s + 1$ vertices labeled $\eta(v_1)$ -- a~contradiction since we assumed $\H$ is $s$-small.
    So $\H_t$ is indeed $s$-small.
    
    For each $t \in [2]$, $\H_t$ encodes the indexed partition $(X^t_1, \ldots, X^t_c)$ of cost $q_t$.
    Thus there is a~pair $(q^\star_t, \H^\star_t) \in \AutomataReps^{c,s}(\Tc, \vec{y_1 x})$ and $(q^\star_2, \H^\star_2) \in \AutomataReps^{c,s}(\Tc, \vec{y_2 x})$ such that $q^\star_t \leq q_t$ and $\H^\star_t \sim^{c,\reps(\oxp)} \H_t$ for each $t \in [2]$.
    Let also, for each $t \in [2]$, $(Y^t_1, \ldots, Y^t_c)$ be an~indexed partition of $S_t$ of cost $q^\star_t$ encoded by $\H^\star_t$.    
    Then take $\H^\star \coloneqq \Combine(\H^\star_1, \H^\star_2)$.
    By \cref{obs:combine-isomorphic}, $\H^\star \sim^{c, \reps(\oxp)} \H$; in particular, $\H^\star$ has at most one nonempty part if and only if $\H$ does.
    By \cref{obs:combine-merges}, $\H^\star$ encodes the indexed partition $(Y_1, \ldots, Y_c)$, where $Y_i = Y^1_i \cup Y^2_i$ for $i \in [c]$.
    This partition has cost $q^\star_1 + q^\star_2 + \delta$ by \cref{cl:merging-costs}, so $(q^\star_1 + q^\star_2 + \delta,\, \H^\star) \in A_\Kc$.
%
    But since $(q, \H)$ has the minimum cost among all pairs in $A_\Kc$, we get
    \[ q \leq q^\star_1 + q^\star_2 + \delta \leq q_1 + q_2 + \delta = q. \]
    Therefore, $q^\star_1 = q_1$ and $q^\star_2 = q_2$ and thus $q = q_1 + q_2 + \delta$ for $(q_1, \H^\star_1) \in \AutomataReps^{c,s}(\Tc, \vec{y_1x})$, $(q_2, \H^\star_2) \in \AutomataReps^{c,s}(\Tc, \vec{y_2x})$ and $\H \sim^{c, \reps(\oxp)} \Combine(\H^\star_1, \H^\star_2)$.
%
%
  \end{claimproof}
  
  Therefore, we compute the set $\AutomataReps^{c,s}(\Tc, \oxp)$ as follows.
  We populate a~set $\Wc$ comprising pairwise different pairs containing a~non-negative integer and a~$(c,\reps(\oxp))$-indexed graph by:
  \begin{itemize}
    \item iterating all pairs $(q_1, \H_1) \in \AutomataReps^{c,s}(\Tc, \vec{y_1 x})$ and $(q_2, \H_2) \in \AutomataReps^{c,s}(\Tc, \vec{y_2 x})$,
    \item computing $\H = \Combine(\H_1, \H_2)$ and $q = q_1 + q_2 + \delta$, where $\delta = 1$ if $\H$ contains at least two nonempty parts, and $\delta = 0$ otherwise, and
    \item if $\H$ is $s$-small, adding a pair $(q, \H)$ to $\Wc$.
  \end{itemize}
  Then we filter $\Wc$ as follows: whenever $\Wc$ contains pairs $(q, \H)$ and $(q', \H')$ such that $q \leq q'$ and $\H \sim^{c, \reps(\oxy)} \H'$, we drop $(q', \H')$ from $\Wc$.
  Naturally, this entire process (the construction of $\Wc$ and its subsequent filtering) can be carried out in time $\Oh[c, s, \ell]{1}$.
  We finally set $\AutomataReps^{c,s}(\Tc, \oxp) \coloneqq \Wc'$.
  Naturally, by \cref{cl:partition-composition}, $\Wc'$ is a~set of $(c,s)$-small representatives of $\Tc$ along $\oxp$.

  
  We conclude the proof by observing that, for every $(q, \H) \in \AutomataReps^{c,s}(\Tc, \oxp)$, we can define the mapping $\Phi((q, \H))$ as any pair $((q_1, \H_1), (q_2, \H_2))$ for which $q = q_1 + q_2 + \delta$ and $\H = \Combine(\H_1, \H_2)$, where $\delta = 1$ if and only if $\H$ contains at least two nonempty parts. (Such a~pair exists by the construction of $\Wc$.)
  Then $\Phi((q, \H))$ satisfies all the requirements of the lemma by \cref{obs:combine-merges,cl:merging-costs}.
\end{proof}

%% file: conclusion.tex
\section{Conclusions}
\label{sec:concl}
We gave a data structure for maintaining bounded-width rank decompositions of dynamic graphs of bounded rankwidth in subpolynomial time per update.
We also used this data structure to give an almost-linear time parameterized algorithm for computing an optimum-width rank decomposition of a given graph.
Along the way, we proved several auxiliary structural and algorithmic results for rankwidth.
An important conceptual contribution of our work appears to be the definition of annotated rank decompositions, together with the efficient algorithms for manipulating them and for translating dynamic programming from other representations of rank decompositions to annotated rank decompositions.
We then discuss future research directions and make some additional remarks about our results.

The obvious interesting open problem is to improve the dynamic algorithm of \Cref{the:dynsimple} to work in $\Oh[k]{\log^{\Oh{1}} n}$ time per update, instead of the current $2^{\Oh[k]{\sqrt{\log n \log \log n}}}$ time.
This would also improve the algorithm of \Cref{the:altrw} to $\Oh[k]{n \log^{\Oh{1}} n} + \Oh{m}$ time.
The same problem is open for dynamic treewidth, so the natural path to solve it would be to first improve the dynamic treewidth algorithm of~\cite{dyntw}, and then generalize the result to rankwidth.
However, we note that the tools developed in \Cref{sec:refi} appear to give a cleaner and more elegant framework for dynamic rankwidth than the framework for dynamic treewidth of~\cite{dyntw} is, so it could make sense to approach dynamic treewidth via dynamic rankwidth, or perhaps via dynamic branchwidth.

In \Cref{the:dynfull} we gave a framework for applying edge updates defined by $\CMSO_1$ sentences.
In this framework, the time required to apply the update is at least linear in the number of vertices incident to the edges updated.
It would be interesting to explore whether this limitation could be lifted for some types of edge updates.
In particular, would there exist a framework for updating many edges at once, where the update time could be sublinear in the number of vertices incident to the edges updated?

Rankwidth of graphs is related to branchwidth of matroids, so it would be interesting to explore whether our techniques could be extended into that setting.
We note that by the connection proved by Oum~\cite{Oum05}, all rankwidth algorithms directly apply to branchwidth of binary matroids when the binary matroid is represented by its fundamental graph, so \Cref{the:altrw} gives an improvement in this setting.
However, our techniques do not seem to directly apply to the more interesting setting of linear matroids represented by matrices.

In \Cref{the:dynfull} we support operations that take some partial vertex-labeling as an input.
We note that \Cref{the:dynsimple,the:dynfull} can be easily extended to the setting where instead of a graph, we maintain a vertex-labeled graph with a bounded number of labels that can be accessed by the $\LinCMSO_1$ formulas.
This extension can be done simply by gadgeteering: We can add some number of degree-1 neighbors to each vertex to encode the label of that vertex.
These gadgeteering techniques also appear applicable for extending our results to the setting of rankwidth/cliquewidth of more general binary relational structures, with an approximation factor depending on the exact definition of rankwidth in that setting.

Lastly, we remark that our dynamic algorithm works in space $\Oh[k]{n}$, and the algorithm of \Cref{the:altrw} in space $\Oh[k]{n}+\Oh{m}$.
In particular, the dynamic algorithm could be interesting from the viewpoint of models of computation with limited space, as its space complexity can be sublinear in the total size $n+m$ of the graph.

%% file: omitted-proofs1.tex
\section{Logarithmic height rank decompositions}
\label{sec:appendix1}

We show that rank decompositions can be turned into logarithmic height, which is based on~\cite{DBLP:conf/wg/CourcelleK07}.

\lemlogdepthdecomp*
\begin{proof}
We assume that the components $C \in \prt$ in the representation of $\lmap$ are represented as pointers so that the representation of $\lmap$ is of size $\Oh{|V(T)|}$.
Let us also assume without loss of generality that $(T,\lmap)$ is unrooted.

We will construct a binary tree $T^*$ of height $\Oh{\log |V(T)|}$ so that
\begin{enumerate}
\item every node $t \in V(T^*)$ is labeled with a subtree $\delta(t)$ of $T$ that contains at least one leaf of $T$,
\item\label{lem:logdepthdecomp:pr2} for each $t \in V(T^*)$ there are at most two edges of $T$ that have one endpoint in $V(\delta(t))$ and another endpoint in $V(T) \setminus V(\delta(t))$, and
\item if $\delta(t)$ contains at least two leaves of $T$, then $t$ has two children $c_1$ and $c_2$ so that $\leafs(T) \cap \leafs(\delta(t))$ is the disjoint union of $\leafs(T) \cap \leafs(\delta(c_1))$ and $\leafs(T) \cap \leafs(\delta(c_2))$.
\end{enumerate}

Before giving the algorithm to construct $T^*$, let us observe that $T^*$ can be transformed into a rooted rank decomposition $(T^*,\lmap^*)$ of $(G,\prt)$ of height $\Oh{\log |V(T)|}$ and width at most $2k$:
Note that for each leaf $l \in \leafs(T^*)$, the subtree $\delta(l)$ contains exactly one leaf of $T$, and these leaves of $T$ are distinct for distinct leaves of $T^*$.
Therefore, there is a natural bijection between $\leafs(T)$ and $\leafs(T^*)$, so we construct $\lmap^*$ simply by following this bijection.
This construction can be implemented in $\Oh{|V(T)|}$ time.
Then, \Cref{lem:logdepthdecomp:pr2} implies that for each $t \in V(T^*)$ (except the root), it holds that $\lparts(T^*)[t] = \lparts(T)[\vec{xy}] \cap \lparts(T)[\vec{zw}]$ for some $\vec{xy},\vec{zw} \in \oE(T)$.
Because of submodularity of $\cutrk_G$, this implies that $\cutrk_G(\lparts(T^*)[t]) \le \cutrk_G(\lparts(T)[\vec{xy}]) + \cutrk_G(\lparts(T)[\vec{zw}]) \le 2k$, which implies that $(T^*,\lmap^*)$ has width at most $2k$.

Then we describe an algorithm to construct such $T^*$ in time $\Oh{|V(T)| \log |V(T)|}$.
The algorithm constructs $T^*$ recursively top-down, in particular, each recursive step takes a subtree $\delta(t)$ of $T$ as an input and if it contains at least two leaves of $T$, constructs the subtrees $\delta(c_1)$ and $\delta(c_2)$ of $T$ for the two children $c_1$ and $c_2$ of $t$, and recurses to $c_1$ and $c_2$.
Alternatively, we can also construct subtrees $\delta(c_1),\delta(c_2),\delta(c_3),\delta(c_4)$ of $T$, where $c_1$ and $c_2$ will be the children of $t$, and $c_3$ and $c_4$ the children of $c_1$, and then recurse to $c_2, c_3$, and $c_4$.

Denote $X = \leafs(T) \cap \leafs(\delta(t))$.
If there is at most one edge of $T$ that has an endpoint in both $V(\delta(t))$ and $V(T) \setminus V(\delta(t))$, we pick an edge $xy \in E(\delta(t))$ so that $|X \cap \leafs(\delta(t))[\oxy]|\le \frac{2}{3}|X|$ and $|X \cap \leafs(\delta(t))[\oyx]|\le \frac{2}{3}|X|$, and let $\delta(c_1)$ and $\delta(c_2)$ be the two connected components of $\delta(t)-xy$.
Such $xy$ can be shown to exist by a simple walking argument on $\delta(t)$.

Then suppose there are two edges of $T$ that have an endpoint in both $V(\delta(t))$ and $V(T) \setminus V(\delta(t))$.
If both of them are incident to the same node $x$ of $\delta(t)$, we can set $\delta(t) \coloneqq \delta(t) - \{x\}$ and apply the case of one edge.
Therefore suppose one of them is incident to a node $x$ of $\delta(t)$ and other to a node $y \neq x$ of $\delta(t)$.
Note that both $x$ and $y$ have degree $2$ in $\delta(t)$.
Let $x = z_1,z_2,\ldots,z_{\ell} = y$ be the unique path between $x$ and $y$ in $\delta(t)$.
Now, each node $z_i$ on this path is incident to exactly one oriented edge $\vec{w_i z_i} \in \oE(\delta(t))$ so that $w_i$ is not on the path, and moreover, the sets $\leafs(\delta(t))[\vec{w_i z_i}]$ form a partition of $X$.
Let us pick the smallest $r$ so that $\sum_{i=1}^r |\leafs(\delta(t))[\vec{w_i z_i}]| \ge \frac{|X|}{3}$.
First, if $z_r \in \{x,y\}$, we let $\delta(c_1)$ be the connected component of $\delta(t) - \{z_r\}$ that contains all vertices on the path except $z_r$, and $\delta(c_2)$ the connected component that is disjoint with the path.
It can be observed that both of them satisfy \Cref{lem:logdepthdecomp:pr2}.
Moreover, we observe that $\delta(c_1)$ contains at most $\frac{2}{3}|X|$ leaves in $X$, and there is at most one edge of $T$ that has endpoints in both $V(\delta(c_2))$ and $V(T)-V(\delta(c_2))$, namely the edge $w_r z_r$.

It remains to consider the case $z_r \notin \{x,y\}$.
We first let $\delta(c_1)$ and $\delta(c_2)$ be the two connected components of $\delta(t) - z_{r} z_{r+1}$, with $x \in V(\delta(c_1))$ and $y \in V(\delta(c_2))$.
Then, we let $\delta(c_3)$ and $\delta(c_4)$ be the two connected components of $\delta(c_1) - \{z_r\}$, with $x \in V(\delta(c_3))$.
We observe that each of the constructed subtrees satisfy \Cref{lem:logdepthdecomp:pr2}.
Moreover, each of $\delta(c_2)$ and $\delta(c_3)$ contain at most $\frac{2}{3}|X|$ leaves in $X$, and there is at most one edge of $T$ that has endpoints in both $V(\delta(c_4))$ and $V(T)-V(\delta(c_4))$, namely the edge $w_r z_r$.

Clearly, each recursive call of this algorithm can be implemented in $\Oh{|\delta(t)|}$ time.
To obtain both the total time complexity $\Oh{|V(T)| \log |V(T)|}$ and the $\Oh{\log |V(T)|}$ height of $T^*$, it remains to bound the height of this recursion tree.
We recall that if there is at most one edge that has endpoints in $V(\delta(t))$ and $V(T)-V(\delta(t))$, then the size of $X$ shrinks by at least a factor $\frac{1}{3}$ when going to the children.
Also, in the other two cases, the only case when we recurse to a child where the size of $X$ does not shrink by a factor of $\frac{1}{3}$ is when there is only one edge of $T$ with endpoints in both the subtree of this child and outside of it.
We conclude that on any path of length $4$ that goes from a node in $T^*$ towards some leaf of $T^*$, the size of $X$ must shrink by a factor of at least $\frac{1}{3}$, implying that the height of $T^*$ is $\Oh{\log |V(T)|}$.
\end{proof}



%% file: cw_appendix.tex
\section{Cliquewidth}
\label{sec:cliquewidth}
In this appendix we give the definition of cliquewidth, show that annotated rank decompositions can be translated into cliquewidth expressions, show that automata working on cliquewidth expressions can be translated into rank decomposition automata, and use this to translate known dynamic programming algorithms on cliquewidth to rank decomposition automata.

\subsection{Definition and \texorpdfstring{$k$}{k}-expressions}
\label{subsec:cliquewidthdef}
Cliquewidth was introduced by Courcelle, Engelfriet, and Rozenberg~\cite{CourcelleER93} and defined in its modern form by Courcelle in~\cite{DBLP:journals/apal/Courcelle95}.
Next we define cliquewidth similarly to~\cite{CourcelleMR00}.
Let $k \in \N$.
A tuple $\pg = (G,V_1,\ldots,V_k)$ is a \emph{$k$-graph} if $G$ is a graph and $V_1,\ldots,V_k$ are disjoint subsets of $V(G)$ whose union equals $V(G)$ (they are not a partition because they are indexed by $[k]$ and allowed to be empty).
We define three types of operations for constructing $k$-graphs.
First, the disjoint union of two $k$-graphs $\pg^1 = (G^1,V^1_1,\ldots,V^1_k)$ and $\pg^2 = (G^2,V^2_1,\ldots,V^2_k)$ where $G^1$ and $G^2$ are disjoint is defined as
\[\pg_1 \oplus \pg_2 = (G^1 \cup G^2, V^1_1 \cup V^2_1, \ldots, V^1_k \cup V^2_k),\text{ where $G^1 \cup G^2$ is the union of $G^1$ and $G^2$.}\]
Then, $\eta(i,j)(\pg)$ for $i,j \in [k]$ with $i \neq j$ denotes the $k$-graph obtained from $\pg = (G,V_1,\ldots,V_k)$ by adding all possible edges between $V_i$ and $V_j$, i.e.,
\[\eta(i,j)(\pg) = (G',V_1,\ldots,V_k), \text{ where } V(G') = V(G) \text{ and } E(G') = E(G) \cup \{uv \mid u \in V_i \wedge v \in V_j\}.\]
Then, $\pi(i,j)(\pg)$ for $i,j \in [k]$ with $i \neq j$ denotes the $k$-graph obtained from $\pg$ by renaming $i$ into $j$, i.e.,
\[\pi(i,j)(\pg) = (G, V'_1, \ldots, V'_k), \text{ where $V'_i = \emptyset$, $V'_j = V_i \cup V_j$, and $V'_l = V_l$ for $l \in [k] \setminus \{i,j\}$.}\]
A graph has cliquewidth at most $k$ if it can be constructed from single-vertex $k$-graphs by using these operations.

More formally, we let $\opkg_k = \{\oplus\} \cup \bigcup_{i,j\in [k] \mid i \neq j} \{\eta(i,j), \pi(i,j)\}$ denote the set of operations on $k$-graphs.
We define that \emph{$k$-expression} is a triple $\Expr = (T,U,\mu)$, where $T$ is a rooted tree whose every node has at most two children and $\mu \colon V(T) \rightarrow U \cup \opkg_k$ is a labeling of its nodes so that
\begin{itemize}
\item the restriction $\funrestriction{\mu}{L(T)}$ of $\mu$ to the leaves of $T$ is a bijection $\funrestriction{\mu}{L(T)} \colon L(T) \to U$,
\item every node $t$ with one child is labeled with $\mu(t) \in \opkg_k \setminus (U \cup \{\oplus\})$ for some $i,j \in [k]$ with $i \neq j$, and
\item every node $t$ with two children is labeled with $\mu(t) = \oplus$.
\end{itemize}

We recursively define that a node $t \in V(T)$ \emph{encodes} a $k$-graph $\zeta(t) = (G,V_1,\ldots,V_k)$ if
\begin{itemize}
\item $t$ is a leaf, $G$ is the graph with a single vertex $\mu(t)$, and $V_1 = V(G)$,
\item $t$ has one child $c$ and $\zeta(t) = \mu(t)(\zeta(c))$, or
\item $t$ has two children $c_1,c_2$ and $\zeta(t) = \zeta(c_1) \oplus \zeta(c_2)$.
\end{itemize}

We say that $\Expr$ encodes a graph $G$ if its root encodes a $k$-graph $(G,V_1,\ldots,V_k)$ for some $V_1,\ldots,V_k$.
We note that if $\Expr$ encodes $G$, then $V(G) = U$.
Now the more formal definition of cliquewidth is that the cliquewidth of $G$ is the smallest $k$ so that there exists a $k$-expression that encodes $G$.

Then we prove that an annotated rank decompositions of width $k$ that encodes a graph $G$ can be turned in $\Oh[k]{n}$ time to a $(2^{k+1}-1)$-expression that encodes $G$.
Our proof follows the original construction of Oum and Seymour~\cite{OumS06}, but optimizes it to linear time in the case of annotated rank decompositions.
The definitions and auxiliary lemmas used for proving this will also be used in the next subsection for translating automata working on $k$-expressions to automata working on annotated rank decompositions.
We will use some definitions that are introduced in \Cref{sec:rdautom}.

Let $\Tc = (T,V(G),\reps,\repse,\dmap)$ be an annotated rank decomposition that encodes a graph $G$ and has width $\ell$.
We start with an observation that allows to optimize the $k$ of the expression by one.
\begin{observation}
Let $\oxy \in \oE(T)$.
There are at most $2^{\ell}-1$ vertices $v \in \reps(\oxy)$ so that $N_{\repse(xy)}(v)$ is non-empty.
\end{observation}
\begin{proof}
Let $M$ be the $|\reps(\oxy)| \times |\reps(\oyx)|$ matrix describing adjacencies of $\repse(xy)$.
We have that the rank of $M$ is at most $\ell$, so it has a row-basis of size $\ell$.
All other rows can be written as linear combinations of this row-basis with coefficients $0$ and $1$, so there are at most $2^{\ell}-1$ different non-zero rows.
\end{proof}

Then let $k = 2 \cdot 2^{\ell}-1$.
We define the $k$-graph associated with an oriented edge $\oxy \in \oE(T)$ to be the $k$-graph
\[\pg(\oxy) = (G(\oxy), V_1(\oxy), \ldots, V_{k}(\oxy)),\]
so that $G(\oxy) = G[\lparts(\Tc)[\oxy]]$ and where the sets $V_1(\oxy), \ldots, V_{k}(\oxy)$ are defined as follows.
Let $\xi_{\oxy} \colon \reps(\oxy) \rightarrow [2^{\ell}]$ be the injective function that maps each $v \in \reps(\oxy)$ to $\xi_{\oxy}(v) \in [2^{\ell}]$ so that
\begin{itemize}
\item if $N_{\repse(xy)}(v) = \emptyset$ then $\xi_{\oxy}(v) = 2^{\ell}$, and
\item otherwise $\xi_{\oxy}(v)$ is the number $i \in [2^{\ell}-1]$ so that there are exactly $i-1$ vertices $u \in \reps(\oxy)$ with $u<v$ and $N_{\repse(xy)}(u) \neq \emptyset$.
\end{itemize}
Let $v \in V(G(\oxy))$.
There exists unique $r_v \in \reps(\oxy)$ so that $N_{G}(r_v) \cap \lparts(\Tc)[\oyx] = N_G(v) \cap \lparts(\Tc)[\oyx]$.
We assign $v$ to the set $V_{\xi_{\oxy}(r_v)}$.
This concludes the definition of $\pg(\oxy)$.
We observe that $\xi_{\oxy}$ can be computed from $\reps(\oxy)$ and $\repse(xy)$ in time $\Oh[\ell]{1}$.

Then we show that these graphs can be inductively constructed on the rank decomposition by operations in $\opkg_k$.

\begin{lemma}
\label{lem:kgraph_ind_trans}
Let $\oxy \in \oE(T)$ be a non-leaf oriented edge and $\vec{c_1 x},\vec{c_2 x}$ be the children of $\oxy$.
The $k$-graph $\pg(\oxy)$ can be produced by a sequence of $\Oh{k^2}$ operations in $\opkg_k$ from the $k$-graphs $\pg(\vec{c_1 x})$ and $\pg(\vec{c_2 x})$.
Moreover, this sequence of operations depends only on the transition signature $\tau(\Tc, \oxy)$ and can be computed given it in $\Oh[\ell]{1}$ time.
\end{lemma}
\begin{proof}
We give the construction of $\pg(\oxy)$ from $\pg(\ocax)$ and $\pg(\ocbx)$.
Because $|\reps(\vec{c_i x})| \le 2^{\ell}$, the sets $V_{2^{\ell}+1}(\vec{c_i x}), \ldots, V_{2 \cdot 2^{\ell}-1}(\vec{c_i x})$ are empty for both $c_i \in \{c_1, c_2\}$.
We start by applying the operations $\pi(j, j+2^{\ell})$ for all $j \in [2^\ell-1]$ to the $k$-graph $\pg(\vec{c_2 x})$.
Let $\pg'(\vec{c_2 x})$ be the resulting $k$-graph.
Then, let $\pg''(\oxy) = \pg(\vec{c_1 x}) \oplus \pg'(\vec{c_2 x})$.
For each $u \in \reps(\vec{c_1 x})$ and $v \in \reps(\vec{c_2 x})$, we know whether $uv \in E(G)$ by inspecting $\dmap(c_1 x c_2)$ and $\repse(x c_2)$, and we know that $uv \notin E(G)$ if $\xi_{\ocax}(u) = 2^{\ell}$ or $\xi_{\ocbx}(v) = 2^{\ell}$.
If $uv \in E(G)$, we apply the operation $\eta(\xi_{\ocax}(u),\xi_{\ocbx}(v)+2^{\ell})$ to $\pg''(\oxy)$.

It remains to rename the labels of the representatives.
Assume $\ell \geq 1$ since otherwise there is nothing to do.
We construct a~function $f : [k] \to [2^\ell]$ so that for each $u \in \reps(\vec{c_1 x})$ we have $f(\xi_{\ocax}(u)) = \xi_{\oxy}(\dmap(c_1 x y)(u)))$; and similarly, for each $v \in \reps(\vec{c_2 x})$ with $\xi_{\ocbx}(v) \neq 2^{\ell}$ we have $f(\xi_{\ocbx}(v) + 2^{\ell}) = \xi_{\oxy}(\dmap(c_2 x y)(v)))$.
Since $k > 2^\ell$, it is straightforward to produce a~sequence of $\Oh{k}$ operations $\pi(\cdot, \cdot)$ that, in total, remaps each label $i \in [k]$ to the label $f(i)$.
\ms{I rewrote this paragraph since the previous argument failed a bit. Does it look good now?}
%

We observe that this sequence of operations depends only on $\tau(\Tc, \oxy)$ and can be computed from it in $\Oh[\ell]{1}$ time.
It remains to prove that it correctly produces the $k$-graph $\pg(\oxy)$.
Let $\pg^* = (G^*, V^*_1, \ldots, V^*_k)$ denote the $k$-graph resulting from the operations.
We prove that $\pg^*$ and $\pg(\oxy)$ are equal.

Let us first check that $G^* = G[\lparts(\Tc)[\oxy]]$.
We have $V(G^*) = V(G(\oxy))$ by construction.
Let $\pg''(\oxy) = (G''(\oxy), V''_1, \ldots, V''_k)$.
We have $V''_1 \cup \ldots \cup V''_{2^\ell-1} \subseteq \lparts(\Tc)[\ocax]$ and $V''_{2^\ell+1} \cup \ldots \cup V''_{2 \cdot 2^\ell-1} \subseteq \lparts(\Tc)[\ocbx]$.
Therefore, our operations did not add edges between the pairs of vertices in $\lparts(\Tc)[\ocax]$, nor between the pairs of vertices in $\lparts(\Tc)[\ocbx]$, so we have that $G^*[\lparts(\Tc)[\ocax]] = G(\oxy)[\lparts(\Tc)[\ocax]]$ and $G^*[\lparts(\Tc)[\ocbx]] = G(\oxy)[\lparts(\Tc)[\ocbx]]$.
It remains to check edges between $\lparts(\Tc)[\ocax]$ and $\lparts(\Tc)[\ocbx]$.
By our construction we have that edges between $u \in \reps(\vec{c_1 x})$ and $v \in \reps(\vec{c_2 x})$ are as claimed.
Suppose that $v \in \lparts(\Tc)[\ocax]$ and $r_v \in \reps(\vec{c_1 x})$ is the node so that $N_G(r_v) \cap \lparts(\Tc)[\oxca] = N_G(v) \cap \lparts(\Tc)[\oxca]$.
We have that $v$ and $r_v$ are in the same set $V''_l$, and therefore $N_{G^*}(v) \cap \lparts(\Tc)[\ocbx] = N_{G^*}(r_v) \cap \lparts(\Tc)[\ocbx]$.
Therefore, because the neighborhood of $r_v$ to $\lparts(\Tc)[\ocbx]$ is correct and $\lparts(\Tc)[\ocbx] \subseteq \lparts(\Tc)[\oxca]$, we deduce that the neighborhood of $v$ to $\lparts(\Tc)[\ocbx]$ is also correct.

Let us then check that $V^*_j = V_j(\oxy)$ for all $j \in [k]$.
Consider $v \in \reps(\ocax)$.
By definitions of annotated rank decompositions we have that $N_G(v) \cap \lparts(\Tc)[\oyx] = N_G(\dmap(c_1 x y)(v)) \cap \lparts(\Tc)[\oyx]$, which readily implies that $v \in V^*_j$ if and only if $v \in V_j(\oxy)$.
Then consider $v \in \lparts(\Tc)[\ocax]$, and again let $r_v \in \reps(\vec{c_1 x})$ be the node so that $N_G(r_v) \cap \lparts(\Tc)[\oxca] = N_G(v) \cap \lparts(\Tc)[\oxca]$.
We have that $v$ and $r_v$ are in the same set $V''_l$, so they end up in the same set $V^*_j$.
Because $\lparts(\Tc)[\oyx] \subseteq \lparts(\Tc)[\oxca]$, we have that $N_G(r_v) \cap \lparts(\Tc)[\oyx] = N_G(v) \cap \lparts(\Tc)[\oyx]$, so the correctness of $v$ follows from the correctness of $r_v$.
The proof for $v \in \lparts(\Tc)[\ocbx]$ is similar.
\end{proof}

Then, with similar arguments we can show that a $k$-graph representing $G$ can be constructed from $\pg(\oxy)$ and $\pg(\oyx)$ for some edge $xy \in E(T)$.
We omit the proof as it is similar to the proof of \Cref{lem:kgraph_ind_trans}.

\begin{lemma}
\label{lem:kgraph_ind_final}
Let $xy \in E(T)$.
The $k$-graph $(G, V(G), \emptyset, \ldots, \emptyset)$ can be produced by a sequence of $\Oh{k^2}$ operations in $\opkg_k$ from the $k$-graphs $\pg(\oxy)$ and $\pg(\oyx)$.
Moreover, this sequence of operations depends only on the edge signature $\sigma(\Tc,\oxy)$.
\end{lemma}

Now we are ready to give the algorithm to translate annotated rank decompositions into $k$-expressions.

\begin{lemma}
\label{lem:rdtocwexp}
There is an algorithm that given an annotated rank decomposition $\Tc$ of width $\ell$ that encodes a graph $G$, in time $\Oh[\ell]{|\Tc|}$ outputs a $(2^{\ell+1}-1)$-expression that encodes $G$.
\end{lemma}
\begin{proof}
Let $\Tc = (T,V(G),\reps,\repse,\dmap)$ and $k = (2^{\ell+1}-1)$, and let us use the definitions introduced in this subsection.
We choose an arbitrary edge $ab \in E(T)$.
By using \Cref{lem:kgraph_ind_trans}, we compute for each non-leaf oriented edge $\oxy \in \pred_T(\oab) \cup \pred_T(\oba)$ a rooted tree with $\Oh{k^2}$ nodes, so that the internal nodes are labeled with operations in $\opkg_k$ and the two leaves are labeled with the two child edges $\ocax$ and $\ocbx$ of $\oxy$, so that it corresponds to a sequence of operations in $\opkg_k$ that turn $\pg(\ocax)$ and $\pg(\ocbx)$ into $\pg(\oxy)$.
We also use \Cref{lem:kgraph_ind_final} to compute the rooted tree with $\Oh{k^2}$ nodes, so that the internal nodes are labeled with operations in $\opkg_k$ and the two leaves are labeled with $a$ and $b$, so that it corresponds to a sequence of operations in $\opkg_k$ that turn $\pg(\oab)$ and $\pg(\oba)$ into $(G, V(G), \emptyset, \ldots, \emptyset)$.
For each leaf edge $\olp \in \leafe(T)$ we compute the $k$-expression with at most one operation that turns the $k$-graph $(G[\reps(\olp), V_1, \emptyset, \ldots, \emptyset])$ into $\pg(\olp)$.
Now, we observe that by gluing these $\Oh{|\Tc|}$ trees we computed together, we obtain a $k$-expression that encodes $G$.
This takes in total $\Oh[\ell]{|\Tc|}$ time.
\end{proof}

\subsection{Automata on \texorpdfstring{$k$}{k}-expressions}
We then define automata working on $k$-expressions.
Our definitions do not strictly follow any literature as they are geared to our notation and the goal of proving \Cref{lem:automata_transl}, but can be seen as equivalent to definitions given by Courcelle and Engelfriet~\cite{CEbook}.

A $k$-expression automaton is a 6-tuple $\autom = (Q,\Gamma,\iota,\chi,\psi,\phi)$ that consists of
\begin{itemize}
\item a state set $Q$,
\item a vertex label set $\Gamma$,
\item an initial mapping $\iota$ that maps a single-vertex graph labeled with $\gamma \in \Gamma$ to a state $\iota(\gamma) \in Q$,
\item a transition mapping $\psi$ that maps every pair of form $(\mu, q)$, where $\mu \in \opkg_k \setminus \{\oplus\}$ and $q \in Q$ to a state $\psi(\mu,q) \in Q$,
\item a transition mapping $\chi$ that maps every pair of states $(q_1,q_2) \in Q \times Q$ to a state $\chi(q_1,q_2) \in Q$, and
\item a final mapping $\phi$ that maps each state $q \in Q$ to a state $\phi(q) \in Q$.
\end{itemize}

The evaluation time of the automaton is the maximum running time to compute the functions $\iota$, $\psi$, $\chi$, and $\phi$ given their arguments.

Let $\Expr = (T,V(G),\mu)$ be a $k$-expression that encodes a graph $G$ and $\alpha \colon V(G) \rightarrow \Gamma$ a vertex-labeling of $G$ with $\Gamma$.
The \emph{run} of $\autom$ on the pair $(\Expr,\alpha)$ is the unique mapping $\rho \colon V(T) \rightarrow Q$ so that
\begin{itemize}
\item for each leaf $l \in \leafs(T)$ it holds that $\rho(l) = \iota(\alpha(\mu(t)))$,
\item for each node $t \in V(T)$ that has one child $c$ it holds that $\rho(t) = \psi(\mu(t),\rho(c))$, and
\item for each node $t \in V(T)$ that has two children $c_1,c_2$ with $c_1 < c_2$ it holds that $\rho(t) = \chi(\rho(c_1),\rho(c_2))$.
\end{itemize}

The valuation of $\autom$ on $(\Expr,\alpha)$ is $\phi(\rho(r))$, where $r$ is the root of $T$.
We say that $\autom$ is \emph{expression-oblivious} if its valuation on $(\Expr,\alpha)$ depends only on the graph $G$ encoded by $\Expr$ and the labeling $\alpha$.
In that case, we call this also the valuation of $\autom$ on $(G,\alpha)$.
The purpose of the final mapping $\phi$ in the definition is to be able to make $k$-expression automata expression-oblivious, for example, if the purpose of $\autom$ is to decide whether $G$ satisfies some graph property, then the image of $\phi$ could be just $\{\bot,\top\}$, while $Q$ could be much larger in order to represent intermediate computations.

We are now ready to prove that $k$-expression automata can be translated into rank decomposition automata.
This is not surprising since the construction of $(2^{\ell+1}-1)$-expression from a rank decomposition of width $\ell$ in \Cref{lem:rdtocwexp} works in a local manner.
The proof uses definitions of rank decomposition automata from \Cref{sec:rdautom}.

\begin{lemma}
\label{lem:automata_transl}
Let $\ell \in \N$ and $k = 2^{\ell+1}-1$.
Given an expression-oblivious $k$-expression automaton $\autom_{\mathsf{ex}} = (Q,\Gamma,\iota,\chi,\psi,\phi)$ with evaluation time $\beta$, it is possible to construct a rank decomposition automaton $\autom_{\mathsf{rd}} = (Q,\Gamma,\iota',\delta,\varepsilon)$ of width $\ell$ and evaluation time $\Oh[\ell]{\beta}$, so that if $\Tc = (T,V(G),\reps,\repse,\dmap)$ is an annotated rank decomposition that encodes a graph $G$ and has width at most $\ell$, $\alpha \colon V(G) \rightarrow \Gamma$ is a vertex-labeling of $G$ with $\Gamma$, and $a,b \in V(T)$ is a pair of adjacent nodes in $T$, then the valuation of $\autom_{\mathsf{rd}}$ on $(\Tc, a, b, \alpha)$ is the same as the valuation of $\autom_{\mathsf{ex}}$ on $(G,\alpha)$.
\end{lemma}
\begin{proof}
We use the definitions of $\pg(\oxy)$ and $\xi_{\oxy}$ introduced in \Cref{subsec:cliquewidthdef}.
By \Cref{lem:kgraph_ind_trans} we can associate with each $\oxy \in \oE(T)$ a $k$-expression $\Expr(\oxy) = (T^{\mathsf{ex}}(\oxy), \lparts(\Tc)[\oxy], \mu(\oxy))$ so that the root of $T^{\mathsf{ex}}(\oxy)$ encodes $\pg(\oxy)$, and if $\oxy$ is non-leaf then $\Expr(\oxy)$ is constructed by combining $\Expr(\ocax)$ and $\Expr(\ocbx)$ by $\Oh[\ell]{1}$ operations in $\opkg_k$ that depend only on $\tau(\Tc, \oxy)$.
In particular, if $\Tpref$ is the prefix of $T^{\mathsf{ex}}(\oxy)$ so that the connected components of $T^{\mathsf{ex}}(\oxy) - \Tpref$ are $T^{\mathsf{ex}}(\ocax)$ and $T^{\mathsf{ex}}(\ocbx)$, then the pair $\Expr(\tau(\Tc, \oxy)) = (T^{\mathsf{ex}}(\oxy)[\Tpref \cup \App(\Tpref)], \funrestriction{\mu(\oxy)}{\Tpref})$ depends only on $\tau(\Tc, \oxy)$.
The tree $T^{\mathsf{ex}}(\oxy)[\Tpref \cup \App(\Tpref)]$ has exactly two leaves that correspond to the roots of $\Expr(\ocax)$ and $\Expr(\ocbx)$, and we let names of these leaves be $l_1$ and $l_2$ so that $l_i$ corresponds to $c_i$ (note that $\tau(\Tc, \oxy)$ includes the subtree $T[\{x,y,c_1,c_2\}]$ so this is allowed).
If $\oxy$ is a leaf edge then $\Tc^{\mathsf{ex}}(\oxy)$ is the $k$-expression consisting of at most two nodes that encodes $\pg(\oxy)$.

Then we define the automaton $\autom_{\mathsf{rd}} = (Q,\Gamma,\iota',\delta,\varepsilon)$.
Like indicated by the notation, the sets $Q$ and $\Gamma$ are the same as for the automaton $\autom_{\mathsf{ex}} = (Q,\Gamma,\iota,\chi,\psi,\phi)$.
The function $\iota'$ is defined as follows:
Let $\sigma$ be an edge signature $\sigma = (\reps^a_\sigma,\reps^b_\sigma,\repse_\sigma)$ and $\gamma$ a function $\gamma \colon \reps^a_\sigma \rightarrow \Gamma$.
If $\reps^a_\sigma$ is a set consisting of a single vertex $v$, we set $\iota'(\sigma,\gamma) = \iota(\gamma(v))$.
Otherwise, we set $\iota'(\sigma,\gamma)$ to be an arbitrary state in $Q$.
Note that $\autom_{\mathsf{rd}}$ is required to work only on annotated rank decompositions that encode graphs, for which the latter case never happens.

The mapping $\delta(\tau,q_1,q_2)$, where $\tau$ is a transition signature and $q_1,q_2 \in Q$ is defined as follows.
We take the pair $\Expr(\tau) = (T^*, \mu^*)$ defined earlier in the course of the proof.
Let $L(T^*) = \{l_1, l_2\}$.
Then we take the run of $\autom_{\mathsf{ex}}$ on $(T^*, \mu^*)$, defined as a function $\rho \colon V(T^*) \rightarrow Q$ so that for the two leaves $l_1,l_2$ we have $\rho(l_1) = q_1$ and $\rho(l_2) = q_2$, and for other nodes the run is defined as per the usual definition of a run of $\autom_{\mathsf{ex}}$.
Then, we set $\delta(\tau,q_1,q_2) = \rho(r)$, where $r$ is the root of $T^*$.
Before defining $\varepsilon$ we can observe that the following claim follows from our construction.
\begin{observation}
Let $\oxy \in \oE(T)$ and $\rho_{\mathsf{rd}} \colon \pred_T(\oxy) \rightarrow Q$ be the run of $\autom_{\mathsf{rd}}$ on $(\Tc,\oxy,\alpha)$.
Let also $\rho_{\mathsf{ex}} \colon V(T^{\mathsf{ex}}(\oxy)) \rightarrow Q$ be the run of $\autom_{\mathsf{ex}}$ on $\Expr(\oxy)$.
Then $\rho_{\mathsf{rd}}(\oxy) = \rho_{\mathsf{ex}}(r(\oxy))$, where $r(\oxy)$ is the root of $T^{\mathsf{ex}}(\oxy)$.
\end{observation}

Next we define $\varepsilon$.
By \Cref{lem:kgraph_ind_trans}, the $k$-graph $\pg = (G, V(G), \emptyset, \ldots, \emptyset)$ can be constructed from the $k$-graphs $\pg(\oxy)$ and $\pg(\oyx)$ by $\Oh[\ell]{1}$ applications of operations in $\opkg_k$ that depend only on the edge signature $\sigma(\Tc,\oxy)$.
Therefore, we can similarly define a $k$-expression $\Expr(x, y) = (T^{\mathsf{ex}}(\oxy), \lparts(\Tc)[\oxy], \mu(\oxy))$ that encodes $\pg$ and is constructed by combining $\Expr(\oxy)$ and $\Expr(\oyx)$ by $\Oh[\ell]{1}$ operations in $\opkg_k$ that depend only on $\sigma(\Tc,\oxy)$.
We can also define a pair $\Expr(\sigma(\Tc,\oxy))$ to describe how exactly these $k$-expressions should be combined.

Now, $\varepsilon(\sigma,q_1,q_2)$ can be constructed from $\Expr(\sigma)$ similarly as $\delta$ was constructed from $\Expr(\tau)$ and finally applying the mapping $\phi$, so that the valuation of $\autom_{\mathsf{rd}}$ on $(\Tc,x,y,\alpha)$ is the same as the valuation of $\autom_{\mathsf{ex}}$ on $(\Expr(x, y),\alpha)$.
Now because $\autom_{\mathsf{ex}}$ is expression-oblivious, the valuation of $\autom_{\mathsf{ex}}$ on $(\Expr(x, y),\alpha)$ is the valuation of $\autom_{\mathsf{ex}}$ on $(G,\alpha)$, which concludes the correctness of the construction.
In the constructions of the functions $\delta$ and $\varepsilon$ we apply the functions $\chi$,$\psi$, and $\phi$ $\Oh[\ell]{1}$ times, so the evaluation time of $\autom_{\mathsf{rd}}$ is $\Oh[\ell]{\beta}$.
\end{proof}

We note that the properties of $\autom_{\mathsf{rd}}$ asserted in the statement of \Cref{lem:automata_transl} imply that it is decomposition-oblivious.

\subsection{\texorpdfstring{$\CMSO_1$}{CMSO1}}
\label{subsec:appendixcmso}
We use definitions of $\CMSO_1$ logic given in \Cref{subsec:maincmso}.
The following theorem was given in~\cite{CourcelleMR00} (see also \cite[Section 6]{CEbook}).
\begin{theorem}[\cite{CourcelleMR00}]
\label{the:cmsokexpr}
There is an algorithm that given a $\CMSO_1$ sentence $\varphi$ with $p$ free set variables and $k \in \N$, in time $\Oh[\varphi,k]{1}$ constructs an decomposition-oblivious $k$-expression automaton $\autom = (Q,\Gamma,\iota,\chi,\psi,\phi)$ so that $\Gamma = 2^{[p]}$, the valuation of $\autom$ on $(G,\alpha)$ is $\top \in Q$ if and only if $(G,\alpha) \models \varphi$, the number of states is $|Q| \le \Oh[\varphi,k]{1}$, and the evaluation time is $\Oh[\varphi,k]{1}$.
\end{theorem}

By combining \Cref{lem:automata_transl,the:cmsokexpr}, we immediately obtain the following.

\cmsordaut*

Then we prove \Cref{lem:lincmsordaut} by using \Cref{lem:cmsordaut}.

\lincmsordaut*
\begin{proof}
Denote $\varphi = (\phi,f)$, where $\phi$ is a $\CMSO_1$ sentence with $p+q$ free variables, where $p$ is the number of free variables of $\varphi$.
Let $f(x_1,\ldots,x_q) = c_0 + c_1 x_1 + \ldots + c_q x_q$.
We first use \Cref{lem:cmsordaut} to turn $\phi$ into a rank decomposition automaton $\autom' = (Q',\Gamma',\iota',\delta',\varepsilon')$ of width $\ell$.

Let $\Tc = (T,V(G),\reps,\repse,\dmap)$ be an annotated rank decomposition that encodes a graph $G$, $\Gamma = 2^{[p]}$, and $\alpha \colon V(G) \rightarrow \Gamma$ a vertex-labeling of $G$.
Then, for a set $X \subseteq V(G)$ and a vertex labeling $\alpha' \colon X \rightarrow 2^{[p+1,p+q]}$, we define $\mathsf{val}(X,\alpha') = f(|X_1|, \ldots, |X_q|)$ where $X_i = \{v \in X \mid i+p \in \alpha'(v)\}$.
We also denote by $\funrestriction{\alpha}{X} \cup \alpha'$ the function $\funrestriction{\alpha}{X} \cup \alpha' \colon X \rightarrow 2^{[p+q]}$ with $(\funrestriction{\alpha}{X} \cup \alpha')(v) = \funrestriction{\alpha}{X}(v) \cup \alpha'(v)$ for all $v \in X$.
Then for every pair $(\oxy,s)$ with $\oxy \in \oE(T)$ and $s \in Q'$, we define $\mathsf{maxval}(\oxy,s)$ to be the maximum value of $\mathsf{val}(\lparts(\Tc)[\oxy], \alpha')$ over all functions $\alpha' \colon \lparts(\Tc)[\oxy] \rightarrow 2^{[p+1,p+q]}$ so that the valuation of $\autom'$ on $(\Tc, \oxy, \alpha \cup \alpha')$ is $s$, or $-\infty$ if no such $\alpha'$ exists.

Now, the state set of $\autom$ is the set of all functions $g \colon Q' \rightarrow \Z \cup \{-\infty\}$, and we can define the transitions of $\autom$ so that the valuation of $\autom$ on $(\Tc, \oxy, \alpha)$ is the function $g_{\oxy}$ that maps each $s \in Q'$ to $\mathsf{maxval}(\oxy, s)$.
In particular, for non-leaf edges $\oxy$ with child edges $\ocax$ and $\ocbx$ this can be done by setting for each $s \in Q'$ the value $g_{\oxy}(s)$ to be the maximum of $g_{\ocax}(s_1)+g_{\ocbx}(s_2)-c_0$ so that $\delta(\tau(\Tc, \oxy), s_1, s_2) = s$.
The construction of the initial mapping $\iota$ is straightforward.
We observe that we can construct the final mapping similarly, so that valuation of $\autom$ on $(G,\alpha)$ is equal to the maximum value of $\mathsf{val}(V(G),\alpha')$ over all functions $\alpha' \colon V(G) \rightarrow 2^{[p+q]}$ so that the valuation of $\autom'$ on $(G,\alpha \cup \alpha')$ is $\top$, and if no such $\alpha'$ exists, the valuation is $\bot$.
This gives evaluation time $\Oh{|Q'|^2 \cdot \beta}$, where $\beta$ is the evaluation time of $\autom'$, resulting in $\Oh[\ell,\varphi]{1}$ evaluation time.
\end{proof}

Let us then also prove \Cref{lem:labelindsub} here.

\lemlabelindsub*
\begin{proof}
We first turn $\Tc$ into an annotated rank decomposition $\Tc' = (T',V(G),\reps',\repse',\dmap')$ that encodes the graph $G$ (instead of the partitioned graph $(G,\prt)$).
This can be done in $\Oh[\ell]{|\Tc|}$ time by adding a subtree of size $\Oh[\ell]{1}$ below each leaf of $\Tc$.

Let $|V(H)| = p$ and let us index the vertices of $H$ by $u_1,\ldots,u_p$.
We write a $\CMSO_1$ sentence $\varphi$ of length $|\varphi| \le \Oh[H]{1}$ with $2p$ free variables so that $(G,X_1,\ldots,X_p,Y_1,\ldots,Y_p) \models \varphi$ if and only if $|X_i| = 1$ and $X_i \subseteq Y_i$ for all $i \in [p]$, and $G[X_1 \cup \ldots \cup X_p]$ is isomorphic to $H$ with an isomorphism that maps the single vertex $v_i \in X_i$ to $u_i$.
We use \Cref{lem:cmsordaut} to construct a rank decomposition automaton $\autom = (Q,\Gamma,\iota,\delta,\varepsilon)$ so that for all adjacent nodes $x,y \in V(T')$, the valuation of $\autom$ on $(\Tc', x, y, \alpha)$ is $\top$ if and only if $\alpha \colon V(G) \rightarrow 2^{[2p]}$ is a vertex-labeling corresponding to $X_1,\ldots,X_p,Y_1,\ldots,Y_p$ so that $(G,X_1,\ldots,X_p,Y_1,\ldots,Y_p) \models \varphi$.

We construct labeling $\alpha \colon V(G) \rightarrow 2^{[2p]}$ so that $\alpha(v) \cap [p] = \emptyset$ and $p+i \in \alpha(v)$ if and only if $u_i \in \gamma(v)$.
Then, if $f$ is a function $f \colon [p] \rightarrow V(G) \cup \{\bot\}$, we denote by $\alpha+f$ the function $(\alpha+f) \colon V(G) \rightarrow 2^{[2p]}$ so that $(\alpha+f)(v) = \alpha(v) \cup \{i \mid f(i) = v\}$.
Now, for each oriented edge $\oxy$ of $T'$ denote by $g_{\oxy}$ the function that maps each $q \in Q$ to a function $g_{\oxy}(q) \colon [p] \rightarrow \lparts(\Tc')[\oxy] \cup \{\bot\}$ so that the valuation of $\autom$ on $(\Tc',\oxy,\alpha+g_{\oxy}(q))$ is $q$, or to $\bot$ if no such function exists.
Now we can construct an auxiliary automaton $\autom'$ that computes $g_{\oxy}$ for each oriented edge $\oxy$ of $T'$ directed towards an arbitrarily chosen root, and finally from that construct a function $f \colon [p] \rightarrow V(G) \cup \{\bot\}$ so that the valuation of $\autom$ on $(\Tc',x,y,(\alpha+f))$ is $\top$, or find that no such $f$ exists.
By construction, such $f$ corresponds to a witness of $H$ as a labeled induced subgraph of $(G,\gamma)$.
\end{proof}

%% file: totally-pure.tex
\section{Totally pure rank decompositions}
\label{sec:totally-pure-decomp}

We now formally introduce the concept of \emph{totally pure rank decompositions} introduced by Jeong, Kim and Oum~\cite{DBLP:journals/siamdm/JeongKO21} and signaled in \cref{ssec:dealternation-prelims}.
Then we will use this definition to prove the existence of optimum-width decompositions of subspace arrangements with bounded-size mixed skeletons.

We reuse the definitions from \cref{ssec:dealternation-prelims} and in the introduction below mostly follow the notation of \cite{DBLP:journals/siamdm/JeongKO21}.

Let $\Tc = (T, \lambda)$ be an~unrooted rank decomposition and $\Tc^b = (T^b, \lambda^b)$ be rooted.
Let also $x \in V(T^b)$.
We say that $\Tc$ is \emph{$x$-disjoint} if either $x$ is the root of $T^b$, or $T$ contains an~edge $uv$ such that $\lparts(\Tc)[\ouv] = \Vc_x$ (equivalently, $\ouv$ is $x$-full and $\ovu$ is $x$-empty).

Let $B_x$ be the boundary space of $x$ in $\Tc^b$, defined as $B_x = \sumof{\lparts(\Tc^b)[\vec{xp}]} \cap \sumof{\lparts(\Tc^b)[\vec{px}]}$, where $p$ is the parent of $x$ in $T^b$; observe that equivalently, $B_x = \sumof{\Vc_x} \cap \sumof{\Vc \setminus \Vc_x}$.
With this in mind, we say that an~edge $uv$ of $T$ is \emph{$x$-degenerate} if the following linear space equality holds:
\[
  \sumof{\lparts_x(\Tc)[\ouv]} \cap B_x = \sumof{\lparts_x(\Tc)[\ovu]} \cap B_x.
\]
Such an~edge is \emph{proper} $x$-degenerate if at least one of the following conditions holds:
\begin{itemize}
  \item either $\ouv$ or $\ovu$ is $x$-empty; or
  \item there exists $y \in V(T^b)$ with $y < x$ such that: (a) there exists a~$y$-degenerate edge in $T$ (possibly different than $uv$) that is not proper, and (b) neither $\ouv$ nor $\ovu$ is $y$-empty.
\end{itemize}
Even though the definition above is recursive, it is defined correctly and uniquely -- the notion of proper $x$-degeneracy only depends on the proper $y$-degeneracy of edges for $y < x$.

An~$x$-degenerate edge that is not proper is called \emph{improper} $x$-degenerate.
If $\Tc$ contains an~improper $x$-degenerate edge, we say that $\Tc$ is \emph{$x$-degenerate}.

Next, an~edge $\ouv$ of $(T, \lambda)$ is \emph{$x$-guarding} (or: $uv$ $x$-guards its end $u$) if the following strict inclusion holds:
\[
  \sumof{\lparts_x(\Tc)[\ouv]} \cap B_x \subsetneq \sumof{\lparts_x(\Tc)[\ovu]} \cap B_x.
\]
In this case, $\ouv$ is \emph{improper} $x$-guarding if all of the following conditions hold: $\deg(u) = 3$; $\ouv$ is $x$-mixed; and if $u_1, u_2$ are the two neighbors of $u$ other than $v$, then neither $\vec{u_1u}$ nor $\vec{u_2u}$ is $x$-empty.
Otherwise, $\ouv$ is \emph{proper} $x$-guarding.

Finally, a~two-edge path $uvw$ of $(T, \lambda)$ is an~\emph{$x$-blocking path} if the following two equalities hold:
\[
\begin{split}
  \sumof{\lparts_x(\Tc)[\vec{uv}]} \cap B_x &= \sumof{\lparts_x(\Tc)[\vec{vw}]} \cap B_x \eqqcolon A_1, \\
  \sumof{\lparts_x(\Tc)[\vec{wv}]} \cap B_x &= \sumof{\lparts_x(\Tc)[\vec{vu}]} \cap B_x \eqqcolon A_2;
\end{split}
\]
and moreover, neither $A_1 \subseteq A_2$ nor $A_2 \subseteq A_1$.
(Note that this implies that $\dim(A_1), \dim(A_2) > 0$, so in particular, neither $\vec{uv}$ nor $\vec{wv}$ is $x$-empty.)
In this case, $uvw$ is an~\emph{improper} $x$-blocking path if $\deg(v) = 3$ and $\vec{v'v}$ is $x$-mixed, for the unique neighbor $v'$ of $v$ other than $u$ and $w$.
Otherwise, $uvw$ is \emph{proper} $x$-blocking.

With this bag of definitions at hand, we say that $\Tc$ is \emph{$x$-pure} if one of the following holds:
\begin{itemize}
  \item $\Tc$ is $x$-degenerate and $x$-disjoint; or
  \item $\Tc$ is not $x$-degenerate, and every $x$-guarding edge $\ouv$ and every $x$-guarding path $uvw$ is proper.
\end{itemize}
Finally, $\Tc$ is \emph{totally pure} with respect to $\Tc^b$ if it is $x$-pure for all $x \in V(T^b)$.

Now, the structure theorem proven by Jeong, Kim and Oum reads as follows:

\koreanrankstructuretheorem*

\paragraph*{Totally pure decompositions imply small mixed skeletons.} Recall now the definition of mixed skeletons from \cref{ssec:dealternation-mixed-skeletons}.
Using \cref{lem:korean-rank-structure-theorem}, we will now give the omitted proof of \cref{lem:small-mixed-skeleton}, which we restate below for convenience.

\smallmixedskeletonlemma*

\begin{proof}
  Let $\Tc = (T, \lambda)$ be a~rooted optimum-width rank decomposition of $\Vc$ that is totally pure with respect to $\Tc^b$; such a~decomposition exists by \cref{lem:korean-rank-structure-theorem}.
  We claim that, for every $x \in V(T^b)$, the height of the $x$-mixed skeleton of $\Tc$ is at most $2\ell + 2$.
  Since mixed skeletons are rooted binary trees, the statement of the lemma will follow immediately.
  
  Fix $x \in V(T^b)$ and let $T^\mix$ be the $x$-mixed skeleton of $\Tc$.
  Assume for contradiction that there exists a~vertical path $P = v_0v_1\ldots{}v_{p+1}$ in $T^\mix$ for some $p \geq 2\ell + 1$, where for each $i \in [p+1]$, the node $v_{i-1}$ is an~ancestor of $v_i$ in $T$.
  For each $i \in [p + 1]$, define $v_i^L$ as the parent of $v_i$ in $T$, and for each $i \in [0, p]$, define $v_i^R$ as the unique child of $v_i$ in $T$ on the simple path between $v_i$ and $v_{i+1}$.
  For each $i \in [p]$, let $v'_i$ be the remaining child of $v_i$ in $T$.
  Note that for each $i \in [p]$, the node $v_i$ is an~$x$-branch point (\cref{lem:mixed-skeleton-binary}), so the edges $\vec{v_i^Rv_i}$ and $\vec{v'_iv_i}$ are $x$-mixed; moreover, the edge $\vec{v_i^Lv_i}$ is $x$-mixed by \cref{lem:mixed-skeleton-mixed-path}.
  
  Recall that $B_x = \sumof{\lparts_x(\Tc^b)[\vec{xp}]} \cap \sumof{\lparts_x(\Tc^b)[\vec{px}]}$, where $p$ is the parent of $x$ in $T^b$.
  Since $\Tc^b$ has width $\ell$, by definition we necessarily have that $\dim(B_x) \leq \ell$.
  Consider the following vector spaces for each $i \in [p + 1]$:
  \[
  \begin{split}
  A_i &= \sumof{\lparts_x(\Tc)[\vec{v_i v_i^L}]} \cap B_x, \\
  B_i &= \sumof{\lparts_x(\Tc)[\vec{v_i^L v_i}]} \cap B_x.
  \end{split}
  \]
  Note that $\vec{v_i^L v_i}$ is a~predecessor of $\vec{v_{i+1}^L v_{i+1}}$ for each $i \in [p]$.
  Therefore we have the following chains of inclusions of vector spaces:
  \[
  \begin{split}
    A_1 \supseteq A_2 \supseteq \ldots \supseteq A_{p+1}, \\
    B_1 \subseteq B_2 \subseteq \ldots \subseteq B_{p+1}.
  \end{split}
  \]
  Each $A_i$ and each $B_i$ is a~vector space of dimension at most $\ell$ since each is a~subspace of  $B_x$.
  Since $p \geq 2\ell + 1$, we find that there exists an~index $t \in [p]$ such that $A_t = A_{t+1}$ and $B_t = B_{t+1}$.
  Because $\sumof{\lparts_x(\Tc)[\vec{v_{t+1} v_{t+1}^L}]} \subseteq \sumof{\lparts_x(\Tc)[\vec{v_t^R v_t}]} \subseteq \sumof{\lparts_x(\Tc)[\vec{v_t v_t^L}]}$ and $\sumof{\lparts_x(\Tc)[\vec{v_t^L v_t}]} \subseteq \sumof{\lparts_x(\Tc)[\vec{v_t v_t^R}]} \subseteq \sumof{\lparts_x(\Tc)[\vec{v_{t+1}^L v_{t+1}}]}$, we have
  \begin{align*}
    &\sumof{\lparts_x(\Tc)[\vec{v_t v_t^L}]} \cap B_x = \sumof{\lparts_x(\Tc)[\vec{v_t^R v_t}]} \cap B_x = A_t
    \quad\text{and}\\
    &\sumof{\lparts_x(\Tc)[\vec{v_t^L v_t}]} \cap B_x = \sumof{\lparts_x(\Tc)[\vec{v_t v_t^R}]} \cap B_x = B_t.
  \end{align*}
  We now consider several cases with regard to the containment relation between $A_t$ and $B_t$.
  \begin{itemize}
    \item If $A_t = B_t$, then the edge $e \coloneqq v_t^L v_t$ is by definition $x$-degenerate.
    
      Suppose first $e$ is improper.
      Then $\Tc$ is $x$-degenerate and so by the total purity of $\Tc$, $\Tc$ is $x$-pure and thus $x$-disjoint (i.e., either $x$ is the root of $T^b$ and then $\Vc_x = \Vc$, or there exists an~edge $pq \in E(T)$ such that $\lparts(\Tc)[\vec{pq}] = \Vc_x$).
      However, by \cref{lem:mixed-skeleton-mixed-path}, the edges $\vec{v_t^L v_t}$ and $\vec{v_t v_t^L}$ are both $x$-mixed.
      This is a~contradiction as in an~$x$-disjoint decomposition, there cannot exist an edge $uv \in E(T)$ such that both $\vec{uv}$ and $\vec{vu}$ are $x$-mixed.
      
      Now assume that $e$ is proper.
      Again by \cref{lem:mixed-skeleton-mixed-path}, the edges $\vec{v_t^L v_t}$ and $\vec{v_t v_t^L}$ are both $x$-mixed.
      By the fact that $e$ is proper, it must be the case that for some $y \in V(T^b)$ with $y < x$, the decomposition $\Tc$ is $y$-degenerate and neither $\vec{v_t^L v_t}$ nor $\vec{v_t v_t^L}$ is $y$-empty.
      By the total purity of $\Tc$, we have that $\Tc$ is $y$-disjoint.
      As previously, it cannot be that both $\vec{v_t^L v_t}$ and $\vec{v_t v_t^L}$ are $y$-mixed.
      Therefore, one of the edges $\vec{v_t^L v_t}, \vec{v_t v_t^L}$ is $y$-full.
      So by \cref{obs:mixedness-relations}, that edge is $x$-full, too -- a~contradiction.
    
    \item If $A_t \subsetneq B_t$, then the edge $e \coloneqq \vec{v_t v_t^L}$ is $x$-guarding by definition.
      But recall that the three edges $\vec{v_t v_t^L}$, $\vec{v_t^R v_t}$ and $\vec{v'_t v_t}$ are $x$-mixed.
      Hence $e$ is improper $x$-guarding by definition, which contradicts the assumption that $\Tc$ is totally pure with respect to $\Tc^b$.
      
    \item If $B_t \subsetneq A_t$, the analogous argument follows, using the $x$-guarding edge $\vec{v_t v_t^R}$ instead.
    
    \item If $A_t \not\subseteq B_t$ and $B_t \not\subseteq A_t$, then the path $v_t^L v_t v_t^R$ is $x$-blocking by definition.
      But since $\vec{v'_t v_t}$ is $x$-mixed, we get that $v_t^L v_t v_t^R$ is improperly $x$-blocking and thus $\Tc$ is not $x$-pure -- a~contradiction.
  \end{itemize}
  Since we reached a~contradiction in each possible case, the proof of the lemma is complete.
\end{proof}